\documentclass[11pt]{article}
\usepackage{natbib}
\makeatother
\usepackage{lipsum}
\usepackage{setspace}
\usepackage{apalike}
\usepackage{makecell}
\usepackage{amsmath}
\usepackage{algorithm,algcompatible,amsmath}
\usepackage{bbm}

\usepackage{amsmath}
\usepackage{booktabs}
\usepackage{enumerate}
\usepackage{scalerel,stackengine}
\DeclareMathOperator*{\argmax}{arg\,max}
\DeclareMathOperator*{\argmin}{arg\,min}
\stackMath
\allowdisplaybreaks
\usepackage{comment}
\usepackage{amssymb,stackengine,graphicx}

\usepackage{graphicx}

\usepackage{amsfonts}
\usepackage{graphicx}
\usepackage[colorlinks=true, allcolors=blue]{hyperref}
\usepackage{amsthm}
\usepackage[toc,page]{appendix}
\usepackage[letterpaper,top=2cm,bottom=2cm,left=2.5cm,right=2.5cm]{geometry}
\usepackage{cite}
\usepackage{longtable}
\usepackage{booktabs}
\usepackage{float}
\usepackage{amsmath}
\usepackage{scalerel,stackengine}
\usepackage{xcolor}
\newcommand{\tcr}{\textcolor{black}}
\stackMath
\newcommand\reallywidehat[1]{%
\savestack{\tmpbox}{\stretchto{%
  \scaleto{%
    \scalerel*[\widthof{\ensuremath{#1}}]{\kern-.6pt\bigwedge\kern-.6pt}%
    {\rule[-\textheight/2]{1ex}{\textheight}}
  }{\textheight}%
}{0.5ex}}%
\stackon[1pt]{#1}{\tmpbox}%
}
\usepackage{graphicx}
\usepackage{fnpct}
\usepackage{amsfonts}
\usepackage{graphicx}
\usepackage{amsthm}
\usepackage{amsfonts}
\usepackage{graphicx}
\usepackage{nccmath}
\usepackage[flushleft]{threeparttable}
\newtheorem{theorem}{Theorem}[section]
\newtheorem{corollary}{Corollary}[section]
\newtheorem{lemma}{Lemma}[section]
\newtheorem{assumption}{Assumption}[section]
\newtheorem{remark}{Remark}

\usepackage{amsfonts}
\usepackage{longtable}
\usepackage{newtxtext}
\usepackage{graphicx}
\usepackage{booktabs}
\usepackage{xr}
\usepackage[colorlinks=true, allcolors=blue]{hyperref}
\counterwithout{figure}{section} 
\counterwithout{table}{section}   

\title{High-dimensional censored MIDAS logistic regression for corporate survival forecasting}

\begin{document}
\date{}
\author{\vspace*{7pt}  Wei Miao$^\dag$, Jad Beyhum$^\ddag$, Jonas Striaukas$^*$ and Ingrid Van Keilegom$^\dag$ \\ 
$^\dag$ORSTAT, KU Leuven\\
$^\ddag$Department of Economics, KU Leuven\\ $^*$Department of Finance, Copenhagen Business School\\ 
}
\onehalfspacing

\maketitle

\begin{abstract}
This paper addresses the challenge of forecasting corporate distress, a problem marked by three
key statistical hurdles: (i) right censoring, (ii) high-dimensional
predictors, and (iii) mixed-frequency data. To overcome these complexities, we introduce a novel high-dimensional censored MIDAS (Mixed Data Sampling) logistic regression. Our approach handles censoring through inverse probability weighting and achieves accurate estimation with numerous mixed-frequency predictors by employing a sparse-group penalty. We establish finite-sample bounds for the estimation error, accounting for censoring, MIDAS approximation error, and heavy tails. 
For statistical inference, we develop a de-sparsified version of the proposed penalized estimator and establish its asymptotic theory, which enables valid statistical inference in high-dimensional settings with censoring. We show that censoring induces a nonstandard variance structure for the de-sparsified estimator, a feature that, to the best of our knowledge, has not been studied in the existing literature.
The superior performance of the method is demonstrated through Monte Carlo simulations. Finally, we present an extensive application of our methodology to predict the financial distress of Chinese-listed firms \tcr{and to identify covariates that are statistically significant for predicting distress.} Our novel procedure is implemented in the R package \texttt{Survivalml}.

\end{abstract}
{\it Keywords:}  Corporate survival analysis; high-dimensional censored data; mixed-frequency data; logistic regression; sparse-group LASSO
\vfill

\section{Introduction}

Regulators, lenders, and investors are increasingly focused on identifying vulnerable firms and developing accurate models to predict firm failures well in advance, as the ability to correctly predict such failures could result in a more resilient financial stability policy and better financial outcomes for market participants. As a result, an extensive body of literature is dedicated to understanding the determinants of firm failures. Traditional statistical models, such as discriminant analysis \citep{almon1965distributed}, logistic regression \citep{ohlson1980financial}, and hazard models \citep{shumway2001forecasting}, along with other time-sensitive approaches \citep{duffie2007multi}, have historically been the main focus of study. However, with the advent of more extensive datasets in recent years, the focus has increasingly shifted toward machine learning methods, which are better equipped to handle high-dimensional data. Such models have shown superior accuracy in predicting firm failures \citep{barboza2017machine} due to their efficient handling of rich data sources. Over time, the task of forecasting corporate survival has gained significant attention due to its critical economic implications and its close connection to other challenges, such as predicting household loan defaults.

In this paper, we focus on the task of predicting the probability that a firm will fail within the first $t$ years after its initial listing, conditional on its survival for the first $s$ years, where $s < t$. This problem presents three significant statistical challenges. First, data are often right-censored.
To illustrate how censoring arises in practice, we consider our empirical dataset on Chinese listed manufacturing firms (see Section~\ref{sec real} for details). Figure~\ref{fig1} depicts firms with different censoring statuses over the observation period from January $1^{\text{st}}$, 1985 to December $31^{\text{st}}$, 2020. For each firm, the survival time $T$ is defined as the number of years from its initial public offering (IPO) date to the first occurrence of financial distress. If a firm is never classified as distressed during the observation period, its survival time $T$ is right-censored. In this case, $T$ exceeds the censoring time $C$, and only $C$, defined as the number of years from the IPO date to the end of the observation period, is observed. In Figure \ref{fig1}, Firms $1$ and $2$ are uncensored, so their survival times are fully observed, whereas Firms $3$ and $4$ are censored and only their censoring times $C$ are observed.
\begin{figure}[hbtp]
\includegraphics[scale=0.55]{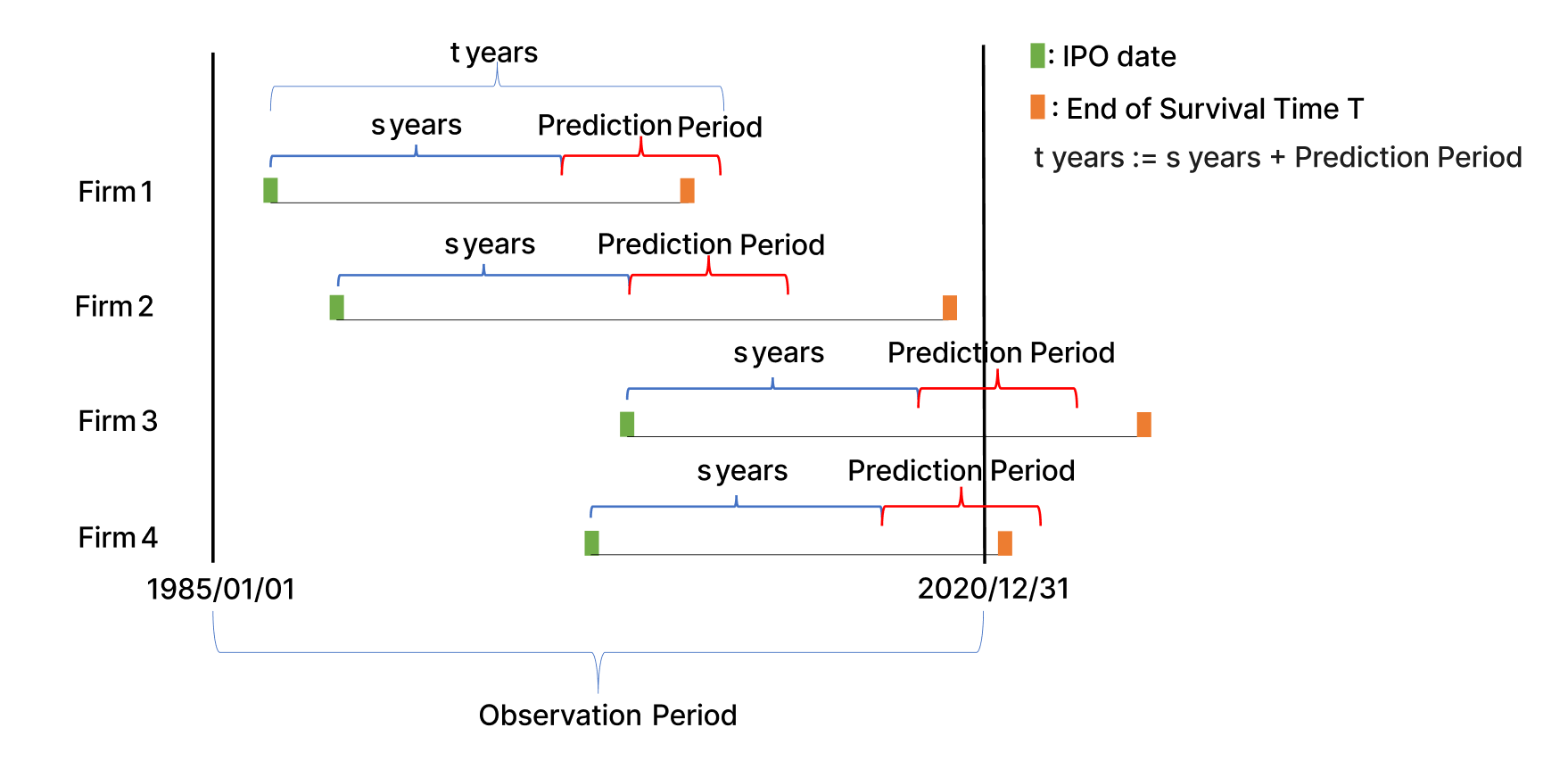}
\caption{Firms with different censoring statuses in the empirical dataset.}
\label{fig1}
\end{figure}

Second, the high dimensionality of the predictors adds complexity. Modern datasets provide a wealth of variables for each listed firm, increasing the analytical burden. Third, the mixed-frequency nature of the data compounds the difficulty. For each potential predictor, we observe numerous lags, exacerbating the challenge of managing the proliferation of parameters.

As highlighted in our review of the literature below, in our view, the existing methods for this prediction task do not adequately address all three challenges. To bridge this gap, we propose a novel high-dimensional censored MIDAS logistic regression method that addresses these complexities. Our approach is based on a high-dimensional logistic regression framework to estimate survival probabilities. To address right-censoring, we make use of a tool from the survival analysis literature called outcome-weighted inverse probability of censoring weighting, as described in \citet{blanche2023logistic}. The mixed-frequency nature of the data is managed using mixed data sampling (MIDAS), an approach developed and popularized by \citet{ghysels2007midas}. This method approximates the coefficients of the lags of each variable using a finite-dimensional series basis, known as the dictionary. Finally, to handle the high dimensionality of the predictors, we apply a sparse-group LASSO penalty. This penalty not only manages the dimensionality of the regressors but also accounts for the group structure of the predictors, which corresponds to the lags of the original variables, as discussed in \citet{babii2022machine}.

We derive finite-sample bounds on the estimation error of our penalized estimator. Notably, these bounds allow for heavy-tailed variables and account for both the approximation error and right-censoring, which are novel contributions to the literature on high-dimensional logistic regression models. Additionally, we propose a de-sparsified procedure for the penalized estimator, which leverages the nodewise regression \citep{van_de_geer_debias}. Under suitable regularity conditions, the de-sparsified estimator is asymptotically normal, enabling valid inference on individual components of the parameter vector while properly accounting for right-censoring, MIDAS approximation error, and heavy-tailed covariates. Notably, the presence of censoring introduces a nonstandard variance structure for the de-sparsified estimator, which has not yet been examined in the existing literature \citep{van_de_geer_debias,caner2023generalized,babii2023machine}.

The finite-sample performance of our method is evaluated through simulations, demonstrating its robustness against natural alternatives.  Furthermore, we showcase the practical advantages of our approach through an application to forecasting the financial distress of Chinese listed firms. In this context, our method significantly outperforms the standard logistic regression benchmark and other competing methods over several horizons, underscoring its empirical effectiveness. Several practical augmented prediction methods, including oversampling and incorporating macro data into the model, are utilized. To further demonstrate the effectiveness of the proposed method that includes censoring information, a comparison is conducted with a method that excludes censored firms. The de-sparsified estimator is also employed to identify statistically significant financial predictors of firm distress. These identified predictors provide practitioners with new empirical insights and help direct attention to specific covariates that are particularly informative for distress prediction.
Finally,  our novel approach is implemented in the R package \texttt{Survivalml} to make it readily accessible for practitioners.\footnote{The package is publicly available at \url{https://github.com/Wei-M-Wei/Survivalml}.}

\paragraph{Literature review.} Let us first review how the existing methods address the three challenges we described, which are inherent in corporate survival forecasting. This review will stress the advantages of our methodology over popular alternatives. Given the extensive literature on this topic, we do not aim to provide an exhaustive review. Instead, we focus on surveying key approaches to the problems at hand. We also cite papers on the related problem of forecasting loan default.

To address the right-censoring of data, many studies restrict their analysis to firms that were first listed more than $t$ years before the end of the follow-up period \citep[see, for instance, ][]{audrino2019predicting,petropoulos2020predicting}. Under the classical assumption of independent censoring, this approach avoids selection bias. However, it discards data on firms listed less than $t$ years ago, leading to a loss of efficiency. Another common strategy is to directly model the hazard rate of firm failure, using methods such as Cox models or single-index models \citep[e.g.,][]{ding2012class,lee2014business,kim2016survival,zhou2022recurrence,li2023corporate}. Although effective in some contexts, this approach has limitations. Typically, the primary interest lies in estimating the probability of failure, not the hazard rate, and, therefore, modeling the survival probability as we do is more natural to solve the problem at hand. Furthermore, none of the aforementioned hazard-based approaches explicitly account for the challenges posed by high-dimensional mixed-frequency data. Instead, they applied their methods to pre-selected low-dimensional sets of predictors and lags, bypassing the complexity of high-dimensional data structures.

Let us now address the challenge of parameter proliferation, which arises from both the high-dimensionality and the mixed-frequency nature of the data. Several studies have used LASSO as a selection tool to predict corporate bankruptcy; see, for example, \citet{petropoulos2020predicting, barbaglia2023forecasting}. However, these studies do not address censoring, lack theoretical results, and do not utilize the MIDAS framework. The application of MIDAS in a logistic regression framework for corporate bankruptcy prediction was explored by \citet{audrino2019predicting}. While their work incorporates the MIDAS approach, it is limited to a low-dimensional set of predictors and does not consider censoring. More closely related to our study, \citet{jiang2021predicting} examined a penalized logistic regression with the norm $\ell_1$. However, their approach does not account for censoring, lacks theoretical underpinnings, and employs what is referred to as unrestricted MIDAS. Unlike our approach, which uses a restricted MIDAS procedure, unrestricted MIDAS includes all lags as predictors, effectively bypassing the dimension reduction benefits of the MIDAS framework.
\footnote{It is worth noting that the term ``unrestricted MIDAS" is somewhat misleading, as this approach directly incorporates all lags as independent variables. Consequently, it does not take advantage of the dimension-reduction capabilities inherent in the MIDAS methodology.}

Finally, we compare our theoretical results to the existing literature for both the penalized estimator and its de-sparsified version. The theory of penalized estimators of the high-dimensional logistic regression model has been extensively studied under various situations. As already mentioned, no existing study allows for censoring or approximation error. \citet{meier2008group, 10.1214/009053607000000929, buhlmann2011statistics, van_de_geer_estimation_2016} analyzed the logistic regression model using fixed design or isotropic conditions of sub-Gaussian covariates, which are often unsuitable for financial data. \citet{van_de_geer_debias} developed a de-sparsified inferential procedure for logistic regression with LASSO under uniformly bounded covariates. More recently, \citet{caner2023generalized} relaxed these assumptions, allowing for random covariate designs with non-normal covariates in the context of penalized Generalized Linear Models (GLM), and established asymptotic theory for the associated de-sparsified estimator. However, compared to the present paper, this work imposes additional assumptions on the shape of the second-order partial derivatives of the loss function. Similarly, \citet{han2023high} developed the theory for GLM with LASSO by establishing local restricted strong convexity of the loss function, which is related to the quadratic margin condition in the present paper; see Online Appendix \ref{one p m}. However, they did not propose an inferential procedure.
In the context of mixed-frequency data, \citet{babii2022machine, babii2023machine} developed the theoretical foundation for high-dimensional time series and panel data linear regression models while accounting for the MIDAS approximation error. However, the theory for penalized logistic regression models incorporating such approximation errors remains unexplored, and none of the aforementioned studies have addressed the challenges posed by censored data. 


%

\paragraph{Outline.} The paper is organized as follows. In Section \ref{sec est}, we first present the model, followed by a discussion on employing the MIDAS weighting technique and incorporating group structure information among variables. Section~\ref{sec the} analyzes the estimation properties of the proposed estimator, introduces the associated inferential procedure, and establishes the asymptotic properties of the de-sparsified estimator. Section~\ref{sec sim} presents simulation studies for both prediction and inference.
 In Section \ref{sec real}, we build a dataset on Chinese firm distress to assess the predictive performance of the proposed method, compare it with other approaches, and test each financial covariate for significance.

\paragraph{Notation.}
For $\ell \in \mathbb{N}$, we define $[\ell]=\{1,2, \ldots, \ell\}$. For a vector $\boldsymbol{b} \in \mathbb{R}^p$, its $\ell_q$ norm is denoted as $|\boldsymbol{b}|_q=\left(\sum_{j \in[p]}\left|b_j\right|^q\right)^{1 / q}$ if $q \in[1, \infty)$, its $\ell_{\infty}$-norm is denoted  $|\boldsymbol{b}|_{\infty}=\max _{j \in[p]}\left|b_j\right|$, and its empirical $\ell_2$-norm is denoted  $\|\boldsymbol{b}\|_N = \sqrt{\left|\boldsymbol{b}\right|_2^2/N}$. Let $diag(b_1, \ldots, b_p)$ be the $p\times p$ diagonal matrix, with diagonal entries equal to $b_1, \ldots, b_p$. For a matrix $\boldsymbol{A}$, let $\boldsymbol{A}^{\top}$ be its transpose and $\lambda_{\min }(\boldsymbol{A})$ be its smallest eigenvalue. The cardinality of a set $S$ is $|S|$, and we use $S^c$ as its complement. For a vector $\Delta \in \mathbb{R}^p$ and a subset $J \subset[p]$, let $\Delta_J$ be a vector in $\mathbb{R}^p$ with the same coordinates as $\Delta$ on $J$ and zero coordinates on $J^c$, where $J^c$ is the complement of the subset $J$.
Let $\Delta_{(J)} := (\Delta_j)_{j \in J} \in \mathbb{R}^{|J|}$ denote the subvector of $\Delta$ corresponding to indices in $J$. For a matrix $\boldsymbol{M} \in \mathbb{R}^{p \times p}$, we define $\boldsymbol{M}_J \in \mathbb{R}^{|J| \times p}$, which is the submatrix consisting of the rows of $\boldsymbol{M}$ corresponding to indices in $J$. For $a, b \in \mathbb{R}$, we put $a \vee b=\max \{a, b\}$ and $a \wedge b=\min \{a, b\}$. We write $a_N \lesssim b_N$ if there exists a (sufficiently large) absolute constant $v$ such that $a_N \leq v b_N$ for all $N \geq 1$ and $a_N \sim b_N$ if $a_N \lesssim b_N$ and $b_N \lesssim a_N$. The indicator function is denoted by $\mathbbm{1}\{\cdot\}$.

\section{High-dimensional censored MIDAS logistic regression} \label{sec est}

\subsection{Logistic regression model}\label{initial model}

In corporate survival analysis, we focus on the survival time $T$ of a firm. The random variable $T$ represents the duration from the firm's Initial Public Offering (IPO) date to the occurrence of financial distress. Specifically, the IPO date refers to the first day the firm’s stock is publicly traded. Since companies are not listed immediately after their creation, the survival time $T$ in our context differs slightly from the typical survival time considered in traditional survival analysis \citep{li2023corporate}.

Our main objective is to predict the probability that a firm will survive up to $t$ years, given that it has already been publicly listed for $s$ years. In practice, the survival time $T$ is right-censored by the censoring time $C$, which denotes the duration between the IPO date and the censoring event, occurring at the end of the follow-up period.\footnote{In the empirical application, all randomness in $C$ arises from the randomness of IPO dates, while the follow-up period's end date is fixed.} Hence, we do not directly observe $T$, but rather the censored value $\widetilde{T} = T \wedge C$, with the indicator $\delta = \mathbbm{1}\{T \leq C\}$.

The financial distress status, indicated by $\mathbbm{1}\{T \leq t\}$, is influenced by covariates $\boldsymbol{Z} \in \mathbb{R}^{K_z}$. We assume, for the moment, that $\boldsymbol{Z}$ has finite variance to ensure well-defined expectations and model the distress indicator $\mathbbm{1}\{T \leq t\}$ using a logistic regression model:

\begin{equation}\label{distributin} P(T \leq t \mid \boldsymbol{Z}, T \geq s) = \frac{\exp \left(\boldsymbol{Z}^{\top}\boldsymbol{\theta}_0(t,s) \right)}{1 + \exp \left(\boldsymbol{Z}^{\top}\boldsymbol{\theta}_0(t,s)\right)}, \end{equation}
where $\boldsymbol{Z}$ is the covariate vector, and $\boldsymbol{\theta}_0(t,s) \in \mathbb{R}^{K_z}$ is the vector of true parameters specific to $t$ and $s$. For convenience, we use $\boldsymbol{\theta}_0$ as shorthand for $\boldsymbol{\theta}_0(t,s)$. Note that in practice, the model includes an intercept term, which enters the variable $\boldsymbol{Z}$.

Model \eqref{distributin} is typically estimated via maximum likelihood estimation. This method is based on the characterization of $\boldsymbol{\theta}_0$ as the solution to the population conditional maximum likelihood problem:

\begin{equation}\label{initial estimator} \boldsymbol{\theta}_0 = \argmax_{\boldsymbol{\theta} \in \mathbb{R}^{K_z}} \mathbbm{E}\left[\left.\mathbbm{1}\{T \leq t\} \boldsymbol{Z}^{\top} \boldsymbol{\theta} - \log\left(1 + \exp(\boldsymbol{Z}^{\top} \boldsymbol{\theta})\right)\right| T \geq s \right]. \end{equation}
In \eqref{initial estimator}, we have rewritten the classical logistic model's likelihood in a simplified form; see Lemma \ref{lm.cens1} for a proof.\footnote{The characterization \eqref{initial estimator} is valid under a full-rank condition stated in Lemma \ref{lm.cens1}.}

However, equation \eqref{initial estimator} cannot be directly used for estimation due to the fact that the survival time $T$ is not always observed. To address this issue, we apply the outcome-weighted inverse probability of censoring weighting (OIPCW) method, as outlined by \citet{blanche2023logistic}.\footnote{An alternative approach for addressing censoring is Inverse Probability Weighting (IPW) \citep{bbbe4e97-7833-3faa-820c-4c61f82fe965, zheng2006application, beyhum2024instrumental, beyhum2024dynamic}. Further details on both OIPCW and IPW can be found in \citet{blanche2023logistic}.} This method relies on two standard assumptions about the censoring mechanism, which we describe below. The first assumption is the assumption of independent censoring:

\begin{assumption}\label{as1} $C$ is independent of $T$ and $\boldsymbol{Z}$. \end{assumption}

This is a standard assumption in survival analysis. We argue that Assumption \ref{as1} is reasonable in corporate survival analysis because the censoring time for a firm is solely determined by the observation period. The second assumption concerns sufficient follow-up:

\begin{assumption}\label{as2} $P\left(\widetilde{T} \geq t\right) \geq C_r$, where $C_r > 0$. \end{assumption}

This assumption implies that some firms have been observed for more than $t$ years without experiencing financial distress, which is necessary for model estimation.

Under Assumptions \ref{as1} and \ref{as2}, we obtain an alternative characterization of $\boldsymbol{\theta}_0$, relying only on observed or estimable quantities:

\begin{equation}\label{equivalent censor} \boldsymbol{\theta}_0 = \argmax_{\boldsymbol{\theta} \in \mathbb{R}^{K_z}} \mathbbm{E}\left[\left. \frac{\delta(t) \mathbbm{1}\{\widetilde{T} \leq t\}}{H(t \wedge \widetilde{T})} \boldsymbol{Z}^{\top} \boldsymbol{\theta} - \log \left(1 + \exp(\boldsymbol{Z}^{\top} \boldsymbol{\theta})\right)\right|\widetilde{T} \geq s\right], \end{equation}
where $H(u) = P(C \geq u|C\ge s)$ is the survival probability of $C$ at time $u$ conditional on $C\ge s$ and $\delta(t) = \mathbbm{1}\{C \geq t \wedge T\}$ is the observation indicator. Equation \eqref{equivalent censor} is proven in Online Appendix \ref{estimator}, see Lemma \ref{lm.cens2}. Essentially, the expectation in \eqref{equivalent censor} weighs the uncensored observations that fail between $s$ and $t$ by the weights $1/H(t \wedge \widetilde{T})$ to ensure they are representative of firms with survival times between $s$ and $t$.\footnote{The expectation in \eqref{equivalent censor} is well-defined since $H(t \wedge \widetilde{T})\ge H(t) =  P(C \geq t \mid C\ge s) =  P\left(C \geq t\right)/P(C\ge s) \ge P(\widetilde{T} \geq t)/P(C\ge s)  > 0$ almost surely by Assumption \ref{as2}.}

The function $H$ is not directly observed but can be estimated under Assumption \ref{as1} using the classical Kaplan-Meier estimator \citep{kaplan1958nonparametric}, as described below.

\subsection{Estimation with mixed-frequency data}\label{sec exp sg}

Consider an i.i.d. sample of firms \(\left(\widetilde{T}_i, \delta_i, \boldsymbol{Z}_i\right), i \in [N]\), such that for all \(i \in [N]\), \(\widetilde{T}_i \geq s\), i.e., all firms in the sample are observed for at least \(s\) years.\footnote{As shown in the previous section, considering firms with at least \(s\) years of observation does not introduce selection bias under the independent censoring assumption.}

For prediction, suppose we have $K$ covariates and for each covariate $k\in [K]$, we have \(s\) years of lagged covariates \(\left\{z_{i, s - \frac{j-1}{m}, k}: j \in [d]\right\}\), where \(d = s \times m\) represents the total number of lags of each covariate and \(m\) is the annual sampling frequency of the covariate observations. The covariates can be observed at varying frequencies, and although not all lags may enter the regression, we omit such cases for simplicity. The \(k\)-th covariate and its lags are represented as 
$$
\widetilde{\boldsymbol{Z}}_{i,k} = \left(z_{i, s, k}, z_{i, s - \frac{1}{m}, k}, \ldots, z_{i, s - \frac{d-1}{m}, k}\right)^{\top} \in \mathbb{R}^{d}, k \in [K].
$$
Then the complete vector of lagged covariates and intercept is denoted as $$
\boldsymbol{Z}_i = \left(1,\widetilde{\boldsymbol{Z}}_{i,1}^\top, \widetilde{\boldsymbol{Z}}_{i,2}^\top, \dots, \widetilde{\boldsymbol{Z}}_{i,K}^\top\right)^\top\in\mathbb{R}^{K_z}
$$
where \(K_z = K \times d + 1\).
The function \(H\) in \eqref{equivalent censor} can be estimated using the Kaplan-Meier estimator:
$$
\widehat{H}(u) = \prod_{j \leq u} \left(1 - \frac{dN(j)}{\widetilde{T}(j)}\right),
$$
where \(N(j) = \sum_{i=1}^N \mathbbm{1}\{\widetilde{T}_i \leq j, \delta_i = 0\}\) is the counting process for censoring events up to time $j$, and \(dN(j) = N(j) - \lim_{j' \to j, j' < j} N(j')\) denotes the jump of the process \(N\) at time \(j\). Additionally, \(\widetilde{T}(j) = \sum_{i=1}^N \mathbbm{1}\{\widetilde{T}_i \geq j\}\) is the number of units who are under observation just before time \(j\).

We consider datasets that are high-dimensional. For instance, in our empirical application, as summarized in Table~\ref{summary}, if we consider firms that have survived \(s = 6\) years, with \(K = 95\) covariates measured \(m = 4\) times per year, the total number of parameters to estimate is \(6 \times 4 \times 95 + 1 = 2281\), including the intercept. When the sample size is not much larger than the number of parameters, the curse of dimensionality arises, complicating computations and reducing estimation precision.

To address this, dimension-reduction techniques are necessary. A common approach is to directly apply the LASSO \citep{tibshirani1996regression} to the original predictors. For an i.i.d. sample \(\{(\widetilde{T}_i, \delta_i, \boldsymbol{Z}_i), i \in [N]\}\), the $\ell_1$-norm penalized estimator solves the empirical version of \eqref{equivalent censor} with the LASSO penalty. 

Here, we follow a different approach \citep{babii2022machine} to reduce the dimension based on Mixed-Data Sampling \citep[MIDAS]{ghysels2006predicting}, which is designed to address parameter proliferation in mixed-frequency data. MIDAS approximates the coefficients of high-frequency lag polynomials using a finite dictionary of functions. Specifically, let us write 
\begin{equation}\label{specify weight}
\boldsymbol{Z}_i^\top \boldsymbol{\theta}_0 =\theta_{0,1}+\frac{1}{d}\sum_{k=1}^K \sum_{j=1}^d \omega_k\left(\frac{j-1}{d}\right) z_{i, s - \frac{j-1}{m}, k}, \quad i \in [N],
\end{equation}
where \(\omega_k:[0,1]\mapsto \mathbb{R}\), \(k \in [K]\), are weight functions for the lag polynomials such that $ \omega_k\left(\frac{j-1}{d}\right)/d=\theta_{0,1+d (k-1) +j}$. Let \(\{w_l: l = 1, \ldots, L\}\) be the dictionary of functions and let $L$ be its size. For each \(k \in [K]\), we assume there exist coefficients \(\boldsymbol{\beta}_{0,k}^* = (\beta_{0,k,1}^*, \beta_{0,k,2}^*,\ldots, \beta_{0,k,L}^*)^\top \in \mathbb{R}^L\) such that:
$$
\omega_k(u) \approx \sum_{l=1}^L \beta_{0,k,l}^* w_l(u), \quad u \in [0, 1].
$$
This reduces the number of parameters from \(K \times d+1 \) to \(K \times L+1 \). The simplest dictionary consists of algebraic power polynomials (e.g., Almon polynomials \citep{almon1965distributed}), but other orthogonal bases of \(L_2[0,1]\) can be used to improve performance with correlated covariates.\footnote{\(L_2[0,1]\) denotes the space of square-integrable functions \(f: [0,1] \to \mathbb{R}\).} With the MIDAS approximation, we construct the following MIDAS-weighted covariate vector:
\[
\boldsymbol{X}_i = \left(1,\widetilde{\boldsymbol{Z}}_{i,1}^\top W, \widetilde{\boldsymbol{Z}}_{i,2}^\top W, \ldots, \widetilde{\boldsymbol{Z}}_{i,K}^\top W\right)^\top \in \mathbb{R}^{p}, \quad i \in [N],
\]
where $p = K \times L + 1$ is the dimension, \(W = \left(w_l\bigl(\frac{j-1}{d}\bigr)/d\right)_{j \in [d], l \in [L]} \in \mathbb{R}^{d \times L}\) is the weighting matrix, and we refer to $W^{\top}\widetilde{\boldsymbol{Z}}_{i,k} \in \mathbb{R}^{L}$ as the $k$-th group of MIDAS-weighted covariates. Although the number of parameters to be estimated is reduced by taking a small MIDAS dictionary size $L$, a large $K$ can still pose a high-dimensional challenge.

\begin{table}[htbp]
\centering
\caption{Collection of notations.}
\medskip
\begin{tabular}{ll}
\toprule
$N$ & Sample size. \\
$K$ & Number of original covariates, each observed with $d$ lags. \\
$L$ & Size of MIDAS dictionary $W$. \\
$d$ & Total number of lags for each covariate. \\
$m$ & Annual sampling frequency of lagged covariates. \\
$s$ & Minimum length of time (in years) that all firms have survived. \\
$p$ & Dimension of MIDAS-weighted covariate vector $\boldsymbol{X}_i$.\\
\bottomrule
\end{tabular}
\label{notationt}
\end{table}

To further reduce dimensionality and exploit the natural grouping of covariates, we employ the sparse-group LASSO \citep{simon2013sparse}. While standard LASSO induces sparsity at the individual-variable level and group LASSO enforces sparsity at the group level, sparse-group LASSO provides a more flexible framework by encouraging sparsity both within and across groups. This is achieved through a convex combination of the LASSO and group LASSO penalties. As a result, sparse-group LASSO can select relevant groups while simultaneously performing variable selection within those groups, making it particularly well suited to high-dimensional settings with hierarchical covariate structures, such as in high-dimensional MIDAS regressions. 
For \(\boldsymbol{\beta}\in\mathbb{R}^{p}\), define the sparse-group LASSO penalty as
$$
\Omega(\boldsymbol{\beta}) = \alpha |\boldsymbol{\beta}|_1 + (1-\alpha) \|\boldsymbol{\beta}\|_{2,1}, \quad \|\boldsymbol{\beta}\|_{2,1} = \sum_{G \in \mathcal{G}} |\boldsymbol{\beta}_G|_{2},
$$
where $0 \leq \alpha \leq 1$ is the weight parameter that balances the LASSO and group LASSO penalties, and $\mathcal{G}$ denotes a pre-specified collection of groups, with each group $G \in \mathcal{G}$ being a subset of $[p]$. 
In our analysis, the specified group structure is \(\mathcal{G} = \{G_k : k \in [K+1]\}\), where $G_1=\{1\}$ corresponds to the intercept and for $k = 2, \ldots, K+1$, \(G_k = \{(1+(k-2)L)+ 1, \ldots, 1+(k-1)L\}\) corresponds to the $(k-1)$-th group of MIDAS-weighted covariates $W^{\top}\widetilde{\boldsymbol{Z}}_{i,k-1}$. This grouping arises by treating each covariate together with all of its lags as a single group, and then applying the MIDAS weighting matrix $W$ to each such group of lags. As shown by \citet{babii2022machine}, this simple and natural group structure enhances prediction performance. Alternative group structures are also possible. For example, groups can be formed by pairing conceptually related covariates, such as Return on Assets (ROA) and Return on Equity (ROE). We explore the robustness of our results to such alternative groupings in the empirical application.

For ease of reference, the notation introduced so far is summarized in Table \ref{notationt}. 

Finally, we propose the sparse-group LASSO estimator, which minimizes
\begin{equation}\label{midas}
\widehat{\boldsymbol{\beta}} = \argmin\limits_{\boldsymbol{\beta} \in \mathbb{R}^{p}} R_N(\boldsymbol{\beta}) + \lambda \Omega(\boldsymbol{\beta}),
\end{equation}
where
$$
R_N(\boldsymbol{\beta}) = \frac{1}{N} \sum_{i=1}^N -\frac{\delta_i(t) \mathbbm{1}\{\widetilde{T}_i \leq t\}}{\widehat{H}(t \wedge \widetilde{T}_i)} \boldsymbol{X}_i^\top \boldsymbol{\beta} + \log \left(1 + \exp(\boldsymbol{X}_i^\top \boldsymbol{\beta})\right).
$$
Here, \(\lambda \geq 0\) controls the regularization, and \(\alpha \in [0,1]\) balances the sparsity and the group structure. Choosing \(\alpha = 1\) recovers LASSO, while \(\alpha = 0\) corresponds to group LASSO.\footnote{In our practical implementation, we do not penalize the intercept coefficient. For simplicity, we do not write this in the equations.} In practice, $\alpha$ can be fixed or selected jointly with $\lambda$ by using cross-validation.
We call \textbf{sg-LASSO-MIDAS} the approach embodied by \eqref{midas}.


In the simulation and empirical sections, we also examine two alternative methods as benchmarks. The first method \textbf{LASSO-UMIDAS} employs unrestricted lag polynomials combined with LASSO. Unlike the MIDAS approach, LASSO-UMIDAS does not impose restrictions on the polynomials, resulting in the need to estimate a significantly larger number of parameters. Moreover, no structural constraints are applied to the coefficients of the lags. As a second alternative, we consider \textbf{LASSO-MIDAS}, which adopts the MIDAS approximation employed in the sg-LASSO-MIDAS method but avoids the group penalty. It corresponds to \eqref{midas} with a fixed parameter of $\alpha = 1$.

\section{Theoretical results} \label{sec the}

In this section, we first outline the main assumptions for the proposed penalized estimator \eqref{midas} and analyze the finite sample properties of the sparse-group LASSO estimator. Then we introduce the de-sparsified procedure and discuss the asymptotic properties of the de-sparsified sparse-group LASSO estimator. Both the (de-sparsified) LASSO and (de-sparsified) group LASSO estimators are covered as special cases.\footnote{We treat $\alpha$ as constant for the theory but optimize it through cross-validation in practice.}  Recall that we have an i.i.d. sample $\{(\widetilde{T}_i,\delta_i,\boldsymbol{Z}_i),\ i\in[N]\}$ such that $\widetilde{T}_i\ge s$ for all $i\in[N].$ We consider an asymptotic regime where $N$ goes to infinity and $p$ goes to infinity as a function of $N$. High-dimensional $\ell_1$-norm penalized logistic regression has been studied in the literature; see, for instance, \citet{10.1214/009053607000000929,van_de_geer_estimation_2016,caner2023generalized} and \citet{han2023high}.  However, none of these studies account for censoring or allow for approximation errors.
Instead, we explicitly take into account the approximation error stemming from the MIDAS approximation defined as
$$
    E_i = \boldsymbol{Z}_{i}^{\top}\boldsymbol{\theta}_0 - \boldsymbol{X}_{i}^{\top}\boldsymbol{\beta}_0, \quad i \in [N],
$$ 
where $\boldsymbol{\beta}_0 = \left(\theta_{0,1}, (\boldsymbol{\beta}_{0,1}^*)^{\top},(\boldsymbol{\beta}_{0,2}^*)^{\top},\ldots,(\boldsymbol{\beta}_{0,K}^*)^{\top}\right)^{\top}\in \mathbb{R}^p$ is the true parameter of interest and we assume it is exactly sparse. We refer to $\theta_{0,1}$ as the intercept term and $\boldsymbol{\beta}_{0,k}^* \in \mathbb{R}^{L}$ as the $k$-th group of MIDAS-weighted parameters. Let $\boldsymbol{E} = (E_1, E_2, \ldots, E_N)^{\top}$ collect all approximation errors.

If one is instead interested in allowing for approximate sparsity of $\boldsymbol{\beta}_{0}$ \citep{van_de_geer_estimation_2016,belloni2018high,babii2022machine}, a straightforward modification is to replace $E_i$ by the following error term $\widetilde{E}_i$ defined by
\[
\widetilde{E}_i
= \underbrace{\boldsymbol{Z}_i^{\top}\boldsymbol{\theta}_0 - \boldsymbol{X}_i^{\top}\boldsymbol{\beta}_0}_{\text{MIDAS approximation error}}
+\;
\underbrace{\boldsymbol{X}_i^{\top}\boldsymbol{\beta}_0 - \boldsymbol{X}_i^{\top}\boldsymbol{\beta}}_{\text{approximate sparsity error}},
\] 
where $\boldsymbol{\beta}$ is a candidate oracle that is exactly sparse, and the assumptions \ref{aseign} and \ref{aseff} below should therefore be imposed on $\boldsymbol{\beta}$ rather than on $\boldsymbol{\beta}_0$.\footnote{Alternatively, by adapting the argument underlying Corollary~12.7 of \citet{van_de_geer_estimation_2016}, one can also derive a finite-sample bound that yields a separate treatment of the MIDAS approximation error and the approximate sparsity error.} In this paper, we stick to exactly sparse $\boldsymbol{\beta}_0$ to simplify the theoretical analysis.
\subsection{Estimation theory}
We start by introducing the following assumptions.
\begin{assumption}\label{asX}
(Data) We have i.i.d. data $\{(\widetilde{T}_i,\delta_i,\boldsymbol{Z}_i),\ i\in[N]\}$, and 
 there exist $q\ge 4$ and $K_0>0$ such that $\max\limits_{|\boldsymbol{u}|_2 = 1}\mathbbm{E}\left(\left|\boldsymbol{X}_{i}^{\top}\boldsymbol{u}\right|^{q}\right) \leq K_0$. 
\end{assumption}
This condition just requires that the variables have at least $4$ finite moments, allowing for polynomial tails commonly observed with financial variables. It is worth noting that Assumption~\ref{asX} also implicitly controls correlations among covariates, as extremely high correlations would lead some linear combination to have an excessively large $q$-th moment, thereby violating the bound.
\begin{assumption}\label{aseign}
    There exists a constant $\gamma_{\mathrm{H}}>0$ such that the minimum eigenvalue
      $$
     \lambda_{\min }\left(\mathbbm{E}\left[\frac{\exp(\boldsymbol{X}_i^{\top}\boldsymbol{\beta}_0 + E_i)}{ \big(1+\exp(\boldsymbol{X}_i^{\top}\boldsymbol{\beta}_0 + E_i)\big)^2 }\boldsymbol{X}_{i} \boldsymbol{X}_{i}^{\top}\right]\right) \geq \gamma_{\mathrm{H}}.
   $$
\end{assumption}

\noindent Assumption \ref{aseign} is similar to the compatibility condition discussed in \citet{10.1214/009053607000000929,van_de_geer_estimation_2016,caner2023generalized}, as well as the restricted Fisher-information matrix eigenvalue condition described in \citet{han2023high}. This is a high-dimensional version of the full-rank condition guaranteeing the asymptotic properties of the maximum likelihood estimator in the (low-dimensional) logistic regression with misspecification error.

Next, we need to introduce additional definitions. Let  $S_{\boldsymbol{\beta}_0}=\{j\in[p]:\ \boldsymbol{\beta}_{0,j}\ne 0\}$ and $ \mathcal{G}_{\boldsymbol{\beta}_0}=\{G \in \mathcal{G}:\ \left(\boldsymbol{\beta}_{0}\right)_{G}\ne \boldsymbol{0}\}$ be the support and 
the group support of the target parameter $\boldsymbol{\beta}_0$. Let $\sqrt{s_{\boldsymbol{\beta}_0}} = \alpha \sqrt{|S_{\boldsymbol{\beta}_0}|} + (1 - \alpha)\sqrt{|\mathcal{G}_{\boldsymbol{\beta}_0}|}$ be the sparsity level and $G^*=\max _{G \in \mathcal{G}_{\boldsymbol{\beta}_0}}|G|$ be the size of the largest group in $\mathcal{G}_{\boldsymbol{\beta}_0}$. For simplicity, we suppose that $s_{\boldsymbol{\beta}_0}\ge 1$ (otherwise, it suffices to replace $s_{\boldsymbol{\beta}_0}$ by $s_{\boldsymbol{\beta}_0}\vee 1$ in all assumptions and bounds involving $s_{\boldsymbol{\beta}_0} $). We impose the following assumption on $s_{\boldsymbol{\beta}_0} $.
\begin{assumption}\label{aseff}
   It holds that 
   $$s_{\boldsymbol{\beta}_0}G^*\left(\frac{p^{\frac{2}{q}}  \log p}{N^{1-\frac{2}{q}}} \vee \frac{p^{\frac{2}{q}}  \sqrt{\log p}}{\sqrt{N}}\right) =o(1),$$ and 
   $$s_{\boldsymbol{\beta}_0} (G^*)^{\frac{3}{2}}\left(\frac{\lambda s_{\boldsymbol{\beta}_0}}{\gamma_{\mathrm{H}}} + \frac{\lambda^{-1}}{N} |\boldsymbol{E}|_1\right)(Np\log p)^{\frac{1}{q}}=o_P(1).$$
\end{assumption}
This is a condition on the degree of sparsity $s_{\boldsymbol{\beta}_0} $, the size of the largest group $G^*$, the $\ell_1$-norm of the approximation error $|\boldsymbol{E}|_1 = \sum_{i = 1}^{N} |E_i|$, and the relative growth rate of $N$ and $p$. The condition is more likely to hold when 
$s_{\boldsymbol{\beta}_0} $, $G^*$, $|\boldsymbol{E}|_1$ or $1/q$ are smaller and $p$ does not grow too quickly with $N$. Note also that if $\lambda$ is too low or too large, the condition might fail to hold. This condition allows establishing a connection between empirical and population effective sparsity, enabling the extension of the quadratic margin condition to its sampled version. \citet{10.1214/09-EJS506} briefly discussed it specifically for data with Gaussian tails in the case of the LASSO for the linear model. We extend this framework to accommodate heavy-tailed data and approximation error in the logistic regression model. When there is no such approximation error, that is $E_i = 0, $ for all $i \in [N]$, a similar assumption imposed on $\lambda$ and $s_{\boldsymbol{\beta}_0}$ is used in \citet{van_de_geer_estimation_2016} and \citet{han2023high}. 

We now establish bounds on the estimation error, presenting two distinct types. The first type pertains to the parameter estimation error $\Omega\left(\widehat{\boldsymbol{\beta}} - \boldsymbol{\beta}_0\right)$, while the second focuses on prediction accuracy. Consider a scenario where, for some $\boldsymbol{z} = \left(1, \tilde{\boldsymbol{z}}_1, \ldots, \tilde{\boldsymbol{z}}_K \right)^{\top} \in \mathbb{R}^{K_z}$, where $\tilde{\boldsymbol{z}}_k \in \mathbb{R}^{d}, k \in [K]$, we aim to estimate $P(\boldsymbol{z}) = P(T \leq t \mid \boldsymbol{Z} = \boldsymbol{z}, T \geq s)$, representing the probability that a firm with covariates $\boldsymbol{z}$, having survived at least $s$ years, fails before $t$. We estimate $P(\boldsymbol{z})$ using $\widehat{P}(\boldsymbol{z})  =  \frac{\exp \left( \boldsymbol{x}^{\top} \widehat{\boldsymbol{\beta}} \right)}{1 + \exp \left(\boldsymbol{x}^{\top}\widehat{\boldsymbol{\beta}}\right)}$ where $\boldsymbol{x} = \left(1,  \tilde{\boldsymbol{z}}_1^{\top}W, \ldots, \tilde{\boldsymbol{z}}_K^{\top}W\right)^{\top}$. Our goal is to provide a bound for the error $\widehat{P}(\boldsymbol{z}) - P(\boldsymbol{z})$. This bound will depend on the term $e = \boldsymbol{z}^\top \boldsymbol{\theta}_0 - \boldsymbol{x}^\top \boldsymbol{\beta}_0$, which represents the MIDAS approximation error at the covariate $\boldsymbol{z}$. The following theorem formally states this result.
\begin{theorem}
\label{the}
Let Assumptions \ref{as1}, \ref{as2}, \ref{asX}, \ref{aseign} and \ref{aseff} hold. If there exists a sufficiently large constant $\mathcal{K}$ such that $\lambda \geq \mathcal{K}p^{\frac{1}{q}}\sqrt{\log p}/N^{\frac{1}{2} - \frac{1}{q}}$, then, with probability going to $1$, we have
\begin{align*}
    \Omega\left(\widehat{\boldsymbol{\beta}} - \boldsymbol{\beta}_0\right) &\lesssim \frac{\lambda s_{\boldsymbol{\beta}_0}}{\gamma_{\mathrm{H}}} +  \lambda^{-1}\frac{1}{N} |\boldsymbol{E}|_1,
\end{align*}
and 
\begin{align*}
    \widehat{P}(\boldsymbol{z}) -  P(\boldsymbol{z}) &\lesssim \frac{\lambda s_{\boldsymbol{\beta}_0}\left|\boldsymbol{x}\right|_{\infty}}{\gamma_{\mathrm{H}}} +  \lambda^{-1}\frac{1}{N} |\boldsymbol{E}|_1 \left|\boldsymbol{x}\right|_{\infty} + |e|.
\end{align*}
\end{theorem}
Let us now discuss the theorem. First, we require that $\lambda$ has at least the same order as
$
p^{\frac{1}{q}} \sqrt{\log p}/N^{\frac{1}{2} - \frac{1}{q}}.
$
For bounds on the LASSO under sub-Gaussian errors, it suffices that \(\lambda\) is of the order of \(\sqrt{\log p / N}\). Our condition is stricter due to the presence of heavy-tailed variables. However, as \( q \to \infty \), the variables are no longer heavy-tailed, and we recover the order \(\sqrt{\log p / N}\) for \(\lambda\). As is standard in the literature, we select \(\lambda\) in practice via cross-validation (see Sections \ref{sec sim} and \ref{sec real}). The dependence of our rates on \(\lambda\) aligns with those for the standard LASSO estimator in high-dimensional regression. 

Our bounds also depend on the \(\ell_1\)-norm \(|\boldsymbol{E}|_1\) of the approximation error. To the best of our knowledge, this work is the first to establish results for the high-dimensional logistic regression model with an approximation error. Such results, however, are well-established for the LASSO in linear regression \citep[see][]{bickel2009simultaneous}. To achieve this result, we bound the difference between the empirical and population loss functions not only by terms related to the empirical process but also by a term dependent on the approximation error \(\boldsymbol{E}\). Addressing this challenge is particularly difficult due to the nonlinearity of the problem. Interested readers are referred to Online Appendix \ref{proof of the} for the detailed proof of Theorem \ref{the}. Regarding the estimated prediction probability, its error bound matches the order of the parameter estimation error, with an additional term which is a function of the MIDAS approximation error \( e \).

As a concluding remark, we note that the Kaplan–Meier estimator converges at the $1/\sqrt{N}$ rate. Hence, estimating the weights $1/H(t \wedge \widetilde{T}_i)$ does not affect the convergence rate of our penalized estimator.

We then consider a regime where the MIDAS approximation error vanishes at a specific rate, as described in the following assumption.
\begin{assumption}\label{asapprerror}
     $\left|\boldsymbol{E}\right|_1/N = O_P(\lambda^2 s_{\boldsymbol{\beta}_0})$.
\end{assumption}
The following corollary on parameter estimation follows immediately.
\begin{corollary}\label{cor rate}
Let Assumptions \ref{as1}, \ref{as2}, \ref{asX}, \ref{aseign}, \ref{aseff}, and \ref{asapprerror} hold. If there exist sufficiently large constants $\mathcal{K}_1$ and $\mathcal{K}_2$ such that $\mathcal{K}_1p^{\frac{1}{q}}\sqrt{\log p}/N^{\frac{1}{2} - \frac{1}{q}} \leq \lambda \leq \mathcal{K}_2p^{\frac{1}{q}}\sqrt{\log p}/N^{\frac{1}{2} - \frac{1}{q}}$, then, with probability going to $1$, we have
\begin{align*}
    \Omega\left(\widehat{\boldsymbol{\beta}} - \boldsymbol{\beta}_0\right) &\lesssim \frac{p^{\frac{1}{q}}\sqrt{\log p}}{N^{\frac{1}{2} - \frac{1}{q}}} s_{\boldsymbol{\beta}_0}.
\end{align*}
\end{corollary}
Corollary \ref{cor rate} characterizes the impact of censoring and heavy-tailed covariates on the estimation accuracy of $\widehat{\boldsymbol{\beta}}$. Unlike Corollary~3.1 of \citet{babii2022machine}, which explicitly incorporates persistence in the dependent variable for a linear time-series model, our setting involves i.i.d. censored data in a logistic regression framework with lagged covariates treated as separate variables. In this context, the $q$-th moment bound in Assumption \ref{asX} implicitly controls the effect of correlations among lagged covariates. Hence, Corollary~\ref{cor rate} captures the impact of covariate correlation strength indirectly. For literature that explicitly incorporates covariate correlations in error bounds for i.i.d. data, we refer to \citet{hebiri2012correlations,lassoprediction}. While these works focus on linear models, their arguments could potentially be extended to nonlinear settings, which we leave for future research.


\subsection{De-sparsified procedure and inference theory} \label{inference theory}
Beyond estimation with the proposed sparse-group LASSO estimator, our second goal is to conduct statistical inference for the underlying parameters. 
Let us use the notation
$
\boldsymbol{\beta}_0 := \left(\beta_{0,1}, \beta_{0,2}, \ldots, \beta_{0,p}\right)^{\top} \in \mathbb{R}^p.
$
We focus on inference on the subvector
$$
\boldsymbol{\beta}_{0,(\mathcal{J})} := (\beta_{0,j})_{j \in \mathcal{J}} \in \mathbb{R}^{|\mathcal{J}|},
$$
where $\mathcal{J} \subseteq [p]$ denotes a fixed index set. As an example, one may wish to test the hypothesis $\mathbb{H}_0: \boldsymbol{\beta}_{0,(\mathcal{J})} = \boldsymbol{0}$. 
However, it is well known that the penalized estimator $\widehat{\boldsymbol{\beta}}$ is asymptotically biased. Hence, a debiasing (or desparsification) procedure is required.

Our de-sparsified methodology builds on the approach of \citet{van_de_geer_debias}, which approximates the inverse of the population Gram matrix $\boldsymbol{\Sigma}_{\boldsymbol{\beta}_0}$ via nodewise regressions and employs this approximation to correct the bias of the penalized estimator, yielding de-sparsified estimates with asymptotically normal components.
Unlike in the linear model, the empirical Gram matrix in the logistic model depends on the penalized estimator $\widehat{\boldsymbol{\beta}}$. In our setting, the sparse-group LASSO estimator $\widehat{\boldsymbol{\beta}}$ in \eqref{midas} is further complicated by censoring, MIDAS approximation errors, and the use of a sparse-group LASSO penalty, distinguishing our approach from the existing literature.

We now describe the procedure to estimate $\boldsymbol{\Sigma}_{\boldsymbol{\beta}_0}^{-1}$. To this end, we define the population Gram matrix in the presence of the MIDAS approximation error $\boldsymbol{E}$:
\begin{equation}\label{populationgram}
\boldsymbol{\Sigma}_{\boldsymbol{\beta}_0}
:= \mathbbm{E}\!\left[ 
\frac{\exp(\boldsymbol{X}_i^{\top}\boldsymbol{\beta}_0 + E_i)}
{\big(1+\exp(\boldsymbol{X}_i^{\top}\boldsymbol{\beta}_0 + E_i)\big)^2}
\boldsymbol{X}_i \boldsymbol{X}_i^{\top}\right]
\in \mathbb{R}^{p \times p}.
\end{equation}
Let $\boldsymbol{\Theta} = (\boldsymbol{\Theta}_1,\ldots,\boldsymbol{\Theta}_p)^{\top}
= \boldsymbol{\Sigma}_{\boldsymbol{\beta}_0}^{-1}$, and let $\widehat{\boldsymbol{\Theta}}
= \left(\widehat{\boldsymbol{\Theta}}_1, \ldots, \widehat{\boldsymbol{\Theta}}_p\right)^{\top}$ denote its estimator.
Each $\widehat{\boldsymbol{\Theta}}_j \in \mathbb{R}^p$, $j \in [p]$, can be obtained by using a weighted nodewise regression. The weighting arises naturally from the logistic model, as the Gram matrix in logistic regression incorporates observation-specific weights. Accordingly, we define the weights used in the nodewise regressions at both the population and sample levels.
$$
\begin{aligned}
W_{\boldsymbol{\beta}_0} & :=\operatorname{diag}\left(w_{\boldsymbol{\beta}_0, 1}, \ldots, w_{\boldsymbol{\beta}_0, i}, \ldots, w_{\boldsymbol{\beta}_0, N}\right) \in \mathbb{R}^{N \times N},\\
W_{\widehat{\boldsymbol{\beta}}}&:=\operatorname{diag}\left(w_{\widehat{\boldsymbol{\beta}}, 1}, \ldots, w_{\widehat{\boldsymbol{\beta}}, i}, \ldots, w_{\widehat{\boldsymbol{\beta}}, N}\right) \in \mathbb{R}^{N \times N},
\end{aligned}
$$ 
where both are $N \times N$ diagonal matrices, with 
\begin{equation}\label{samllwandwhat}
\begin{aligned}
w_{\boldsymbol{\beta}_0, i}&:=\sqrt{\frac{\exp(\boldsymbol{X}_i^{\top}\boldsymbol{\beta}_0+E_i)}{ \big(1+\exp(\boldsymbol{X}_i^{\top}\boldsymbol{\beta}_0+E_i)\big)^2 }}, \quad
w_{\widehat{\boldsymbol{\beta}}, i}:=\sqrt{\frac{\exp(\boldsymbol{X}_i^{\top}\widehat{\boldsymbol{\beta}})}{ \big(1+\exp(\boldsymbol{X}_i^{\top}\widehat{\boldsymbol{\beta}})\big)^2 }}, \quad i \in [N].
\end{aligned}
\end{equation}
For $\boldsymbol{\beta} \in \mathbb{R}^{p}$, which may represent either $\boldsymbol{\beta}_0$ or $\widehat{\boldsymbol{\beta}}$, let $\boldsymbol{X}_{\boldsymbol{\beta}}:=W_{\boldsymbol{\beta}} \boldsymbol{X} \in \mathbb{R}^{N \times p}$, and denote the $j$-th column of that matrix as $\boldsymbol{X}_{\boldsymbol{\beta}, j} \in \mathbb{R}^{N}$, and the matrix containing all the columns of $\boldsymbol{X}_{\boldsymbol{\beta}}$ except the $j$-th one as $\boldsymbol{X}_{\boldsymbol{\beta},-j} \in \mathbb{R}^{N \times (p-1)}$.

To implement the weighted nodewise regressions, for each $j \in [p]$ we solve
\begin{equation}\label{nodewiseregression}
    \begin{aligned}
    \widehat{\boldsymbol{\gamma}}_{\widehat{\boldsymbol{\beta}}, j}=\underset{\boldsymbol{\gamma}_j \in \mathbb{R}^{p-1}}{\operatorname{argmin}} \left\|\boldsymbol{X}_{\widehat{\boldsymbol{\beta}}, j}-\boldsymbol{X}_{\widehat{\boldsymbol{\beta}}, -j} \boldsymbol{\gamma}_j\right\|_N^2+2 \lambda_j\left|\boldsymbol{\gamma}_j\right|_1,
    \end{aligned}
\end{equation}
where $ \widehat{\boldsymbol{\gamma}}_{\widehat{\boldsymbol{\beta}}, j} = ( \widehat{\gamma}_{\widehat{\boldsymbol{\beta}}, j,1}, \ldots, \widehat{\gamma}_{\widehat{\boldsymbol{\beta}}, j,j-1},\widehat{\gamma}_{\widehat{\boldsymbol{\beta}}, j,j+1},\ldots, \widehat{\gamma}_{\widehat{\boldsymbol{\beta}}, j,p})^{\top} \in \mathbb{R}^{p-1}$ and $\lambda_j>0$ is the penalty parameter for each regression, which can be selected using cross-validation in practice. Then the estimator of $\boldsymbol{\Theta}$ is
$$
\widehat{\boldsymbol{\Theta}} = \left(\widehat{\boldsymbol{\Theta}}_1, \ldots, \widehat{\boldsymbol{\Theta}}_p\right)^{\top} = \widehat{\boldsymbol{B}}^{-1}\widehat{\boldsymbol{C}},
$$ 
where
\begin{equation}\label{hatBC}
\begin{aligned}
\widehat{\boldsymbol{B}} &= diag\left[\frac{\boldsymbol{X}_{\widehat{\boldsymbol{\beta}}, 1}^{\top}\left(\boldsymbol{X}_{\widehat{\boldsymbol{\beta}}, 1}-\boldsymbol{X}_{\widehat{\boldsymbol{\beta}}, -1}\widehat{\boldsymbol{\gamma}}_{\widehat{\boldsymbol{\beta}}, 1}\right)}{N}, \ldots, \frac{\boldsymbol{X}_{\widehat{\boldsymbol{\beta}}, p}^{\top}\left(\boldsymbol{X}_{\widehat{\boldsymbol{\beta}}, p}-\boldsymbol{X}_{\widehat{\boldsymbol{\beta}}, -p} \widehat{\boldsymbol{\gamma}}_{\widehat{\boldsymbol{\beta}}, p}\right)}{N}\right],\\
\widehat{\boldsymbol{C}} & = \left(\widehat{\boldsymbol{C}}_1, \ldots, \widehat{\boldsymbol{C}}_p\right)^{\top} 
= \left(\begin{array}{cccc}
1 & -\widehat{\gamma}_{\widehat{\boldsymbol{\beta}}, 1,2} & \ldots & -\widehat{\gamma}_{\widehat{\boldsymbol{\beta}}, 1, p} \\
-\widehat{\gamma}_{\widehat{\boldsymbol{\beta}}, 2,1} & 1 & \ldots &  -\widehat{\gamma}_{\widehat{\boldsymbol{\beta}}, 2, p} \\
\vdots & \vdots & \ddots & \vdots \\
-\widehat{\gamma}_{\widehat{\boldsymbol{\beta}}, p,1} & \ldots & -\widehat{\gamma}_{\widehat{\boldsymbol{\beta}}, p,p-1} & 1
\end{array}\right).
\end{aligned}
\end{equation}
Each $\widehat{\boldsymbol{\Theta}}_j$ can also be written as 
\begin{equation}
\widehat{\boldsymbol{\Theta}}_j
= \widehat{\boldsymbol{C}}_j / \hat{\tau}_j^2,
\qquad \text{where} \quad
\hat{\tau}_j^2
= \frac{\boldsymbol{X}_{\widehat{\boldsymbol{\beta}}, j}^{\top}
(\boldsymbol{X}_{\widehat{\boldsymbol{\beta}}, j}
- \boldsymbol{X}_{\widehat{\boldsymbol{\beta}},-j}
\widehat{\boldsymbol{\gamma}}_{\widehat{\boldsymbol{\beta}}, j})}{N}.
\end{equation}

To establish the convergence rate of $\widehat{\boldsymbol{\Theta}}$, we introduce the following notation and assumptions. For each $j \in [p]$, consider the population nodewise regressions
\begin{equation}\label{populationnodewise1}
    \begin{aligned}
   \boldsymbol{\gamma}_{\boldsymbol{\beta}_0, j} = \underset{\boldsymbol{\gamma}_j \in \mathbb{R}^{p-1}}{\operatorname{argmin}} \quad \mathbbm{E}\left(\left\|\boldsymbol{X}_{\boldsymbol{\beta}_0, j}-\boldsymbol{X}_{\boldsymbol{\beta}_0, -j} \boldsymbol{\gamma}_j\right\|_N^2\right),
    \end{aligned}
\end{equation}
where $ \boldsymbol{\gamma}_{\boldsymbol{\beta}_0, j} = ( \boldsymbol{\gamma}_{\boldsymbol{\beta}_0, j, 1}, \ldots,\boldsymbol{\gamma}_{\boldsymbol{\beta}_0, j, j-1},\boldsymbol{\gamma}_{\boldsymbol{\beta}_0, j, j+1},\ldots,  \boldsymbol{\gamma}_{\boldsymbol{\beta}_0, j, p})^{\top} \in \mathbb{R}^{p-1}$ and define the error as
\begin{equation}\label{populationnodewise}
\boldsymbol{\eta}_{\boldsymbol{\beta}_0,j} = \boldsymbol{X}_{\boldsymbol{\beta}_0, j}-\boldsymbol{X}_{\boldsymbol{\beta}_0, -j} \boldsymbol{\gamma}_{\boldsymbol{\beta}_0,j} := \left(\eta_{\boldsymbol{\beta}_0,j,1}, \ldots, \eta_{\boldsymbol{\beta}_0,j,N}\right)^{\top} \in \mathbb{R}^{N}.
\end{equation}

\begin{assumption}\label{aseign_2}
Let $q > 12$. (i) $\|\boldsymbol{E}\|_{N}=o_{p}(N^{-1/2-3/(2q)})$; 
(ii) $\max\limits_{|\boldsymbol{u}|_2 = 1}\mathbbm{E}\left(\left|\boldsymbol{X}_{i}^{\top}\boldsymbol{u}\right|^{q}\right) \leq K_0$;
(iii) $\max_{j\in[p]}\mathbbm{E}(|\eta_{\boldsymbol{\beta}_0,j,i}|^{q})\le C_{\eta}$; 
(iv) $\max_{j\in[p]}\mathbbm{E}(|\boldsymbol{\Theta}_{j}^{\top}\boldsymbol{X}_i|^{q})\leq K_{\Theta}$.
\end{assumption}
Assumption \ref{aseign_2} (i) requires that the MIDAS approximation errors vanish sufficiently fast asymptotically.\footnote{If each row of $\boldsymbol{\Theta}$ has a bounded $\ell_1$-norm, we only require $\|\boldsymbol{E}\|_{N}=o_{p}(N^{-1/2})$ in Assumption \ref{aseign_2} (i).} 
Assumption~\ref{aseign_2} (ii) imposes a stronger moment condition on the covariates than Assumption~\ref{asX}, requiring $q>12$. This requirement differs from much of the existing literature. For example, Condition~(D4) in \citet{van_de_geer_debias} and Assumption~3 in \citet{caner2023generalized} either assume uniformly bounded covariates or impose a uniform lower bound on $w^2_{\boldsymbol{\beta}_0,i} \ge c$ for all $i \in [N]$, where $c$ is a strictly positive constant. Such assumptions may be restrictive or unrealistic when covariates have heavy tails. In contrast, our framework allows for heavy-tailed covariates without requiring boundedness or uniform positivity of $w^2_{\boldsymbol{\beta}_0,i}$. This added generality comes at the cost of a more restrictive sparsity condition, as formalized in Assumption~\ref{asrate}.
Assumption \ref{aseign_2} (iii) can be interpreted as a weaker version of the second statement of Condition (D1) in \citet{van_de_geer_debias}, where the latter needed $|\boldsymbol{\eta}_{\boldsymbol{\beta}_0,j}|_{\infty}$ to be bounded. Assumption~\ref{aseign_2} (iv) weakens Assumption~4.1 (ii) in \citet{babii2023machine}, which imposes that every row of $\boldsymbol{\Theta}$ has a uniformly bounded $\ell_1$-norm. Together with Assumption~\ref{aseign_2} (ii), Assumption~4.1 (ii) in \citet{babii2023machine} implies our Assumption~\ref{aseign_2} (iv), so our requirement is less restrictive.\footnote{If each row of $\boldsymbol{\Theta}$ has a bounded $\ell_1$-norm, we have $\mathbbm{E}(|\boldsymbol{\Theta}_{j}^{\top}\boldsymbol{X}_i|^{q}) = |\boldsymbol{\Theta}_{j}|^q_2\mathbbm{E}(|\boldsymbol{\Theta}_{j}^{\top}\boldsymbol{X}_i/|\boldsymbol{\Theta}_{j}|_2|^{q}) \leq |\boldsymbol{\Theta}_{j}|^q_1K_0$ where the inequality follows from Assumption~\ref{aseign_2} (ii) and the fact that $|\boldsymbol{\Theta}_{j}|_2 \leq |\boldsymbol{\Theta}_{j}|_1$.}

Define $\bar{s} = \max_{j \in \mathcal{J}} |S_j|$, where $S_j$ denotes the index set of nonzero elements of $\boldsymbol{\Theta}_j$, and set $s^{*} = \bar{s} \vee s_{\boldsymbol{\beta}_0}$. We have the following assumption for the sparsity level $s^*$.
\begin{assumption}\label{asrate}
    (Rate of sparsity level for inference)
     $s^* =  o\left(N^{\frac{1}{4}-\frac{3}{q}}/\left(p^{\frac{2}{q}}(\log p)^{\frac{2}{3} + \frac{2}{q}}\right)\right)$.
\end{assumption}
When the covariates are bounded, and there is no MIDAS approximation error, the admissible sparsity level for valid inference is typically of the order $o\big(\sqrt{N}/\log p\big)$ \citep{van_de_geer_debias}. The sparsity condition imposed in this paper is more restrictive, reflecting the presence of MIDAS approximation error and heavy-tailed covariates.

\begin{lemma}\label{nodewise}
Suppose that Assumptions \ref{asapprerror}, \ref{aseign_2}, \ref{asrate}, and all conditions of Theorem \ref{the} hold. If there exist sufficiently large constants $\mathcal{K}_1$, $\tilde{\mathcal{K}}_1$, $\mathcal{K}_2$, and $\tilde{\mathcal{K}}_2$ such that
\[
\mathcal{K}_1\frac{p^{1/q} \sqrt{\log p}}{N^{1/2 - 1/q}} \leq \lambda \leq \tilde{\mathcal{K}}_1\frac{p^{1/q} \sqrt{\log p}}{N^{1/2 - 1/q}} ,
\]
and
\[
\mathcal{K}_2\frac{(p-1)^{2/q} \sqrt{\log p}}{N^{1/2 - 2/q}} \leq \lambda_j \leq \tilde{\mathcal{K}}_2\frac{(p-1)^{2/q} \sqrt{\log p}}{N^{1/2 - 2/q}} \quad \text{for all } j \in \mathcal{J},
\] 
then the following holds:
    \begin{equation}
    \begin{aligned}
        \max _{j \in \mathcal{J}}\left|\widehat{\boldsymbol{\Theta}}_j-\boldsymbol{\Theta}_j\right|_1  = o_P(1).
        \end{aligned}
    \end{equation}
\end{lemma}
Although lemmas of this type for nodewise regression are not new in the literature, this paper makes several distinct contributions. \citet{van_de_geer_debias} established a related result for nodewise regression with the LASSO. However, their analysis relies on strong assumptions, including uniformly bounded covariates and uniformly bounded products of nodewise regression coefficients and covariates. In addition, their framework does not accommodate censored data or the MIDAS approximation error. Similarly, \citet{caner2023generalized} derived a general result for weighted nodewise regression using a structured sparsity estimator. Nevertheless, their analysis neither considers censored data nor accounts for the MIDAS approximation error, and it assumes that $w^2_{\boldsymbol{\beta}_0,i}$ is uniformly bounded away from zero. Such an assumption can be restrictive and is generally violated in the presence of unbounded covariates.

Given the estimator $\widehat{\boldsymbol{\Theta}}$, the de-sparsified sparse-group LASSO estimator 
$\widehat{\boldsymbol{b}} = (\widehat{b}_1, \ldots, \widehat{b}_p)^{\top}$ is constructed componentwise as
\begin{equation}\label{debiased estimator}
\widehat{b}_j =  \widehat{\beta}_j - \widehat{\boldsymbol{\Theta}}_j^{\top}  \frac{1}{N} \sum_{i=1}^N \bigg( - \frac{\delta_i(t)\mathbbm{1}{\{\widetilde{T}_i \leq t\}}}{\widehat{H}\left(t \wedge \widetilde{T}_i\right)} +  \frac{\exp(\boldsymbol{X}_{i}^{\top}\widehat{\boldsymbol{\beta}})}{1+\exp(\boldsymbol{X}_{i}^{\top}\widehat{\boldsymbol{\beta}}) }\bigg) \boldsymbol{X}_{i},
\end{equation}
where $\widehat{\beta}_j$ is the $j$-th coordinate of the sparse-group LASSO estimator $\widehat{\boldsymbol{\beta}}$. We call \textbf{de-sparsified sg-LASSO-MIDAS} the approach embodied by \eqref{debiased estimator}.
To derive the limiting distribution of the de-sparsified estimator, one more assumption is required.
\begin{assumption}\label{asinfeign}
    There exists a constant $\gamma_{\mathrm{L}}>0$ such that the minimum eigenvalue
      $$
      \begin{aligned}
     \lambda_{\min }\left(\mathbbm{E}\left[\left(- \frac{\delta_i(t)}{ H\left(t \wedge \widetilde{T}_i\right)} \mathbbm{1}{\{\widetilde{T}_i \leq t\}} +  \frac{\exp(\boldsymbol{X}_{i}^{\top}\boldsymbol{\beta}_0+E_i)}{1+\exp(\boldsymbol{X}_{i}^{\top}\boldsymbol{\beta}_0+E_i) }\right)^2\boldsymbol{X}_{i} \boldsymbol{X}_{i}^{\top}\right]\right)  \geq \gamma_{\mathrm{L}}.
      \end{aligned}
   $$
\end{assumption}
Assumption \ref{asinfeign} is a typical assumption for high-dimensional logistic regression, see a similar condition (vi) in Theorem $3.3$ of \citet{van_de_geer_debias} or Assumption $8$ of \citet{caner2023generalized}, which is used to show Lyapunov’s condition for the limiting distribution of the de-sparsified sparse-group LASSO estimator. We extend it to the case of censored data. 

The following theorem establishes the asymptotic distribution of the proposed de-sparsified estimator, which is applicable to high-dimensional censored data. The detailed proof is provided in Online Appendix \ref{appinf}. In order to state the theorem, we introduce
$$
\xi_{KM, i} = \mathbbm{E}\left(\boldsymbol{X}_{(2)}\frac{\delta_{(2)}(t) \mathbbm{1}{\{\widetilde{T}_{(2)} \leq t\}}}{H^2\left(t \wedge \widetilde{T}_{(2)}\right)}v\left(\widetilde{T}_i, \delta_i,t \wedge \widetilde{T}_{(2)}\right)\bigg| \boldsymbol{X}_i, \widetilde{T}_i, \delta_i\right),
$$
where $\left(\widetilde{T}_{(2)}, \delta_{(2)}, \boldsymbol{X}_{(2)}\right)$ denotes an independent copy drawn from the same distribution as the sample $\left\{\left(\widetilde T_i, \delta_{i}, \boldsymbol{X}_i\right), i \in [N]\right\}$. Here, $v\left(\widetilde{T}_i, \delta_i, \cdot \right)$ is the influence function of the Kaplan-Meier estimator of $H$ for observation $\left(\widetilde T_i, \delta_i\right)$ \citep[e.g.,][]{lo1986product,gill2006lectures}. Specifically, for any positive value $z$, we have 
\begin{equation}\label{IFKM}
\begin{aligned}
v\left(\widetilde{T}_k, \delta_k, z\right)=H(z)\left[\frac{ \mathbbm{1}\left\{\widetilde{T}_k \leq z, \delta_k = 0\right\}}{\bar{P}\left(\widetilde{T}_k\right)}-\int_0^{z \wedge \widetilde{T}_k} \frac{P_0(du)}{\bar{P}^2(u)}\right],
\end{aligned}
\end{equation}
where $\bar{P}(u)=P\left(\widetilde{T}>u\right)$ and $P_0(du)=\frac{\partial P\left(\widetilde{T} \leq u, \delta = 0\right)}{\partial u}du$. 
\begin{theorem}
\label{inf the}
Suppose all assumptions and conditions of Theorem \ref{the} and Lemma \ref{nodewise} hold. Also suppose that Assumption \ref{asinfeign} holds. We have
\begin{equation}
\sqrt{N} \left(\widehat{\boldsymbol{b}}_{(\mathcal{J})}-\boldsymbol{\beta}_{0,(\mathcal{J})}\right)  \xrightarrow[]{d} \mathcal{N}(\boldsymbol{0},\boldsymbol{V}_\mathcal{J}), 
\end{equation}
where $\widehat{\boldsymbol{b}}_{(\mathcal{J})} \in \mathbb{R}^{|\mathcal{J}|}$ is $\widehat{\boldsymbol{b}}$ indexed by $\mathcal{J}$, and 
\tcr{
$$
\boldsymbol{V}_\mathcal{J} = \boldsymbol{\Theta}_\mathcal{J} Var\left(\left(- \frac{\delta_i(t)}{H\left(t \wedge \widetilde{T}_i\right)} \mathbbm{1}{\{ \widetilde{T}_i \leq t\}} 
  +  \frac{\exp(\boldsymbol{X}_{i}^{\top}\boldsymbol{\beta}_0+E_i)}
{1+\exp(\boldsymbol{X}_{i}^{\top}\boldsymbol{\beta}_0+E_i)}\right)\boldsymbol{X}_i + \xi_{KM, i}\right)\boldsymbol{\Theta}_\mathcal{J}^{\top}.
$$
}
\end{theorem}
Theorem \ref{inf the} establishes the limiting distribution of the de-sparsified sparse-group LASSO estimator $\widehat{\boldsymbol{b}}_{(\mathcal{J})}$. In contrast to traditional de-sparsified estimators studied in \citet{van_de_geer_debias}, \citet{caner2023generalized}, and \citet{babii2023machine}, the variance of $\widehat{\boldsymbol{b}}_{(\mathcal{J})}$ is influenced by censored data, a phenomenon that, to the best of our knowledge, has not been previously studied. 

Finally, we provide a plug-in estimator of the variance $\boldsymbol{V}_\mathcal{J}$:
\[
\widehat{\boldsymbol{V}}_\mathcal{J} = \widehat{\boldsymbol{\Theta}}_\mathcal{J}\left(\frac{1}{N}\sum_{i=1}^{N}\widehat{\boldsymbol{\sigma}}_i\widehat{\boldsymbol{\sigma}}_i^{\top}\right)\widehat{\boldsymbol{\Theta}}_\mathcal{J}^{\top},
\]
where
\tcr{
\[
\widehat{\boldsymbol{\sigma}}_i = \left(-\frac{\delta_i(t) \mathbbm{1}\{\widetilde{T}_i \leq t\}}{\widehat{H}(t \wedge \widetilde{T}_i)}+\frac{\exp(\boldsymbol{X}_{i}^{\top}\widehat{\boldsymbol{\beta}})}{1+\exp(\boldsymbol{X}_{i}^{\top}\widehat{\boldsymbol{\beta}}) }\right)\boldsymbol{X}_{i} + \frac{1}{N}\sum_{k=1}^{N}\boldsymbol{X}_{k}\frac{\delta_k(t) \mathbbm{1}{\{\widetilde{T}_k \leq t\}}}{\widehat{H}^2\left(t \wedge \widetilde{T}_k\right)}\hat{v}\left(\widetilde{T}_i, \delta_i,t \wedge \widetilde{T}_k\right),
\]
}
and $\hat{v}\left(\widetilde{T}_i, \delta_i, \cdot\right)$ is the plug-in estimator of $v\left(\widetilde{T}_i, \delta_i, \cdot\right)$. Specifically, for any positive value $z$,
\begin{equation}\label{IFestimate}
\hat v\left(\widetilde{T}_i, \delta_i, z\right)
=\widehat{H}(z)\left[\frac{ \mathbbm{1}\left\{\widetilde{T}_i \leq z, \delta_i = 0\right\}}{\bar{P}_N\left(\widetilde{T}_i\right)}- \frac{1}{N}\sum_{j=1}^N\frac{\mathbbm{1}\{\widetilde{T}_j \leq z \wedge \widetilde{T}_i, \delta_j = 0\}}{\bar{P}_N^2(\widetilde{T}_j)}\right],
\end{equation}
where $\bar{P}_N(u) = \frac{1}{N}\sum_{l=1}^N\mathbbm{1}\{\widetilde{T}_l > u\}$. \footnote{For $\hat{v}$, its implementation can be found in the R package \texttt{survival}; see the argument \texttt{influence} of the function \texttt{survfit} in \citet{therneau2015package}.} The term $\frac{1}{N}\sum_{j=1}^N\frac{\mathbbm{1}\{\widetilde{T}_j \leq z \wedge \widetilde{T}_i, \delta_j = 0\}}{\bar{P}_N^2(\widetilde{T}_j)}$ in \eqref{IFestimate} is an estimate of the term $\int_0^{z \wedge \widetilde{T}_i} \frac{P_0(du)}{\bar{P}^2(u)}$ in \eqref{IFKM} because we have
$
\int_0^{z \wedge \widetilde{T}_i} \frac{1}{\bar{P}^2(u)}P_0(du) = \mathbbm{E}\left(\frac{1}{\bar{P}^2(\widetilde{T})}\mathbbm{1}\{\widetilde{T} \leq z \wedge \widetilde{T}_i\}\mathbbm{1}\{\delta=0\}\right),
$ since $P_0(du)$ is the probability mass of $\widetilde{T}$ in an infinitesimal interval around $u$, computed only over the event $\delta=0$.

\section{Simulations} \label{sec sim}

We first evaluate the predictive performance of three methods through simulations: i) LASSO-UMIDAS, an unstructured LASSO estimator with unrestricted MIDAS weights, ii) LASSO-MIDAS, an unstructured LASSO estimator using MIDAS weights, and iii) sg-LASSO-MIDAS, a structured sparse-group LASSO estimator with MIDAS weights. The sg-LASSO-MIDAS approach, specifically, highlights the advantages of leveraging group structures and dictionaries within a high-dimensional framework, offering a compelling comparison to LASSO-MIDAS and LASSO-UMIDAS \citep[see, e.g.,][]{babii2022machine}. The prediction simulation results showcase the method's strengths in achieving superior prediction accuracy with finite sample data. In addition to predictive performance, we evaluate the inference performance of de-sparsified sg-LASSO-MIDAS. The inference simulation results indicate that it provides accurate finite-sample inference, delivering empirical rejection rates close to the nominal level across various scenarios.

\subsection{Simulation design}\label{est simulation}

Let us describe the data-generating process. There are $K=50$ high-frequency covariates, but only the first two enter the model. All the observations in the simulated dataset have survived at least $s$ years, and we are interested in a quarterly sampling frequency $ m = 4$. We consider $s=6$ years of lagged data. 

For the generation of $z_{i, \frac{j}{m}, k}, j \in [d]$, we first initiate the processes by letting $\left(z_{i,\frac{1}{m},1}, \ldots, z_{i,\frac{1}{m},K}\right)^{\top}$ follow a $\mathcal{N}\left(\boldsymbol{0}, \Sigma \right)$ distribution with $\Sigma_{u, v}=\rho_{0}^{|u-v|}, u,v \in [K]$. Then, the high-frequency covariates $z_{i, \frac{j}{m},k}, j \in [d], k \in [K]$ are generated according to the following scenarios:

\newcounter{scenario}  
\medskip
\refstepcounter{scenario} 
\noindent\textit{Scenario \thescenario}:\label{s1} $z_{i, \frac{j}{m},k}=\rho z_{i,\frac{j-1}{m},k}+\nu_{i,k}$, $k \in [K], j \in \{2,3,\ldots,d\}$, and $(\nu_{i,1}, \ldots, \nu_{i,K})^{\top}  \sim_{\text {i.i.d }} \mathcal{N}\left(\boldsymbol{0}, \Sigma(1-\rho^2)\right)$ with $\rho=0.1$ and $\rho_0 = 0.1$.

\medskip

\refstepcounter{scenario} 
\noindent\textit{Scenario \thescenario}:\label{s2} $z_{i, \frac{j}{m},k}=\rho z_{i,\frac{j-1}{m},k}+\nu_{i,k}$, $k \in [K], j \in \{2,3,\ldots,d\}$, and $(\nu_{i,1}, \ldots, \nu_{i,K})^{\top}  \sim_{\text {i.i.d }} \mathcal{N}\left(\boldsymbol{0}, \Sigma(1-\rho^2)\right)$ with $\rho=0.9$ and $\rho_0 = 0.1$.

\medskip
It is clear that $\rho$ regulates the degree of series dependence among original lagged covariates, while $\rho_0$ represents the level of cross-covariate dependence across all $K$ covariates. We refer to $\rho$ as the autocorrelation strength of original lags and $\rho_0$ as cross-covariate correlation strength. Inspired by the empirical application, we consider two more scenarios that allow the covariates to have heavy tails. In these cases, we first initiate the processes with $\left(z_{i,\frac{1}{m},1}, \ldots, z_{i,\frac{1}{m},K}\right)^{\top} \sim$ student-$t(2)$ with the scale matrix $\Sigma_{u, v}=\rho_{0}^{|u-v|}, u,v \in [K]$. Then, the third and fourth scenarios are as follows.

\medskip
\refstepcounter{scenario} 
\noindent\textit{Scenario \thescenario}:\label{s3} $z_{i, \frac{j}{m},k}=\rho z_{i,\frac{j-1}{m},k}+\nu_{i,k}$, $k \in [K], j \in \{2,3,\ldots,d\}$, and $(\nu_{i,1}, \ldots, \nu_{i,K})^{\top}  \sim_{\text {i.i.d}}$ student-$t$ with degree $2$ and its scale matrix $\Sigma(1-\rho^2)$, with $\rho=0.1$ and $\rho_0 = 0.1$.

\medskip
\refstepcounter{scenario} 
\noindent\textit{Scenario \thescenario}:\label{s4} $z_{i, \frac{j}{m},k}=\rho z_{i,\frac{j-1}{m},k}+\nu_{i,k}$, $k \in [K], j \in \{2,3,\ldots,d\}$, and $(\nu_{i,1}, \ldots, \nu_{i,K})^{\top}  \sim_{\text {i.i.d}}$ student-$t$ with degree $2$ and its scale matrix $\Sigma(1-\rho^2)$, with $\rho=0.9$ and $\rho_0 = 0.1$.
\medskip

Furthermore, we examine how strong cross-covariate correlation affects the prediction and introduce the following scenario with an increased $\rho_0$.

\medskip
\refstepcounter{scenario} 
\noindent
\textit{Scenario \thescenario}:\label{s5}  
$z_{i, \frac{j}{m},k}=\rho z_{i,\frac{j-1}{m},k}+\nu_{i,k}$, $k \in [K], j \in \{2,3,\ldots,d\}$, and $(\nu_{i,1}, \ldots, \nu_{i,K})^{\top}  \sim_{\text {i.i.d }} \mathcal{N}\left(\boldsymbol{0}, \Sigma(1-\rho^2)\right)$ with $\rho=0.9$ and $\rho_0 = 0.9$.

\medskip

It is important to note that these scenarios generate raw lagged covariates that enter directly into LASSO-UMIDAS, whereas LASSO-MIDAS and sg-LASSO-MIDAS require a MIDAS weighting transformation. Section \ref{dis midas} discusses how this transformation affects these scenarios and provides a detailed comparison across them.

To generate $T_i$, we first transform the covariates to their absolute values, which ensures that the distribution functions will be increasing in $t$ for all $\boldsymbol{Z}$. Then we let \begin{equation*}
T_{i} = s + \exp \left(\frac{\log(\frac{\mathbf{\zeta}}{1-\mathbf{\zeta}}) -\left(1 + \sum\limits_{j\in[d]} \widetilde{\omega}_1\left(\frac{j-1}{d}\right) \left|z_{i, s - \frac{j-1}{m}, 1}\right| -\sum\limits_{j\in[d]}
 \widetilde{\omega}_2\left(\frac{j-1}{d}\right) \left|z_{i, s - \frac{j-1}{m}, 2}\right|\right)}{1 + \sum_{k=1}^{2} \sum\limits_{j\in[d]}\widetilde{\omega}_k\left(\frac{j-1}{d}\right) \left|z_{i, s - \frac{j-1}{m}, k}\right|}\right),
\end{equation*}
for all $i\in[N]$, 
where $\mathbf{\zeta} \sim \text{Uniform}(0,1)$.
The weighting schemes $\widetilde{\omega}_{k}(u), u \in [0,1]$ for $k =1, 2$ correspond to $\boldsymbol{\operatorname{Beta}}(1,3)$ and $\boldsymbol{\operatorname{Beta}}(2,3)$ densities, respectively; see \citet{ghysels2007midas,ghysels2019estimating,babii2022machine} for further details. This generation scheme guarantees that conditional distribution function of $T$ satisfies \eqref{distributin}, where $ \boldsymbol{\theta}_0(t,s) $ is such that $ \boldsymbol{\theta}_{0,1}(t,s) =1 + \log(t-s) $, $\boldsymbol{\theta}_{0,1+j}(t,s) =(1+\log(t-s))\widetilde{\omega}_1\left(\frac{j-1}{d}\right),\ j\in[d]$, $\boldsymbol{\theta}_{0,1+d+j}(t,s) =(\log(t-s)-1)\widetilde{\omega}_2\left(\frac{j-1}{d}\right),\ j\in[d]$ and $\boldsymbol{\theta}_{0,k}(t,s) =0$ for all $k\in\{2d+2,\dots,K_z\}$. Note that only the first two high-frequency covariates are relevant. 
The censoring time $C_i$, for $i \in [N]$, is generated from a shifted exponential distribution. Specifically, we let $C_i = s + \text{Exp}_i$, where $\text{Exp}_i \sim \mathrm{Exp}(\gamma)$ is an exponential random variable with rate parameter $\gamma > 0$. The parameter $\gamma$ is chosen so that the resulting censoring rate, $(1/N)\sum_{i=1}^N \mathbbm{1}\{T_i > C_i\}$, is approximately $81\%$ in the simulated dataset, matching the censoring rate observed in the real dataset (see Section~\ref{real data}).

For the choice of the MIDAS weight function $W$ in the LASSO-MIDAS and sg-LASSO-MIDAS models, we use a dictionary of orthogonal polynomials shifted to the interval $[0,1]$, with size $L=3$, given by $W=\{w_1(u), w_2(u), w_3(u)\}$ for $u\in[0,1]$. The dictionary is constructed using a Gram–Schmidt orthogonalization of the power polynomials $\{1,x,x^2\}$ with respect to the measure $\mathrm{d}\mu(x)=(1-x)^{\alpha_{\text{poly}}}(1+x)^{\beta_{\text{poly}}}\mathrm{d}x$ on $[-1,1]$ where  $\alpha_{\text{poly}} = \beta_{\text{poly}} =-1/2$.
To obtain the basis on the unit interval, we apply the transformation $x=2u-1$, which yields the shifted polynomials $\{w_1(u), w_2(u), w_3(u)\}$. This construction is commonly referred to as a Gegenbauer polynomial dictionary \citep{babii2022machine} in the MIDAS literature.\footnote{The function \texttt{gb} in the \texttt{midasml} R package computes an orthonormal polynomial basis for a given dictionary size and parameters $\alpha_{\text{poly}}$ and $\beta_{\text{poly}}$; see Online Appendix~A of \citet{babii2022machine}.}
The use of orthogonal polynomial dictionaries helps reduce multicollinearity and improve numerical stability \citep{babii2022machine}, though it may also change the variance of the MIDAS-weighted covariates, as discussed in Section~\ref{dis midas}.

Regarding $t$, we set it to the following percentiles $t = \{t_1 = 10^{\text{th}}, t_2 = 30^{\text{th}}, t_3 = 50^{\text{th}}\}$ of the distribution of $\{T_i: T_i \text{ is uncensored}, i \in [N]\}$.

Concerning the evaluation of classification performance, Receiver Operating Characteristic (ROC) curves are widely used in the literature. However, traditional ROC curves are not fully suitable in this context due to censoring, where the status of firms is only partially observed. To address this limitation, we use the ROC curve estimator developed by \citet{heagerty2000time}, which was specifically designed to evaluate classification performance effectively in the presence of censoring.

\subsection{Evaluation metric: ROC curves with censoring}\label{{time depedent auc}}

Recalling the definitions of sensitivity and specificity in the ROC curves, we see that in our model, both sensitivity, or the “true positive rate” (TPR), and specificity, or the “true negative rate” (TNR), are also functions that depend on $t$:
\begin{equation}\label{metric}
\begin{aligned}
Se(c, t)=P[ \Upsilon_i > c \mid T_i \leq t], \quad Sp(c, t)=P[ \Upsilon_i \leq c \mid T_i > t],
\end{aligned}
\end{equation}
where $\Upsilon_i:= p\left(\widehat{\boldsymbol{\beta}}, \boldsymbol{X}_i\right) = \frac{\exp \left(\boldsymbol{X}_i^{\top}\widehat{\boldsymbol{\beta}}\right)}{1+\exp \left(\boldsymbol{X}_i^{\top}\widehat{\boldsymbol{\beta}}\right)}$ is the estimated probability.\footnote{When defining $\Upsilon_i$, we treat $\widehat{\beta}$ as fixed because it is estimated on the training set. The probabilities in $Se(c, t)$ and $Sp(c, t)$ are over the distribution of the test set.} The threshold $c$ is used to classify a firm as distressed if $\Upsilon_i > c $, or as non-distressed if $\Upsilon_i \leq c$, with $\mathbbm{1}\{T \leq t\}$ indicating whether the firm has failed by time $t$.

A ROC curve illustrates the full range of True Positive Rates (TPR) and False Positive Rates (FPR = $1 -$ TNR ) across all possible threshold values $c$. A larger area under the ROC curve (AUC) signifies better performance in distinguishing between firms that have failed and those that have not. In practice, the status $\mathbbm{1}\{T \leq t\}$ in \eqref{metric} cannot be fully observed due to censoring. To address this issue, various ROC curve estimators have been proposed in \citet{heagerty2000time,cai2006sensitivity,heagerty2005survival,10.1093/biomet/asaa080}. 

Here, we employ the Nearest Neighbor estimator \citep{heagerty2000time} to account for the censored data and evaluate the ROC curves. Let
$$
\widehat{S}_{\kappa_N}(c, t)=\frac{1}{N} \sum_{i=1}^N \widehat{S}_{\kappa_N}\left(t \mid \Upsilon_i\right) \mathbbm{1}\{\Upsilon_i>c\},
$$ where $\widehat{S}_{\kappa_N}\left(t \mid \Upsilon_i\right)$ is a suitable estimator of the  conditional survival function characterized by a parameter $\kappa_N$:
$$
\widehat{S}_{\kappa_N}\left(t \mid \Upsilon_i\right)=\prod_{a \in \mathcal{T}_N, a \leq t}\left\{1-\frac{\sum_j 
 \Psi_{\kappa_N}\left(\Upsilon_j, \Upsilon_i\right) \mathbbm{1}\{\widetilde{T}_j=a\} \delta_j}{\sum_j \Psi_{\kappa_N}\left(\Upsilon_j, \Upsilon_i\right) \mathbbm{1}\{\widetilde{T}_j \geq a\}}\right\},
$$
where $\mathcal{T}_N$ is a set of the unique values of $\widetilde{T}_i$ for observed events, $\delta_i = \mathbbm{1}\{T_i \leq C_i\}$ and $\Psi_{\kappa_N}\left(\Upsilon_j, \Upsilon_i\right)$ is a kernel function that depends on a smoothing parameter $\kappa_N$. Following the approach in \citet{heagerty2000time}, we use a $0 / 1$ nearest neighbor kernel \citep{akritas1994nearest}, $\Psi_{\kappa_N}\left(\Upsilon_j, \Upsilon_i\right)= \mathbbm{1}\{-\kappa_N<\widehat{F}_{\Upsilon}\left(\Upsilon_i\right)-\widehat{F}_{\Upsilon}\left(\Upsilon_j\right)<\kappa_N\}$ where $\widehat{F}_{\Upsilon}(\cdot)$ is the empirical distribution function of ${\Upsilon}$ and $2 \kappa_N \in(0,1)$ represents the percentage of individuals that are included in each neighborhood (boundaries). The resulting sensitivity and specificity are defined by:
$$
\begin{aligned}
\widehat{S e}(c, t)=\frac{\left(1-\widehat{F}_{\Upsilon}(c)\right)-\widehat{S}_{\kappa_N}(c, t)}{1-\widehat{S}_{\kappa_N}(t)}, \quad \widehat{S p}(c, t)=1-\frac{\widehat{S}_{\kappa_N}(c, t)}{\widehat{S}_{\kappa_N}(t)},
\end{aligned}
$$
where $\widehat{S}_{\kappa_N}(t)=\widehat{S}_{\kappa_N}(-\infty, t)$. Both sensitivity and specificity above are monotone and bounded in $[0,1]$. 

\citet{heagerty2000time} used bootstrap resampling to estimate the confidence intervals for this ROC curve estimator. Motivated by the results of \citet{akritas1994nearest} and \citet{cai2011robust}, \citet{hung2010optimal} discussed the asymptotic properties of the estimator and concluded that bootstrap resampling techniques can be used to estimate the variances of the proposed ROC curve. In practice, \citet{heagerty2000time} suggested that the value for $\kappa_N$ is chosen to be $O\left(N^{-\frac{1}{3}}\right)$. In the present paper, we use the default value of the $\kappa_N$ provided in the documentation of the R package 'SurvivalROC', which is consistent with the choice found in \citet{blanche2013estimating}. For further details on other ROC curve estimators in the survival analysis, we refer to \citet{kamarudin2017time}.

\subsection{Discussion on the MIDAS-weighted covariates}\label{dis midas}
In the simulation scenarios, increasing the autocorrelation strength $\rho$ of the original lags does not necessarily amplify the correlation within each group of MIDAS-weighted covariates. For example, although Scenario~\ref{s2} features stronger correlation among the original lagged covariates, the correlation matrix of the corresponding MIDAS-weighted covariates is nearly identical to that in Scenario~\ref{s1} (see Figure~\ref{correlation}). This is because the orthogonal weighting functions in the MIDAS dictionary $W$ redistribute the dependence structure along orthogonal directions, so higher autocorrelation in the raw lags does not translate into stronger correlation among the transformed covariates.

Instead, the transformation could inflate the variance of these MIDAS-weighted covariates.
 To illustrate this, consider the $k$-th covariate with its $d$ high-frequency lags \[\widetilde{\boldsymbol{Z}}_{i,k} = \left(z_{i, s, k}, z_{i, s - \frac{1}{m}, k}, \ldots, z_{i, s - \frac{d-1}{m}, k}\right)^{\top},\]
 where each component has variance $1$ in both Scenarios \ref{s1} and \ref{s2}. After multiplying by the MIDAS weighting matrix $W = \left(w_1, \ldots, w_L\right) \in \mathbb{R}^{d \times L}$, with $w_l \in \mathbb{R}^{d}$, we obtain the $k$-th group of MIDAS-weighted covariates $$W^{\top}\widetilde{\boldsymbol{Z}}_{i,k} = \left(\widetilde{\boldsymbol{Z}}_{i,k}^{\top}w_1, \ldots,  \widetilde{\boldsymbol{Z}}_{i,k}^{\top}w_L \right)^{\top},$$
which is used in estimation. Unlike the original lags, the components of this transformed vector generally do not have variance equal to $1$. For example,
$$Var\left(\widetilde{\boldsymbol{Z}}_{i,k}^{\top}w_1\right) = w_1^{\top}Cov\left(\widetilde{\boldsymbol{Z}}_{i,k} \right)w_1,$$ 
which depends on $W$ and the covariance structure of $\left(z_{i, s, k}, z_{i, s - \frac{1}{m}, k}, \ldots, z_{i, s - \frac{d-1}{m}, k}\right)^{\top}$, and hence varies with $\rho$ in the simulation design. In Scenarios~\ref{s1} and~\ref{s2}, the stronger dependence among the original lagged covariates leads to larger variance in the corresponding MIDAS–weighted covariates.\footnote{This pattern depends jointly on the covariance structure of $\widetilde{\boldsymbol{Z}}_{i,k}$ and the choice of the weighting matrix $W$. For instance, under a Toeplitz covariance structure combined with Chebyshev-based weights, specific patterns arise. In empirical applications, the true covariance structure of covariates is always unknown, so whether MIDAS necessarily increases their variance remains an open question.} Table \ref{ill} reports the average variance ratio of the $L$ elements in the first group ($k=1$) of MIDAS-weighted covariates between Scenario \ref{s2} and Scenario \ref{s1} based on the simulation, with similar patterns observed for other groups. Thus, a strong correlation among the raw lagged covariates increases the variance of the MIDAS-weighted covariates. This increase in the variance strengthens the signal relative to noise, thereby enhancing the predictive performance of MIDAS-based methods.
\begin{figure}[!htbp]
    \centering
    \includegraphics[width=0.7\linewidth]{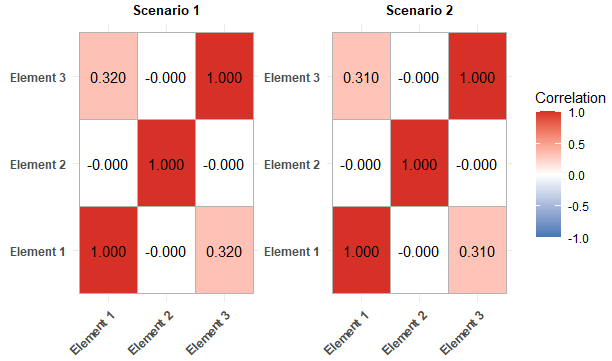}
    \caption{Correlation matrix of the $L = 3$ elements in the first group ($k=1$) of MIDAS-weighted covariates in Scenarios \ref{s1} and \ref{s2} for $N = 800$. Results are based on $100$ simulation repetitions, and the MIDAS weighting matrix $W$ is described in Section~\ref{est simulation}.}
    \label{correlation}
\end{figure}
\begin{table}[!htbp]
\centering
  \caption{Average variance ratio of the $L = 3$ elements in the first group ($k=1$) of MIDAS-weighted covariates between Scenario \ref{s2} and Scenario \ref{s1} for $N = 800$. Results are based on $100$ simulation repetitions, and the MIDAS weighting matrix $W$ is described in Section~\ref{est simulation}.
}
  \medskip
    \begin{tabular}{lccc}
    \toprule
          & \multicolumn{1}{l}{Element $1$} & \multicolumn{1}{l}{Element $2$} & \multicolumn{1}{l}{Element $3$} \\
\cmidrule{2-4}    Variance ratio & 7.07  & 4.44  & 2.63 \\
    \bottomrule
    \end{tabular}%
  \label{ill}%
\end{table}%

Notably, the MIDAS transformation is applied within each group, that is, to each original covariate and its lags, and does not change the correlation structure across different groups of covariates. Hence, we can evaluate how strong cross-covariate correlation affects the prediction of MIDAS-based methods by comparing Scenario \ref{s2} and Scenario \ref{s5}.

\subsection{Prediction simulation results} \label{sec sim res}
We compute results for the three different LASSO-type regression methods. In the structured approach, sg-LASSO-MIDAS, each covariate and its high-frequency lags share the same group. Therefore, we have $K+1$ groups (one group corresponding to the intercept). Table \ref{number est} presents the number of parameters (including the intercept) to be estimated in each of the three methods. It is evident that the two methods using MIDAS weights help mitigate the high-dimensional problem when $s \times m$ exceeds $L$.
\begin{table}[htbp]
  \centering
  \caption{Number of parameters to be estimated in different methods.}
  \medskip
  \scalebox{1}{
    \begin{tabular}{cc}
    \toprule
    Methods & Number of estimated parameters \\
    \cmidrule{1-2}
    LASSO-UMIDAS & $1 + K \times s \times m$ \\
    LASSO-MIDAS & $1 + K \times L$ \\
    sg-LASSO-MIDAS & $1 + K \times L$ \\
    \bottomrule
    \end{tabular}%
    }
  \label{number est}%
\end{table}%
We start by comparing the prediction results for sample sizes $N \in \{800, 1200\}$ across all simulation scenarios, followed by examining the recovery of the MIDAS weight function. To assess the prediction performance, we randomly split the simulated dataset into a training dataset ($80\%$) and a test dataset ($20\%$), ensuring that both sets maintain the same proportion of the event indicator $\delta_i(t) \mathbbm{1}{\{\widetilde{T}_i \leq t\}}$. We then calculate the estimated AUC in the test dataset, with the average estimated AUC obtained from $100$ simulated datasets for each sample size. The tuning parameters in the sg-LASSO-MIDAS and LASSO-MIDAS models are selected using $5$-fold stratified cross-validation to maximize the AUC on the validation fold, and the same procedure is applied to the LASSO-UMIDAS model. Specifically, we perform a grid search over the regularization parameter $\alpha$ in the sparse-group LASSO penalty, with values in the set $\{0, 0.1, 0.3, 0.5, 0.7, 0.9, 1\}$ and, as standard, $\lambda$ is chosen in a grid which follows \citet{liang_sparsegl}.
\begin{table}[htp!]
  \centering
  \caption{Estimated average AUCs (standard deviation) in the test dataset of the three different methods: LASSO-UMIDAS (LASSO-U), LASSO-MIDAS (LASSO-M), sg-LASSO-MIDAS (sg-LASSO-M). We set $s = 6$ and $t = \{t_1 = 10^{\text{th}}, t_2 = 30^{\text{th}}, t_3 = 50^{\text{th}}\}$ to denote the corresponding percentiles of the set $\{T_i: T_i \text{ is uncensored }, i \in [N]\}$. We use $\rho$ to measure the autocorrelation strength of the original lags, while $\rho_0$ measures the cross-covariate correlation. }
  \medskip
    \scalebox{0.9}{\begin{tabular}{cccccccc}
    \toprule
          & \multicolumn{7}{c}{Scenario \ref{s1}: $\rho = 0.1$, $\rho_0 = 0.1$} \\
\cmidrule{2-8}          
          & \multicolumn{3}{c}{$N = 800$} &       & \multicolumn{3}{c}{$N = 1200$} \\
\cmidrule{2-4}\cmidrule{6-8}          
          & $t = t_1$ & $t = t_2$ & $t = t_3$ &       & $t = t_1$ & $t = t_2$ & $t = t_3$ \\
\cmidrule{2-4}\cmidrule{6-8}    
Oracle      & 0.974 (0.026) & 0.913 (0.038) & 0.853 (0.043) 
            &       & 0.965 (0.028) & 0.927 (0.033) & 0.886 (0.036) \\
LASSO-U     & 0.575 (0.181) & 0.578 (0.129) & 0.576 (0.102) 
            &       & 0.604 (0.141) & 0.648 (0.097) & 0.635 (0.090) \\
LASSO-M     & 0.829 (0.139) & 0.864 (0.076) & 0.813 (0.071) 
            &       & 0.886 (0.096) & 0.891 (0.051) & 0.866 (0.058) \\
sg-LASSO-M  & 0.867 (0.116) & 0.888 (0.068) & 0.846 (0.075)
            &       & 0.917 (0.061) & 0.914 (0.043) & 0.889 (0.051) \\
\cmidrule{2-8}

          & \multicolumn{7}{c}{Scenario \ref{s2}: $\rho = 0.9$, $\rho_0 = 0.1$} \\
\cmidrule{2-8}          
          & \multicolumn{3}{c}{$N = 800$} &       & \multicolumn{3}{c}{$N = 1200$} \\
\cmidrule{2-4}\cmidrule{6-8}          
          & $t = t_1$ & $t = t_2$ & $t = t_3$ &       & $t = t_1$ & $t = t_2$ & $t = t_3$ \\
\cmidrule{2-4}\cmidrule{6-8}    
Oracle      & 0.986 (0.016) & 0.955 (0.020) & 0.916 (0.029) 
            &       & 0.988 (0.011) & 0.968 (0.015) & 0.941 (0.019) \\
LASSO-U     & 0.824 (0.142) & 0.870 (0.067) & 0.845 (0.075) 
            &       & 0.872 (0.095) & 0.908 (0.043) & 0.899 (0.042) \\
LASSO-M     & 0.903 (0.081) & 0.921 (0.043) & 0.905 (0.056) 
            &       & 0.920 (0.058) & 0.943 (0.027) & 0.932 (0.031) \\
sg-LASSO-M  & 0.905 (0.081) & 0.937 (0.041) & 0.922 (0.046) 
            &       & 0.930 (0.048) & 0.947 (0.027) & 0.944 (0.029) \\
\cmidrule{2-8}

          & \multicolumn{7}{c}{Scenario \ref{s3}: $\rho = 0.1$, $\rho_0 = 0.1$, heavy-tailed covariates} \\
\cmidrule{2-8}          
          & \multicolumn{3}{c}{$N = 800$} &       & \multicolumn{3}{c}{$N = 1200$} \\
\cmidrule{2-4}\cmidrule{6-8}          
          & $t = t_1$ & $t = t_2$ & $t = t_3$ &       & $t = t_1$ & $t = t_2$ & $t = t_3$ \\
\cmidrule{2-4}\cmidrule{6-8}    
Oracle      & 0.985 (0.011) & 0.943 (0.025) & 0.890 (0.034) 
            &       & 0.989 (0.009) & 0.961 (0.018) & 0.917 (0.028) \\
LASSO-U     & 0.600 (0.191) & 0.588 (0.123) & 0.564 (0.099) 
            &       & 0.622 (0.150) & 0.627 (0.097) & 0.572 (0.089) \\
LASSO-M     & 0.770 (0.188) & 0.804 (0.105) & 0.760 (0.101) 
            &       & 0.835 (0.121) & 0.860 (0.077) & 0.828 (0.082) \\
sg-LASSO-M  & 0.782 (0.184) & 0.825 (0.096) & 0.793 (0.105)
            &       & 0.845 (0.123) & 0.881 (0.088) & 0.851 (0.093) \\
\cmidrule{2-8}    

          & \multicolumn{7}{c}{Scenario \ref{s4}: $\rho = 0.9$, $\rho_0 = 0.1$, heavy-tailed covariates} \\
\cmidrule{2-8}          
          & \multicolumn{3}{c}{$N = 800$} &       & \multicolumn{3}{c}{$N = 1200$} \\
\cmidrule{2-4}\cmidrule{6-8}          
          & $t = t_1$ & $t = t_2$ & $t = t_3$ &       & $t = t_1$ & $t = t_2$ & $t = t_3$ \\
\cmidrule{2-4}\cmidrule{6-8}    
Oracle      & 0.991 (0.008) & 0.972 (0.018) & 0.945 (0.025) 
            &       & 0.993 (0.007) & 0.984 (0.010) & 0.964 (0.013) \\
LASSO-U     & 0.777 (0.166) & 0.815 (0.096) & 0.787 (0.102) 
            &       & 0.811 (0.122) & 0.870 (0.057) & 0.850 (0.060) \\
LASSO-M     & 0.799 (0.190) & 0.854 (0.097) & 0.841 (0.097) 
            &       & 0.845 (0.117) & 0.901 (0.059) & 0.866 (0.070) \\
sg-LASSO-M  & 0.818 (0.176) & 0.864 (0.088) & 0.842 (0.098)
            &       & 0.848 (0.121) & 0.910 (0.052) & 0.882 (0.075) \\
\cmidrule{2-8} 

          & \multicolumn{7}{c}{Scenario \ref{s5}: $\rho = 0.9$, $\rho_0 = 0.9$} \\
\cmidrule{2-8}          
          & \multicolumn{3}{c}{$N = 800$} &       & \multicolumn{3}{c}{$N = 1200$} \\
\cmidrule{2-4}\cmidrule{6-8}          
          & $t = t_1$ & $t = t_2$ & $t = t_3$ &       & $t = t_1$ & $t = t_2$ & $t = t_3$ \\
\cmidrule{2-4}\cmidrule{6-8}    
Oracle      & 0.977 (0.022) & 0.933 (0.031) & 0.870 (0.042) 
            &       & 0.974 (0.022) & 0.943 (0.025) & 0.903 (0.033) \\
LASSO-U     & 0.686 (0.161) & 0.745 (0.096) & 0.719 (0.093) 
            &       & 0.749 (0.123) & 0.796 (0.088) & 0.778 (0.072) \\
LASSO-M     & 0.751 (0.143) & 0.805 (0.101) & 0.757 (0.087) 
            &       & 0.829 (0.105) & 0.850 (0.056) & 0.820 (0.065) \\
sg-LASSO-M  & 0.756 (0.146) & 0.824 (0.102) & 0.765 (0.094)
            &       & 0.843 (0.098) & 0.869 (0.052) & 0.847 (0.059) \\
    \bottomrule
    \end{tabular}}%
  \label{AUC6}%
\end{table}

Table \ref{AUC6} reports the estimated average AUCs in the test dataset. To understand how different simulation scenarios affect the performance of a given method, we highlight two key factors. First, different simulation scenarios correspond to distinct models, as their true coefficients differ due to the prediction horizons $t$ being selected based on the generated data. For this reason, we report oracle AUCs for each scenario, computed by evaluating the AUC at the true data-generating parameters. The corresponding results have been added to Table~\ref{AUC6}. Second, differences in covariate correlation strength across scenarios may affect the accuracy with which the true parameters are estimated. In what follows, we separately discuss the impact of covariate correlation on MIDAS-based methods (such as LASSO-MIDAS and sg-LASSO-MIDAS), drawing on the analysis in Section~\ref{dis midas}, and on LASSO-UMIDAS.

We first focus on LASSO-MIDAS and sg-LASSO-MIDAS. As explained in Section \ref{dis midas}, it is expected that LASSO-MIDAS and sg-LASSO-MIDAS perform better in Scenario \ref{s2} than in Scenario \ref{s1}.
The oracle AUCs in Scenario \ref{s2} are even higher than those in Scenario \ref{s1} across different prediction horizons, which further explains the superior performance of LASSO-MIDAS and sg-LASSO-MIDAS. Similar patterns arise when comparing Scenario \ref{s3} and Scenario \ref{s4}, with the only difference being that the covariates follow heavy-tailed distributions. As expected, Scenario~\ref{s5} features stronger correlations across groups of MIDAS-weighted covariates, leading to noticeably weaker performance of LASSO-MIDAS and sg-LASSO-MIDAS compared with Scenario~\ref{s2}. The lower oracle AUCs in Scenario~\ref{s5} further reinforce this pattern. Although Scenarios~\ref{s3} and~\ref{s4} exhibit the highest oracle AUCs overall, LASSO-MIDAS and sg-LASSO-MIDAS perform worse there than in Scenarios~\ref{s1} and~\ref{s2} due to the heavy-tailed covariates. 

As for LASSO-UMIDAS, it exhibits performance patterns similar to those of LASSO-MIDAS and sg-LASSO-MIDAS across the scenarios. It attains a higher AUC in Scenario~\ref{s2} than in Scenario~\ref{s1}, partly reflecting the higher oracle AUCs. This is intuitive: Scenario~\ref{s2} features stronger within-group correlations among the active covariates but similarly low cross-group correlations, which allows LASSO-UMIDAS to approximate the predictive signal by selecting a few representative covariates in the active groups and therefore yields more stable predictions. In contrast, under the weaker within-group correlation of Scenario~\ref{s1}, the signal is more diffusely distributed across the active covariates, leading LASSO-UMIDAS to select too few predictors and produce less stable predictions with lower AUC. To quantify differences in signal strength across scenarios, Table~\ref{tab:snr_cross_t} reports the relative signal-to-noise ratios (SNRs) for LASSO-UMIDAS (shown for \(N = 800\), with similar results for \(N = 1200\)). We measure the SNR as $3\operatorname{Var}\!\left(\boldsymbol{Z}^{\top}\boldsymbol{\theta}_0\right)/\pi^2
$ using the standard normalization that the logistic error has variance $\pi^2/3$.
As noted by \citet{chardon2024finite}, stronger signal strength facilitates classification. The substantially higher SNRs in Scenario~\ref{s2} explain why LASSO-UMIDAS achieves more accurate predictions there compared with Scenario~\ref{s1}.
Although we do not report SNRs for Scenarios~\ref{s3} and~\ref{s4}, the same qualitative pattern is observed: these scenarios differ from Scenarios~\ref{s1} and~\ref{s2} only in using heavy-tailed covariates, and the resulting predictions follow a similar trend. Table~\ref{tab:snr_cross_t} further supports that LASSO-UMIDAS performs better in Scenario~\ref{s5} than in Scenario~\ref{s1}, as the SNRs are higher in Scenario~\ref{s5}, where the strong correlation among active lags appears to be particularly beneficial. Although high within-group correlation among active covariates tends to improve prediction, strong cross-group correlation between active and inactive groups can diminish predictive performance. The decrease in AUCs from Scenario~\ref{s2} to Scenario~\ref{s5} reflects the effect of stronger cross-group dependence. When comparing Scenarios~\ref{s1} and~\ref{s2} with Scenarios~\ref{s3} and~\ref{s4}, the degraded performance in the latter can be attributed to the heavy-tailed covariates. For $N = 800$, LASSO-UMIDAS performs similarly in Scenarios~\ref{s1} and~\ref{s3}, but the AUCs across all three prediction horizons are consistently slightly closer to the oracle in Scenario~\ref{s1}, indicating a modest advantage in this scenario. LASSO-UMIDAS also consistently outperforms in Scenario~\ref{s2} compared with Scenario~\ref{s4}, as expected. When the sample size increases to $N = 1200$, the same pattern persists: LASSO-UMIDAS shows better performance in Scenarios~\ref{s1} and~\ref{s2}, while Scenarios~\ref{s3} and~\ref{s4} remain affected by heavy-tailed covariates.

\begin{table}[htbp]
  \centering
  \caption{Relative signal-to-noise ratios for LASSO-UMIDAS with Gaussian covariates ($N=800$). Scenario~\ref{s1} values are set to $1$ as benchmarks, and other scenarios are normalized by their respective horizons. Results are averaged over $100$ simulation repetitions.}
  \medskip
  \scalebox{1}{
    \begin{tabular}{cccc}
    \toprule
    Prediction horizon & Scenario \ref{s1} & Scenario \ref{s2} & Scenario \ref{s5} \\
    \midrule
    $t_1$  & 1.00 & 7.48 & 2.59 \\
    $t_2$  & 1.00 & 6.68 & 2.00 \\
    $t_3$  & 1.00 & 6.24 & 1.78 \\
    \bottomrule
    \end{tabular}%
    }
  \label{tab:snr_cross_t}%
\end{table}%

For general context, as discussed in \citet{hebiri2012correlations,lassoprediction}, the predictive performance of Lasso-type methods can remain stable under various levels of correlation provided that the penalty parameter is appropriately tuned, for example, through cross-validation. Moreover, since the five scenarios correspond to different underlying data-generating models, our primary objective is to compare how the MIDAS-based methods perform relative to LASSO-UMIDAS within each scenario. 

We therefore focus on method comparisons within the same scenario. As shown, sg-LASSO-MIDAS achieves the highest AUCs across different simulation scenarios. Both sg-LASSO-MIDAS and LASSO-MIDAS, using weight function approximations, outperform LASSO-UMIDAS. LASSO-UMIDAS generally demonstrates the poorest predictive performance across all scenarios. As expected, the predictive performance improves with an increase in sample size $N$. These results remain robust as the autocorrelation of original lagged covariates $\rho$ increases from $0.1$ to $0.9$ or the cross-covariate correlation $\rho_0$ increases from $0.1$ to $0.9$. Although all three methods perform less effectively with heavy-tailed covariates, sg-LASSO-MIDAS continues to outperform the others. 
In Tables \ref{recovery1}, \ref{recovery2}, \ref{recovery3}, \ref{recovery4}, and \ref{recovery5} of Online Appendix \ref{additional}, we report additional results for the estimation accuracy of the true parameters. 
It is worth noting that the increase of the parameter estimation accuracy with the sample size is not particularly large, as the high censoring rate in the simulated datasets limits the average increase in the number of uncensored firms ($\mathbbm{1}\{T_i \leq C_i\} = 1$) to only about $76$ as the sample size $N$ grows from $800$ to $1200$.\footnote{In simulation results not shown for brevity, when the dataset's censoring rate is approximately $30\%$, notable improvements both in parameter estimation and estimated AUC are observed as the sample size $N$ increases from $800$ to $1200$.}

In addition to the prediction performance, we examine the variable selection performance of sg-LASSO-MIDAS across different scenarios. 
\begin{table}[htbp]
  \centering
  \caption{Average True Positive Rate for sg-LASSO-MIDAS. We set $s = 6$ and $t = \{t_1 = 10^{\text{th}}, t_2 = 30^{\text{th}}, t_3 = 50^{\text{th}}\}$ percentile of the set $\{T_i: T_i \text{ is uncensored }, i \in [N]\}$. We use $\rho$ to measure the autocorrelation strength of the original lags and $\rho_0$ to measure the cross-covariate correlation.}
  \medskip
    \scalebox{1}{\begin{tabular}{cccccccc}
    \toprule       
          & \multicolumn{3}{c}{$N = 800$} &       & \multicolumn{3}{c}{$N = 1200$} \\
\cmidrule{2-4}\cmidrule{6-8}          
          & $t = t_1$ & $t = t_2$ & $t = t_3$ &       & $t = t_1$ & $t = t_2$ & $t = t_3$ \\
\cmidrule{2-4}\cmidrule{6-8} 
          & \multicolumn{7}{c}{Scenario \ref{s1}: $\rho = 0.1$, $\rho_0 = 0.1$} \\
\cmidrule{2-8}          
sg-LASSO-M   & 0.895 & 0.985 & 0.980 &       & 0.995 & 1.000 & 0.995 \\
\cmidrule{2-8}          
          & \multicolumn{7}{c}{Scenario \ref{s2}: $\rho = 0.9$, $\rho_0 = 0.1$} \\
\cmidrule{2-8}          
sg-LASSO-M     & 0.925 & 1.000 & 1.000 &       & 0.985 & 1.000 & 1.000 \\
\cmidrule{2-8}          
          & \multicolumn{7}{c}{Scenario \ref{s3}: $\rho = 0.1$, $\rho_0 = 0.1$, heavy-tailed covariates} \\
\cmidrule{2-8}          
sg-LASSO-M    & 0.720 & 0.935 & 0.920 &       & 0.850 & 0.965 & 0.955 \\
\cmidrule{2-8}      
          & \multicolumn{7}{c}{Scenario \ref{s4}: $\rho = 0.9$, $\rho_0 = 0.1$, heavy-tailed covariates} \\
\cmidrule{2-8}          
sg-LASSO-M    & 0.850 & 0.970 & 0.960 &       & 0.955 & 0.995 & 0.990 \\
\cmidrule{2-8} 
          & \multicolumn{7}{c}{Scenario \ref{s5}: $\rho = 0.9$, $\rho_0 = 0.9$} \\
\cmidrule{2-8}          
sg-LASSO-M    & 0.895 & 0.970 & 0.930 &       & 0.960 & 0.995 & 1.000 \\
    \bottomrule
    \end{tabular}}%
  \label{tab:variable_selection_metrics}%
\end{table}
Recall that we have $K = 50$ covariate groups, with the first two groups being active and the remaining $48$ groups being inactive. We define a group as selected if any estimated coefficient in the group is nonzero. 
That is, for group $k$:
\[
A_k = 
\begin{cases}
1, & \text{if any estimated coefficient in group $k$ is nonzero} \\
0, & \text{otherwise}
\end{cases}.
\] 
The performance metrics for variable selection are defined as follows:
\begin{align*}
& \text{True Positive Rate (TPR)} = \frac{A_1 + A_2}{2}.
\end{align*}

Table~\ref{tab:variable_selection_metrics} reports the TPR for sg-LASSO-MIDAS across different scenarios. The pattern of variable selection is similar to the trends observed in the estimated AUCs in Table~\ref{AUC6}. As expected, sg-LASSO-MIDAS achieves its best TPR in Scenario~\ref{s2}, where high within-group correlation is effectively handled, and the SNR is higher as a result of the MIDAS transformation relative to the other scenarios. In contrast, variable selection deteriorates in Scenarios~\ref{s3}, ~\ref{s4}, and~\ref{s5}, due to the presence of heavy-tailed covariates and high cross-covariate correlations. Across all scenarios, performance improves with larger sample sizes, as increased data stabilizes estimation and variable selection.
Overall, the simulation evidence strongly supports the advantage of using MIDAS weighting and incorporating the internal structure of covariates in high-dimensional settings.

\subsection{Inference simulation results}\label{inf sim}
To assess the inference performance, we keep the data generation process in Section \ref{est simulation}, set the first group of MIDAS-weighted parameters as the group of interest, and test 
\begin{equation}\label{test}
\mathbbm{H}_0: \boldsymbol{\beta}_{0,1}^* = \boldsymbol{0}, \quad \mathbbm{H}_1: \boldsymbol{\beta}_{0,1}^* \neq \boldsymbol{0}.
\end{equation}
We implement a Wald test, using a nominal significance level of $5\%$ in this simulation. We scale the Beta density function $\widetilde{\omega}_{1}(u)$ by multiplying it with a constant $a \in \{ 0, 0.1 \}$. The tuning parameters for both sg-LASSO-MIDAS and the nodewise LASSO are selected using $5$-fold cross-validation to maximize the likelihood score on the validation folds. For the regularization parameter $\alpha$ in the sparse-group LASSO, we use a relatively small grid set $\{0, 0.5, 1\}$ to alleviate computational burden. The sample size is $N = 1200$ and the simulation is repeated $500$ times. Unlike for the prediction exercise, we do not split the data here and use the full dataset to estimate. All other settings follow those in Section \ref{est simulation}. 

As shown in Table \ref{inf1}, the proposed de-sparsified sg-LASSO-MIDAS attains an empirical test size close to the nominal $5\%$ level in both Scenario \ref{s1} and Scenario \ref{s2}. The empirical power in Scenario \ref{s2} is higher than in Scenario \ref{s1}, which is consistent with the larger signal-to-noise ratio discussed in Section \ref{dis midas}. In Scenarios~\ref{s3} and~\ref{s4}, which involve heavy-tailed covariates, a similar pattern holds: the method performs better in terms of both size and power in Scenario~\ref{s4}, which has a larger signal-to-noise ratio. Empirical power drops markedly when cross-covariate correlation increases, as in the comparison of Scenario~\ref{s2} and Scenario~\ref{s5}. Comparing Scenario~\ref {s1} with Scenario \ref{s3}, the latter shows higher empirical power but also a substantial deviation of empirical size from the nominal $5\%$ level, indicating notable size inflation. \tcr{Finally, Scenario~\ref{s2} exhibits both higher empirical power and more accurate size than Scenario~\ref{s4} which features heavy-tailed covariates.}
\begin{table}[htbp]
  \centering
  \caption{Inference simulation results: $N=1200$, $s = 6$, and $t = 30^{\text{th}}$ percentile of the set $\{T_i: T_i \text{ is uncensored }, i \in [N]\}$. The results are based on $500$ simulation repetitions.}
  \medskip
  \scalebox{1}{
    \begin{tabular}{lcc}
    \toprule
    Simulation scenario & $a=0$ (empirical size) & $a=0.1$ (empirical power) \\
    \midrule
    Scenario \ref{s1}: $\rho = 0.1$, $\rho_0 = 0.1$ & 0.048 & 0.760 \\
    \midrule
    Scenario \ref{s2}: $\rho = 0.9$, $\rho_0 = 0.1$ & 0.056 & 0.998 \\
    \midrule
    Scenario \ref{s3}: heavy-tailed covariates, $\rho = 0.1$, $\rho_0 = 0.1$ & 0.080 & 0.964 \\
    \midrule
    Scenario \ref{s4}: heavy-tailed covariates, $\rho = 0.9$, $\rho_0 = 0.1$ & 0.064 & 0.978 \\
    \midrule
    Scenario \ref{s5}: $\rho = 0.9$, $\rho_0 = 0.9$ & 0.056 & 0.872 \\
    \bottomrule
    \end{tabular}%
    }
  \label{inf1}%
\end{table}%


\section{Empirical application}\label{sec real}

\subsection{Data}\label{real data}

We construct a dataset of all publicly traded Chinese manufacturing firms listed on the Shanghai and Shenzhen Stock Exchanges. These firms' financial statuses are classified as either Special Treatment (ST) or No-ST.\footnote{The initial public offering (IPO) dates of these firms fall between January $1^{\text{st}}$, $1985$, and December $31^{\text{st}}$, $2015$.} A firm is designated as an ST firm if it meets any of the following criteria: i) two consecutive years of earnings are negative; ii) one recent year of earnings is negative and the most recent year of equity is negative; iii) the most recent year’s audited financial statements conclude with substantial doubt; and iv) other situations identified by the stock exchange as abnormal activities or a high risk of delisting. According to \citet{li2021chinese}, ST status is a reliable indicator of financial distress in China. Therefore, we use the ST indicator as a proxy for a firm's financial distress.

The dataset is sourced from the IFIND database \url{https://www.hithink.com/ifind.html}, one of China's leading financial data providers. The database contains mostly manually
extracted data, covering financial data such as stocks, bonds, funds, futures, and indexes. Detailed information about the dataset can be found in Online Appendix \ref{data details}. Additionally, we have developed an R package, \texttt{Survivalml}, which is publicly available at \url{https://github.com/Wei-M-Wei/Survivalml}, and its instructions are provided as well.

The raw dataset consists of $1614$ companies, of which $299$ were classified as ST and $1315$ as No-ST. The data cover the period from January $1^{\text{st}}, 1985$ to December $31^{\text{st}}, 2020$, with firms entering the sample from their respective IPO dates within this window.\footnote{There are no mergers in the dataset.} The dataset exhibits a censoring rate of approximately $81\%$. We collect $57$ financial variables measured quarterly, categorized into $8$ types (number of covariates in each type), as follows: Operation-Related $(6)$, Debt-Related $(10)$, Profit-Related $(16)$, Potential-Related $(6)$, Z-score Related $(5)$ \citep{altman1968financial}, Capital-Related $(6)$, Stock-Related $(5)$, and Cash-Related $(3)$. Table \ref{variable} provides detailed information on these financial variables; see Online Appendix \ref{data details} for further details. Figure \ref{summarypic} presents the distribution of IPO, first ST designations, and censored firms across different years of the raw dataset. Many of them were listed in $2010$ and $2011$, several firms were publicly listed in $2013$, and the financially distressed firms seem to be distributed evenly between $1999$ and $2020$.\footnote{Since China froze IPOs in $2013$, there were only a limited number of IPO firms in this year.}

\begin{figure}[htbp]
    \centering
    \includegraphics[width=0.8\linewidth]{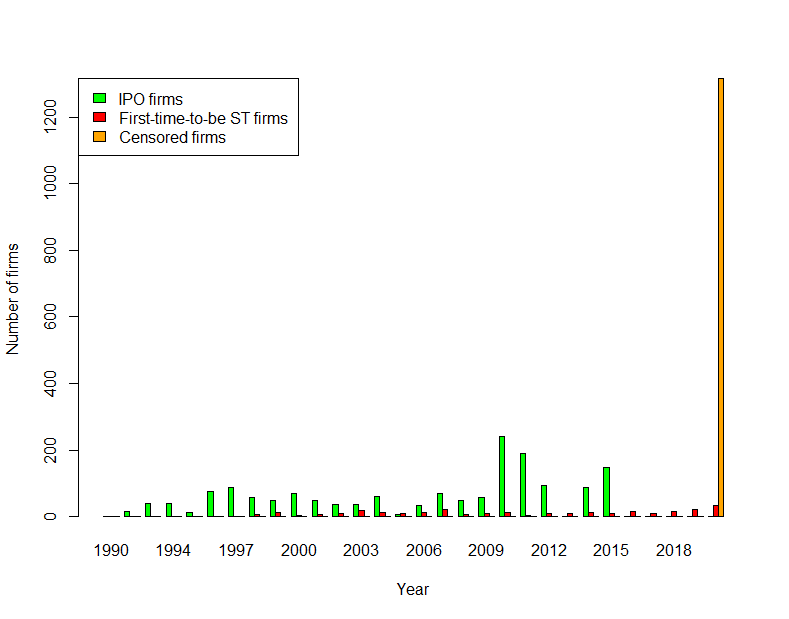}
    \caption{Numbers of IPOs, first ST designations, and censored firms across different years in the raw dataset.}
    \label{summarypic}
\end{figure}

We construct the sub-dataset in which all firms have survived for at least $s$ years. The goal is to use these $s$ years of information to predict whether a firm will fail within $t$ years. 


\subsection{Estimation procedure}\label{emp set}

We now describe the estimation procedure in the empirical application. First, we note that all public firms report their financial information with a one-quarter delay. Consequently, if a firm has survived for $s$ years, only $s \times 4 - 1$ quarters' worth of financial information will be available for analysis.

Let $z_{i, s - \frac{j-1}{m}, k}$ represent the $k$-th financial covariate of firm $i$, measured at time $s - \frac{j-1}{m}$, where $j = 2, \ldots, d$, and $d = s \times m$. We organize all the lags of the covariate into a group vector $\widetilde{\boldsymbol{Z}}_{i,k}$:
$$
\widetilde{\boldsymbol{Z}}_{i,k} = \left(z_{i, s - \frac{1}{m}, k}, z_{i, s - \frac{2}{m}, k}, \ldots, z_{i, \frac{1}{m} ,k}\right)^{\top}, \quad i \in [N], \quad k \in [K],
$$
where $z_{i,\frac{1}{m}, k}$ refers to the $k$-th covariate measured in the next quarter following the firm's IPO date, and $m = 4$ denotes the quarterly frequency of the financial covariates.

Next, we aggregate the lagged covariate vector $\widetilde{\boldsymbol{Z}}_{i,k}$ using a dictionary $W$, which consists of Gegenbauer polynomials shifted to the interval $[0,1]$ with parameter $\alpha_{\text{poly}} = \beta_{\text{poly}} = -\frac{1}{2}$ and size $L = 3$.\footnote{These polynomials are also known as Chebyshev polynomials on the interval $[0,1]$.} This choice of $W$ coincides with the specification used in Section~\ref{est simulation}.

Finally, we construct the covariate matrix $\boldsymbol{X}$ as follows:
$$\boldsymbol{X} =\left(\boldsymbol{X}_{1} , \boldsymbol{X}_{2}, \ldots, \boldsymbol{X}_{N}\right)^{\top},$$ 
where each $\boldsymbol{X}_i$ $=\left(1, \widetilde{\boldsymbol{Z}}_{i,1}^{\top} W, \widetilde{\boldsymbol{Z}}_{i,2}^{\top} W, \ldots, \widetilde{\boldsymbol{Z}}_{i,K}^{\top} W\right)^{\top}, i \in [N]$. This matrix $\boldsymbol{X}$ is then used in sg-LASSO-MIDAS and LASSO-MIDAS. Notice that we include the intercept term but do not penalize it in the estimation procedure.

We compare the performance of firm distress predictions using the following methods.

\medskip 

\noindent\textbf{Logistic regression.} 
As a benchmark, we consider a simple unpenalized logistic regression using only the lags of the first Z-score covariate. We solve the empirical version of \eqref{equivalent censor}, in which the function $H$ is estimated using the Kaplan-Meier estimator. This is considered a reasonable starting point for distress predictions. The total number of parameters, including the intercept, to estimate is $1 + (d - 1)$, where $d-1$ is the number of available lags for each covariate.

\medskip

\noindent\textbf{LASSO-U (LASSO-UMIDAS).} We estimate $ d - 1 = s \times m -1$ coefficients per group covariate $\widetilde{\boldsymbol{Z}}_{i,k} \in \mathbb{R}^{d-1}, k \in [K]$, using the unstructured LASSO estimator. The total number of parameters to estimate is $1 + K \times (d - 1)$, where $s$ is the number of years survived by the firm, and $m$ represents the annual sampling frequency of each covariate.

\medskip

\noindent\textbf{LASSO-M (LASSO-MIDAS).} Each high-frequency covariate and its $d-1$ lags are grouped into $\widetilde{\boldsymbol{Z}}_{i,k} \in \mathbb{R}^{d-1}, k \in [K]$. We aggregate the group covariate $\widetilde{\boldsymbol{Z}}_{i,k}$ using Gegenbauer polynomials $W \in \mathbb{R}^{(d-1) \times L}$. We apply a Lasso penalty to induce sparsity. The total number of parameters to estimate is $1 + K \times L$, where $L$ is the size of the Gegenbauer polynomial dictionary.

\medskip

\noindent\textbf{sg-LASSO-M (sg-LASSO-MIDAS).} Similarly to LASSO-MIDAS, each high-frequency covariate and its $d-1$ lags form a group $\widetilde{\boldsymbol{Z}}_{i,k} \in \mathbb{R}^{d-1}, k \in [K]$, which is aggregated using Gegenbauer polynomials $W \in \mathbb{R}^{(d-1) \times L}$. Instead of using a Lasso penalty, we use the sparse-group Lasso penalty to induce sparsity in the group covariates. The group structure matches that used in the simulation, where each group corresponds to the MIDAS-weighted covariates. Alternative group structures are also examined, and the corresponding results are reported in Online Appendix \ref{additional}.
The total number of parameters to estimate is $1 + K \times L$.

\medskip

The choice of $s$ dictates the historical information captured in the covariate matrix $\boldsymbol{X}$, while $t$ denotes the prediction horizon. The existing literature on firm distress prediction, particularly in the United States, often examines prediction periods $t-s$ ranging from $1$ quarter to $2$ years \citep{cole2012deja}. In practical applications, such as for bank regulators, models need to identify potential failures well in advance. For example, \citet{audrino2019predicting} developed a MIDAS-type method with prediction periods $t-s$ equal to $1$ and $2$ years.

For our empirical application, we select a reference period of $s = 6$ years. Using the firm classification criteria for Special Treatment outlined in Section \ref{real data}, we establish prediction horizons of $t = 8, 8.5, 9$ years to forecast firm distress within these intervals. To investigate longer forecast periods, we also consider an additional case with $s = 10$ years and prediction horizons of $t = 13, 13.5, 14$ years. Longer prediction horizons provide information on the risks of long-term financial distress. As highlighted by \citet{li2021chinese}, these prediction horizons are critical for accurately forecasting firm financial distress in China. They also offer meaningful and practical benchmarks for evaluating firm failure prediction models.

In practice, missing data in financial variables can arise due to various factors, including inconsistent reporting practices between firms, differing regulatory requirements, incomplete disclosures, and delays in data availability after IPOs. Given the substantial amount of missing data in the raw dataset, we construct a complete sub-dataset for each $s$ by selecting firms with consistent $s$-year observations. While common approaches for handling missing data, such as removing variables or firms with missing values, are widely used, these methods often result in retaining too few firms or variables for meaningful analysis. Furthermore, the firms with observable status in the sub-dataset play a critical role in the predictive modeling process. To address these challenges, we propose an algorithm that balances dimensionality and the number of uncensored firms in the selected sub-dataset, as outlined in \tcr{Algorithm~\ref{algorithm 1} of Online Appendix \ref{data process appendix}}.

To evaluate the prediction performance of the different methods, we randomly split the dataset into in-sample ($80\%$) and out-of-sample ($20\%$) sets, ensuring that both sets maintain the same proportion of the event indicator $\delta_i(t) \mathbbm{1}\{\widetilde{T}_i \leq t\}$. The tuning parameters for sg-LASSO-MIDAS, LASSO-MIDAS, and LASSO-UMIDAS are selected using stratified $5$-fold cross-validation, where the optimal parameters are those that maximize the AUC on the validation fold.\footnote{In this paper, unless specified otherwise, the default choice for cross-validation is to maximize the AUC on the validation fold.} Additionally, as an alternative, cross-validation to maximize the likelihood score is also investigated. The AUC estimator employed in this procedure follows the method described in Section \ref{{time depedent auc}}. Specifically, we perform a grid search over the regularization parameter $\alpha$ in the sparse-group LASSO penalty, with values in the set $\{0, 0.1, 0.3, 0.5, 0.7, 0.9, 1\}$ and, as standard, $\lambda$ is chosen in a grid which follows \citet{liang_sparsegl}. This process is repeated $10$ times as a robustness check, each time using a different random split of the data. All models are trained on the same training set and evaluated on the same test set. 

For each split, the AUC is computed on the out-of-sample data, and the out-of-sample data are then bootstrapped $1000$ times to calculate the AUC for each bootstrap sample. The AUC values for each bootstrap sample are subsequently averaged across the $10$ different splits, resulting in $1000$ averaged AUC values. The final performance is reported as the overall average AUC, along with a two-sided $95\%$ confidence interval, which is calculated based on these $1000$ bootstrapped averages. This approach ensures a robust performance evaluation by accounting for variability in the data and model performance.\footnote{We note that the bootstrap approach is not theoretically validated for the regularized estimators we consider in this paper.}

On top of the simple logistic regression and the LASSO-UMIDAS, LASSO-MIDAS, and sg-LASSO-MIDAS with cross-validation for the AUC or the likelihood score, we consider other alternative approaches. 

\paragraph{Macro data augmented prediction.} We first assess whether incorporating macroeconomic data can enhance the accuracy of distress prediction models. The macroeconomic dataset for China is sourced from the Federal Reserve Bank of Atlanta’s China Macroeconomy Project (\url{https://www.atlantafed.org/cqer/research/china-macroeconomy#Tab2}), which provides a comprehensive set of macroeconomic variables relevant to the Chinese economy. The dataset includes $98$ macroeconomic variables, measured quarterly, and spans the same time period as the financial data collected for the firms in our study.

To merge the macroeconomic data with the financial dataset, we select only those macroeconomic variables that do not have missing values across all firms within each financial sub-dataset. Since the sub-datasets differ based on the value of $s$, the set of macroeconomic variables selected will vary accordingly for each sub-dataset. Furthermore, we use the same MIDAS dictionary $W$ for the macroeconomic covariates as for the financial covariates, ensuring consistency in the aggregation of high-frequency data over time.

Table \ref{summary} summarizes the details of the two sub-datasets categorized by different values of $s$. For the sub-dataset with $s=6$ years, we use all available information across each firm’s entire survival period, allowing us to leverage the maximum historical data available for firms with $6$ years of survival. In contrast, for the sub-dataset with $s=10$ years, we restrict the covariates to those from the last $4$ years of each firm’s survival period. This adjustment is necessary because firms that have survived for $10$ years were generally listed in the $1990$s, and significant missing data is often observed in the early years following their IPOs. By focusing on the most recent $4$ years, we ensure better data quality and a more robust analysis.
\begin{table}[htbp]
  \centering
  \caption{Summary information of the dataset with $s = 6$ and $s = 10$ years.}
  \medskip
  \scalebox{1}{
    \begin{tabular}{lll}
    \toprule
          & $s= 6$ years & $s = 10$ years \\
          \cmidrule{2-3}    Number of firms $N$   & 901  & 784  \\
    Number of uncensored firms & 67    &  80 \\
        Number of financial covariates $K_{\text{financial}}$ (including lags)   &  32 (736) & 36 (540)   \\
    Number of macro covariates $K_{\text{macro}}$ (including lags)  &  63 (1449) &  63 (945) \\
 \cmidrule{2-3}  $30^{\text{th}}$ percentile of $\widetilde{T}$ (years) & 9.512 &  13.369 \\
    $50^{\text{th}}$ percentile of $\widetilde{T}$ (years)& 10.285 &  15.411 \\
 \cmidrule{2-3}  $30^{\text{th}}$ percentile of $T$ among uncensored firms (years)& 7.789  &  11.032 \\
    $50^{\text{th}}$ percentile of $T$ among uncensored firms (years)& 8.934  &  13.844 \\
    \bottomrule
    \end{tabular}%
    }
  \label{summary}%
\end{table}%

\paragraph{Oversampling.} In addition, we apply an oversampling technique to address the imbalance in the dataset caused by the high censoring rate, which results in an unequal proportion of firms experiencing distress versus those that are not. This imbalance could adversely affect the performance of distress prediction models, as the minority class (distressed firms) may be underrepresented. Since the empirical dataset has a high censoring rate, we face a class imbalance between those firms that eventually experience distress $ \mathbbm{1}\{T_i \leq C_i\}\mathbbm{1}\{\widetilde{T}_i \leq t\} = 1 $ and those that do not or we do not observe $\mathbbm{1}\{T_i \leq C_i\}\mathbbm{1}\{\widetilde{T}_i \leq t\} = 0$. To balance this, for the training dataset, we randomly duplicate the observations from the minority class (firms that experience distress) until the proportion of distressed firms reaches $15\%$ of the training dataset. This step helps mitigate the imbalance and ensures that the model is exposed to a sufficient number of distressed firms during training. Tuning parameters are selected using $5$-fold stratified cross-validation, where the optimal parameters maximize the likelihood score.

\paragraph{Does censoring matter for prediction?} We compare with an approach that applies LASSO-UMIDAS, LASSO-MIDAS, and sg-LASSO-MIDAS to the sub-dataset where censored firms with censoring time smaller than $t$ ($\mathbbm{1}\{ C_i < T_i\}\mathbbm{1}\{ C_i < t\}=1$) have been removed.\footnote{Given the prediction horizon $t$, the distress status of censored firms cannot be observed if their censoring time is shorter than $t$ and, clearly, censored firms with $\mathbbm{1}\{C_i \geq t\}$ are not distressed.} This is the approach usually taken in the literature, since it allows for ignoring censoring; see the discussion in the literature review of the introduction. The limitation of this procedure is that it does not use all observations, resulting in a loss of precision. 

\subsection{Application results}\label{app res}

The results are presented in Tables \ref{empirical AUC Augment6} and \ref{empirical AUC Augment10}. The LASSO-MIDAS and sg-LASSO-MIDAS generally outperform the LASSO-UMIDAS and the logistic regression benchmark, which aligns with our expectations. For $s=6$ years, sg-LASSO-MIDAS and LASSO-MIDAS perform essentially equally, with only minor differences across prediction horizons. When $s=10$ years, however, sg-LASSO-MIDAS exhibits a slight but consistent performance advantage over LASSO-MIDAS, indicating that the sparse-group Lasso regularization is particularly beneficial when incorporating a larger historical window of data.

When we compare the performance of models based on cross-validation using different metrics, we observe that cross-validation based on the AUC generally yields better results than cross-validation based on likelihood scores. This is logical since our target measure is the AUC itself.

Additionally, while integrating macroeconomic data does not improve prediction performance over the purely financial model when $s = 6$ years, it enhances performance when $s = 10$ years. This suggests that macroeconomic variables become more relevant with a larger historical window, offering supplementary information that helps improve prediction accuracy, especially for firms with longer survival periods. However, oversampling does not seem to provide any additional benefit in improving prediction performance.

When we remove censored firms with $C_i < t$, the performance of our methodologies deteriorates across all scenarios, emphasizing the importance of properly accounting for censoring in predictive modeling.
\begin{table}[htbp]
  \centering
  \caption{(Distress prediction performance) Estimated average AUCs ($95\%$ confidence intervals) in the out-of-sample set with $s = 6$ years and prediction horizons $t = 8, 8.5, 9$ years.}
  \medskip
  \scalebox{0.7}{
    \begin{tabular}{cccc}
    \toprule
          & \multicolumn{3}{c}{$s = 6$ years} \\
\cmidrule{2-4}          & $t = 8$ years & $t = 8.5$ years & $t = 9$ years \\
\cmidrule{2-4}          & \multicolumn{3}{c}{Benchmark } \\
\cmidrule{2-4}    Logistic reg. &  0.623 [0.563, 0.673]   &  0.587 [0.528, 0.647]   &  0.606 [0.576, 0.641]\\
    \cmidrule{2-4}          & \multicolumn{3}{c}{Cross-validation for the AUC} \\
\cmidrule{2-4}    LASSO-U &  0.797 [0.755, 0.844]   &  0.756 [0.717, 0.801]   & 0.765 [0.734, 0.802]\\
    LASSO-M &  0.838 [0.793, 0.872]   &   0.817 [0.774, 0.864]    & 0.811 [0.778, 0.843] \\
    sg-LASSO-M &  0.823 [0.789, 0.865]    &   0.821 [0.778, 0.861]  &  0.806 [0.776, 0.840] \\
    \cmidrule{2-4}          & \multicolumn{3}{c}{Cross-Validation for the likelihood score} \\
\cmidrule{2-4}    LASSO-U &  0.710 [0.671, 0.753]  &  0.644 [0.587, 0.708]   & 0.761 [0.718, 0.804] \\
    LASSO-M &  0.701 [0.665, 0.738]   &   0.851 [0.817, 0.894]   & 0.808 [0.774, 0.843]\\
    sg-LASSO-M &  0.782 [0.747, 0.820]    &   0.813 [0.773, 0.862]   &  0.795 [0.760, 0.835] \\
\cmidrule{2-4}          & \multicolumn{3}{c}{Macro Data Augmented} \\
\cmidrule{2-4}    LASSO-U    &  0.790 [0.767, 0.819]      & 0.740 [0.718, 0.772] &  0.721 [0.698, 0.748]\\
    LASSO-M     & 0.823 [0.797, 0.848]     & 0.810 [0.786, 0.836]  &  0.782 [0.761, 0.804]\\
    sg-LASSO-M    &   0.820 [0.800, 0.846]   & 0.806 [0.783, 0.830]  &  0.798 [0.778, 0.822]\\
    \cmidrule{2-4}          & \multicolumn{3}{c}{Oversampling with Financial Data} \\
\cmidrule{2-4}    LASSO-U     &  0.833 [0.800, 0.863]   & 0.760 [0.716, 0.800] &  0.773 [0.746, 0.808]\\
    LASSO-M     &  0.801 [0.760, 0.838]    & 0.832 [0.806, 0.864] & 0.825 [0.799, 0.855] \\
    sg-LASSO-M    & 0.810 [0.768, 0.851]     & 0.834 [0.808, 0.861]  &  0.822 [0.800, 0.851]\\
            \cmidrule{2-4}          & \multicolumn{3}{c}{Data without censored firms satisfying $C_i < t$} \\
\cmidrule{2-4}    LASSO-U &  0.707 [0.659, 0.768]   &  0.792 [0.759, 0.829]    &  0.731 [0.695, 0.776]\\
    LASSO-M &  0.787 [0.741, 0.829]     &   0.817 [0.777, 0.854]   &  0.744 [0.701, 0.791]\\
    sg-LASSO-M &   0.753 [0.702, 0.816]  &  0.820 [0.778, 0.856]   &  0.776 [0.733, 0.821] \\
    \bottomrule
    \end{tabular}%
    }
  \label{empirical AUC Augment6}%
\end{table}%

\begin{table}[htbp]
  \centering
  \caption{(Distress prediction performance) Estimated average AUCs ($95\%$ confidence intervals) in the out-of-sample set with $s = 10$ years and prediction horizons $t = 13, 13.5, 14$ years.}
  \medskip
  \scalebox{0.7}{
    \begin{tabular}{cccc}
    \toprule
          & \multicolumn{3}{c}{$s = 10$ years} \\
\cmidrule{2-4}          & $t = 13$ years & $t = 13.5$ years & $t = 14$ years \\
\cmidrule{2-4}          & \multicolumn{3}{c}{Benchmark} \\
\cmidrule{2-4}    Logistic reg. &  0.549 [0.504, 0.592]   &  0.573 [0.523, 0.627]   &  0.617 [0.576, 0.661]\\
    \cmidrule{2-4}          & \multicolumn{3}{c}{Cross-Validation for the AUC} \\
\cmidrule{2-4}    LASSO-U &  0.566 [0.514, 0.633]   &  0.628 [0.591, 0.674]   & 0.669 [0.637, 0.709] \\
    LASSO-M &  0.773 [0.738, 0.818]    &   0.653 [0.615, 0.700]    & 0.688 [0.654, 0.726]\\
    sg-LASSO-M &   0.818 [0.781, 0.847]    &   0.671 [0.635, 0.718]   &  0.702 [0.667, 0.737] \\
            \cmidrule{2-4}          & \multicolumn{3}{c}{Cross-Validation for the likelihood score} \\
\cmidrule{2-4}    LASSO-U &  0.572 [0.537, 0.625]  &  0.555 [0.520, 0.603]    & 0.622 [0.593, 0.663]\\
    LASSO-M &  0.787 [0.754, 0.831]   &   0.609 [0.579, 0.668]    & 0.701 [0.671, 0.739]\\
    sg-LASSO-M &  0.815 [0.779, 0.853]     &    0.657 [0.630, 0.710]    &  0.701 [0.677, 0.739] \\
\cmidrule{2-4}          & \multicolumn{3}{c}{Macro Data Augmented} \\
\cmidrule{2-4}    LASSO-U     &  0.573 [0.542, 0.611]     & 0.702 [0.677, 0.727]  & 0.678 [0.659, 0.701]\\
    LASSO-M     & 0.747 [0.728, 0.779]    & 0.670 [0.647, 0.703]   & 0.691 [0.671, 0.716]\\
    sg-LASSO-M    &   0.773 [0.750, 0.795]     & 0.707 [0.687, 0.738]   & 0.725 [0.706, 0.749]\\
    \cmidrule{2-4}          & \multicolumn{3}{c}{Oversampling with Financial Data} \\
\cmidrule{2-4}    LASSO-U     & 0.655 [0.602, 0.711]     & 0.636 [0.590, 0.675]  &  0.697 [0.668, 0.731]\\
    LASSO-M     & 0.704 [0.649, 0.753]     & 0.627 [0.590, 0.677] &  0.679 [0.640, 0.714]\\
    sg-LASSO-M    & 0.685 [0.637, 0.738]     & 0.615 [0.578, 0.664]  & 0.669 [0.634, 0.706]\\
            \cmidrule{2-4}          & \multicolumn{3}{c}{Data without censored firms satisfying $C_i < t$} \\
\cmidrule{2-4}    LASSO-U &  0.764 [0.726, 0.809]   &  0.619 [0.577, 0.665]   &  0.695 [0.667, 0.733]\\
    LASSO-M &  0.759 [0.726, 0.815]     &   0.624 [0.588, 0.669]   &  0.616 [0.575, 0.658]\\
    sg-LASSO-M &   0.802 [0.764, 0.845] &  0.625 [0.592, 0.676]    &  0.642 [0.610, 0.680] \\
    \bottomrule
    \end{tabular}%
    }
  \label{empirical AUC Augment10}%
\end{table}%

We also evaluate how different covariate group structures influence the performance of sg-LASSO-MIDAS, with results reported in Tables \ref{alt group structures s=6} and \ref{alt group structures s=10} of Online Appendix \ref{additional}. The performance of sg-LASSO-MIDAS is generally robust to variations in group structure, consistently outperforming LASSO-UMIDAS and, under certain group specifications, surpassing LASSO-MIDAS more frequently than under the baseline grouping; see Online Appendix~\ref{additional} for a description of the group structures and further discussion.

To further assess the performance difference, we conduct a pairwise comparison test, a widely used method for comparing two AUCs. Slightly modifying the approach in \citet{han, robin_proc_2011}, we specifically test whether the estimated AUC of sg-LASSO-MIDAS is no larger than those of LASSO-MIDAS and LASSO-UMIDAS.\footnote{We have $1000$ bootstrapped average AUCs for each method as described before. The p-value is calculated as the proportion of sg-LASSO-MIDAS's AUC values that are smaller than those of another method.} Table~\ref{pair-wise test} shows that the improvement of the sg-LASSO-MIDAS over the LASSO-UMIDAS is statistically significant at least at the $10\%$ significance level across all scenarios, with the largest gap observed when $s = 10$ and $t = 13$ years. As for the comparison between sg-LASSO-MIDAS and LASSO-MIDAS, no significant difference is found when $s = 6$ years. However, when $s =10$ years, sg-LASSO-MIDAS significantly outperforms LASSO-MIDAS at the $10\%$ significance level for prediction horizons $t = 13$ and $13.5$ years. We also conduct a pairwise comparison between the sg-LASSO-MIDAS applied to the dataset with and without censored firms satisfying $C_i < t$. For the scenarios where $s = 6$, $t = 8.5$ years, and $s = 10$, $t = 13$ years, sg-LASSO-MIDAS performs better on the full dataset than on the dataset without censored firms satisfying $C_i < t$, though the difference is not statistically significant. However, in other scenarios, including censoring significantly improves model performance, with results being statistically significant, at least at the $5\%$ significance level. These findings strongly support the advantages of using MIDAS weights, considering the group structure of covariates, and incorporating the censoring information in practice. Overall, the empirical results highlight the superiority of the sg-LASSO-MIDAS across different scenarios. 
\begin{table}[htbp]
  \centering
  \caption{Pairwise difference test: Null hypothesis $H_0$: estimated AUC of the first prediction approach is no larger than that of the second. We use $*$ and $**$ to indicate $10\%$ and $5\%$ significance, respectively.}
  \medskip
  \scalebox{1}{\begin{tabular}{ccccccc}
\toprule  \multicolumn{3}{c}{$s = 6$ years} &       & \multicolumn{3}{c}{$s = 10$ years} \\
\cmidrule{1-3}\cmidrule{5-7}    $t = 8$ years & $t = 8.5$ years & $t = 9$ years &       & $t = 13$ years & $t = 13.5$ years & $t = 14$ years \\
\cmidrule{1-3}\cmidrule{5-7}    \multicolumn{3}{c}{sg-LASSO-M $vs.$ LASSO-U} &       & \multicolumn{3}{c}{sg-LASSO-M $vs.$ LASSO-U} \\
    $0.098^{*}$ & $0.000^{**}$ & $0.000^{**}$     &   &    $0.000^{**}$ & $0.010^{**}$ & $0.065^{*}$\\
\cmidrule{1-3}\cmidrule{5-7}    \multicolumn{3}{c}{sg-LASSO-M $vs.$ LASSO-M} &       & \multicolumn{3}{c}{sg-LASSO-M $vs.$ LASSO-M} \\
    $0.757$ & $0.462$ & $0.588$     &   &    $0.000^{**}$ & $0.093^{*}$ & $0.111$\\
\cmidrule{1-3}\cmidrule{5-7}    \multicolumn{3}{c}{\makecell{sg-LASSO-M with $vs.$ without \\ censored firms satisfying $C_i < t$}} &       & \multicolumn{3}{c}{\makecell{sg-LASSO-M with $vs.$ without \\ censored firms satisfying $C_i < t$}} \\
     $0.001^{**}$ & $0.495$ & $0.021^{**}$     &   &    $0.238$ & $0.022^{**}$ & $0.001^{**}$ \\
    \bottomrule
    \end{tabular}}%
  \label{pair-wise test}%
\end{table}%

To better understand which covariates are useful for prediction, we examine the financial types selected by the sg-LASSO-MIDAS, as illustrated in Figure \ref{selected6} of Online Appendix \ref{additional}. Financial variables related to the $Z$-score appear to play a pivotal role across all prediction horizons in forecasting firm distress. This observation aligns with prior research \citep{altman1968financial}, as the $Z$-score model has been widely employed in both academic studies and industry to predict corporate defaults \citep{altman2017financial}. Further details on the selected financial covariates are presented in Figures \ref{sel 6} and \ref{sel 10} of Online Appendix \ref{additional}.

Furthermore, following Section~\ref{inference theory}, we separately test each financial covariate for statistical significance at the $5\%$ level in distress prediction.
Specifically, for each covariate, we conduct a Wald test on its associated vector of MIDAS coefficients, estimated using the de-sparsified sg-LASSO-MIDAS on the full empirical dataset, as in Section \ref{inf sim}. All the MIDAS settings remain the same as those in Section \ref{emp set}. For the regularization parameter $\alpha$ in the sparse-group LASSO, we consider the set $\{0, 0.1, 0.3, 0.5, 0.7, 0.9, 1\}$. Tuning parameters for sg-LASSO-MIDAS and the nodewise LASSO are chosen via $5$-fold stratified cross-validation to maximize the validation likelihood. As shown in Figures~\ref{sig 6} and~\ref{sig 10} of Online Appendix~\ref{additional}, statistically significant covariates are marked in blue for each prediction horizon. Consistent with the literature \citep{altman1968financial, altman2017financial}, we find that the five Z-score covariates (denoted by $X_1,\ldots,X_5$) play prominent roles, while several other financial covariates also have a notable influence. \tcr{When $s = 6$ years, $X_2$ (\emph{Retained Earnings\,/\,Total Assets}) and $X_4$ (\emph{Market Value of Equity\,/\,Total Liabilities}) are statistically significant across all prediction horizons, whereas $X_1$ (\emph{Working Capital\,/\,Total Assets}) is significant at only one prediction horizon. Beyond the Z-score covariates, several other covariates are also significant. For example, \emph{Earnings Before Interest, Taxes, Depreciation, and Amortization\,/\,Total Liabilities} is consistently significant across all horizons. When $s = 10$ years, $X_4$ remains significant across all horizons, $X_2$ and $X_3$ (\emph{Earnings Before Interest and Taxes\,/\,Total Assets}) are significant at two prediction horizons, while $X_5$ (\emph{Sales\,/\,Total Assets}) is significant at only one prediction horizon. Other financial covariates, such as \emph{Net Profit on Assets} and \emph{Interest-Bearing Debt Ratio}, are also influential at most prediction horizons.}

\section{Conclusion}

This paper presents a novel approach to corporate survival analysis, addressing the challenges of high-dimensional censored data sampled at both consistent and mixed frequencies. 

The first major contribution is the introduction of the sparse-group LASSO estimator for high-dimensional censored MIDAS logistic regressions, along with a corresponding de-sparsified estimator for inference. The proposed framework effectively accommodates hierarchical data structures and facilitates variable selection both within and across groups, unifying classical LASSO and group LASSO.

Secondly, we develop the theory for logistic regression with high-dimensional censored data sampled at different frequencies. To extend the existing literature, which assumes either a fixed design or isotropic sub-Gaussian covariates, we develop the non-asymptotic properties of the proposed sparse-group LASSO estimator for censored, heavy-tailed data. 
The framework is readily extendable to generalized linear models with structured sparsity estimators. We also explicitly consider the approximation error, which, to the best of our knowledge, is a novel contribution in the context of logistic regression. This error may arise from various sources, including approximations in the MIDAS weight function and/or deviations from exact sparsity. In addition to the penalized estimator, we analyze the corresponding de-sparsified sparse-group LASSO estimator, showing that it is asymptotically unbiased while properly accounting for censoring, MIDAS approximation error, and heavy-tailed covariates, and that its asymptotic variance is affected by censoring, which has not been previously studied in the literature.

A key practical contribution is an application to a comprehensive dataset of publicly traded Chinese manufacturing firms, integrating survival and censoring time information alongside numerous high-frequency financial covariates. Empirical findings show that sg-LASSO-MIDAS consistently outperforms unstructured LASSO methods across various scenarios. Notably, the inclusion of censoring information significantly enhances prediction performance, providing valuable insights for predicting firm distress under real-world conditions. Beyond prediction, the de-sparsified sg-LASSO-MIDAS further enables valid statistical inference, allowing us to identify financial covariates that are statistically significant predictors of firm distress and thereby providing practical guidance for decision makers.

Overall, the methodologies developed in this paper have broad applicability beyond corporate distress prediction. The integration of logistic models, MIDAS, and regularized machine learning techniques holds promise for applications in areas such as disease diagnosis, solvency evaluation, fraud detection, customer churn analysis, and labor market studies.

\section*{Acknowledgments} The authors thank the co-editor Viktor Todorov, the Associate Editor, two anonymous referees, Christophe Croux, Geert Dhaene, Daniel Gutknecht, Onno Kleen and Yoshimasa Uematsu for helpful comments as well as seminar participants at Tilburg University and  Copenhagen Business School, and conference attendees at the International Association for Applied Econometrics 2025, the Leuven Statistics Days 2025, the International Conference on Econometrics and Statistics 2025, and Financial Econometrics Meets Machine Learning 2025. The authors are listed in order of contribution.

		{
		\section*{Funding sources}
Wei Miao gratefully acknowledges financial support from the China Scholarship Council
through the grant 202306130036. Jad Beyhum gratefully acknowledges financial support from the Research Fund KU Leuven through the grant STG/23/014. Jonas Striaukas gratefully acknowledges the financial support from the European Commission, MSCA-2022-PF Individual Fellowship, Project 101103508. Project Acronym: MACROML. Ingrid Van Keilegom acknowledges support from the FWO (Research Foundation Flanders) through the projects G0I3422N and G047524N.
	}

    \section*{Supplementary material} \label{supp m}
    
    \paragraph{Online Appendix:} Proofs of Theorem \ref{the}, Lemma \ref{nodewise}, and Theorem \ref{inf the}, empirical dataset pre-processing algorithm, additional empirical and simulation results, and details on the empirical dataset (.pdf file).
    \paragraph{R package and Replication files:} The R package \texttt{Survivalml}, which implements our method, and all simulation and application code are available at \url{https://github.com/Wei-M-Wei/Survivalml}. 

\bibliographystyle{apalike}
\bibliography{ref_new}

\appendix
\begin{appendices}

\section{On Section \ref{sec est}}  \label{estimator}

\begin{lemma}\label{lm.cens1} If 
$$\mathbbm{E}\left[\left. \frac{\exp(\boldsymbol{Z}^{\top} \boldsymbol{\theta}_0)}{(1+ \exp(\boldsymbol{Z}^{\top} \boldsymbol{\theta}_0))^2} \boldsymbol{Z}\boldsymbol{Z}^{\top}\right| T \geq s \right]$$ is positive definite,
    we have $$\boldsymbol{\theta}_0 = \argmax_{\boldsymbol{\theta} \in \mathbb{R}^{K_z}} \mathbbm{E}\left[\left.\mathbbm{1}\{T \leq t\} \boldsymbol{Z}^{\top} \boldsymbol{\theta} - \log\left(1 + \exp(\boldsymbol{Z}^{\top} \boldsymbol{\theta})\right)\right| T \geq s \right].$$
\end{lemma}
\begin{proof} Let us use the notation
\begin{align*}L(\boldsymbol{\theta}) &= \mathbbm{E}\left[\left.\mathbbm{1}\{T \leq t\} \boldsymbol{Z}^{\top} \boldsymbol{\theta} - \log\left(1 + \exp(\boldsymbol{Z}^{\top} \boldsymbol{\theta})\right)\right| T \geq s \right].
\end{align*}
By the law of iterated expectations, it holds that 
\begin{align*}L(\boldsymbol{\theta}) &= \mathbbm{E}\left[\left.\mathbbm{E}\left[ \mathbbm{1}\{T \leq t\}|\boldsymbol{Z},T \geq s \right] \boldsymbol{Z}^{\top} \boldsymbol{\theta} - \log\left(1 + \exp(\boldsymbol{Z}^{\top} \boldsymbol{\theta})\right)\right| T \geq s \right]\\
&= \mathbbm{E}\left[\left.\frac{\exp(\boldsymbol{Z}^{\top} \boldsymbol{\theta}_0)}{1+ \exp(\boldsymbol{Z}^{\top} \boldsymbol{\theta}_0)} \boldsymbol{Z}^{\top} \boldsymbol{\theta} - \log\left(1 + \exp(\boldsymbol{Z}^{\top} \boldsymbol{\theta})\right)\right| T \geq s \right].
\end{align*}
Notice that the derivative of $L(\cdot)$ is $$\dot{L}(\boldsymbol{\theta})=\mathbbm{E}\left[\left.\frac{\exp(\boldsymbol{Z}^{\top} \boldsymbol{\theta}_0)}{1+ \exp(\boldsymbol{Z}^{\top} \boldsymbol{\theta}_0)} \boldsymbol{Z}^{\top}- \frac{\exp(\boldsymbol{Z}^{\top} \boldsymbol{\theta})}{1+ \exp(\boldsymbol{Z}^{\top} \boldsymbol{\theta})} \boldsymbol{Z}^{\top}\right| T \geq s \right].$$
Therefore, $\dot{L}(\boldsymbol{\theta}_0)=0.$ Moreover, the second derivative of $L(\cdot)$ at $\boldsymbol{\theta}_0$ is $$\ddot{L}(\boldsymbol{\theta}_0)=-\mathbbm{E}\left[\left. \frac{\exp(\boldsymbol{Z}^{\top} \boldsymbol{\theta}_0)}{(1+ \exp(\boldsymbol{Z}^{\top} \boldsymbol{\theta}_0))^2} \boldsymbol{Z}\boldsymbol{Z}^{\top}\right| T \geq s \right]$$ is negative definite.
Hence, $\boldsymbol{\theta}_0$  is indeed the global maximum of $L(\cdot).$
\end{proof}

\begin{lemma}\label{lm.cens2}
    Under Assumptions \ref{as1} and \ref{as2}, 
 if 
$$\mathbbm{E}\left[\left. \frac{\exp(\boldsymbol{Z}^{\top} \boldsymbol{\theta}_0)}{(1+ \exp(\boldsymbol{Z}^{\top} \boldsymbol{\theta}_0))^2} \boldsymbol{Z}\boldsymbol{Z}^{\top}\right| T \geq s \right]$$ is positive definite, we have

\begin{equation} \boldsymbol{\theta}_0 = \argmax_{\boldsymbol{\theta} \in \mathbb{R}^{K_z}} \mathbbm{E}\left[\left. \frac{\delta(t) \mathbbm{1}\{\widetilde{T} \leq t\}}{H(t \wedge \widetilde{T})} \boldsymbol{Z}^{\top} \boldsymbol{\theta} - \log \left(1 + \exp(\boldsymbol{Z}^{\top} \boldsymbol{\theta})\right)\right|\widetilde{T} \geq s\right].\end{equation}
\end{lemma}
\begin{proof} Remark that 
\begin{equation}\label{2112241}
\begin{aligned}
    &\mathbbm{E}\left[\left. \frac{\delta(t) \mathbbm{1}\{\widetilde{T} \leq t\}}{H(t \wedge \widetilde{T})} \boldsymbol{Z}^{\top} \boldsymbol{\theta} - \log \left(1 + \exp(\boldsymbol{Z}^{\top} \boldsymbol{\theta})\right)\right|\widetilde{T} \geq s\right]\\
    &= \mathbbm{E}\left[\mathbbm{1}\{\widetilde{T} \ge s\}\left\{ \frac{\delta(t) \mathbbm{1}\{\widetilde{T} \leq t\}}{H(t \wedge \widetilde{T})} \boldsymbol{Z}^{\top} \boldsymbol{\theta} - \log \left(1 + \exp(\boldsymbol{Z}^{\top} \boldsymbol{\theta})\right)\right\}\right]/P(\widetilde{T} \geq s)\\
    &= \frac{1}{P(\widetilde{T} \geq s)}\mathbbm{E}\left[\mathbbm{1}\{\widetilde{T} \ge s\}\frac{\delta(t) \mathbbm{1}\{\widetilde{T} \leq t\}}{H(t \wedge \widetilde{T})} \boldsymbol{Z}^{\top} \boldsymbol{\theta} \right]\\
    &\quad - \frac{1}{P(\widetilde{T} \geq s)}\mathbbm{E}\left[ \mathbbm{1}\{\widetilde{T} \ge s\}\log \left(1 + \exp(\boldsymbol{Z}^{\top} \boldsymbol{\theta})\right)\right].
\end{aligned}
\end{equation}
We have 
\begin{equation}\label{2112242}
\begin{aligned}
&\frac{1}{P(\widetilde{T} \geq s)}\mathbbm{E}\left[ \mathbbm{1}\{\widetilde{T} \ge s\}\log \left(1 + \exp(\boldsymbol{Z}^{\top} \boldsymbol{\theta})\right)\right]\\
&= \frac{1}{P(T \geq s,C\ge s)}\mathbbm{E}\left[ \mathbbm{1}\{T \ge s\}\mathbbm{1}\{C \ge s\}\log \left(1 + \exp(\boldsymbol{Z}^{\top} \boldsymbol{\theta})\right)\right]\\
&=  \frac{1}{P(T \geq s)P(C\ge s)}\mathbbm{E}[\mathbbm{1}\{C \ge s\}]\mathbbm{E}\left[ \mathbbm{1}\{T \ge s\}\log \left(1 + \exp(\boldsymbol{Z}^{\top} \boldsymbol{\theta})\right)\right]\\
&= \mathbbm{E}\left[ \left. \mathbbm{1}\{T \ge s\}\log \left(1 + \exp(\boldsymbol{Z}^{\top} \boldsymbol{\theta})\right)\right|T\ge s\right],
\end{aligned}
\end{equation}where, in the second equality, we used Assumption \ref{as1}.
Moreover, it holds that
\begin{equation}\label{2112243}
\begin{aligned}
     &\frac{1}{P(\widetilde{T} \geq s)}\mathbbm{E}\left[\mathbbm{1}\{\widetilde{T} \ge s\}\frac{\delta(t) \mathbbm{1}\{\widetilde{T} \leq t\}}{H(t \wedge \widetilde{T})} \boldsymbol{Z}^{\top} \boldsymbol{\theta} \right]\\
     &= \frac{1}{P(T \geq s,C\ge s)}\mathbbm{E}\left[\mathbbm{1}\{T \ge s\}\frac{\mathbbm{1}\{C \ge T\wedge t\}\mathbbm{1}\{T \leq t\}}{H(T )} \boldsymbol{Z}^{\top} \boldsymbol{\theta} \right]\\
     &= \frac{1}{P(T \geq s)P(C\ge s)}\mathbbm{E}\left[\mathbbm{E}\left[\left.\mathbbm{1}\{C \ge T\}\right|\boldsymbol{Z},T\right]\mathbbm{1}\{T \ge s\}\frac{\mathbbm{1}\{T \leq t\}}{H(T )}\boldsymbol{Z}^{\top} \boldsymbol{\theta} \right]\\
          &= \frac{1}{P(T \geq s)}\mathbbm{E}\left[\mathbbm{1}\{T \ge s\}\mathbbm{1}\{T \leq t\}\boldsymbol{Z}^{\top} \boldsymbol{\theta} \right]\\
          &= \mathbbm{E}\left[\mathbbm{1}\{T \leq t\}\boldsymbol{Z}^{\top} \boldsymbol{\theta}|T \ge s \right],
\end{aligned}
\end{equation}
where, in the first equality, we used that when $\widetilde{T}\le t$ and $C\ge T\wedge t$, then $\widetilde{T}=T\le t$, in the second equality, we leveraged the law of iterated expectations and Assumption \ref{as1} and, in the third equality, we used that for $u\ge s$, $H(u)=P(C\ge u|C\ge s)=P(C\ge u)/P(C\ge s)$ and Assumption \ref{as1} to obtain that $\mathbbm{E}\left[\left.\mathbbm{1}\{C \ge T\}\right|\boldsymbol{Z},T\right]/H(T)=P(C\ge s).$
Combining equations \eqref{2112241}, \eqref{2112242} and \eqref{2112243} and Lemma \ref{lm.cens1}, we obtain the result.
\end{proof}

\section{Estimation theory} \label{proof of the}

This section contains elements allowing us to prove Theorem \ref{the}. In Section \ref{norm}, we present a Lemma about the properties of the sparse-group LASSO norm. In Section \ref{definition}, we define different types of effective sparsity. Section \ref{one p m} introduces the one point margin condition for the conditional population risk function. Section \ref{fuk} presents technical Lemmas which are related to concentration inequalities.
The subsequent two sections \ref{process} and \ref{effective} are dedicated to addressing two key questions:
\begin{itemize}
    \item[(i)] In Section \ref{process}, we establish probability inequalities for the empirical process 
    $$\sup_{\boldsymbol{\beta}\in\mathbb{R}^p: \Omega\left(\boldsymbol{\beta}-\boldsymbol{\beta}_0\right) \leq M}\bigg|\left[R_N(\boldsymbol{\beta})-R(\boldsymbol{\beta}|\boldsymbol{X})\right]-\left[R_N(\boldsymbol{\beta}_0)-R\left(\boldsymbol{\beta}_0|\boldsymbol{X}\right)\right]\bigg|,$$
    where $M > 0 $ is a constant, $R_N(\cdot)$ is the empirical risk function and $R(\cdot|\boldsymbol{X})$ is the conditional population risk function. The presence of censoring and approximation errors further complicates this empirical process.
    \item[(ii)] Section \ref{effective}: We show that the sample effective sparsity can approach the population effective sparsity. Several foundational results were established by \citet{10.1214/15-EJS983}, and these findings were later extended to encompass more general sparsity-inducing norms, as discussed in \citet{van_de_geer_estimation_2016}. It is important to note that these works primarily consider the fixed design or isotropic condition of the covariates. In contrast, our study advances these results by accommodating scenarios involving heavy-tailed data.
\end{itemize}
Notably, all the results we obtain can be readily extended to generalized linear models with structured sparsity estimators. Finally, Section \ref{subsec.proof} contains the proof of Theorem \ref{the}, which builds on the previous sections. 

\paragraph{Notation.}
For a $N \times p$ matrix $\mathbf{A} = (a_{i,j})$, let $\operatorname{vec}(\mathbf{A}) \in \mathbb{R}^{N p}$ be its vectorization consisting of all elements and we denote its entry wise max norm $|\mathbf{A}|_{\infty}=|\operatorname{vec}(\mathbf{A})|_{\infty} = \max _{i, j}\left|a_{i ,j}\right|$.

\subsection{Sparse-group LASSO norm}\label{norm}
We consider the sparse-group LASSO penalty $\Omega(\cdot)$. Note that $\Omega(\cdot)$ can be decomposed as a sum of two seminorms $\Omega(\boldsymbol{b})=\Omega^+(\boldsymbol{b})+\Omega^-(\boldsymbol{b}),$ for all $ \boldsymbol{b} \in \mathbb{R}^{p}$ with
$$
\Omega^+(\boldsymbol{b})=\alpha\left|\boldsymbol{b}_{S_{\boldsymbol{\beta}_0}}\right|_1+(1-\alpha) \sum_{G \in \mathcal{G}_{\boldsymbol{\beta}_0}}\left|\boldsymbol{b}_G\right|_2 , \quad \Omega^-(\boldsymbol{b})=\alpha\left|\boldsymbol{b}_{S_{\boldsymbol{\beta}_0}^{c}}\right|_1 + (1-\alpha) \sum_{G \in \mathcal{G}_{\boldsymbol{\beta}_0}^{c}}\left|\boldsymbol{b}_{G}\right|_2.
$$

\begin{lemma}\label{cor1}
Denote by $\Omega_*(\cdot)$ the dual norm of $\Omega(\cdot)$, it satisfies
\begin{enumerate}
\item For any $x, y \in \mathbb{R}^p$, $\left|x^{\top} y\right| \leq \Omega_*(x) \Omega(y) $. 

\item For any $z \in \mathbb{R}^p$, we have $\Omega_*(z) \leq \alpha|z|_1^*+(1-\alpha)|z|_{2,1}^*,$ where $|\cdot|_1^*$ is the dual norm of $|\cdot|_1$ and $|\cdot|_{2,1}^*$ is the dual norm of $|\cdot|_{2,1}$. 
Furthermore, we also have $\Omega_*(z) \leq \sqrt{G^*}|z|_\infty.$ 

\item For any $ X \in \mathbb{R}^{N \times p}$ and $ z \in \mathbb{R}^p$, we have $\Omega_*(Xz) \leq G^* \left|X\right|_\infty \Omega(z)$.

\end{enumerate}
\end{lemma}

 \begin{proof}
 The first statement is a direct consequence of the definition of the dual norm. For the second statement, $\Omega(\cdot)$ is a norm, and by the convexity of $x \mapsto x^{-1}$ on $(0, \infty)$, we have
$$
\begin{aligned}
\Omega_*(z) & =\sup _{b \neq 0} \frac{|\langle z, b\rangle|}{\Omega(b)} \leq \sup _{b \neq 0}\left\{\alpha \frac{|\langle z, b\rangle|}{|b|_1}+(1-\alpha) \frac{|\langle z, b\rangle|}{|b|_{2,1}}\right\} \\
& \leq \alpha \sup _{b \neq 0} \frac{|\langle z, b\rangle|}{|b|_1}+(1-\alpha) \sup _{b \neq 0} \frac{|\langle z, b\rangle|}{|b|_{2,1}} \\
& =\alpha|z|_1^*+(1-\alpha)|z|_{2,1}^*.
\end{aligned}
$$
We also know that $|z|_1^* = |z|_\infty$ and $|z|_{2,1}^* = \left(\sum_{G \in \mathcal{G}}\left|z_G\right|_2\right)^*=\max\limits_{G \in \mathcal{G}}\left|z_G\right|_2 \leq \sqrt{G^*}|z|_\infty$, see also the appendix of \citet{babii2023machine}. Then
$$
\Omega_*(z) \leq \alpha|z|_1^*+(1-\alpha)|z|_{2,1}^* \leq \alpha|z|_{\infty} + (1-\alpha)\sqrt{G^*}|z|_\infty \leq \sqrt{G^*}|z|_\infty.
$$
The third statement comes from
$$
\begin{aligned} 
\Omega_*(Xz) & \leq \alpha|Xz|_{\infty}+(1-\alpha) \max _{G \in \mathcal{G}}\left|[Xz]_G\right|_2 \\ 
& \leq \alpha|X|_{\infty}|z|_1+(1-\alpha)\sqrt{G^*} |X|_{\infty}\left|z\right|_1 \\ 
& \leq |X|_{\infty}\left(\alpha|z|_1+(1-\alpha) G^*\left|z\right|_{2,1}\right) \\ 
& \leq G^* |X|_{\infty} \Omega(z),
\end{aligned}
$$
since $\max _{G \in \mathcal{G}}|[Xz]_{G}|_2 \leq \sqrt{G^*} |Xz|_2 \leq \sqrt{G^*} |X|_{\infty}|z|_1$ and $|z|_1 \leq \sqrt{G^*} |z|_{2,1} = \sqrt{G^*} \sum_{G \in \mathcal{G}}\left|z_G\right|_2$.
 \end{proof}

 \subsection{Effective sparsity}\label{definition}

First, we present several definitions that are inspired by \citet{van_de_geer_estimation_2016}. For all $ \boldsymbol{\beta}, \Delta \in \mathbb{R}^p$, we define the pseudo-norm $\widehat{\tau}_{\boldsymbol{\beta}}(\cdot)$ and its population version $\tau_{\boldsymbol{\beta}}(\cdot)$
$$
\widehat{\tau}_{\boldsymbol{\beta}}^2(\Delta):= \frac{1}{N} \sum_{i=1}^N \frac{\exp(\boldsymbol{X}_i^{\top}\boldsymbol{\beta} + E_{i})}{ \big(1+\exp(\boldsymbol{X}_i^{\top}\boldsymbol{\beta}+ E_{i})\big)^2 } \left|\boldsymbol{X}_{i}^{\top}\Delta\right|_2^2,\quad
\tau_{\boldsymbol{\beta}}^2(\Delta):=\mathbbm{E} \left[\widehat{\tau}_{\boldsymbol{\beta}}^2(\Delta)\right].
$$
Furthermore, for $M\ge 0$, let 
$$
\frac{1}{\mathrm{C}_M^2(\boldsymbol{X}_i)}=\left(\frac{1}{1+\exp\left(\boldsymbol{X}_{i}^{\top}\boldsymbol{\beta}_0 + E_{i} + M\Omega_{*}(\boldsymbol{X}_i)\right)}\right)\left(1 - \frac{1}{1+\exp\left(\boldsymbol{X}_{i}^{\top}\boldsymbol{\beta}_0 + E_{i} - M\Omega_{*}(\boldsymbol{X}_i)\right)}\right).
$$
It follows that for all $\boldsymbol{\beta}\in\mathbb{R}^p$ that satisfies $\Omega\left(\boldsymbol{\beta}-\boldsymbol{\beta}_0\right) \leq M$
$$
\begin{aligned}
& \frac{\exp(\boldsymbol{X}_i^{\top}\boldsymbol{\beta} + E_{i})}{ \big(1+\exp(\boldsymbol{X}_i^{\top}\boldsymbol{\beta} + E_{i})\big)^2 } \\
& = \left(\frac{1}{1+\exp\left(\boldsymbol{X}_{i}^{\top}\left(\boldsymbol{\beta}- \boldsymbol{\beta}_0 \right)+ \boldsymbol{X}_{i}^{\top}\boldsymbol{\beta}_0 + E_{i}\right)}\right)\left(1 - \frac{1}{1+\exp\left(\boldsymbol{X}_{i}^{\top}\left(\boldsymbol{\beta} - \boldsymbol{\beta}_0 \right)+ \boldsymbol{X}_{i}^{\top}\boldsymbol{\beta}_0 + E_{i}\right)}\right)  \\
& \geq \left(\frac{1}{1+\exp\left(\left|\boldsymbol{X}_{i}^{\top}\left(\boldsymbol{\beta}- \boldsymbol{\beta}_0 \right)\right| + \boldsymbol{X}_{i}^{\top}\boldsymbol{\beta}_0 +E_{i}\right)}\right)\left(1 - \frac{1}{1+\exp\left(-\left|\boldsymbol{X}_{i}^{\top}\left(\boldsymbol{\beta} - \boldsymbol{\beta}_0 \right)\right|+ \boldsymbol{X}_{i}^{\top}\boldsymbol{\beta}_0 + E_{i}\right)}\right)\\
& \geq \frac{1}{\mathrm{C}_M^2(\boldsymbol{X}_i)},
\end{aligned}
$$
since $\left|\boldsymbol{X}_{i}^{\top}\left(\boldsymbol{\beta} - \boldsymbol{\beta}_0 \right)\right| \leq \Omega\left(\boldsymbol{\beta}-\boldsymbol{\beta}_0\right)\Omega_{*}(\boldsymbol{X}_i) \leq M\Omega_{*}(\boldsymbol{X}_i)$.
We also define 
$$
\widehat{\tau}_{M}^2(\Delta):= \frac{1}{N} \sum_{i=1}^N 
\frac{\left|\boldsymbol{X}_{i}^{\top}\Delta\right|_2^2}{\mathrm{C}_M^2(\boldsymbol{X}_i)}.
$$
Notice that, for all $\boldsymbol{\beta}$ that satisfies $\Omega\left(\boldsymbol{\beta}-\boldsymbol{\beta}_0\right) \leq M$, we have $\widehat{\tau}_{\boldsymbol{\beta}}^2(\Delta) \geq \widehat{\tau}_{M}^2(\Delta)$.
For a mapping $\tau:\Delta\in\mathbb{R}^p\mapsto \mathbb{R}$, let us define the effective sparsity
\begin{equation}\label{ccpg}
\Gamma^{2}\left(\tau(\cdot)\right):=\left(\min \left\{\tau^2(\Delta): \Delta \in \mathbb{R}^p, \Omega^{+}(\Delta)=1, \Omega^{-}(\Delta) \leq 2\right\}\right)^{-1}.
\end{equation}
For all $\boldsymbol{\beta}\in \mathbb{R}^p$ and $\Delta$ which satisfies $\Omega^{-}(\Delta) \leq 2 \Omega^{+}(\Delta)$, it holds that
$$
\Omega^{+}(\Delta) \leq \widehat{\tau}_{\boldsymbol{\beta}}\left(\Delta\right)\Gamma\left(\widehat{\tau}_{\boldsymbol{\beta}}\left(\cdot\right)\right),
$$
where we assume $\widehat{\tau}_{\boldsymbol{\beta}}\left(\Delta\right)$ and $\Gamma\left(\widehat{\tau}_{\boldsymbol{\beta}}\left(\cdot\right)\right)$ are positive.
The following is a useful lemma to tie
$\Gamma^2\left(\tau_{\boldsymbol{\beta}_0}(\cdot)\right)$ to Assumption \ref{aseign}.

\begin{lemma}\label{eigen}
    Suppose that Assumption \ref{aseign} holds, we have $\Gamma^2\left(\tau_{\boldsymbol{\beta}_0}(\cdot)\right) \leq \frac{s_{\boldsymbol{\beta}_0}}{\gamma_{\mathrm{H}}}.$
\end{lemma}
\begin{proof}
Take $\Delta\in\mathbb{R}^p$ such that $\Omega^{+}(\Delta)=1$ and $\Omega^{-}(\Delta) \leq 2.$
By the property of the smallest eigenvalue of a symmetric matrix and Assumption \ref{aseign}, we have
    $$
    \gamma_{\mathrm{H}} \leq \min_{|u|_2 = 1} u^{\top}\mathbbm{E}\left[\frac{\exp(\boldsymbol{X}_i^{\top}\boldsymbol{\beta}_0 + E_{i})}{ \big(1+\exp(\boldsymbol{X}_i^{\top}\boldsymbol{\beta}_0+ E_{i})\big)^2 }\boldsymbol{X}_{i} \boldsymbol{X}_{i}^{\top}\right]u.
    $$
    By Jensen's inequality and the inequality of Cauchy-Schwarz, we then see $$
    \sqrt{|S_{\boldsymbol{\beta}_0}|\Delta^{\top}\Delta} \geq \sqrt{|S_{\boldsymbol{\beta}_0}|\Delta_{S_{\boldsymbol{\beta}_0}}^{\top}\Delta_{S_{\boldsymbol{\beta}_0}}} \geq |\Delta_{S_{\boldsymbol{\beta}_0}}|_1, $$ 
    and 
    $$
    \begin{aligned}    
    \sqrt{|\mathcal{G}_{\boldsymbol{\beta}_0}|\Delta^{\top}\Delta} 
    & \geq \sqrt{|\mathcal{G}_{\boldsymbol{\beta}_0}|\sum_{G \in \mathcal{G}_{\boldsymbol{\beta}_0}}\Delta_G^{\top}\Delta_G} \geq \sum_{G \in \mathcal{G}_{\boldsymbol{\beta}_0}}\left|\Delta_G\right|_2.
    \end{aligned}
    $$ 
    


\noindent Overall, since $\sqrt{s_{\boldsymbol{\beta}_0}} = \alpha\sqrt{|S_{\boldsymbol{\beta}_0}|} + (1-\alpha)\sqrt{|\mathcal{G}_{\boldsymbol{\beta}_0}|}$, it is easy to see that
    \begin{equation}\label{omega}
    \begin{aligned}
    \sqrt{s_{\boldsymbol{\beta}_0}\Delta^{\top}\Delta} & = \alpha\sqrt{|S_{\boldsymbol{\beta}_0}|}\sqrt{\Delta^{\top}\Delta} + (1-\alpha)\sqrt{|\mathcal{G}_{\boldsymbol{\beta}_0}|}\sqrt{\Delta^{\top}\Delta} \\
    & \geq \alpha |\Delta_{S_{\boldsymbol{\beta}_0}}|_1 + (1-\alpha)\sum_{G \in \mathcal{G}_{\boldsymbol{\beta}_0}}\left|\Delta_G\right|_2 = \Omega^{+}(\Delta) = 1.
    \end{aligned}
     \end{equation}
   Then, we have 
    $$
    s_{\boldsymbol{\beta}_0}\tau_{\boldsymbol{\beta}_0}^2(\Delta)= \sqrt{s_{\boldsymbol{\beta}_0}}\Delta^{\top}\mathbbm{E}\left[\frac{\exp(\boldsymbol{X}_i^{\top}\boldsymbol{\beta}_0+ E_{i})}{ \big(1+\exp(\boldsymbol{X}_i^{\top}\boldsymbol{\beta}_0+ E_{i})\big)^2 }\boldsymbol{X}_{i} \boldsymbol{X}_{i}^{\top}\right]\Delta\sqrt{s_{\boldsymbol{\beta}_0}} \geq  \gamma_{\mathrm{H}}
    $$
    since $\sqrt{s_{\boldsymbol{\beta}_0}\Delta^{\top}\Delta} \geq 1$, which yields the result.
\end{proof}
  
 \subsection{One point margin condition for conditional risk function}\label{one p m}
For $M>0$, let $\mathbbm{B}_{local} :=\left\{\boldsymbol{\beta}\in\mathbb{R}^p: \Omega\left(\boldsymbol{\beta}-\boldsymbol{\beta}_0\right) \leq M \right\}$, which is a convex neighborhood of the true parameter $\boldsymbol{\beta}_0$. We consider the relationship between the empirical risk function $R_N(\boldsymbol{\beta})$ and the population risk function $R\left(\boldsymbol{\beta}|\boldsymbol{X}\right)$.
We define 
 $$\widehat{f}_i(t) = \frac{\delta_i(t)}{\widehat{H}\left(t \wedge \widetilde{T}_i\right)} \mathbbm{1}{\{\widetilde{T}_i \leq t\}},\
 f_i(t) = \frac{\delta_i(t)}{H\left(t \wedge \widetilde{T}_i\right)} \mathbbm{1}{\{ \widetilde{T}_i \leq t\}}.
 $$ 
Then, we see that
$$
R_N(\boldsymbol{\beta}) = \frac{1}{N} \sum_{i=1}^N \bigg(- \widehat{f}_i(t) \boldsymbol{X}_{i}^{\top}\boldsymbol{\beta} +  \log \left(1+\exp(\boldsymbol{X}_{i}^{\top}\boldsymbol{\beta} )\right)\bigg),
$$
and we can define the conditional population risk as
$$
R(\boldsymbol{\beta}|\boldsymbol{X}) =\frac1N \sum_{i=1}^N   -\left\{\mathbbm{E}[f_i(t)|\boldsymbol{X}_i]\left(\boldsymbol{X}_{i}^{\top}\boldsymbol{\beta} +E_i\right) +  \log \left(1+\exp(\boldsymbol{X}_{i}^{\top}\boldsymbol{\beta} + E_{i})\right)\right\}.
$$
Their gradient and Hessian matrix functions are as follows
\begin{equation}\label{defdotR}
\begin{aligned}
 \dot{R}_{N}(\boldsymbol{\beta})& = \frac{1}{N} \sum_{i=1}^N \bigg( - \widehat{f}_i(t) +  \frac{\exp(\boldsymbol{X}_{i}^{\top}\boldsymbol{\beta})}{1+\exp(\boldsymbol{X}_{i}^{\top}\boldsymbol{\beta}) }\bigg) \boldsymbol{X}_{i}, \\
 \dot{R}(\boldsymbol{\beta}|\boldsymbol{X})& =\frac1N \sum_{i=1}^N   \bigg(-\mathbbm{E}[f_i(t)|\boldsymbol{X}_i] +  \frac{\exp(\boldsymbol{X}_{i}^{\top}\boldsymbol{\beta}+E_i)}{1+\exp(\boldsymbol{X}_{i}^{\top}\boldsymbol{\beta}+E_i)} \bigg) \boldsymbol{X}_i,\\
  \ddot{R}_{N}(\boldsymbol{\beta}) &= \frac{1}{N} \sum_{i=1}^N \frac{\exp(\boldsymbol{X}_i^{\top}\boldsymbol{\beta})}{ \big(1+\exp(\boldsymbol{X}_i^{\top}\boldsymbol{\beta})\big)^2 } \boldsymbol{X}_{i} \boldsymbol{X}_{i}^{\top},
 \\  \ddot{R}(\boldsymbol{\beta}|\boldsymbol{X})& =  \frac{1}{N} \sum_{i=1}^N \frac{\exp(\boldsymbol{X}_i^{\top}\boldsymbol{\beta}+E_i)}{ \big(1+\exp(\boldsymbol{X}_i^{\top}\boldsymbol{\beta}+E_i)\big)^2 } \boldsymbol{X}_{i} \boldsymbol{X}_{i}^{\top}.
\end{aligned}
\end{equation}
 
 \begin{lemma}
   \label{one point}
   (One point margin condition for the conditional risk)  
   We consider the conditional theoretical risk function $R(\boldsymbol{\beta}|\boldsymbol{X})$ here. For all $\tilde{\boldsymbol{\beta}} \in \mathbbm{B}_{local}$, one has the one point margin condition:
$$
R(\tilde{\boldsymbol{\beta}}|\boldsymbol{X})-R\left(\boldsymbol{\beta}_0|\boldsymbol{X}\right) \geq \frac{\widehat{\tau}^2_{M}\left(\tilde{\boldsymbol{\beta}} - \boldsymbol{\beta}_0\right)}{2}.
$$
\end{lemma}
\begin{proof}
    By the Taylor expansion at $\boldsymbol{\beta}_0$, we obtain
    $$
    \begin{aligned}
    & R\left(\tilde{\boldsymbol{\beta}}|\boldsymbol{X}_i\right)-R\left(\boldsymbol{\beta}_0|\boldsymbol{X}\right) \\
    & = \dot{R}\left(\boldsymbol{\beta}_0|\boldsymbol{X}_i\right)^{\top}\left(\tilde{\boldsymbol{\beta}} - \boldsymbol{\beta}_0\right) + \frac{1}{2}\left(\tilde{\boldsymbol{\beta}} - \boldsymbol{\beta}_0\right)^{\top} \ddot{R}\left(\boldsymbol{\beta}^{\prime}|\boldsymbol{X}_i\right) \left(\tilde{\boldsymbol{\beta}} - \boldsymbol{\beta}_0\right) \\
    & = \dot{R}\left(\boldsymbol{\beta}_0|\boldsymbol{X}_i\right)^{\top}\left(\tilde{\boldsymbol{\beta}} - \boldsymbol{\beta}_0\right) \\
    &\quad  + \frac{1}{2}\left(\tilde{\boldsymbol{\beta}} - \boldsymbol{\beta}_0\right)^{\top}\left(\frac{1}{N}\sum_{i = 1}^{N} \frac{\exp\left(\boldsymbol{X}_i^{\top}\boldsymbol{\beta}^{\prime} + E_{i}\right)}{ \big(1+\exp\left(\boldsymbol{X}_i^{\top}\boldsymbol{\beta}^{\prime} + E_{i}\right)\big)^2 }\boldsymbol{X}_i\boldsymbol{X}_i^{\top}\right)\left(\tilde{\boldsymbol{\beta}} - \boldsymbol{\beta}_0\right) \\
    & \geq \dot{R}\left(\boldsymbol{\beta}_0|\boldsymbol{X}_i\right)^{\top}\left(\tilde{\boldsymbol{\beta}} - \boldsymbol{\beta}_0\right) +\frac{\widehat{\tau}^2_{M}\left(\tilde{\boldsymbol{\beta}} - \boldsymbol{\beta}_0\right)}{2},
    \end{aligned}
    $$
    where $\boldsymbol{\beta}^{\prime} \in  \mathbbm{B}_{local}$ since $ \mathbbm{B}_{local}$ is convex. To obtain the result, it suffices to show that 
    $\dot{R}\left(\boldsymbol{\beta}_0|\boldsymbol{X}_i\right) = 0 $. This is obtained through
    $$
    \begin{aligned}
    \dot{R}\left(\boldsymbol{\beta}_0|\boldsymbol{X}_i\right) 
    & = \frac{1}{N} \sum_{i=1}^{N} \left( -\frac{\exp(\boldsymbol{X}_{i}^{\top}\boldsymbol{\beta}_0+ E_{i})}{1+\exp(\boldsymbol{X}_{i}^{\top}\boldsymbol{\beta}_0+ E_{i}) }\boldsymbol{X}_{i} +  \frac{\exp(\boldsymbol{X}_{i}^{\top}\boldsymbol{\beta}_0+ E_{i})}{1+\exp(\boldsymbol{X}_{i}^{\top}\boldsymbol{\beta}_0+ E_{i}) }\boldsymbol{X}_{i} \right)  = 0,
    \end{aligned}
    $$
where we used that
\begin{equation}\label{expf}
\begin{aligned}
\mathbbm{E}[f_i(t)|\boldsymbol{X}_i] & = \mathbbm{E}\left[\left.\frac{\delta_i(t)}{H\left(t \wedge \widetilde{T}_i\right)} \mathbbm{1}{\{\widetilde{T}_i \leq t\}}\right|\boldsymbol{X}_i\right]  = \mathbbm{E}\left[\left.\frac{\delta_i(t)}{H\left(t \wedge T_i\right)} \mathbbm{1}\{T_i \leq t\}\right|\boldsymbol{X}_i\right] \\
& = \mathbbm{E}\left[\left.\frac{\mathbbm{1}{\{T_i \leq t\}}}{H\left(t \wedge T_i\right)} \mathbbm{E}\left[\mathbbm{1}{\{C_i \geq t \wedge T_i \}}|T_i,\boldsymbol{X}_i\right]\right|\boldsymbol{X}_i\right]  = \mathbbm{E}\left[\mathbbm{1}{\{T_i \leq t\}}|\boldsymbol{X}_i\right] \\
& = \frac{\exp \left(\boldsymbol{X}_i^{\top} \boldsymbol{\beta}_0+ E_{i}\right)}{1+\exp \left(\boldsymbol{X}_i^{\top} \boldsymbol{\beta}_0+ E_{i}\right)}
\end{aligned}
\end{equation}
since $\delta_i(t)\mathbbm{1}{\{\widetilde{T}_i \leq t\}} = \mathbbm{1}{\{C_i \geq t \wedge T_i\}}\mathbbm{1}{\{T_i \wedge C_i \leq t\}} =1$ together with $\mathbbm{1}{\{T_i \leq t\}} = 1$ mean $\widetilde{T}_i = T_i$, and $H\left(t \wedge T_i\right) = P(C \geq t \wedge T_i) = \mathbbm{E}\left[\mathbbm{1}{\{C_i \geq t \wedge T_i \}}|T_i, \boldsymbol{X}_i\right]$. 
    \end{proof}

\subsection{A Fuk-Nagaev concentration inequality and some technical lemmas}\label{fuk}
This section presents some useful lemmas for the proof of Theorem \ref{the}.
\begin{theorem}\label{key}
Let $\widehat{\mu}_j=\frac{1}{N} \sum_{i=1}^N Z_{i, j}$ and $\mu_j=\mathbbm{E} [Z_{i, j}],$ where $\boldsymbol{Z}_i = \left(Z_{i,1}, \ldots, Z_{i,p}\right)^{\top}, i\in[N]$ are i.i.d. $p$-dimensional random vectors. There exist universal constants $C_1,\tilde{K}_2>0$ such that 
\begin{equation*}
P\left(\max _{1 \leq j \leq p}\left|\widehat{\mu}_j-\mu_j\right|\right. \left.\geq K_2\frac{\sqrt{\log p}}{\sqrt{N}} + K_{U,1}\frac{\log p}{N^{1-\frac{2}{q}}} + K_{U,2}\frac{\sqrt{\log p}}{\sqrt{N}}\right) 
 \leq \frac{C_{1}}{\log p},
\end{equation*}
where $q \geq 4 $ is defined in Assumption \ref{asX} and 
\begin{align*}
K_2 &= 2\tilde{K}_2 \max\limits_{1 \leq j \leq p}\mathbbm{E}\left[\left|Z_{i,j}\right|^2\right]^{\frac{1}{2}};\\
K_{U,1} &= 2\tilde{K}_2 \mathbbm{E}\left[\max\limits_{1 \leq j \leq p}\left|Z_{i, j}\right|^{\frac{q}{2}}\right]^{\frac{2}{q}};\\
K_{U,2}& = \sqrt{\mathbbm{E}\left[\max\limits_{1 \leq j \leq p}\left|Z_{i, j}\right|^2\right]}.
\end{align*}
\end{theorem}

\begin{proof}
First, let 
 $$\sigma^2 = \max\limits_{1 \leq j \leq p} \sum_{i=1}^{N} \mathbbm{E}[|Z_{i,j}|^2],K_U = \mathbbm{E}\left[\max\limits_{1 \leq j \leq p}\left|Z_{i, j}\right|^2\right],
\tilde{K}_U = \mathbbm{E}\left[\max\limits_{1 \leq j \leq p}\left|Z_{i, j}\right|^{\frac{q}{2}}\right].
 $$
Then we use the following maximal inequality, which follows from Lemma D.2 of \citet{chernozhukov2019inference}, for any $a > 0$,
\begin{equation}\label{maxi}
P\left[\max _{1 \leq j \leq p}\left|\widehat{\mu}_j-\mu_j\right| \geq 2 \mathbbm{E} \left(\max _{1 \leq j \leq p}\left|\widehat{\mu}_j-\mu_j\right|\right) + \frac{a}{N} \right] \leq \exp \left(-\frac{a^2 }{3 \sigma^2}\right)+K_1 \frac{NK_{U}}{a^2},
\end{equation}
where $K_1>0$ is a constant.

Lemma D.3 of \citet{chernozhukov2019inference} shows that with $\tilde{K}_2>0$ a universal positive constant, we have
$$
\begin{aligned}
\mathbbm{E} \left[\max _{1 \leq j \leq p}\left|\widehat{\mu}_j-\mu_j\right|\right] & \leq \tilde{K}_2\left[\frac{\sigma\sqrt{\log p}}{N}+\frac{\sqrt{\mathbbm{E}\left[\max\limits_{1 \leq i \leq N}\max\limits_{1 \leq j \leq p}\left|Z_{i,j}\right|^2]\right]} \log p}{N}\right]\\
& \leq \tilde{K}_2\left[\frac{\sigma\sqrt{\log p}}{N}+\frac{N^{\frac{2}{q}}(\tilde{K}_{U})^{\frac{2}{q}} \log p}{N}\right]
\end{aligned}
$$
since
$$
\begin{aligned}
\mathbbm{E}\left[\max _{1 \leq i \leq N}\max _{1 \leq j \leq p}\left|Z_{i,j}\right|^2\right] & = \mathbbm{E}\left[\max _{1 \leq i \leq N}\left(\max _{1 \leq j \leq p}\left|Z_{i,j}\right|^2\right)\right] \\
& \leq \left(\mathbbm{E}\left[\max _{1 \leq i \leq N}\left(\max _{1 \leq j \leq p}\left|Z_{i,j}\right|^2\right)^{\frac{q}{4}}\right]\right)^{\frac{4}{q}}\\
& \leq \left(\mathbbm{E}\left[N\left(\max _{1 \leq j \leq p}\left|Z_{i,j}\right|^2\right)^{\frac{q}{4}}\right]\right)^{\frac{4}{q}}\\
& = N^{\frac{4}{q}} \left(\mathbbm{E}\left[\left(\max _{1 \leq j \leq p}\left|Z_{i,j}\right|^{\frac{q}{2}}\right)\right]\right)^{\frac{4}{q}}.
\end{aligned}
$$
Let $a=\sqrt{K_U N \log p}$ in \eqref{maxi}, then we have
$$
\begin{aligned}
& P\left(\max _{1 \leq j \leq p}\left|\widehat{\mu}_j-\mu_j\right|\right. \left.\geq 2 \tilde{K}_2\left[\frac{\sigma \sqrt{\log p}}{N}+\frac{(\tilde{K}_{U})^{\frac{2}{q}} \log p}{N^{1-\frac{2}{q}}}\right]+\frac{\sqrt{K_U}\sqrt{\log p}}{\sqrt{N}}\right) \\
& \leq \exp \left(\frac{-K_U N\log p}{3 \sigma^2}\right)+ \frac{K_1}{\log p}.
\end{aligned}
$$
Since $\sigma^2 \le N \max\limits_{ 1 \leq j \leq p}\mathbbm{E}\left[|Z_{i, j}|^2\right] \leq  NK_U$,
we have $\frac{\sigma^2}{N} \leq \max\limits_{ 1 \leq j \leq p}\mathbbm{E}\left[|Z_{i, j}|^2\right] \leq K_U$, so that
$$
\begin{aligned}
&P\left(\max _{1 \leq j \leq p}\left|\widehat{\mu}_j-\mu_j\right|\right. \left.\geq K_2\frac{\sqrt{\log p}}{\sqrt{N}} + K_{U,1}\frac{\log p}{N^{1-\frac{2}{q}}} + K_{U,2}\frac{\sqrt{\log p}}{\sqrt{N}}\right) \\
& = P\left(\max _{1 \leq j \leq p}\left|\widehat{\mu}_j-\mu_j\right|\right. \left.\geq 2 \tilde{K}_2\left[\max\limits_{ 1 \leq j \leq p}\sqrt{\mathbbm{E}\left[|Z_{i, j}|^2\right]}\frac{\sqrt{\log p}}{\sqrt{N}}+\frac{(\tilde{K}_{U})^{\frac{2}{q}} \log p}{N^{1-\frac{2}{q}}}\right] + \frac{\sqrt{K_U}\sqrt{\log p}}{\sqrt{N}}\right) \\ 
& \leq \exp \left(\frac{-K_U N\log p}{3 \sigma^2}\right)+ \frac{K_1}{\log p} \leq \exp \left(-\frac{\log p}{3}\right)+ \frac{K_1}{\log p} \leq \frac{3}{\log p} + \frac{K_1}{\log p} = \frac{C_1}{\log p},
\end{aligned}
$$
where $K_1, \tilde{K}_2$ are universal constants, $C_1 = K_1 + 3$, $K_2 = 2\tilde{K}_2\max\limits_{ 1 \leq j \leq p}\sqrt{\mathbbm{E}\left[|Z_{i, j}|^2\right]}$, $K_{U,1} = 2\tilde{K}_2 (\widetilde{K}_U)^{\frac2q}$, and $K_{U,2} = \sqrt{\mathbbm{E}\left[\max\limits_{1 \leq j \leq p}\left|Z_{i, j}\right|^2\right]}=\sqrt{K_U}$ and we used that 
$$\exp \left(-\frac{\log p}{3}\right)=\frac{1}{\exp \left(\frac{\log p}{3}\right)}\le \frac{1}{\frac{\log p}{3}}=\frac{3}{\log p}.$$
\end{proof}

To simplify presentation, let $\boldsymbol{X}_i := \left(X_{i,1}, \ldots, X_{i,p}\right)^{\top} \in \mathbb{R}^p$. We use this definition in the following Lemmas.
\begin{lemma}
\label{lemmakey1}
Recall that $ f_i(t) = \frac{\delta_i(t)}{H\left(t \wedge \widetilde{T}_i\right)} \mathbbm{1}{\{\widetilde{T}_i \leq t\}}, i \in [N]$. Under Assumptions \ref{as1}, \ref{as2} and \ref{asX}, there exist constants $C_2,A_1,A_2,A_3>0$ such that the event
\begin{equation}
\left| \frac{1}{N} \sum_{i = 1}^{N} f_i(t)\boldsymbol{X}_i - \mathbbm{E}\left[f_i(t)\boldsymbol{X}_i\right]\right|_{\infty} \geq  A_1\frac{\sqrt{\log p}}{\sqrt{N}} + A_2\frac{p^{\frac{1}{q}} \log p}{N^{1-\frac{2}{q}}} + A_3\frac{p^{\frac{1}{q}} \sqrt{\log p}}{\sqrt{N}}
\end{equation}
holds with probability at most $\frac{C_2}{\log p}$. 
\end{lemma}
\begin{proof}
    Let $\boldsymbol{Z}_i$ in Theorem \ref{key} be $f_i(t)\boldsymbol{X}_i$ and by Theorem \ref{key}, we know that 
$$
\left| \frac{1}{N} \sum_{i = 1}^{N} f_i(t)\boldsymbol{X}_i - \mathbbm{E}\left[f_i(t)\boldsymbol{X}_i\right]\right|_{\infty} \geq  K_{4}\frac{\sqrt{\log p}}{\sqrt{N}} + K_{U,3}\frac{ \log p}{N^{1-\frac{2}{q}}} + K_{U,4}\frac{\sqrt{\log p}}{\sqrt{N}}
$$
holds with probability at most $\frac{C_2}{\log p}$, where $C_2, \tilde{K}_4$ are universal constants, and 
\begin{align*}K_4 &= 2\tilde{K}_4 \max\limits_{1 \leq j \leq p}\mathbbm{E}\left[\left|f_i(t)X_{i,j}\right|^2\right]^{\frac{1}{2}} \\
K_{U,3}& = 2\tilde{K}_4 \mathbbm{E}\left[\max\limits_{1 \leq j \leq p}\left|f_i(t)X_{i, j}\right|^{\frac{q}{2}}\right]^{\frac{2}{q}}\\
K_{U,4} &= \sqrt{\mathbbm{E}\left[\max\limits_{1 \leq j \leq p}\left|f_i(t)X_{i, j}\right|^2\right]}.\end{align*}
Next, we see that
    $$
\begin{aligned}
\mathbbm{E}[f_i^{\frac{q}{2}}(t)|\boldsymbol{X}_i] & = \mathbbm{E}\left[\frac{\delta_i^{\frac{q}{2}}(t)}{H^{\frac{q}{2}}\left(t \wedge \widetilde{T}_i\right)} \mathbbm{1}^{\frac{q}{2}}{\{\widetilde{T}_i \leq t\}}\bigg|\boldsymbol{X}_i\right]  = \mathbbm{E}\left[\frac{\delta_i(t)}{H^{\frac{q}{2}}\left(t \wedge T_i\right)} \mathbbm{1}{\{T_i \leq t\}}\bigg|\boldsymbol{X}_i\right] \\
& = \mathbbm{E}\left[\frac{\mathbbm{1}{\{T_i \leq t\}}}{H^{\frac{q}{2}}\left(t \wedge T_i\right)} \mathbbm{E}\left[\mathbbm{1}{\{C_i \geq t \wedge T_i \}}|T,\boldsymbol{X}_i\right]|\boldsymbol{X}_i\right] = \mathbbm{E}\left[\frac{\mathbbm{1}{\{T_i \leq t\}}}{H^{{\frac{q}{2}}-1}\left(t \wedge T_i\right)}\bigg|\boldsymbol{X}_i\right] \\
& \leq \frac{1}{C_r^{{\frac{q}{2}}-1}}\frac{\exp \left(\boldsymbol{X}_i^{\top} \boldsymbol{\beta}_0+ E_{i}\right)}{1+\exp \left(\boldsymbol{X}_i^{\top} \boldsymbol{\beta}_0+ E_{i}\right)}
\end{aligned}
$$
since $\delta_i(t)\mathbbm{1}{\{\widetilde{T}_i \leq t\}} = \mathbbm{1}{\{C_i \geq t \wedge T_i\}}\mathbbm{1}{\{T_i \wedge C_i \leq t\}} =1$ together with $\mathbbm{1}{\{T_i \leq t\}} = 1$ mean $\widetilde{T}_i = T_i$, and $H\left(t \wedge T_i\right) = P(C \geq t \wedge T_i) = \mathbbm{E}\left[\mathbbm{1}{\{C_i \geq t \wedge T_i \}}|T_i, \boldsymbol{X}_i\right]\ge C_r$ by Assumption \ref{as2}.
Then, we have
    $$
    \begin{aligned}
     \mathbbm{E}\left[\max\limits_{1 \leq j \leq p}|f_i(t)X_{i,j}|^{\frac{q}{2}}\right] 
    & = \mathbbm{E}\left[f_i^{\frac{q}{2}}(t)\left(\max\limits_{1 \leq j \leq p}|X_{i,j}|^{\frac{q}{2}}\right)\right] \\
    & = \mathbbm{E}\left[\mathbbm{E}\left[f_i^{\frac{q}{2}}(t)\left(\max\limits_{1 \leq j \leq p}|X_{i,j}|^{\frac{q}{2}}\right)\bigg|\boldsymbol{X}_i\right]\right] \\
    & = \mathbbm{E}\left[\left(\max\limits_{1 \leq j \leq p}|X_{i,j}|^{\frac{q}{2}}\right)\mathbbm{E}\left[f_i^{\frac{q}{2}}(t)|\boldsymbol{X}_i\right]\right]  \\
    & \leq \mathbbm{E}\left[\left(\frac{1}{C_r^{\frac{q}{2}-1}} \frac{\exp \left(\boldsymbol{X}_i^{\top} \boldsymbol{\beta}_0+ E_{i}\right)}{1+\exp \left(\boldsymbol{X}_i^{\top} \boldsymbol{\beta}_0+ E_{i}\right)}\right)\left(\max\limits_{1 \leq j \leq p}|X_{i,j}|^{\frac{q}{2}}\right)\right] \\
    & \leq \frac{1}{{C_r^{\frac{q}{2}-1}}} \mathbbm{E}\left[\max\limits_{1 \leq j \leq p}|X_{i,j}|^{\frac{q}{2}}\right]
    \end{aligned}
    $$
since $\frac{\exp \left(\boldsymbol{X}_i^{\top} \boldsymbol{\beta}_0\right)}{1+\exp \left(\boldsymbol{X}_i^{\top} \boldsymbol{\beta}_0\right)} \leq 1$. 
By Holder’s inequality and Assumption \ref{asX}, we then see that, since $q \geq 4$,
\begin{equation}\label{holder}\begin{aligned}
\mathbbm{E}\left[\max\limits_{1 \leq j \leq p}\left|X_{i, j}\right|^2\right]^{\frac{1}{2}} &\leq \mathbbm{E}\left[\max\limits_{1 \leq j \leq p}\left|X_{i, j}\right|^{\frac{q}{2}}\right]^{\frac{2}{q}}\\
& \le \mathbbm{E}\left[\max\limits_{1 \leq j  \leq p}\left|X_{i, j}\right|^q\right]^{\frac{1}{q}}  \leq \left(p \max\limits_{1 \leq j \leq p} \mathbbm{E}\left[\left|X_{i, j}\right|^q\right]\right)^{\frac{1}{q}} \leq (pK_0)^{\frac{1}{q}},
\end{aligned}
\end{equation}
which means $\mathbbm{E}\left[\max\limits_{1 \leq j \leq p}|f_i(t)X_{i,j}|^{\frac{q}{2}}\right]^{\frac{2}{q}} \leq \frac{(pK_0)^{\frac{1}{q}}}{C_r^{1-\frac{2}{q}}}$ that is $K_{U,3} \leq 2\tilde{K}_4 \frac{(pK_0)^{\frac{1}{q}}}{C_r^{1-\frac{2}{q}}} \sim p^{\frac{1}{q}}$ and $K_{U,4} \leq \frac{(pK_0)^{\frac{1}{q}}}{\sqrt{C_r}} \sim p^{\frac{1}{q}}$. 

Moreover, for $\max\limits_{1 \leq j \leq p}\mathbbm{E}\left[\left|f_i(t)X_{i,j}\right|^2\right]$, we can similarly prove that 
$$
\max\limits_{1 \leq j \leq p}\mathbbm{E}\left[\left|f_i(t)X_{i,j}\right|^2\right] \leq \max\limits_{1 \leq j \leq p}\frac{\mathbbm{E}\left[\left|X_{i,j}\right|^2\right]}{C_r} \leq \frac{ \max\limits_{1 \leq j \leq p}\mathbbm{E}\left[\left|X_{i,j}\right|^q\right]^{\frac{2}{q}}}{C_r}, \quad q \geq 4,
$$
which means $K_4 \leq 2\tilde{K}_4\sqrt{\frac{K_0^{\frac{2}{q}}}{C_r}}$.
Notice that  $K_4$, $K_{U,3}$ and $K_{U,4}$ are all positively related to $K_0$ and negatively related to $C_r$. 

Overall, under Assumptions \ref{as1}, \ref{as2}, and \ref{asX}, there exist a universal constant 
$C_2$ and $A_1, A_2, A_3$ are finite constants depend on $K_0$, such that, we have 
$$
\begin{aligned}
& P\left(\left| \frac{1}{N} \sum_{i = 1}^{N} f_i(t)\boldsymbol{X}_i - \mathbbm{E}\left[f_i(t)\boldsymbol{X}_i\right]\right|_{\infty} \geq  A_1\frac{\sqrt{\log p}}{\sqrt{N}} + A_2\frac{p^{\frac{1}{q}} \log p}{N^{1-\frac{2}{q}}} + A_3\frac{p^{\frac{1}{q}}\sqrt{\log p}}{\sqrt{N}}\right) \\
& \leq P\left(\left| \frac{1}{N} \sum_{i = 1}^{N} f_i(t)\boldsymbol{X}_i - \mathbbm{E}\left[f_i(t)\boldsymbol{X}_i\right]\right|_{\infty} \geq  K_{4}\frac{\sqrt{\log p}}{\sqrt{N}} + K_{U,3}\frac{ \log p}{N^{1-\frac{2}{q}}} + K_{U,4}\frac{\sqrt{\log p}}{\sqrt{N}}\right)\\
& \leq \frac{C_2}{\log p}.
\end{aligned}
$$
\end{proof}

\begin{lemma}
\label{lemmakey2} Under Assumption \ref{asX}, there exist constants $C_3,A_4,A_5,A_6>0$ such that the event
\begin{equation}\notag
\begin{aligned}
&\left| \frac{1}{N} \sum_{i = 1}^{N} \frac{\exp \left(\boldsymbol{X}_i^{\top} \boldsymbol{\beta}+ E_{i}\right)}{1+\exp \left(\boldsymbol{X}_i^{\top} \boldsymbol{\beta}+ E_{i}\right)}\boldsymbol{X}_i - \mathbbm{E}\left[\frac{\exp \left(\boldsymbol{X}_i^{\top} \boldsymbol{\beta}+ E_{i}\right)}{1+\exp \left(\boldsymbol{X}_i^{\top} \boldsymbol{\beta}+ E_{i}\right)}\boldsymbol{X}_i\right]\right|_{\infty}\\
& \geq  A_4\frac{\sqrt{\log p}}{\sqrt{N}} + A_5\frac{p^{\frac{1}{q}} \log p}{N^{1-\frac{2}{q}}} + A_6\frac{p^{\frac{1}{q}} \sqrt{\log p}}{\sqrt{N}}
\end{aligned}
\end{equation}
holds with probability at most $\frac{C_3}{\log p}$.
\end{lemma}
\begin{proof}
    For all $i$, let $\boldsymbol{Z}_i$ in Theorem \ref{key} be $\frac{\exp \left(\boldsymbol{X}_i^{\top} \boldsymbol{\beta} + E_{i}\right)}{1+\exp \left(\boldsymbol{X}_i^{\top} \boldsymbol{\beta} + E_{i}\right)}\boldsymbol{X}_i$ and then we know that 
$$
\begin{aligned}
&\left| \frac{1}{N} \sum_{i = 1}^{N} \frac{\exp \left(\boldsymbol{X}_i^{\top} \boldsymbol{\beta} + E_{i}\right)}{1+\exp \left(\boldsymbol{X}_i^{\top} \boldsymbol{\beta} + E_{i}\right)}\boldsymbol{X}_i - \mathbbm{E}\left[\frac{\exp \left(\boldsymbol{X}_i^{\top} \boldsymbol{\beta} + E_{i}\right)}{1+\exp \left(\boldsymbol{X}_i^{\top} \boldsymbol{\beta} + E_{i}\right)}\boldsymbol{X}_i\right]\right|_{\infty} \\
&\geq  K_{6}\frac{\sqrt{\log p}}{\sqrt{N}} + K_{U,5}\frac{ \log p}{N^{1-\frac{2}{q}}} + K_{U,6}\frac{\sqrt{\log p}}{\sqrt{N}} \end{aligned}
$$
holds with probability at most $\frac{C_3}{\log p}$, where $C_3, \tilde{K}_6$ are universal constants. It follows from the proof of Lemma \ref{lemmakey1} that $$\medmath{
K_6 = 2\tilde{K}_6 \max\limits_{1 \leq j \leq p}\mathbbm{E}\left[\left|\frac{\exp \left(\boldsymbol{X}_i^{\top} \boldsymbol{\beta} + E_{i}\right)}{1+\exp \left(\boldsymbol{X}_i^{\top} \boldsymbol{\beta} + E_{i}\right)}X_{i,j}\right|^2\right]^{\frac{1}{2}} \leq 2\tilde{K}_6 \max\limits_{1 \leq j \leq p}\mathbbm{E}\left[\left|X_{i,j}\right|^2\right]^{\frac{1}{2}} \leq 2\tilde{K}_6 K_0^{\frac{1}{q}}};$$
$$
K_{U,5} = 2\tilde{K}_6 \mathbbm{E}\left[\max\limits_{1 \leq j \leq p}\left|\frac{\exp \left(\boldsymbol{X}_i^{\top} \boldsymbol{\beta} + E_{i}\right)}{1+\exp \left(\boldsymbol{X}_i^{\top} \boldsymbol{\beta} + E_{i}\right)}X_{i, j}\right|^{\frac{q}{2}}\right]^{\frac{2}{q}} \leq 2\tilde{K}_6 \mathbbm{E}\left[\max\limits_{1 \leq j \leq p}\left|X_{i, j}\right|^{\frac{q}{2}}\right]^{\frac{2}{q}} \leq  2\tilde{K}_6(pK_0)^{\frac{1}{q}},$$
and 
$$
K_{U,6} = \sqrt{\mathbbm{E}\left[\max\limits_{1 \leq j \leq p}\left|\frac{\exp \left(\boldsymbol{X}_i^{\top} \boldsymbol{\beta} + E_{i}\right)}{1+\exp \left(\boldsymbol{X}_i^{\top} \boldsymbol{\beta} + E_{i}\right)}X_{i, j}\right|^2\right]} \leq \sqrt{\mathbbm{E}\left[\max\limits_{1 \leq j \leq p}\left|X_{i, j}\right|^2\right]} \leq (pK_0)^{\frac{1}{q}}.$$
Therefore, there exist constants $C_3,A_4, A_5, A_6 >0$ such that \begin{align*}
&\left| \frac{1}{N} \sum_{i = 1}^{N} \frac{\exp \left(\boldsymbol{X}_i^{\top} \boldsymbol{\beta} + E_{i}\right)}{1+\exp \left(\boldsymbol{X}_i^{\top} \boldsymbol{\beta} + E_{i}\right)}\boldsymbol{X}_i - \mathbbm{E}\left[\frac{\exp \left(\boldsymbol{X}_i^{\top} \boldsymbol{\beta} + E_{i}\right)}{1+\exp \left(\boldsymbol{X}_i^{\top} \boldsymbol{\beta} + E_{i}\right)}\boldsymbol{X}_i\right]\right|_{\infty} \\
&\geq  A_4\frac{\sqrt{\log p}}{\sqrt{N}} + A_5\frac{p^{\frac{1}{q}} \log p}{N^{1-\frac{2}{q}}} + A_6\frac{p^{\frac{1}{q}} \sqrt{\log p}}{\sqrt{N}} 
\end{align*}
holds with probability at most $\frac{C_3}{\log p}$.
\end{proof}

\begin{lemma}\label{use sample1}
Under Assumption \ref{asX}, for all $\boldsymbol{\beta}\in\mathbb{R}^p$, there exist constants $C_4,B_1,B_2,B_3$ such that the event
\begin{equation}\label{U_H}
\begin{aligned}
& \left|\frac{1}{N} \sum_{i=1}^N  \frac{\exp(\boldsymbol{X}_i^{\top}\boldsymbol{\beta} + E_{i})}{ \big(1+\exp(\boldsymbol{X}_i^{\top}\boldsymbol{\beta} + E_{i})\big)^2 } \boldsymbol{X}_{i} \boldsymbol{X}_{i}^{\top}  
 - \mathbbm{E}\left[\frac{\exp(\boldsymbol{X}_i^{\top}\boldsymbol{\beta} + E_{i})}{ \big(1+\exp(\boldsymbol{X}_i^{\top}\boldsymbol{\beta} + E_{i})\big)^2 }   \boldsymbol{X}_{i} \boldsymbol{X}_{i}^{\top}  \right]\right|_{\infty} \\
 & \geq B_1\frac{\sqrt{\log p}}{\sqrt{N}} + B_2\frac{p^{\frac{2}{q}}  \log p}{N^{1-\frac{2}{q}}} + B_3\frac{p^{\frac{2}{q}}  \sqrt{\log p}}{\sqrt{N}} 
 \end{aligned}
\end{equation}
holds with probability at most $\frac{C_4}{\log p}$.
\end{lemma}
\begin{proof}
    Let $\boldsymbol{Z}_i$ in Theorem \ref{key} become $\frac{\exp(\boldsymbol{X}_i^{\top}\boldsymbol{\beta} + E_{i})}{ \big(1+\exp(\boldsymbol{X}_i^{\top}\boldsymbol{\beta} + E_{i})\big)^2 }\operatorname{vec}(\boldsymbol{X}_i\boldsymbol{X}_i^{\top})$. By Theorem \ref{key}, we have
    \begin{equation}
  \begin{aligned}  
& \left|\frac{1}{N} \sum_{i=1}^N  \frac{\exp(\boldsymbol{X}_i^{\top}\boldsymbol{\beta} + E_{i})}{ \big(1+\exp(\boldsymbol{X}_i^{\top}\boldsymbol{\beta} + E_{i})\big)^2 } \boldsymbol{X}_{i} \boldsymbol{X}_{i}^{\top}  
 - \mathbbm{E}\left[\frac{\exp(\boldsymbol{X}_i^{\top}\boldsymbol{\beta} + E_{i})}{ \big(1+\exp(\boldsymbol{X}_i^{\top}\boldsymbol{\beta} + E_{i})\big)^2 }   \boldsymbol{X}_{i} \boldsymbol{X}_{i}^{\top}  \right]\right|_{\infty} \\
 & \geq K_8\frac{\sqrt{2\log p}}{\sqrt{N}} + K_{U,7}\frac{2\log p}{N^{1-\frac{2}{q}}} + K_{U,8}\frac{\sqrt{2\log p}}{\sqrt{N}} 
 \end{aligned}
\end{equation}
holds with probability at most $\frac{C_4^{\prime}}{2\log p}$, where $C_4^{\prime}, \tilde{K}_8>0$ are universal constants and, by Assumption \ref{asX},
\begin{align*}K_8 &= 2\tilde{K}_8 \max\limits_{1 \leq j,l \leq p}\mathbbm{E}\left[\left(\frac{\exp(\boldsymbol{X}_i^{\top}\boldsymbol{\beta} + E_{i})}{ \big(1+\exp(\boldsymbol{X}_i^{\top}\boldsymbol{\beta} + E_{i})\big)^2 }\right)^2\left|X_{i,j}X_{i,l}\right|^2\right]^{\frac{1}{2}} \\
&\leq 2\tilde{K}_8 \max\limits_{1 \leq j,l \leq p}\mathbbm{E}\left[\left|X_{i,j}X_{i,l}\right|^2\right]^{\frac{1}{2}} \leq 2\tilde{K}_8K_0^{\frac{2}{q}};
\end{align*}
$$
\begin{aligned}
& K_{U,7} = 2\tilde{K}_8 \mathbbm{E}\left[\max\limits_{1 \leq j,l \leq p}\left(\frac{\exp(\boldsymbol{X}_i^{\top}\boldsymbol{\beta} + E_{i})}{ \big(1+\exp(\boldsymbol{X}_i^{\top}\boldsymbol{\beta} + E_{i})\big)^2 }\right)^{\frac{q}{2}}\left|X_{i, j}X_{i,l}\right|^{\frac{q}{2}}\right]^{\frac{2}{q}} \\
& \leq 2\tilde{K}_8 \mathbbm{E}\left[\max\limits_{1 \leq j,l \leq p}\left|X_{i, j}X_{i,l}\right|^{\frac{q}{2}}\right]^{\frac{2}{q}} \leq 2\tilde{K}_8p^{\frac{2}{q}}K_0^{\frac{2}{q}}
\end{aligned}
$$
and 
\begin{align*}K_{U,8} &= \sqrt{\mathbbm{E}\left[\max\limits_{1 \leq j,l \leq p}\left(\frac{\exp(\boldsymbol{X}_i^{\top}\boldsymbol{\beta} + E_{i})}{ \big(1+\exp(\boldsymbol{X}_i^{\top}\boldsymbol{\beta} + E_{i})\big)^2 }\right)^2\left|X_{i, j}X_{i,l}\right|^2\right]}\\
&\leq \sqrt{\mathbbm{E}\left[\max\limits_{1 \leq j,l \leq p}\left|X_{i, j}X_{i,l}\right|^2\right]}  \leq p^{\frac{2}{q}}K_0^{\frac{2}{q}},\end{align*}
since 
\begin{equation}\label{holder2}
\begin{aligned}
\mathbbm{E}\left[\max\limits_{1 \leq j,l \leq p}\left|X_{i, j}X_{i,l}\right|^2\right]^{\frac{1}{2}}  &\leq  \mathbbm{E}\left[\max\limits_{1 \leq j \leq p}\left|X_{i, j}X_{i,j}\right|^{\frac{q}{2}}\right]^{\frac{2}{q}}  \leq \mathbbm{E}\left[\sum_{j = 1}^{p} \left|X_{i, j}X_{i,j}\right|^{\frac{q}{2}}\right]^{\frac{2}{q}}  \\
& \leq \left(p \max\limits_{1 \leq j \leq p} \mathbbm{E}\left[\left|X_{i, j}X_{i,j}\right|^{\frac{q}{2}}\right]\right)^{\frac{2}{q}} \leq p^{\frac{2}{q}}K_0^{\frac{2}{q}}.
\end{aligned}
\end{equation}
\end{proof}

\begin{lemma}\label{use sample2}
Under Assumption \ref{asX}, for all $\boldsymbol{\beta}\in\mathbb{R}^p$,  there exist constants $C_5,B_4, B_5, B_6>0 $ such that the event
\begin{equation}\label{U_H_2}
\left|\frac{1}{N} \sum_{i=1}^N   \boldsymbol{X}_{i} \boldsymbol{X}_{i}^{\top}  
 - \mathbbm{E}\left[ \boldsymbol{X}_{i} \boldsymbol{X}_{i}^{\top}  \right]\right|_{\infty} \geq B_4\frac{\sqrt{\log p}}{\sqrt{N}} + B_5\frac{p^{\frac{2}{q}}  \log p}{N^{1-\frac{2}{q}}} + B_6\frac{p^{\frac{2}{q}}  \sqrt{\log p}}{\sqrt{N}} 
\end{equation}
holds with probability at most $\frac{C_5}{\log p}$.
\end{lemma}
\begin{proof}
   Applying Theorem \ref{key} with $ \boldsymbol{Z}_i=\operatorname{vec}(\boldsymbol{X}_i\boldsymbol{X}_i^{\top})$, we obtain that
    \begin{equation}
\left|\frac{1}{N} \sum_{i=1}^N   \boldsymbol{X}_{i} \boldsymbol{X}_{i}^{\top}  
 - \mathbbm{E}\left[   \boldsymbol{X}_{i} \boldsymbol{X}_{i}^{\top}  \right]\right|_{\infty} \geq K_{10}\frac{\sqrt{2\log p}}{\sqrt{N}} + K_{U,9}\frac{2\log p}{N^{1-\frac{2}{q}}} + K_{U,10}\frac{\sqrt{2\log p}}{\sqrt{N}} 
\end{equation}
holds with probability at most $\frac{C_5^{\prime}}{2\log p}$, where $C_5^{\prime}, \tilde{K}_8>0$ are universal constants and, using arguments similar to that of the proof of Lemma \ref{lemmakey1}, it holds that
$$K_{10} = 2\tilde{K}_{10} \max\limits_{1 \leq j,l \leq p}\mathbbm{E}\left[\left|X_{i,j}X_{i,l}\right|^2\right]^{\frac{1}{2}}  \lesssim 2\tilde{K}_{10}K_0^{\frac{2}{q}},
$$
$$K_{U,9} = 2\tilde{K}_{10} \mathbbm{E}\left[\max\limits_{1 \leq j,l \leq p}\left|X_{i, j}X_{i,l}\right|^{\frac{q}{2}}\right]^{\frac{2}{q}} \leq 2\tilde{K}_{10}p^{\frac{2}{q}}K_0^{\frac{2}{q}},
$$
and 
$$K_{U,10} = \sqrt{\mathbbm{E}\left[\max\limits_{1 \leq j,l \leq p}\left|X_{i, j}X_{i,l}\right|^2\right]} \leq p^{\frac{2}{q}}K_0^{\frac{2}{q}}.$$
\end{proof}

\subsection{Probability inequalities for empirical process}\label{process}

    \begin{theorem}
    \label{the3}
Under Assumptions \ref{as1}, \ref{as2} and \ref{asX}, there exist constants $c_1,a_1, a_2, a_3>0$ such that, for all $M>0$ and $\epsilon > 0 $, there exists $ U_{\epsilon} > 0$ such that

\begin{equation*}
\begin{aligned}
& P\left(\sup_{\boldsymbol{\beta}\in\mathbb{R}^p: \Omega\left(\boldsymbol{\beta}-\boldsymbol{\beta}_0\right) \leq M}\left|\left[R_N(\boldsymbol{\beta})-R(\boldsymbol{\beta}|\boldsymbol{X})\right]-\left[R_N(\boldsymbol{\beta}_0)-R\left(\boldsymbol{\beta}_0|\boldsymbol{X}\right)\right]\right| \leq \varphi M + \frac{1}{N} \sum_{i=1}^N 2|E_{i}|  \right)\\
& \geq 1 - \epsilon - \frac{c_1}{\log p },
\end{aligned}
\end{equation*}
and 
$$ \varphi = a_1\frac{\sqrt{\log p}}{\sqrt{N}} + a_2\frac{p^{\frac{1}{q}} \log p}{N^{1-\frac{2}{q}}} + a_3\frac{p^{\frac{1}{q}} \sqrt{\log p}}{\sqrt{N}}+ U_{\epsilon}\frac{p^{\frac{1}{q}}(\log p)^{\frac{1}{q}}}{N^{\frac{1}{2} - \frac{1}{q}}}.$$

\end{theorem}

\begin{proof}
First, remark that, by the Taylor expansion, we have that for all $ x,\varepsilon \in \mathbb{R}$, there exists $ x^{\prime} \in  \mathbb{R}$ which satisfies 
$$|\log(1+\exp(x+\varepsilon)) - \log(1+\exp(x)) | = \left|\frac{\exp(x^{\prime})}{1+\exp(x^{\prime})}\varepsilon\right| \leq |\varepsilon|.$$
This implies that
\begin{equation}
\begin{aligned}
& \sup_{\boldsymbol{\beta}\in\mathbb{R}^p: \Omega\left(\boldsymbol{\beta}-\boldsymbol{\beta}_0\right) \leq M} \left|\left[R_N(\boldsymbol{\beta})-R(\boldsymbol{\beta}|\boldsymbol{X})\right]-\left[R_N(\boldsymbol{\beta}_0)-R(\boldsymbol{\beta}_0|\boldsymbol{X})\right]\right|  \\
& \le \sup_{\boldsymbol{\beta}\in\mathbb{R}^p: \Omega\left(\boldsymbol{\beta}-\boldsymbol{\beta}_0\right) \leq M} \left| \frac{1}{N} \sum_{i=1}^N \left( - \widehat{f}_i(t) +  \mathbbm{E}\left[f_i(t)|\boldsymbol{X}_i \right]\right)\boldsymbol{X}_{i}^{\top} \left(\boldsymbol{\beta}-\boldsymbol{\beta}_0 \right)\right| \\
& +  \sup_{\boldsymbol{\beta}\in\mathbb{R}^p: \Omega\left(\boldsymbol{\beta}-\boldsymbol{\beta}_0\right) \leq M} \left| \frac{1}{N} \sum_{i=1}^N\left(\log\left(1+\exp(\boldsymbol{X}_{i}^{\top}\boldsymbol{\beta})\right) - \log\left(1+\exp(\boldsymbol{X}_{i}^{\top}\boldsymbol{\beta} + E_i)\right)\right)\right| \\
& + \sup_{\boldsymbol{\beta}\in\mathbb{R}^p: \Omega\left(\boldsymbol{\beta}-\boldsymbol{\beta}_0\right) \leq M} \left|\frac{1}{N} \sum_{i=1}^N\left(\log\left(1+\exp(\boldsymbol{X}_{i}^{\top}\boldsymbol{\beta}_0 )\right) - \log\left(1+\exp(\boldsymbol{X}_{i}^{\top}\boldsymbol{\beta}_0 + E_i)\right)\right)\right|\\
&\le  \sup_{\boldsymbol{\beta}\in\mathbb{R}^p: \Omega\left(\boldsymbol{\beta}-\boldsymbol{\beta}_0\right) \leq M} \left| \frac{1}{N} \sum_{i=1}^N \left( - \widehat{f}_i(t) +  \mathbbm{E}\left[f_i(t)|\boldsymbol{X}_i \right]\right)\boldsymbol{X}_{i}^{\top} \left(\boldsymbol{\beta}-\boldsymbol{\beta}_0 \right)\right| + \frac{1}{N} \sum_{i=1}^N 2|E_{i}|.
\end{aligned}
\end{equation}
Next, using the second statement of Lemma \ref{cor1}, this leads to
\begin{equation}\label{912241}
\begin{aligned}
& \sup_{\boldsymbol{\beta}\in\mathbb{R}^p: \Omega\left(\boldsymbol{\beta}-\boldsymbol{\beta}_0\right) \leq M} \left|\left[R_N(\boldsymbol{\beta})-R(\boldsymbol{\beta}|\boldsymbol{X})\right]-\left[R_N(\boldsymbol{\beta}_0)-R(\boldsymbol{\beta}_0|\boldsymbol{X})\right]\right| \\
& \leq \Omega_*\left( \frac{1}{N} \sum_{i=1}^N \left( - \widehat{f}_i(t) +  \mathbbm{E}\left[f_i(t)|\boldsymbol{X}_i \right]\right)\boldsymbol{X}_{i}   \right) M+ \frac{1}{N} \sum_{i=1}^N 2|E_{i}|\\
& \leq \sqrt{G^*} \left| \frac{1}{N} \sum_{i=1}^N \left(\widehat{f}_i(t)- \mathbbm{E}[f_i(t)|\boldsymbol{X}_i]\right) \boldsymbol{X}_{i} \right|_{\infty} M + \frac{1}{N} \sum_{i=1}^N 2|E_{i}|.
\end{aligned}
\end{equation}
Then, remark that
\begin{equation}\label{912242}
\begin{aligned}
& \left| \frac{1}{N} \sum_{i=1}^N \left(\widehat{f}_i(t)- \mathbbm{E}[f_i(t)|\boldsymbol{X}_i]\right) \boldsymbol{X}_{i} \right|_{\infty}\\
&  = \left| \frac{1}{N} \sum_{i=1}^N \left(\widehat{f}_i(t)- f_i(t) + f_i(t) - \mathbbm{E}[f_i(t)|\boldsymbol{X}_i]\right) \boldsymbol{X}_{i} \right|_{\infty} \\
& \leq \left| \frac{1}{N} \sum_{i=1}^N \left(\widehat{f}_i(t)- f_i(t)\right) \boldsymbol{X}_{i} \right|_{\infty} \\
&\quad + \left| \frac{1}{N} \sum_{i=1}^N \left(f_i(t)\boldsymbol{X}_i - \mathbbm{E}\left[f_i(t)\boldsymbol{X}_i\right] + \mathbbm{E}\left[f_i(t)\boldsymbol{X}_i\right] - \mathbbm{E}[f_i(t)|\boldsymbol{X}_i] \boldsymbol{X}_{i}\right) \right|_{\infty}\\
& \le (a1) + (a2) + (a3),
\end{aligned}
\end{equation}
where 
\begin{align*}
  (a1) &=  \left| \frac{1}{N} \sum_{i=1}^N \left(\widehat{f}_i(t)- f_i(t)\right) \boldsymbol{X}_{i} \right|_{\infty}; \\
  (a2) &= \left| \frac{1}{N} \sum_{i=1}^N \left(f_i(t)\boldsymbol{X}_i -\mathbbm{E}\left[f_i(t)\boldsymbol{X}_i\right]\right)\right|_{\infty};\\
  (a3) &=  \left| \frac{1}{N} \sum_{i=1}^N \left(\mathbbm{E}[f_i(t)|\boldsymbol{X}_i] \boldsymbol{X}_{i} - \mathbbm{E}\left[f_i(t)\boldsymbol{X}_i\right]\right)\right|_{\infty}.\\
\end{align*}
First, we bound the term $(a1)$. By definition of $\widehat{f}_i(t)$ and $f_i(t)$, we have 
$$
\begin{aligned}
\left|\widehat{f}_i(t) - f_i(t)\right| & = \left|\frac{\delta_i(t)}{\widehat{H}\left(t \wedge \widetilde{T}_i\right)} \mathbbm{1}{\{\widetilde{T}_i \leq t\}} - \frac{\delta_i(t)}{H\left(t \wedge \widetilde{T}_i\right)} \mathbbm{1}{\{\widetilde{T}_i \leq t\}}\right|\\
& \leq \left|\frac{1}{\widehat{H}\left(t \wedge \widetilde{T}_i\right)}-\frac{1}{H\left(t \wedge \widetilde{T}_i\right)}\right| \\
& = \left|\frac{\widehat{H}\left(t \wedge \widetilde{T}_i\right) - H\left(t \wedge \widetilde{T}_i\right)}{\widehat{H}\left(t \wedge \widetilde{T}_i\right)H\left(t \wedge \widetilde{T}_i\right)}\right|.
\end{aligned}
$$
By standard properties of the Kaplan-Meier estimator $\widehat{H}$, (see Theorem 1.1 in \citet{gill1983large}), it holds
\begin{equation}\label{kmrate}
\sup _{x \in[0, t]}|H(x)-\widehat{H}(x)|=O_P\left(N^{-1 / 2}\right),
\end{equation}
where we used Assumption \ref{as2}. Moreover, by Assumption \ref{as2}, we have $H\left(t \wedge \widetilde{T}_i\right)\ge C_r$, which leads to
\begin{equation}\label{kmrate2}
\max\limits_{1 \leq i \leq N}\left| \frac{1}{\widehat{H}\left(t \wedge \widetilde{T}_i\right)}-\frac{1}{H\left(t \wedge \widetilde{T}_i\right)}\right|= O_P\left(N^{-1 / 2}\right),
\end{equation}
which means
\begin{equation}\label{pa1}
\max\limits_{1 \leq i \leq N}\left|\widehat{f}_i(t) - f_i(t)\right| = O_P\left(N^{-1 / 2}\right).
\end{equation}
Then, for all $\epsilon > 0$, there exists $ U_{\epsilon}^{\prime} > 0$ such that
\begin{equation}
    \label{pa23}
    P\left(\max\limits_{1 \leq i \leq N}\left|\widehat{f}_i(t) - f_i(t)\right| \leq \frac{U_{\epsilon}^{\prime}}{\sqrt{N}}\right) \geq 1 -\epsilon.
\end{equation}
For all $k > 0$, by Markov's inequality, we have
$$
\begin{aligned}
P\left(\max_{1 \leq i \leq N} \max_{1 \leq j \leq p} \left|X_{i,j}\right| > k\right) & \leq \sum_{1 \leq i \leq N}\sum_{1 \leq j \leq p} P\left(\left|X_{i,j}\right| > k\right) \\
& \leq Np \max_{1 \leq i \leq N, 1 \leq j \leq p} P\left(\left|X_{i,j}\right| > k\right) \\
& \leq Np\frac{K_0}{k^q},
\end{aligned}
$$
since for all $ i \in [N], \max\limits_{1 \leq j \leq p} \mathbbm{E}\left(\left|X_{i,j}\right|^q\right) \leq K_0$ from Assumption \ref{asX}. Let $k = (K_0Np\log p)^{\frac{1}{q}}$, then, we have

\begin{equation}\label{pa3}
P\left(\max_{1 \leq i \leq N, 1 \leq j \leq p} \left|X_{i,j}\right| > (K_0Np\log p)^{\frac{1}{q}}\right) \leq Np\frac{\max\limits_{1 \leq i \leq N,1 \leq j \leq p} \mathbbm{E}\left(\left|X_{i,j}\right|^q\right)}{NK_0p \log p} \leq \frac{1}{\log p}.
\end{equation}
Combining \eqref{pa23} and \eqref{pa3}, we have that for all $ \epsilon > 0$, there exists $ U_{\epsilon} = U_{\epsilon}^{\prime}\sqrt{G^*}K_0^{\frac{1}{q}} > 0$
\begin{equation}
    \label{pa4}
\sqrt{G^*} \times (a1) \leq  \sqrt{G^*} \times \max\limits_{1 \leq i \leq N}\left|\widehat{f}_i(t) - f_i(t)\right|  \max\limits_{1 \leq i \leq N} \left|\boldsymbol{X}_i\right|_{\infty} \leq \frac{U_{\epsilon}p^{\frac{1}{q}}(\log p)^{\frac{1}{q}}}{N^{\frac{1}{2} - \frac{1}{q}}}
\end{equation}
holds with probability $1 - \epsilon - \frac{1}{\log p}$.

Second, we consider the term $(a2)$. By Lemma \ref{lemmakey1}, there exists constants $C_2,A_1,A_2,A_3>0$ such that the event
\begin{equation}
\label{pa5}
\left| \frac{1}{N} \sum_{i=1}^N \left(f_i(t)\boldsymbol{X}_i -\mathbbm{E}\left[f_i(t)\boldsymbol{X}_i\right]\right)\right|_{\infty} \geq  A_1\frac{\sqrt{\log p}}{\sqrt{N}} + A_2\frac{p^{\frac{1}{q}}\log p}{N^{1-\frac{2}{q}}} + A_3\frac{p^{\frac{1}{q}}\sqrt{\log p}}{\sqrt{N}}
\end{equation}
holds with probability at most $C_2/\log p$.
Finally, concerning $(a3)$, we know that
$$
\mathbbm{E}[f_i(t)|\boldsymbol{X}_i]  = \frac{\exp \left(\boldsymbol{X}_i^{\top} \boldsymbol{\beta}_0 +E_{i}\right)}{1+\exp \left(\boldsymbol{X}_i^{\top} \boldsymbol{\beta}_0 +E_{i} \right)},
$$
and
$$
\mathbbm{E}[f_i(t)\boldsymbol{X}_i] = \mathbbm{E}\left[\mathbbm{E}[f_i(t)\boldsymbol{X}_i|\boldsymbol{X}_i]\right] = \mathbbm{E}\left[\boldsymbol{X}_i\mathbbm{E}[f_i(t)|\boldsymbol{X}_i]\right] = \mathbbm{E}\left[\frac{\exp \left(\boldsymbol{X}_i^{\top} \boldsymbol{\beta}_0 +E_{i}\right)}{1+\exp \left(\boldsymbol{X}_i^{\top} \boldsymbol{\beta}_0 +E_{i}\right)}\boldsymbol{X}_i\right].
$$
Applying Lemma \ref{lemmakey2}, we obtain that there exist constant $C_3,A_4,A_5,A_6$ such that the event
\begin{equation}
\label{pa6}
\left|\frac{1}{N} \sum_{i=1}^N \mathbbm{E}[f_i(t)|\boldsymbol{X}_i] \boldsymbol{X}_{i} - \mathbbm{E}\left[f_i(t)\boldsymbol{X}_i\right] \right|_{\infty}  \geq  A_4\frac{\sqrt{\log p}}{\sqrt{N}} + A_5\frac{p^{\frac{1}{q}}\log p}{N^{1-\frac{2}{q}}} + A_6\frac{p^{\frac{1}{q}}\sqrt{\log p}}{\sqrt{N}}
\end{equation}
holds with probability at most $C_3/\log p$.

The result follows by combining \eqref{912241}, \eqref{912242}, \eqref{pa4}, \eqref{pa5} and \eqref{pa6}, with $a_1 = \sqrt{G^*}(A_1 + A_4)$, $a_2 = \sqrt{G^*}(A_2 + A_5)$, $a_3 = \sqrt{G^*}(A_3 + A_6)$, $c_1 = 1 + C_2 + C_3$.
\end{proof}

\subsection{Relationship between population and sample effective sparsity}\label{effective}
 
Now, we tie the population effective sparsity $\Gamma^{-2}\left(\tau_{\boldsymbol{\beta}}(\cdot)\right)$ and the sample effective sparsity $\Gamma^{-2}\left(\widehat{\tau}_{\boldsymbol{\beta}}(\cdot)\right)$. 

\begin{lemma}\label{sample1} Under Assumption \ref{asX}, there exist universal constants $C_4,B_1,B_2,B_3>0$ such that
 $$ P\left[\Gamma^{-2}\left(\widehat{\tau}_{\boldsymbol{\beta}_0}(\cdot)\right) \geq  \Gamma^{-2}\left(\tau_{\boldsymbol{\beta}_0}(\cdot)\right) - U_{H_1}\right]\ge 1 - \frac{C_4}{\log p},
    $$
where $U_{H_1} = 9G^*\left(B_1\frac{\sqrt{\log p}}{\sqrt{N}} + B_2\frac{p^{\frac{2}{q}}  \log p}{N^{1-\frac{2}{q}}} + B_3\frac{p^{\frac{2}{q}}  \sqrt{\log p}}{\sqrt{N}}\right). $
\end{lemma}
\begin{proof}
   Let $\widehat{\Sigma} = \frac{1}{N} \sum_{i = 1}^{N} \frac{\exp(\boldsymbol{X}_i^{\top}\boldsymbol{\beta} + E_{i})}{ \big(1+\exp(\boldsymbol{X}_i^{\top}\boldsymbol{\beta} + E_{i})\big)^2 }\boldsymbol{X}_{i} \boldsymbol{X}_{i}^{\top}$and $\Sigma = \mathbbm{E}\left[\widehat{\Sigma}\right]$. We have 
 $$
 \begin{aligned}
 \left|\Delta^{\top} \widehat{\Sigma}\Delta\right| & \geq \left|\Delta^{\top} \mathbbm{E}\left[\widehat{\Sigma}\right]\Delta\right| - \left|\Delta^{\top} \left(\widehat{\Sigma} - \mathbbm{E}\left[\widehat{\Sigma}\right]\right)\Delta\right| \\
 &\geq\left|\Delta^{\top} \mathbbm{E}\left[\widehat{\Sigma}\right] \Delta\right|-G^*\Omega^2\left(\Delta\right)\left|\widehat{\Sigma}-\mathbbm{E}\left[\widehat{\Sigma}\right]\right|_{\infty},
 \end{aligned}
 $$
 where we used that, by the first and third statements of Lemma \ref{cor1}, it holds that
$$
 \begin{aligned}
 \left|\Delta^{\top}\left(\widehat{\Sigma} - \mathbbm{E}\left[\widehat{\Sigma}\right]\right)\Delta\right| & \leq \Omega\left(\Delta\right)\Omega_*\left(\left(\widehat{\Sigma} - \mathbbm{E}\left[\widehat{\Sigma}\right]\right)\Delta\right)\\
 & \leq \Omega\left(\Delta\right) G^* \left|\left(\widehat{\Sigma} - \mathbbm{E}\left[\widehat{\Sigma}\right]\right)\right|_{\infty}\Omega\left(\Delta\right)\\
 & =  G^* \Omega^2\left(\Delta\right)  \left|\widehat{\Sigma} - \mathbbm{E}\left[\widehat{\Sigma}\right]\right|_{\infty}.
  \end{aligned}
$$
Remark that, for all $\Delta\in\mathbb{R}^p$ such that $\Omega^{+}(\Delta)=1$ and $\Omega^{-}(\Delta) \leq 2$, we have $\Omega\left(\Delta\right) = \Omega^{+}(\Delta) + \Omega^{-}(\Delta) \leq 3$ and, therefore, 
$$
\left|\Delta^{\top} \widehat{\Sigma}\Delta\right| \geq \left|\Delta^{\top} \mathbbm{E}\left[\widehat{\Sigma}\right] \Delta\right|-9G^*\left|\widehat{\Sigma}-\mathbbm{E}\left[\widehat{\Sigma}\right]\right|_{\infty}.
$$
By Lemma \ref{use sample1}, we have
$$
P\left(9G^*\left|\widehat{\Sigma}-\mathbbm{E}\left[\widehat{\Sigma}\right]\right|_{\infty} \leq U_{H_1}\right) \geq 1 - \frac{C_4}{\log p}.
$$
Then, by definition $\widehat{\tau}_{\boldsymbol{\beta}_0}^2\left(\Delta\right)$ and $\tau_{\boldsymbol{\beta}_0}^2\left(\Delta\right) = \mathbbm{E}\left(\widehat{\tau}_{\boldsymbol{\beta}_0}^2\left(\Delta\right)\right)$, we obtain,  for all $\Delta\in\mathbb{R}^p$ such that $\Omega^{+}(\Delta)=1$ and $\Omega^{-}(\Delta) \leq 2$, that
\begin{equation}\label{mini}
P\left[\widehat{\tau}_{\boldsymbol{\beta}_0}^2\left(\Delta\right) \geq \tau_{\boldsymbol{\beta}_0}^2\left(\Delta\right) - U_{H_1}\right]\ge 1 - \frac{C_4}{\log p}.
\end{equation}
Recall the definition of \eqref{ccpg} and minimizing both the $\widehat{\tau}_{\boldsymbol{\beta}_0}^2\left(\Delta\right)$ and $\tau_{\boldsymbol{\beta}_0}^2\left(\Delta\right)$ with respect to $\{\Delta\in\mathbb{R}^p:\text{$\Omega^{+}(\Delta)=1$ and $\Omega^{-}(\Delta) \leq 2$}\}$, we obtain
$$
P\left[\Gamma^{-2}\left(\widehat{\tau}_{\boldsymbol{\beta}_0}(\cdot)\right) \geq \Gamma^{-2}(\tau_{\boldsymbol{\beta}_0}(\cdot))- U_{H_1}\right]\ge 1 - \frac{C_4}{\log p}.
$$
\end{proof}

\begin{lemma}\label{sample2}Under Assumption \ref{asX}, there exist universal constants $ C_5, B_4, B_5, B_6>0$ such that, for all $M\ge 0$ and $\boldsymbol{\beta}\in \mathbb{R}^p$, such that $\Omega\left(\boldsymbol{\beta} - \boldsymbol{\beta}_0 \right) \leq M$,
$$
P\left[\Gamma^{-2}( \widehat{\tau}_{\boldsymbol{\beta}}(\cdot)) \geq \Gamma^{-2}\left(\widehat{\tau}_{\boldsymbol{\beta}_0}(\cdot)\right)-U_{H_2}\right] \geq 1 - \frac{C_5}{\log p}- \frac{1}{\log p},
$$
where
$$U_{H_2}= 18(G^*)^{\frac{3}{2}}\left(\left(B_4\frac{\sqrt{\log p}}{\sqrt{N}} + B_5\frac{p^{\frac{2}{q}} \log p}{N^{1-\frac{2}{q}}} + B_6\frac{p^{\frac{2}{q}} \sqrt{\log p}}{\sqrt{N}}\right) + K_0^{\frac{2}{q}}\right)M(K_0Np\log p)^{\frac{1}{q}}$$
Furthermore, it holds that
$$
P\left[\Gamma^{-2}\left(\widehat{\tau}_{M}(\cdot)\right) \geq \Gamma^{-2}\left(\widehat{\tau}_{\boldsymbol{\beta}_0}(\cdot)\right)-U_{H_2}\right] \geq 1 - \frac{C_5}{\log p}- \frac{1}{\log p}.
$$
\end{lemma}
\begin{proof}
Let $\widehat{\Sigma}_{M} = \frac{1}{N} \sum_{i = 1}^{N} \frac{\boldsymbol{X}_{i} \boldsymbol{X}_{i}^{\top}}{ \mathrm{C}_M^2(\boldsymbol{X}_i)}$ and $\widehat{\Sigma} = \frac{1}{N} \sum_{i = 1}^{N} \frac{\exp(\boldsymbol{X}_i^{\top}\boldsymbol{\beta}_0+ E_{i})}{ \big(1+\exp(\boldsymbol{X}_i^{\top}\boldsymbol{\beta}_0+ E_{i})\big)^2 }\boldsymbol{X}_{i} \boldsymbol{X}_{i}^{\top}$. For any $\Delta\in\mathbb{R}^p$ such that $\Omega^{+}(\Delta)=1$ and $\Omega^{-}(\Delta) \leq 2$, we have
 $$
 \begin{aligned}
 \left|\Delta^{\top} \widehat{\Sigma}_{M} \Delta\right| & \geq \left|\Delta^{\top} \widehat{\Sigma} \Delta\right| - \left|\Delta^{\top} \left(\widehat{\Sigma}_{M} - \widehat{\Sigma} \right)\Delta\right|.
 \end{aligned}
 $$
Since
$$
\begin{aligned}
\left|\Delta^{\top} \left(\widehat{\Sigma}_{M} - \widehat{\Sigma} \right)\Delta\right| & = \left|\frac{1}{N} \sum_{i = 1}^{N} \left(\frac{1}{ \mathrm{C}_M^2(\boldsymbol{X}_i)} -   \frac{\exp(\boldsymbol{X}_i^{\top}\boldsymbol{\beta}_0+ E_{i})}{ \big(1+\exp(\boldsymbol{X}_i^{\top}\boldsymbol{\beta}_0+ E_{i})\big)^2 }\right) \Delta^{\top}\boldsymbol{X}_{i} \boldsymbol{X}_{i}^{\top}\Delta  \right|\\
& \leq \max_{i \in [N]} \left|\frac{1}{ \mathrm{C}_M^2(\boldsymbol{X}_i)} -   \frac{\exp(\boldsymbol{X}_i^{\top}\boldsymbol{\beta}_0+ E_{i})}{ \big(1+\exp(\boldsymbol{X}_i^{\top}\boldsymbol{\beta}_0+ E_{i})\big)^2 }\right|\left|\frac{1}{N}  \sum_{i = 1}^{N} \Delta^{\top}\boldsymbol{X}_{i} \boldsymbol{X}_{i}^{\top}\Delta\right|.
\end{aligned}
$$
Moreover, letting $A_i=\boldsymbol{X}_{i}^{\top}\boldsymbol{\beta}_0+ E_{i}$,
we have 
\begin{equation*}
\begin{aligned}
& \left|\frac{1}{ \mathrm{C}_M^2(\boldsymbol{X}_i)} - \frac{\exp(A_i)}{ \big(1+\exp(A_i)\big)^2 }\right|\\
& \leq \left|\frac{1}{ \mathrm{C}_M^2(\boldsymbol{X}_i)} - \left(\frac{1}{1+\exp\left(A_i+ M\Omega_{*}(\boldsymbol{X}_i)\right)}\right)\left(1 - \frac{1}{1+\exp\left(A_i\right)}\right) \right|\\
& + \left|\left(\frac{1}{1+\exp\left(A_i+ M\Omega_{*}(\boldsymbol{X}_i)\right)}\right)\left(1 - \frac{1}{1+\exp\left(A_i\right)}\right) - \frac{\exp(A_i)}{ \big(1+\exp(A_i)\big)^2 } \right|\\
& = \left|\left(\frac{1}{1+\exp\left(A_i + M\Omega_{*}(\boldsymbol{X}_i)\right)}\right)\left(\frac{1}{1+\exp\left(A_i\right)} - \frac{1}{1+\exp\left(A_i - M\Omega_{*}(\boldsymbol{X}_i)\right)}\right)\right| \\
& + \left| \left(\frac{1}{1+\exp\left(A_i + M\Omega_{*}(\boldsymbol{X}_i)\right)} - \frac{1}{1+\exp\left(A_i\right)} \right)\left(1 - \frac{1}{1+\exp\left(A_i\right)}\right)\right|\\
& \leq \left|\frac{1}{1+\exp\left(A_i\right)} - \frac{1}{1+\exp\left(A_i - M\Omega_{*}(\boldsymbol{X}_i)\right)}\right|  + \left| \frac{1}{1+\exp\left(A_i + M\Omega_{*}(\boldsymbol{X}_i)\right)} - \frac{1}{1+\exp\left(A_i\right)} \right|\\
& \leq 2\left| M\Omega_{*}(\boldsymbol{X}_i)\right|,
\end{aligned}
\end{equation*}
where the last inequality comes from that, for all $ x, \epsilon \in \mathbb{R}$, there exists $ x^{\prime}$ satisfying
$$
\left|\frac{1}{1+\exp(x+e)} - \frac{1}{1+\exp(x)}\right| = \left|\frac{-\exp(x^{\prime})}{(1+\exp(x^{\prime}))^2}e\right| \leq |e|,
$$
by a Taylor expansion.
For all $ k > 0$, by Markov's inequality, we see that
$$
\begin{aligned}
P\left(\max_{1 \leq i \leq N} \max_{1 \leq j \leq p} \left|X_{i,j}\right| > k\right) & \leq \sum_{1 \leq i \leq N}\sum_{1 \leq j \leq p} P\left(\left|X_{i,j}\right| > k\right) \\
& \leq Np \max_{1 \leq i \leq N, 1 \leq j \leq p} P\left(\left|X_{i,j}\right| > k\right) \\
& \leq Np\frac{\max\limits_{1 \leq i \leq N, 1 \leq j \leq p} \mathbbm{E}\left(\left|X_{i,j}\right|^q\right)}{k^q},
\end{aligned}
$$
since for all $ i \in [N], \max\limits_{1 \leq j \leq p} \mathbbm{E}\left(\left|X_{i,j}\right|^q\right) \leq K_0$ by Assumption \ref{asX}. Let $k = (K_0Np\log p)^{\frac{1}{q}}$. Then, we have
\begin{equation}\label{maxmax}
P\left(\max_{i \in [N], j \in[p]} \left|X_{i,j}\right| > (K_0Np\log p)^{\frac{1}{q}}\right) \leq Np\frac{\max\limits_{i \in [N], j \in[p]} \mathbbm{E}\left(\left|X_{i,j}\right|^q\right)}{K_0Np\log p} \leq \frac{1}{\log p}.
\end{equation}
By the second statement of Lemma \ref{cor1} and \eqref{maxmax}, we obtain that
\begin{equation}\label{pa2}
\max_{i \in [N]} \Omega_*\left(\boldsymbol{X}_i\right) \leq \max_{i \in [N]} \sqrt{G^*} \left|\boldsymbol{X}_i\right|_{\infty} \leq \sqrt{G^*}(K_0Np\log p)^{\frac{1}{q}} 
\end{equation}
holds with probability at least $1 - \frac{1}{\log p}$.
It follows that
\begin{equation}\label{1}
\begin{aligned}
\max_{i \in [N]} \left|\frac{1}{ \mathrm{C}_M^2(\boldsymbol{X}_i)} -   \frac{\exp(\boldsymbol{X}_i^{\top}\boldsymbol{\beta}_0+E_i)}{ \big(1+\exp(\boldsymbol{X}_i^{\top}\boldsymbol{\beta}_0+E_i)\big)^2 }\right| & \leq 2\max_{i \in [N]} M\Omega_*\left( \boldsymbol{X}_i \right)\\
& \leq 2M\sqrt{G^*}(K_0Np\log p)^{\frac{1}{q}} 
\end{aligned}
\end{equation}
holds with probability at least $1 - \frac{1}{\log p}$.

As for $\left|\frac{1}{N}  \sum_{i = 1}^{N} \Delta^{\top}\boldsymbol{X}_{i} \boldsymbol{X}_{i}^{\top}\Delta\right|$, since $\Omega^{+}(\Delta)=1$ and $\Omega^{-}(\Delta) \leq 2$, we can see $\Omega\left(\Delta\right) = \Omega^{+}(\Delta) + \Omega^{-}(\Delta) \leq 3$. Then, we have that
\begin{equation}\label{2}
\begin{aligned}
 &\left|\frac{1}{N}  \sum_{i = 1}^{N} \Delta^{\top}\boldsymbol{X}_{i} \boldsymbol{X}_{i}^{\top}\Delta\right| \\
 &= \left|\frac{1}{N}  \sum_{i = 1}^{N} \Delta^{\top}\left(\boldsymbol{X}_{i} \boldsymbol{X}_{i}^{\top} -\mathbbm{E}\left(\boldsymbol{X}_{i} \boldsymbol{X}_{i}^{\top}\right)\right)\Delta + \Delta^{\top}\mathbbm{E}\left(\boldsymbol{X}_{i} \boldsymbol{X}_{i}^{\top}\right)\Delta\right|\\
& \leq G^*\Omega^2(\Delta)\left| \frac{1}{N}  \sum_{i = 1}^{N} \boldsymbol{X}_{i} \boldsymbol{X}_{i}^{\top} -\mathbbm{E}\left(\boldsymbol{X}_{i} \boldsymbol{X}_{i}^{\top}\right)\right|_{\infty} + \max_{i \in [N]} G^*\Omega^2(\Delta)\left|\mathbbm{E}\left(\boldsymbol{X}_{i} \boldsymbol{X}_{i}^{\top}\right)\right|_{\infty}\\
&\leq 9G^* \left(B_4\frac{\sqrt{\log p}}{\sqrt{N}} + B_5\frac{p^{\frac{2}{q}} \log p}{N^{1-\frac{2}{q}}} + B_6\frac{p^{\frac{2}{q}} \sqrt{\log p}}{\sqrt{N}}\right) + 9G^*K_0^{\frac{2}{q}},
\end{aligned}
\end{equation}
holds with probability at least $1-\frac{C_{5}}{\log p}$ for some constants $B_4,B_5,B_6>0$. The first inequality comes from the third statement of Lemma \ref{cor1}. The first part of the last inequality is obtained from Lemma \ref{use sample2} and the second part is from Assumption \ref{asX}, which implies $\max\limits_{1 \leq j,l \leq p} \mathbbm{E}\left(\left|X_{i,j}X_{i,l}\right|\right) \leq  \max\limits_{1 \leq j,l \leq p} \mathbbm{E}\left(\left|X_{i,j}X_{i,l}\right|^{\frac{q}{2}}\right)^{\frac{2}{q}} \le K_0^{\frac{2}{q}}$.

Combining \eqref{1} and \eqref{2}, we conclude that
$$
 \left|\Delta^{\top} \left(\widehat{\Sigma}_{M} - \widehat{\Sigma} \right)\Delta\right| \leq U_{H_2}
$$
 holds with probability at least $1 -\frac{C_5}{\log p}- \frac{1}{\log p}$.
  By the same argument as in Lemma \ref{sample1}, we have 
 \begin{equation}\label{mini2}
P\left(\widehat{\tau}_{M}^2\left(\Delta\right) \geq \widehat{\tau}_{\boldsymbol{\beta}_0}^2\left(\Delta\right) - U_{H_2}\right) \geq 1 -\frac{C_5}{\log p}- \frac{1}{\log p}.
\end{equation}
By definition of $\widehat{\tau}_{M}^2(\Delta)$, for all $ \boldsymbol{\beta}$ satisfying $\Omega\left(\boldsymbol{\beta}-\boldsymbol{\beta}_0\right) \leq M$, we have $\widehat{\tau}_{\boldsymbol{\beta}}^2(\Delta) \geq \widehat{\tau}_{M}^2(\Delta)$. It follows
 \begin{equation}\label{mini3}
P\left(\widehat{\tau}_{\boldsymbol{\beta}}^2(\Delta) \geq \widehat{\tau}_{\boldsymbol{\beta}_0}^2\left(\Delta\right) - U_{H_2}\right) \geq 1 -\frac{C_5}{\log p}- \frac{1}{\log p}.
\end{equation}
Minimizing both left and right sides of the events in \eqref{mini2} and \eqref{mini3} with respect to $\{\Delta\in\mathbb{R}^p:\text{$\Omega^{+}(\Delta)=1$ and $\Omega^{-}(\Delta) \leq 2$}\}$, we obtain
$$
P\left[\Gamma^{-2}\left(\widehat{\tau}_{M}(\cdot)\right) \geq \Gamma^{-2}\left(\widehat{\tau}_{\boldsymbol{\beta}_0}(\cdot)\right)-U_{H_2}\right] \geq 1 - \frac{C_5}{\log p}- \frac{1}{\log p},
$$
and
$$
P\left[\Gamma^{-2}(\widehat{\tau}_{\boldsymbol{\beta}}(\cdot)) \geq \Gamma^{-2}\left(\widehat{\tau}_{\boldsymbol{\beta}_0}(\cdot)\right)-U_{H_2}\right] \geq 1 - \frac{C_5}{\log p}- \frac{1}{\log p}.
$$
\end{proof}

Combining Lemmas \ref{sample1} and \ref{sample2}, we obtain the following result.
\begin{lemma}\label{constant}Under Assumption \ref{asX}, for all $M\ge 0$, we have
$$
P\left[\Gamma^{-2}\left(\widehat{\tau}_{M}(\cdot)\right) \geq \Gamma^{-2}(\tau_{\boldsymbol{\beta}_0}(\cdot)) - U_{H_1} - U_{H_2}\right] \geq 1 -\frac{C_6}{\log p},
$$
where $U_{H_1},U_{H_2},$ and $C_6 = C_4 + C_5 + 1$ are defined in Lemmas \ref{sample1} and \ref{sample2}.
\end{lemma}

Finally, we have the following result.
\begin{lemma}\label{constant2}Let Assumptions \ref{asX} and \ref{aseign} hold. For all $M\ge 0$, if 
$(U_{H_1}+U_{H_2})\frac{s_{\boldsymbol{\beta}_0}}{\gamma_{\mathrm{H}}}\le \frac12,$
then
$$ P\left(\Gamma^{2}\left( \widehat{\tau}_{M}(\cdot)\right) \leq \frac{2s_{\boldsymbol{\beta}_0}}{\gamma_{\mathrm{H}}}\right) \geq 1 -\frac{C_6}{\log p},$$
where $U_{H_1},U_{H_2},$ and $C_6 = C_4 + C_5 + 1$ are defined in Lemmas \ref{sample1} and \ref{sample2}.
\end{lemma}
\begin{proof}
We work on the event $$
\mathcal{E}=\left\{\Gamma^{-2}\left(\widehat{\tau}_{M}(\cdot)\right) \geq \Gamma^{-2}\left( \tau_{\boldsymbol{\beta}_0}(\cdot)\right) - U_{H_1} - U_{H_2}\right\},
$$ which has probability at least $1 -\frac{C_6}{\log p}$ by Lemma \ref{constant}. Remark that, by Lemma \ref{eigen}, we know that $\Gamma^{2}\left(\tau_{\boldsymbol{\beta}_0}(\cdot)\right) \leq \frac{s_{\boldsymbol{\beta}_0}}{\gamma_{\mathrm{H}}}$, so that $(U_{H_1} + U_{H_2})\Gamma^{2}\left(\tau_{\boldsymbol{\beta}_0}(\cdot)\right) \leq \frac{1}{2}$ since $(U_{H_1}+U_{H_2})\frac{s_{\boldsymbol{\beta}_0}}{\gamma_{\mathrm{H}}}\le \frac12$ by assumption.
Therefore, on $\mathcal{E}$, we have
 \begin{align*}
     \Gamma^{2}\left( \widehat{\tau}_{M}(\cdot)\right)&\le \frac{\Gamma^{2}\left(\tau_{\boldsymbol{\beta}_0}(\cdot)\right)}{1 - (U_{H_1} + U_{H_2})\Gamma^{2}\left(\tau_{\boldsymbol{\beta}_0}(\cdot)\right)} \\
     &\le  2 \Gamma^{2}\left(\tau_{\boldsymbol{\beta}_0}(\cdot)\right)\le  \frac{2s_{\boldsymbol{\beta}_0}}{\gamma_{\mathrm{H}}}.
     \end{align*}
This yields the result.
\end{proof}

\subsection{Proof of Theorem \ref{the}}\label{subsec.proof}

Let 
$$M_{\boldsymbol{\beta}_0}=18 \frac{\lambda s_{\boldsymbol{\beta}_0}}{\gamma_{\mathrm{H}}} + \lambda^{-1}\frac{24}{N} |\boldsymbol{E}|_1$$
and $$\mathbbm{B}_{local} =\left\{\boldsymbol{\beta}\in\mathbb{R}^p: \Omega\left(\boldsymbol{\beta}-\boldsymbol{\beta}_0\right) \leq M_{\boldsymbol{\beta}_0}\right\}$$ be a convex neighborhood of the target parameter $\boldsymbol{\beta}_0$. 
Define $\tilde{\boldsymbol{\beta}}:=d \widehat{\boldsymbol{\beta}}+(1-d) \boldsymbol{\beta}_0$, where $d:=\frac{M_{\boldsymbol{\beta}_0}}{M_{\boldsymbol{\beta}_0}+\Omega\left(\widehat{\boldsymbol{\beta}}-\boldsymbol{\beta}_0\right)}$.
We have
\begin{equation}\label{leqM}
\begin{aligned}
\Omega\left(\tilde{\boldsymbol{\beta}}-\boldsymbol{\beta}_0\right) & =\Omega\left(d \widehat{\boldsymbol{\beta}}+(1-d) \boldsymbol{\beta}_0-\boldsymbol{\beta}_0\right)=d \Omega\left(\widehat{\boldsymbol{\beta}}-\boldsymbol{\beta}_0\right) \\
& =\frac{M_{\boldsymbol{\beta}_0} \Omega\left(\widehat{\boldsymbol{\beta}}-\boldsymbol{\beta}_0\right)}{M_{\boldsymbol{\beta}_0}+\Omega\left(\widehat{\boldsymbol{\beta}}-\boldsymbol{\beta}_0\right)} \leq M_{\boldsymbol{\beta}_0},
\end{aligned}
\end{equation}
which shows that $\Omega\left(\tilde{\boldsymbol{\beta}}-\boldsymbol{\beta}_0\right) \leq M_{\boldsymbol{\beta}_0} $. 

The rest of the proof shows that $\Omega\left(\widehat{\boldsymbol{\beta}}-\boldsymbol{\beta}_0\right) \leq M_{\boldsymbol{\beta}_0} $ with large probability. Since $R_N(\cdot)+\lambda \Omega(\cdot)$ is convex and $\widehat{\boldsymbol{\beta}}$ is the minimizer of \eqref{midas}, we get
\begin{equation}\label{eq1}
\begin{aligned}
R_N(\tilde{\boldsymbol{\beta}})+\lambda \Omega(\tilde{\boldsymbol{\beta}}) & \leq d R_N(\widehat{\boldsymbol{\beta}})+d \lambda \Omega(\widehat{\boldsymbol{\beta}})+(1-d) R_N(\boldsymbol{\beta}_0)+(1-d) \lambda \Omega\left(\boldsymbol{\beta}_0\right) \\
& \leq R_N(\boldsymbol{\beta}_0)+\lambda \Omega\left(\boldsymbol{\beta}_0\right) .
\end{aligned}
\end{equation}
Adding the term $\left(R(\tilde{\boldsymbol{\beta}}|\boldsymbol{X})-R\left(\boldsymbol{\beta}_0|\boldsymbol{X}\right)\right)$ on both sides of \eqref{eq1}, we obtain
\begin{equation}
\label{eq2}
\begin{aligned}
&R(\tilde{\boldsymbol{\beta}}|\boldsymbol{X})-R\left(\boldsymbol{\beta}_0|\boldsymbol{X}\right)\\
&\leq-\left[\left(R_N(\tilde{\boldsymbol{\beta}})-R(\tilde{\boldsymbol{\beta}}|\boldsymbol{X})\right)-\left(R_N(\boldsymbol{\beta}_0)-R\left(\boldsymbol{\beta}_0|\boldsymbol{X}\right)\right)\right]+\lambda \Omega\left(\boldsymbol{\beta}_0\right)-\lambda \Omega(\tilde{\boldsymbol{\beta}}) .
\end{aligned}
\end{equation}
By definition of $\Omega(\cdot),\Omega^+(\cdot)$ and $\Omega^-(\cdot)$, we have
\begin{equation}\label{eq2u}
\begin{aligned}
\Omega(\boldsymbol{\beta}_0) - \Omega(\tilde{\boldsymbol{\beta}}) & = \Omega^+(\boldsymbol{\beta}_0) + \Omega^-(\boldsymbol{\beta}_0) - \Omega^+(\tilde{\boldsymbol{\beta}}) - \Omega^-(\tilde{\boldsymbol{\beta}})\\
& = \Omega^+(\boldsymbol{\beta}_0) - \Omega^+(\tilde{\boldsymbol{\beta}}) - \Omega^-(\tilde{\boldsymbol{\beta}}) \\
&= \Omega^+(\boldsymbol{\beta}_0) - \Omega^+(\tilde{\boldsymbol{\beta}}) - \Omega^-\left(\tilde{\boldsymbol{\beta}} - \boldsymbol{\beta}_0\right)\\
& \leq \Omega^{+}\left(\tilde{\boldsymbol{\beta}} - \boldsymbol{\beta}_0\right) - \Omega^-\left(\tilde{\boldsymbol{\beta}} - \boldsymbol{\beta}_0\right).
\end{aligned}
\end{equation}
Let us define the event 
$$\mathcal{A}=\left\{R(\tilde{\boldsymbol{\beta}}|\boldsymbol{X})-R\left(\boldsymbol{\beta}_0|\boldsymbol{X}\right)\leq \frac{\lambda}{12}M_{\boldsymbol{\beta}_0}+ \frac{2}{N} |\boldsymbol{E}|_1+  \lambda \left\{\Omega^{+}\left(\tilde{\boldsymbol{\beta}} - \boldsymbol{\beta}_0\right) - \Omega^{-}\left(\tilde{\boldsymbol{\beta}} - \boldsymbol{\beta}_0\right)\right\}\right\}.$$
By Theorem \ref{the3}, together with Assumption \ref{aseff}, \eqref{eq2} and \eqref{eq2u}, since $\varphi=O\left(p^{\frac{1}{q}}\sqrt{\log p}/N^{\frac{1}{2} - \frac{1}{q}}\right) \lesssim \lambda$, we have $P\left(\mathcal{A} \right)\to 1.$
 Let also

$$
\mathcal{B}=\left\{\Gamma^{2}\left( \widehat{\tau}_{M_{\boldsymbol{\beta}_0}}\left(\cdot\right)\right) \leq \frac{2s_{\boldsymbol{\beta}_0}}{\gamma_{\mathrm{H}}}\right\}.
$$
Note that, by Lemma \ref{constant2} and Assumption \ref{aseff} which implies that $s_{\boldsymbol{\beta}_0}U_{H_1}=o(1)$ and $\frac{s_{\boldsymbol{\beta}_0}U_{H_1} + s_{\boldsymbol{\beta}_0}U_{H_2}}{\gamma_{\mathrm{H}}} \leq 1/2$, 
we have
$P\left(\mathcal{B} \right)\to 1$.

Let us introduce the event 
$$\mathcal{C}=\left\{ \lambda \Omega\left(\tilde{\boldsymbol{\beta}} - \boldsymbol{\beta}_0\right)\leq \frac{9\lambda^2  s_{\boldsymbol{\beta}_0}}{2\gamma_{\mathrm{H}}} + \frac{\lambda}{4} M_{\boldsymbol{\beta}_0} + \frac{6}{N} |\boldsymbol{E}|_1\right\}.$$
We will show that $\mathcal{C}$ holds almost surely on the event $\mathcal{A}\cap\mathcal{B}$. To do so, we work on $\mathcal{A}\cap\mathcal{B}$ and consider two cases.

\medskip

\paragraph{Case $1$.} In this case, we assume that $\lambda \Omega^{+}\left(\tilde{\boldsymbol{\beta}} - \boldsymbol{\beta}_0\right) \leq \frac{\lambda}{12} M_{\boldsymbol{\beta}_0} + \frac{2}{N} |\boldsymbol{E}|_1.$ 

\noindent Since we work on the event $\mathcal{A}$, we obtain
 $$\lambda \Omega^{-}\left(\tilde{\boldsymbol{\beta}} - \boldsymbol{\beta}_0\right) + R(\tilde{\boldsymbol{\beta}}|\boldsymbol{X}) - R(\boldsymbol{\beta}_0|\boldsymbol{X}) \leq 2 \frac{\lambda}{12} M_{\boldsymbol{\beta}_0} + \frac{4}{N} |\boldsymbol{E}|_1.$$
Since $R(\tilde{\boldsymbol{\beta}}|\boldsymbol{X})-R\left(\boldsymbol{\beta}_0|\boldsymbol{X}\right) \geq 0$ by Lemma \ref{one point}, this yields $$
 \lambda \Omega^{-}\left(\tilde{\boldsymbol{\beta}} - \boldsymbol{\beta}_0\right) \leq \frac{\lambda}{6} M_{\boldsymbol{\beta}_0} + \frac{4}{N} |\boldsymbol{E}|_1.
 $$
Using $\lambda \Omega^{+}\left(\tilde{\boldsymbol{\beta}} - \boldsymbol{\beta}_0\right) \leq \frac{\lambda}{12} M_{\boldsymbol{\beta}_0} + \frac{2}{N} |\boldsymbol{E}|_1$, we obtain that, 
 \begin{equation}\label{c1}
 \lambda \Omega\left(\tilde{\boldsymbol{\beta}} - \boldsymbol{\beta}_0\right) \leq \frac{\lambda}{4} M_{\boldsymbol{\beta}_0} + \frac{6}{N} |\boldsymbol{E}|_1.
 \end{equation}
 so that $\mathcal{C}$ holds.
\paragraph{Case $2$.} In this case, we assume that we assume that $\lambda \Omega^{+}\left(\tilde{\boldsymbol{\beta}} - \boldsymbol{\beta}_0\right) >\frac{\lambda}{12} M_{\boldsymbol{\beta}_0} + \frac{2}{N} |\boldsymbol{E}|_1.$

\noindent Since we work on $\mathcal{A}$, we obtain
 \begin{equation}
\label{eq4}
 \begin{aligned}
 R(\tilde{\boldsymbol{\beta}}|\boldsymbol{X})-R\left(\boldsymbol{\beta}_0|\boldsymbol{X}\right) + \lambda \Omega^{-}\left(\tilde{\boldsymbol{\beta}} - \boldsymbol{\beta}_0\right) &\leq \frac{\lambda}{12} M_{\boldsymbol{\beta}_0}+ \frac{1}{N} \sum_{i=1}^N 2|E_{i}| + \lambda \Omega^{+}\left(\tilde{\boldsymbol{\beta}} - \boldsymbol{\beta}_0\right) \\
 & \leq 2\lambda \Omega^{+}\left(\tilde{\boldsymbol{\beta}} - \boldsymbol{\beta}_0\right) .
 \end{aligned}
\end{equation}
Since $R(\tilde{\boldsymbol{\beta}}|\boldsymbol{X})-R\left(\boldsymbol{\beta}_0|\boldsymbol{X}\right) \geq 0$ by Lemma \ref{one point}, we get $\Omega^{-}\left(\tilde{\boldsymbol{\beta}} - \boldsymbol{\beta}_0\right) \leq 2\Omega^{+}\left(\tilde{\boldsymbol{\beta}} - \boldsymbol{\beta}_0\right)$. 
By definition of the effective sparsity, we obtain
$$
\Omega^{+}\left(\tilde{\boldsymbol{\beta}} - \boldsymbol{\beta}_0\right) \leq \widehat{\tau}_{M_{\boldsymbol{\beta}_0}}\left(\tilde{\boldsymbol{\beta}} - \boldsymbol{\beta}_0\right)\Gamma\left(  \widehat{\tau}_{M_{\boldsymbol{\beta}_0}}\left(\cdot\right)\right).
$$
Then adding $\frac{1}{2} \lambda \Omega^{+}\left(\tilde{\boldsymbol{\beta}} - \boldsymbol{\beta}_0\right) $ on both sides  of \eqref{eq4} and using the inequality $|ab|\le (a^2+b^2)/2$, we get
\begin{equation}\label{any}
\begin{aligned}
& R(\tilde{\boldsymbol{\beta}}|\boldsymbol{X})-R\left(\boldsymbol{\beta}_0|\boldsymbol{X}\right) + \lambda \Omega^{-}\left(\tilde{\boldsymbol{\beta}} - \boldsymbol{\beta}_0\right) +\frac{1}{2} \lambda \Omega^{+}\left(\tilde{\boldsymbol{\beta}} - \boldsymbol{\beta}_0\right) \\
& \leq \frac{3}{2} \lambda \Omega^{+}\left(\tilde{\boldsymbol{\beta}} - \boldsymbol{\beta}_0\right) + \frac{\lambda}{12} M_{\boldsymbol{\beta}_0} + \frac{1}{N} \sum_{i=1}^N 2|E_{i}|\\
& \leq \frac{3}{2} \lambda \widehat{\tau}_{M_{\boldsymbol{\beta}_0}}\left(\tilde{\boldsymbol{\beta}} - \boldsymbol{\beta}_0\right)\Gamma\left( \widehat{\tau}_{M_{\boldsymbol{\beta}_0}}\left(\cdot\right)\right) +\frac{\lambda}{12}M_{\boldsymbol{\beta}_0} + \frac{1}{N} \sum_{i=1}^N 2|E_{i}|\\
& \leq \frac{\widehat{\tau}^2_{M_{\boldsymbol{\beta}_0}}\left(\tilde{\boldsymbol{\beta}} - \boldsymbol{\beta}_0\right)}{2} + \frac{9}{8}\lambda^2 \Gamma^2\left(  \widehat{\tau}_{M_{\boldsymbol{\beta}_0}}\left(\cdot\right)\right)  + \frac{\lambda}{12}M_{\boldsymbol{\beta}_0}+ \frac{1}{N} \sum_{i=1}^N 2|E_{i}|.
\end{aligned}
\end{equation}
Since $\tilde{\boldsymbol{\beta}} \in  \mathbbm{B}_{local}$, by Lemma \ref{one point}, we know that 
$$
R(\tilde{\boldsymbol{\beta}}|\boldsymbol{X})-R\left(\boldsymbol{\beta}_0|\boldsymbol{X}\right) \geq \frac{\widehat{\tau}^2_{M_{\boldsymbol{\beta}_0}}\left(\tilde{\boldsymbol{\beta}} - \boldsymbol{\beta}_0\right)}{2}.
$$
Then, we have
\begin{equation}\label{last key}
\begin{aligned}
 & \lambda \Omega^{-}\left(\tilde{\boldsymbol{\beta}} - \boldsymbol{\beta}_0\right) +\frac{1}{2} \lambda \Omega^{+}\left(\tilde{\boldsymbol{\beta}} - \boldsymbol{\beta}_0\right) \\
 & \leq  \frac{9}{8}\lambda^2 \Gamma^2\left(   \widehat{\tau}_{M_{\boldsymbol{\beta}_0}}\left(\cdot\right)\right) + \frac{\lambda}{12}M_{\boldsymbol{\beta}_0} + \frac{2}{N} |\boldsymbol{E}|_1.
 \end{aligned}
\end{equation}
We obtain that $\mathcal{C}$ holds since we work on $\mathcal{B}$. 

\medskip

We have now shown that $\mathcal{C}$ holds on $\mathcal{A}\cap\mathcal{B}$. Thus, we obtain that $P\left( \mathcal{C}\right) \to 1$ since $P(\mathcal{A}\cap\mathcal{B})\to 1.$
Now, let us work on the event $\mathcal{C}$. On $\mathcal{C}$, we have 
\begin{equation}\label{M}
\begin{aligned}
\frac{M_{\boldsymbol{\beta}_0}}{M_{\boldsymbol{\beta}_0} + \Omega\left(\widehat{\boldsymbol{\beta}}-\boldsymbol{\beta}_0\right)} \Omega\left(\widehat{\boldsymbol{\beta}} - \boldsymbol{\beta}_0\right)   & = \Omega\left(\tilde{\boldsymbol{\beta}} - \boldsymbol{\beta}_0\right)  \\
& \leq \frac{9\lambda  s_{\boldsymbol{\beta}_0}}{2\gamma_{\mathrm{H}}} + \frac{\frac{6}{N} |\boldsymbol{E}|_1}{\lambda} + \frac{\lambda}{4\lambda} M_{\boldsymbol{\beta}_0}\\
& \leq  \frac{1}{4}M_{\boldsymbol{\beta}_0} + \frac{1}{4}M_{\boldsymbol{\beta}_0} = \frac{1}{2}M_{\boldsymbol{\beta}_0}.
\end{aligned}
\end{equation}
Since
$$
\frac{M_{\boldsymbol{\beta}_0}}{M_{\boldsymbol{\beta}_0} + \Omega\left(\widehat{\boldsymbol{\beta}}-\boldsymbol{\beta}_0\right)} \Omega\left(\widehat{\boldsymbol{\beta}} - \boldsymbol{\beta}_0\right) \leq \frac{M_{\boldsymbol{\beta}_0}}{2} \\
\Leftrightarrow 2 \Omega\left(\widehat{\boldsymbol{\beta}} - \boldsymbol{\beta}_0\right) \leq M_{\boldsymbol{\beta}_0} + \Omega\left(\widehat{\boldsymbol{\beta}} - \boldsymbol{\beta}_0\right),
$$
it holds that
\begin{equation*}\label{neighbor}
P\left(\Omega\left(\widehat{\boldsymbol{\beta}} - \boldsymbol{\beta}_0\right)\le M_{\boldsymbol{\beta}_0}\right)\to 1.
\end{equation*}
Furthermore, using Lemma \ref{cor1}, we conclude that with probability going to $1$, it holds that
\begin{equation*}
\begin{aligned}
\widehat{P}(\boldsymbol{z}) -  P(\boldsymbol{z}) & = \frac{\exp \left(\boldsymbol{x}^{\top}\widehat{\boldsymbol{\beta}} \right)}{1 + \exp \left(\boldsymbol{x}^{\top}\widehat{\boldsymbol{\beta}}\right)} -  \frac{\exp \left(\boldsymbol{x}^{\top}\boldsymbol{\beta}_0 + e \right)}{1 + \exp \left(\boldsymbol{x}^{\top}\boldsymbol{\beta}_0 + e\right)}\\
& \leq \left|\boldsymbol{x}^{\top}\widehat{\boldsymbol{\beta}} -\boldsymbol{x}^{\top}\boldsymbol{\beta}_0 - e\right|\\
& \leq \left|\boldsymbol{x}^{\top}\widehat{\boldsymbol{\beta}} -\boldsymbol{x}^{\top}\boldsymbol{\beta}_0\right| + \left|e\right|\\
& \leq \Omega_{*}\left(\boldsymbol{x}\right) \Omega\left(\widehat{\boldsymbol{\beta}} - \boldsymbol{\beta}_0\right) + \left|e\right| \\
& \lesssim \frac{\lambda s_{\boldsymbol{\beta}_0}\left|\boldsymbol{x}\right|_{\infty}}{\gamma_{\mathrm{H}}} +  \lambda^{-1}\left|\boldsymbol{x}\right|_{\infty}\frac{1}{N} |\boldsymbol{E}|_1 + |e|,
\end{aligned}
\end{equation*}
where the first inequality comes from that, for all $ x, \varepsilon \in \mathbb{R}$, there exists $ x^{\prime}$ satisfying
$$
\left|\frac{\exp(x+\varepsilon)}{1+\exp(x+\varepsilon)} - \frac{\exp(x)}{1+\exp(x)}\right| = \left|\frac{\exp(x^{\prime})}{(1+\exp(x^{\prime}))^2}\varepsilon\right| \leq |\varepsilon|,
$$
by a Taylor expansion. The last inequality comes from the second statement of Lemma \ref{cor1} and the value of $M_{\boldsymbol{\beta}_0}$.

{\color{black}
\section{Inference theory}\label{appinf}
\subsection{Proof of Theorem \ref{inf the}}
We begin with the de-sparsified sparse-group LASSO estimator $\widehat{\boldsymbol{b}}_{(\mathcal{J})}$. Recall the definition of $\dot{R}_N(\widehat{\boldsymbol{\beta}})$ in \eqref{defdotR}, we have
\begin{equation}\label{preexp}
\begin{aligned}
\widehat{\boldsymbol{b}}_{(\mathcal{J})} - \boldsymbol{\beta}_{0,(\mathcal{J})} 
&= \left(\widehat{\boldsymbol{\beta}}_{(\mathcal{J})} - \boldsymbol{\beta}_{0,(\mathcal{J})}\right) 
    - \widehat{\boldsymbol{\Theta}}_\mathcal{J}\left[\dot{R}_N(\widehat{\boldsymbol{\beta}})\right] \\
&= \left(\widehat{\boldsymbol{\beta}}_{(\mathcal{J})} - \boldsymbol{\beta}_{0,(\mathcal{J})}\right) 
   - \widehat{\boldsymbol{\Theta}}_\mathcal{J}\left[\dot{R}_N(\boldsymbol{\beta}_0)
   + \ddot{R}_N(\tilde{\boldsymbol{\beta}})\left(\widehat{\boldsymbol{\beta}} - \boldsymbol{\beta}_0\right)\right] \\
&= - \widehat{\boldsymbol{\Theta}}_\mathcal{J}\left[\dot{R}_N(\boldsymbol{\beta}_0)\right]
   - \left(\widehat{\boldsymbol{\Theta}} \ddot{R}_N(\widehat{\boldsymbol{\beta}}) - \boldsymbol{I}_p\right)_\mathcal{J}
     \left(\widehat{\boldsymbol{\beta}} - \boldsymbol{\beta}_0\right) \\
&\quad - \widehat{\boldsymbol{\Theta}}_\mathcal{J}
   \left[\ddot{R}_N(\tilde{\boldsymbol{\beta}}) - \ddot{R}_N(\widehat{\boldsymbol{\beta}})\right]
   \left(\widehat{\boldsymbol{\beta}} - \boldsymbol{\beta}_0\right),
\end{aligned}
\end{equation}
where we use the mean value theorem and $\tilde{\boldsymbol{\beta}}$ is between $\widehat{\boldsymbol{\beta}}$ and $\boldsymbol{\beta}_0$, and $\boldsymbol{I}_p$ denotes the $p \times p$ identity matrix.

To simplify the notation, let
\begin{equation}\label{rho_hat}
\widehat{\rho}_N(\boldsymbol{\beta}) 
= \frac{1}{N} \sum_{i=1}^{N} 
  \left(- \widehat{f}_i(t) 
  +  \frac{\exp(\boldsymbol{X}_{i}^{\top}\boldsymbol{\beta}+E_i)}
{1+\exp(\boldsymbol{X}_{i}^{\top}\boldsymbol{\beta}+E_i)}\right)\boldsymbol{X}_i,
\end{equation}
and
\begin{equation}
\rho_N(\boldsymbol{\beta}) 
= \frac{1}{N} \sum_{i=1}^{N} 
  \left(- f_i(t) 
  +  \frac{\exp(\boldsymbol{X}_{i}^{\top}\boldsymbol{\beta}+E_i)}
  {1+\exp(\boldsymbol{X}_{i}^{\top}\boldsymbol{\beta}+E_i)}\right)\boldsymbol{X}_i.
\end{equation}
With this notation, we can rewrite the previous expression \eqref{preexp} as
\begin{equation}\label{start eq}
\begin{aligned}
N^{1/2}\left(\widehat{\boldsymbol{b}}_{(\mathcal{J})} - \boldsymbol{\beta}_{0,(\mathcal{J})}\right)
&= - N^{1/2}\widehat{\boldsymbol{\Theta}}_\mathcal{J}
   \left[\dot{R}_N(\boldsymbol{\beta}_0) - \widehat{\rho}_N(\boldsymbol{\beta}_0)\right] \\
&\quad - N^{1/2}\left(\widehat{\boldsymbol{\Theta}}_\mathcal{J} - \boldsymbol{\Theta}_\mathcal{J}\right)
   \widehat{\rho}_N(\boldsymbol{\beta}_0)
   - N^{1/2}\boldsymbol{\Theta}_\mathcal{J} \widehat{\rho}_N(\boldsymbol{\beta}_0)  \\
   & \quad - N^{1/2}\left(\widehat{\boldsymbol{\Theta}}\ddot{R}_N(\widehat{\boldsymbol{\beta}}) - \boldsymbol{I}_p\right)_\mathcal{J}
   \left(\widehat{\boldsymbol{\beta}} - \boldsymbol{\beta}_0\right) \\
&\quad - N^{1/2}\widehat{\boldsymbol{\Theta}}_\mathcal{J}
   \left[\ddot{R}_N(\tilde{\boldsymbol{\beta}}) - \ddot{R}_N(\widehat{\boldsymbol{\beta}})\right]
   \left(\widehat{\boldsymbol{\beta}} - \boldsymbol{\beta}_0\right).
\end{aligned}
\end{equation}

To establish the asymptotic normality of the de-sparsified estimator $\widehat{\boldsymbol{b}}_{(\mathcal{J})}$, we decompose the right-hand side of \eqref{start eq} into five separate components:
\begin{equation}
\begin{aligned}
tt_1 &:= N^{1/2} 
   \left|\widehat{\boldsymbol{\Theta}}_\mathcal{J}
   \left[\dot{R}_N(\boldsymbol{\beta}_0) - \widehat{\rho}_N(\boldsymbol{\beta}_0)\right]\right|_{\infty}, \\
tt_2 &:= N^{1/2}
   \left|\widehat{\boldsymbol{\Theta}}_\mathcal{J}\widehat{\rho}_N(\boldsymbol{\beta}_0)
   - \boldsymbol{\Theta}_\mathcal{J}\widehat{\rho}_N(\boldsymbol{\beta}_0)\right|_{\infty}, \\
tt_3 &:= N^{1/2}
   \boldsymbol{\Theta}_\mathcal{J} \widehat{\rho}_N(\boldsymbol{\beta}_0) \\
tt_4 &:= 
\left|N^{1/2}\left(\widehat{\boldsymbol{\Theta}}\ddot{R}_N(\widehat{\boldsymbol{\beta}}) - \boldsymbol{I}_p\right)_\mathcal{J}
   \left(\widehat{\boldsymbol{\beta}} - \boldsymbol{\beta}_0\right)\right|_{\infty}, \\
tt_5 &:= 
   \left|N^{1/2}\widehat{\boldsymbol{\Theta}}_\mathcal{J}
   \left[\ddot{R}_N(\tilde{\boldsymbol{\beta}}) - \ddot{R}_N(\widehat{\boldsymbol{\beta}})\right]
   \left(\widehat{\boldsymbol{\beta}} - \boldsymbol{\beta}_0\right)\right|_{\infty}.
\end{aligned}
\end{equation}
Theorem \ref{inf the} is a direct result of the following Lemmas \ref{inferencelemma}, \ref{ustat} which show
$$
tt_1 = o_{P}(1), tt_2 = o_{P}(1), tt_3 \xrightarrow{d}\mathcal{N}(\boldsymbol{0},\boldsymbol{V}_\mathcal{J}),  tt_4 = o_{P}(1), tt_5 = o_{P}(1)
,
$$
and Slutsky’s lemma.
\subsection{Auxiliary lemmas}
\begin{lemma} \label{inferencelemma}
Given a fixed group $\mathcal{J} \subseteq[p]$. Under the conditions of Theorem \ref{the} and Assumptions \ref{asapprerror}, \ref{aseign_2}, \ref{asrate}, and \ref{asinfeign}, the following results hold:
$$
tt_1 = o_{P}(1), tt_2 = o_{P}(1),  tt_4 = o_{P}(1), tt_5 = o_{P}(1).
$$
\end{lemma}
\begin{proof}
     We first show that $tt_1 = o_{P}(1)$. By substituting the definitions of 
$\dot{R}_N(\boldsymbol{\beta}_0)$ and $\widehat{\rho}_N(\boldsymbol{\beta}_0)$ from \eqref{defdotR} and \eqref{rho_hat}, we obtain
\[
\begin{aligned}
tt_1 
&= N^{1/2} \Bigg| \widehat{\boldsymbol{\Theta}}_\mathcal{J} \frac{1}{N} \sum_{i=1}^{N} 
   \left( 
       \frac{\exp(\boldsymbol{X}_{i}^{\top}\boldsymbol{\beta}_0)}{1+\exp(\boldsymbol{X}_{i}^{\top}\boldsymbol{\beta}_0)} 
       - \frac{\exp(\boldsymbol{X}_{i}^{\top}\boldsymbol{\beta}_0+E_i)}{1+\exp(\boldsymbol{X}_{i}^{\top}\boldsymbol{\beta}_0+E_i)}  
   \right) \boldsymbol{X}_i \Bigg|_{\infty} \\
&= N^{1/2} \Bigg| \frac{1}{N} \sum_{i=1}^{N} 
   \frac{\exp(\boldsymbol{X}_{i}^{\top}\tilde{\boldsymbol{\beta}}_i)}{\left(1+\exp(\boldsymbol{X}_{i}^{\top}\tilde{\boldsymbol{\beta}}_i)\right)^2} 
   E_i \, \widehat{\boldsymbol{\Theta}}_\mathcal{J} \boldsymbol{X}_i \Bigg|_{\infty} \\
&\leq N^{1/2} \sqrt{\frac{\sum_i \left(\frac{\exp(\boldsymbol{X}_{i}^{\top}\tilde{\boldsymbol{\beta}}_i)}{\left(1+\exp(\boldsymbol{X}_{i}^{\top}\tilde{\boldsymbol{\beta}}_i)\right)^2} E_i \right)^2}{N}} \max_{j \in \mathcal{J}} \sqrt{\frac{\sum_i \left(\widehat{\boldsymbol{\Theta}}_j^{\top}\boldsymbol{X}_i\right)^2}{N}} \\
&\leq N^{1/2} \left\| \boldsymbol{E} \right\|_N 
       \max_{j \in \mathcal{J}} \sqrt{\frac{\sum_i \left(\widehat{\boldsymbol{\Theta}}_j^{\top}\boldsymbol{X}_i\right)^2}{N}} \\
&= o_P(1).
\end{aligned}
\]
Here, the first equality follows from the mean-value theorem and $\boldsymbol{X}_{i}^{\top}\tilde{\boldsymbol{\beta}}_i$ is between  $\boldsymbol{X}_{i}^{\top}\boldsymbol{\beta}_0$ and $\boldsymbol{X}_{i}^{\top}\boldsymbol{\beta}_0 + E_i$, the first inequality is a consequence of the Cauchy--Schwarz inequality, and the second inequality uses the bound 
$\frac{\exp(\cdot)}{(1+\exp(\cdot))^2} \leq 1$. Finally, the conclusion 
$tt_1 = o_P(1)$ follows from Assumption \ref{aseign_2} (i), which implies 
$\|\boldsymbol{E}\|_N = o_P(1/N^{1/2+1.5/q})$, together with Lemma \ref{xtheta}, 
which establishes that $\max_{i \in [N]} \left| \boldsymbol{X}_i^{\top}\widehat{\boldsymbol{\Theta}}_{j}\right| = O_P(N^{1.5/q})$.

Second, we show that $tt_2 = o_{P}(1)$. To see this, we start by writing 
\begin{equation}
    \begin{aligned}
        tt_2 & =  N^{1 / 2}\left| \widehat{\boldsymbol{\Theta}}_\mathcal{J}\widehat{\rho}_N(\boldsymbol{\beta}_0)- \boldsymbol{\Theta}_\mathcal{J} \widehat{\rho}_N(\boldsymbol{\beta}_0)\right|_{\infty} \\
        &  \leq N^{1 / 2}\max_{j \in \mathcal{J}}\left| \widehat{\boldsymbol{\Theta}}_j - \boldsymbol{\Theta}_j \right|_1 \left| \widehat{\rho}_N(\boldsymbol{\beta}_0)\right|_{\infty}\\
        &  \leq N^{1 / 2}\max_{j \in \mathcal{J}}\left| \widehat{\boldsymbol{\Theta}}_j - \boldsymbol{\Theta}_j \right|_1 \left| \frac{1}{N} \sum_{i=1}^{N} \left(- \widehat{f}_i(t) +  \frac{\exp(\boldsymbol{X}_{i}^{\top}\boldsymbol{\beta}_0+E_i)}{1+\exp(\boldsymbol{X}_{i}^{\top}\boldsymbol{\beta}_0+E_i) }\right)\boldsymbol{X}_i \right|_{\infty}.
    \end{aligned}
\end{equation}
Next, we observe that
$$
\begin{aligned}
\left| \frac{1}{N} \sum_{i=1}^{N} \left(- \widehat{f}_i(t) +  \frac{\exp(\boldsymbol{X}_{i}^{\top}\boldsymbol{\beta}_0+E_i)}{1+\exp(\boldsymbol{X}_{i}^{\top}\boldsymbol{\beta}_0+E_i) }\right)\boldsymbol{X}_i \right|_{\infty} = \left| \frac{1}{N} \sum_{i=1}^{N} \left(- \widehat{f}_i(t) +  \mathbbm{E}(f_i(t)|\boldsymbol{X}_i) \right)\boldsymbol{X}_i \right|_{\infty}
\end{aligned}
$$
which corresponds to the term in \eqref{912242}. We then have 
\begin{equation}\label{a1 plus}
\begin{aligned}
\left| \frac{1}{N} \sum_{i=1}^{N} \left(- \widehat{f}_i(t) + \mathbbm{E}(f_i(t)|\boldsymbol{X}_i) \right)\boldsymbol{X}_i \right|_{\infty} = O_P\left( \frac{p^{\frac{1}{q}}(\log p)^{\frac{1}{2}}}{N^{\frac{1}{2} - \frac{1}{q}}} \right),
\end{aligned}
\end{equation}
where the rate comes from the analysis of \eqref{912242} in Theorem \ref{the3}.
Combining \eqref{a1 plus}, Lemma \ref{nodewise}, and Assumption \ref{asrate}, we conclude that
$$
tt_2 = o_P\left(\frac{1}{N^{1/8+1/(2q)}(\log p)^{1/q}}\right) = o_P(1).
$$

Third, we consider $tt_4$. By definition, we have
\begin{equation}\label{tt4}
\begin{aligned}
tt_4 
&= \left| N^{1/2} \left(\widehat{\boldsymbol{\Theta}} \ddot{R}_N(\widehat{\boldsymbol{\beta}}) - \boldsymbol{I}_p \right)_\mathcal{J} 
      \left( \widehat{\boldsymbol{\beta}} - \boldsymbol{\beta}_0 \right) \right|_{\infty} \\
&\leq \max_{j \in \mathcal{J}} \left| N^{1/2} \left( \ddot{R}_N(\widehat{\boldsymbol{\beta}})\widehat{\boldsymbol{\Theta}}_j - e_j \right) \right|_{\infty} 
      \left| \widehat{\boldsymbol{\beta}} - \boldsymbol{\beta}_0 \right|_1 \\
&\sim \max_{j \in \mathcal{J}} \left| N^{1/2} \left( \widehat{\boldsymbol{\Sigma}}_{\widehat{\boldsymbol{\beta}}}\widehat{\boldsymbol{\Theta}}_j - e_j \right) \right|_{\infty} 
      \, \Omega \left( \widehat{\boldsymbol{\beta}} - \boldsymbol{\beta}_0 \right) \\
&= o_P(1),
\end{aligned}
\end{equation}
where $e_j \in \mathbb{R}^{p}$ has $1$ in the $j$-th element and $0$ for others.
Here, the first inequality follows from the first statement of Lemma \ref{cor1}, by taking the norm as $|\cdot|_1$. The second inequality uses the equivalence 
\(\Omega(\widehat{\boldsymbol{\beta}} - \boldsymbol{\beta}_0) \sim |\widehat{\boldsymbol{\beta}} - \boldsymbol{\beta}_0|_1\), which holds since the size of the group $\mathcal{J}$ is fixed and \(\ddot{R}_N(\widehat{\boldsymbol{\beta}}) = \widehat{\boldsymbol{\Sigma}}_{\widehat{\boldsymbol{\beta}}}\) by definition. The last equality in \eqref{tt4} follows from Lemma \ref{thetasigma}, Theorem \ref{the},  Assumption \ref{asrate}, and $s^* \geq 1$, which provide the sufficient conditions to establish the stated order.

Lastly, we analyze $tt_5$. By definition, we have
\begin{equation}
\begin{aligned}
tt_5 
&= \left| N^{1/2} \widehat{\boldsymbol{\Theta}}_\mathcal{J} 
      \left[ \ddot{R}_N(\tilde{\boldsymbol{\beta}}) - \ddot{R}_N(\widehat{\boldsymbol{\beta}}) \right] 
      \left( \widehat{\boldsymbol{\beta}} - \boldsymbol{\beta}_0 \right) \right|_{\infty} \\
& \leq N^{1/2} \max_{j \in \mathcal{J}} \left|  \frac{1}{N} \sum_{i=1}^N \left(
          \frac{\exp(\theta_i)(1-\exp(\theta_i))}{(1+\exp(\theta_i))^3} 
          \left| \boldsymbol{X}_i^{\top} \widehat{\boldsymbol{\beta}} - \boldsymbol{X}_i^{\top} \boldsymbol{\beta}_0 \right| \widehat{\boldsymbol{\Theta}}_j^{\top}\boldsymbol{X}_i\boldsymbol{X}_i^{\top}\left( \widehat{\boldsymbol{\beta}} - \boldsymbol{\beta}_0 \right)\right)\right|\\
&  \leq N^{1/2} \max_{i \in [N], j \in \mathcal{J}} \left|\widehat{\boldsymbol{\Theta}}_j^{\top}\boldsymbol{X}_i\right| \left| \frac{1}{N} \sum_{i=1}^N 
          \left(\frac{\exp(\theta_i)(1-\exp(\theta_i))}{(1+\exp(\theta_i))^3} 
          \left| \boldsymbol{X}_i^{\top} \widehat{\boldsymbol{\beta}} - \boldsymbol{X}_i^{\top} \boldsymbol{\beta}_0 \right|\boldsymbol{X}_i^{\top}\left( \widehat{\boldsymbol{\beta}} - \boldsymbol{\beta}_0 \right)\right)\right|\\
& \leq N^{1/2} \max_{i \in [N], j \in \mathcal{J}} \left|\widehat{\boldsymbol{\Theta}}_j^{\top}\boldsymbol{X}_i\right| \left| \frac{1}{N} \sum_{i=1}^N 
          \left|\frac{\exp(\theta_i)(1-\exp(\theta_i))}{(1+\exp(\theta_i))^3} \right|
          \left| \boldsymbol{X}_i^{\top} \widehat{\boldsymbol{\beta}} - \boldsymbol{X}_i^{\top} \boldsymbol{\beta}_0 \right|^2\right|\\
& \leq N^{1/2}N^{1.5/q}\left\| \boldsymbol{X}_i^{\top} \widehat{\boldsymbol{\beta}} - \boldsymbol{X}_i^{\top} \boldsymbol{\beta}_0 \right\|_N^2\\
& \leq N^{1/2+1.5/q}\lambda^2s^2_{\boldsymbol{\beta}_0},
\end{aligned}
\end{equation}
where the first inequality is justified by the mean-value theorem and $\left| \boldsymbol{X}_i^{\top} \widehat{\boldsymbol{\beta}} - \boldsymbol{X}_i^{\top} \tilde{\boldsymbol{\beta}}_0 \right| \leq \left| \boldsymbol{X}_i^{\top} \widehat{\boldsymbol{\beta}} - \boldsymbol{X}_i^{\top} \boldsymbol{\beta}_0 \right|,$
the fourth inequality comes from Lemma \ref{xtheta} and the bound 
$$\left|\frac{\exp(\theta_i)(1-\exp(\theta_i))}{(1+\exp(\theta_i))^3}\right| \leq 1,$$
and the last inequality is justified by Lemma \ref{lneedextra}. Under the rate conditions in Assumption \ref{asrate}, this is sufficient to conclude that 
\[
tt_5 = o_P(1).
\]

\end{proof}

\begin{lemma}\label{ustat}
Given a fixed group $\mathcal{J} \subseteq [p]$. Under the conditions of Theorem \ref{the} and Assumptions \ref{asapprerror}, \ref{aseign_2}, \ref{asrate}, and \ref{asinfeign}, there exists a covariance matrix $\boldsymbol{V}_\mathcal{J}$, such that 
$$
\sqrt{N}\boldsymbol{\Theta}_\mathcal{J} \widehat{\rho}_N(\boldsymbol{\beta}_0) \xrightarrow{d}\mathcal{N}(\boldsymbol{0},\boldsymbol{V}_\mathcal{J}).
$$
Furthermore, let $\left(\widetilde{T}_{(2)}, \delta_{(2)},\boldsymbol{X}_{(2)}\right)$ denote an independent copy drawn from the same distribution as the sample $\left\{\left(\widetilde T_i, \delta_{i}, \boldsymbol{X}_i\right), i \in [N]\right\}$ and 
$$
\xi_{KM, i} = 
\mathbbm{E}\!\left(
\boldsymbol{X}_{(2)} \frac{\delta_{(2)}(t)\,\mathbbm{1}\{\widetilde{T}_{(2)} \leq t\}}{H^2(t \wedge \widetilde{T}_{(2)})}v\left(\widetilde{T}_{i}, \delta_{i}, t \wedge \widetilde{T}_{(2)}\right) \,\big|\, \boldsymbol{X}_{i}, \widetilde{T}_{i}, \delta_{i}
\right),
$$
where $v(\widetilde{T}_i, \delta_i, \cdot)$ is the influence function of the Kaplan-Meier estimator of $H$ for observation $i$ and for any positive value $z$
\begin{equation*}
v\left(\widetilde{T}_i, \delta_i, z\right)=H(z)\left[\frac{ \mathbbm{1}\left\{\widetilde{T}_i \leq z, \delta_i = 0\right\}}{\bar{P}\left(\widetilde{T}_i\right)}-\int_0^{z \wedge \widetilde{T}_i} \frac{P\left(u \leq \widetilde{T} \leq u+du, \delta = 0\right)}{\bar{P}^2(u)}\right],
\end{equation*}
where $\bar{P}(u)=P\left(\widetilde{T}>u\right)$.
We then have 
\tcr{
$$
\boldsymbol{V}_\mathcal{J} = \boldsymbol{\Theta}_\mathcal{J} Var\left(\left(- f_i(t) 
  +  \frac{\exp(\boldsymbol{X}_{i}^{\top}\boldsymbol{\beta}_0+E_i)}
{1+\exp(\boldsymbol{X}_{i}^{\top}\boldsymbol{\beta}_0+E_i)}\right)\boldsymbol{X}_i + \xi_{KM, i}\right)\boldsymbol{\Theta}_\mathcal{J}^{\top}.
$$
}
\end{lemma}
\begin{proof}
We decompose $\boldsymbol{\Theta}_\mathcal{J} \widehat{\rho}_N(\boldsymbol{\beta}_0)$ into two components:
\[
A_1 = \boldsymbol{\Theta}_\mathcal{J} \left(\widehat{\rho}_N(\boldsymbol{\beta}_0) - \rho_N(\boldsymbol{\beta}_0)\right),
 \quad
A_2 = \boldsymbol{\Theta}_\mathcal{J} \rho_N(\boldsymbol{\beta}_0).
\]
Our objective is to show that $\boldsymbol{\Theta}_\mathcal{J} \widehat{\rho}_N(\boldsymbol{\beta}_0)$ can be expressed as a U-statistic plus a negligible remainder term. 

To this end, we first consider $A_1$. 
\tcr{
    $$
    \begin{aligned}
        A_1 & = \frac{1}{N}\boldsymbol{\Theta}_\mathcal{J}\sum_{i=1}^{N}\left(f_i(t) - \widehat{f}_i(t)\right)\boldsymbol{X}_i\\
        & = \frac{1}{N}\boldsymbol{\Theta}_\mathcal{J}\sum_{i=1}^{N}\boldsymbol{X}_i\delta_i(t) \mathbbm{1}{\{\widetilde{T}_i \leq t\}}\left(\frac{1}{H\left(t \wedge \widetilde{T}_i\right)} - \frac{1}{\widehat{H}\left(t \wedge \widetilde{T}_i\right)}\right)\\
        & = \frac{1}{N}\boldsymbol{\Theta}_\mathcal{J}\sum_{i=1}^{N}\boldsymbol{X}_i\delta_i(t) \mathbbm{1}{\{\widetilde{T}_i \leq t\}}\\
        & \quad \times \left(\frac{\widehat{H}\left(t \wedge \widetilde{T}_i\right) - H\left(t \wedge \widetilde{T}_i\right)}{H\left(t \wedge \widetilde{T}_i\right)\widehat{H}\left(t \wedge \widetilde{T}_i\right)} - \frac{\widehat{H}\left(t \wedge \widetilde{T}_i\right) - H\left(t \wedge \widetilde{T}_i\right)}{H^2\left(t \wedge \widetilde{T}_i\right)} +\frac{\widehat{H}\left(t \wedge \widetilde{T}_i\right) - H\left(t \wedge \widetilde{T}_i\right)}{H^2\left(t \wedge \widetilde{T}_i\right)}\right)\\
        & = \frac{1}{N}\boldsymbol{\Theta}_\mathcal{J}\sum_{i=1}^{N}\boldsymbol{X}_i\delta_i(t) \mathbbm{1}{\{\widetilde{T}_i \leq t\}}\frac{\widehat{H}\left(t \wedge \widetilde{T}_i\right) - H\left(t \wedge \widetilde{T}_i\right)}{H\left(t \wedge \widetilde{T}_i\right)}\left(\frac{1}{\widehat{H}\left(t \wedge \widetilde{T}_i\right)} - \frac{1}{H\left(t \wedge \widetilde{T}_i\right)}\right)\\
        & \quad + \frac{1}{N}\boldsymbol{\Theta}_\mathcal{J}\sum_{i=1}^{N}\boldsymbol{X}_i\frac{\delta_i(t) \mathbbm{1}{\{\widetilde{T}_i \leq t\}}}{H^2\left(t \wedge \widetilde{T}_i\right)}\left(\widehat{H}\left(t \wedge \widetilde{T}_i\right) - H\left(t \wedge \widetilde{T}_i\right)\right)\\
        & = \frac{1}{N}\boldsymbol{\Theta}_\mathcal{J}\sum_{i=1}^{N}\boldsymbol{X}_i\delta_i(t) \mathbbm{1}{\{\widetilde{T}_i \leq t\}} O_P(N^{-1}) \\
        & \quad + \frac{1}{N}\boldsymbol{\Theta}_\mathcal{J}\sum_{i=1}^{N}\boldsymbol{X}_i\frac{\delta_i(t) \mathbbm{1}{\{\widetilde{T}_i \leq t\}}}{H^2\left(t \wedge \widetilde{T}_i\right)}\left(\widehat{H}\left(t \wedge \widetilde{T}_i\right) - H\left(t \wedge \widetilde{T}_i\right)\right)\\
        & = \frac{1}{N}\boldsymbol{\Theta}_\mathcal{J}\sum_{i=1}^{N}\boldsymbol{X}_i\frac{\delta_i(t) \mathbbm{1}{\{\widetilde{T}_i \leq t\}}}{H^2\left(t \wedge \widetilde{T}_i\right)}\left(\widehat{H}\left(t \wedge \widetilde{T}_i\right) - H\left(t \wedge \widetilde{T}_i\right)\right) +  O_P(N^{-1}),
    \end{aligned}
    $$
    }
    where the second-to-last equality follows from \eqref{kmrate} and \eqref{kmrate2}, and the last equality relies on the facts that $\delta_i(t)\mathbbm{1}\{\widetilde{T}_i \leq t\}$ is bounded, $1/H^2(t \wedge \widetilde{T}_i)$ is bounded by Assumption \ref{as2}, and that, for any $j \in \mathcal{J}$, we have 
\[
\frac{1}{N}\sum_{i=1}^{N}\boldsymbol{\Theta}_j^{\top}\boldsymbol{X}_i = O_P(1),
\]
which follows from Assumption~\ref{aseign_2} (iv) implying $\left(\mathbbm{E}\left(\left|\boldsymbol{\Theta}_{j}^{\top}\boldsymbol{X}_i\right|\right)\right)^{q} \leq \mathbbm{E}\left(\left|\boldsymbol{\Theta}_{j}^{\top}\boldsymbol{X}_i\right|^q\right) = O\left(1\right)$ and Markov’s inequality. Combining these results, it follows that for any $j \in \mathcal{J}$,
\begin{equation}\label{mar1}
\frac{1}{N}\boldsymbol{\Theta}_j^{\top}\sum_{i=1}^{N}\boldsymbol{X}_i
\frac{\delta_i(t)\mathbbm{1}\{\widetilde{T}_i \leq t\}}{H^2(t \wedge \widetilde{T}_i)} 
= O_P(1).
\end{equation}
Then, by Theorem~1 of \citet{appendix_lo1986product}, we obtain
\begin{equation}\label{mar2}
\widehat{H}\left(t \wedge \widetilde{T}_i\right) - H\left(t \wedge \widetilde{T}_i\right)
= \frac{1}{N}\sum_{k=1}^{N}v\left(\widetilde{T}_k, \delta_k, t \wedge \widetilde{T}_i\right) + o_P\left(N^{-1/2}\right),
\end{equation}
where for any positive value $z$
\begin{equation}\label{influence_function}
\begin{aligned}
v\left(\widetilde{T}_k, \delta_k, z\right)=H(z)\left[\frac{ \mathbbm{1}\left\{\widetilde{T}_k \leq z, \delta_k = 0\right\}}{\bar{P}\left(\widetilde{T}_k\right)}-\int_0^{z \wedge \widetilde{T}_k} \frac{P_1(du)}{\bar{P}^2(u)}\right].
\end{aligned}
\end{equation}
Here, $\bar{P}(u)=P\left(\widetilde{T}>u\right)$  and $P_1(du)=P\left(u \leq \widetilde{T} \leq u+ du, \delta = 0\right)$. 
Let 
\tcr{
\[
\rho_{i,k} = 
\boldsymbol{X}_i \frac{\delta_i(t)\mathbbm{1}\{\widetilde{T}_i \leq t\}}
{H^2\left(t \wedge \widetilde{T}_i\right)}
v\left(\widetilde{T}_k, \delta_k, t \wedge \widetilde{T}_i\right),
\]
}
so that we can write
\[
A_1 = 
\frac{1}{N^2}\boldsymbol{\Theta}_\mathcal{J}
\sum_{i=1}^{N}\sum_{k=1}^{N}\rho_{i,k}
+ o_P\left(N^{-1/2}\right) + O_P(N^{-1}),
\]
where we use \eqref{mar1} and \eqref{mar2}.

Next, consider the term $A_2$. We define
\[
\begin{aligned}
A_2 
&= \frac{1}{N}\boldsymbol{\Theta}_\mathcal{J} \sum_{i=1}^{N}
\left(
- f_i(t) 
+ \frac{\exp(\boldsymbol{X}_{i}^{\top}\boldsymbol{\beta}_0 + E_i)}
{1 + \exp(\boldsymbol{X}_{i}^{\top}\boldsymbol{\beta}_0 + E_i)}
\right)\boldsymbol{X}_i \\
&= \frac{1}{N}\boldsymbol{\Theta}_\mathcal{J}\sum_{i=1}^{N}\phi_i,
\end{aligned}
\]
where 
\[
\phi_i = 
\left(
- f_i(t) 
+ \frac{\exp(\boldsymbol{X}_{i}^{\top}\boldsymbol{\beta}_0 + E_i)}
{1 + \exp(\boldsymbol{X}_{i}^{\top}\boldsymbol{\beta}_0 + E_i)}
\right)\boldsymbol{X}_i.
\]
Hence, combining $A_1$ and $A_2$, we have
\[
\boldsymbol{\Theta}_\mathcal{J} \widehat{\rho}_N(\boldsymbol{\beta}_0) = A_1 + A_2 =
\frac{1}{N^2}\boldsymbol{\Theta}_\mathcal{J}
\sum_{i=1}^{N}\sum_{k=1}^{N}\left(\rho_{i,k} + \phi_i\right)
+ o_P\left(N^{-1/2}\right) + O_P(N^{-1}).
\]

We now show that $A_1 + A_2$ can be rewritten as a symmetric sum plus a negligible term. To proceed, we define the following quantities:
  $$
  \tilde{\xi}_{i,k} = \frac{1}{2}\left(\rho_{i,k} + \phi_i + \rho_{k,i} + \phi_k \right).
  $$
Note that
  $$
  \begin{aligned}
  & A_1 +A_2 \\
  & = \frac{1}{N^2}\boldsymbol{\Theta}_\mathcal{J}\sum_{i=1}^{N}\sum_{k=1}^{N}\tilde{\xi}_{i,k} + o_P(N^{-1/2})+  O_P(N^{-1})\\
  &  = \frac{2}{N^2}\boldsymbol{\Theta}_\mathcal{J}\sum_{1\leq i<k \leq N}\tilde{\xi}_{i,k} + \frac{1}{N^2}\boldsymbol{\Theta}_\mathcal{J}\sum_{i=1}^{N}\tilde{\xi}_{i,i} + o_P(N^{-1/2}) + O_P(N^{-1})\\
  & = \frac{2}{N^2-N}\left(1-\frac{1}{N}\right)\boldsymbol{\Theta}_\mathcal{J}\sum_{1\leq i<k \leq N}\tilde{\xi}_{i,k}  + \frac{1}{N^2}\boldsymbol{\Theta}_\mathcal{J}\sum_{i=1}^{N}\tilde{\xi}_{i,i} + o_P(N^{-1/2}) + O_P(N^{-1})\\
  & = \frac{2}{N^2-N}\boldsymbol{\Theta}_\mathcal{J}\sum_{1\leq i<k \leq N}\tilde{\xi}_{i,k} -\frac{2}{N(N^2-N)}\boldsymbol{\Theta}_\mathcal{J}\sum_{1\leq i<k \leq N}\tilde{\xi}_{i,k}\\
  & \quad + \frac{1}{N^2}\boldsymbol{\Theta}_\mathcal{J}\sum_{i=1}^{N}\tilde{\xi}_{i,i} + o_P(N^{-1/2}) +O_P(N^{-1}).
  \end{aligned}
  $$
  We first bound the diagonal term
\[
\frac{1}{N^2}\boldsymbol{\Theta}_\mathcal{J}\sum_{i=1}^{N}\tilde{\xi}_{i,i} 
= \frac{1}{N^2}\sum_{i=1}^{N}\boldsymbol{\Theta}_\mathcal{J}\left(\rho_{i,i} + \phi_i\right).
\]
For each $j \in G$, $\mathbbm{E}\left(\left|\boldsymbol{\Theta}_j^{\top}\phi_i\right|\right)$ is bounded since we have 
$$
\mathbbm{E}\left(\left|\boldsymbol{\Theta}_j^{\top}\phi_i\right|\right) \lesssim \left(\mathbbm{E}\left(\left|\boldsymbol{\Theta}_{j}^{\top}\boldsymbol{X}_i\right|^q\right)\right)^{\frac{1}{q}} = O\left(1\right)
$$
by Jensen’s inequality, Assumptions~\ref{as2} and~\ref{aseign_2}~(iv). Obviously, $Var\left(\boldsymbol{\Theta}_j^{\top}\phi_i\right)$ is also bounded.
Then Chebyshev's inequality implies that 
\begin{equation}\label{thetaphi}
\boldsymbol{\Theta}_j^{\top}\frac{1}{N}\sum_{i=1}^{N}\phi_i = O_P(1).
\end{equation}
Recalling the definition \eqref{influence_function}, it is immediate that $v_i\left(t \wedge \widetilde{T}_i\right)$ is uniformly bounded, since $H\left(t \wedge \widetilde{T}_i\right) \leq 1$ by definitions and $\inf_{s\leq u \leq t}P(\widetilde{T}\geq u \mid \widetilde{T} \geq s) \geq C_r$ by Assumption \ref{as2}.
It follows that 
\tcr{
\begin{equation}\label{thetarhoi}
\begin{aligned}
\frac{1}{N^2}\boldsymbol{\Theta}_j^{\top}\sum_{i=1}^{N}\rho_{i,i} 
&= \frac{1}{N^2}\sum_{i=1}^{N}  
\boldsymbol{\Theta}_j^{\top}\boldsymbol{X}_i
\frac{\delta_i(t)\mathbbm{1}\{\widetilde{T}_i \leq t\}}
{H^2\left(t \wedge \widetilde{T}_i\right)}v\left(\widetilde{T}_i, \delta_i, t \wedge \widetilde{T}_i\right) \\
&= O_P\left(N^{-1}\right),
\end{aligned}
\end{equation}
}
where we use \eqref{mar1}.
Hence, combining \eqref{thetaphi} and \eqref{thetarhoi}, we have
\begin{equation}\label{diag}
\frac{1}{N^2}\boldsymbol{\Theta}_\mathcal{J}\sum_{i=1}^{N}\tilde{\xi}_{i,i} = O_P\left(N^{-1}\right).
\end{equation}
Next, we bound the off-diagonal term:
\[
\begin{aligned}
\frac{2}{N(N^2 - N)}\boldsymbol{\Theta}_\mathcal{J}\sum_{1 \leq i < k \leq N}\tilde{\xi}_{i,k} 
&= \frac{1}{N(N^2 - N)}\boldsymbol{\Theta}_\mathcal{J}\sum_{i}\sum_{k}\tilde{\xi}_{i,k} 
- \frac{1}{N(N^2 - N)}\boldsymbol{\Theta}_\mathcal{J}\sum_{i}\tilde{\xi}_{i,i} \\
&= \frac{1}{N(N^2 - N)}\boldsymbol{\Theta}_\mathcal{J}\sum_{i}\sum_{k}\left(\rho_{i,k} + \phi_i\right) 
- O_P\left(N^{-2}\right).
\end{aligned}
\]
Recalling~\eqref{kmrate} and \eqref{mar2}, we have $\frac{1}{N}\sum_{k=1}^{N}v\left(\widetilde{T}_k, \delta_k, t \wedge \widetilde{T}_i\right) = O_P(N^{-1/2})$. Therefore,
\tcr{
\begin{equation}\label{thetarhoik}
\begin{aligned}
\frac{1}{N^2}\boldsymbol{\Theta}_j^{\top}\sum_{i=1}^{N}\sum_{k=1}^{N}\rho_{i,k} 
&= \frac{1}{N}\sum_{i=1}^{N}  
\boldsymbol{\Theta}_j^{\top}\boldsymbol{X}_i
\frac{\delta_i(t)\mathbbm{1}\{\widetilde{T}_i \leq t\}}
{H^2\left(t \wedge \widetilde{T}_i\right)} 
\times \frac{1}{N}\sum_{k=1}^{N}v\left(\widetilde{T}_k, \delta_k, t \wedge \widetilde{T}_i\right) \\
&= O_P\left(N^{-1/2}\right)
\times \frac{1}{N}\sum_{i=1}^{N}  
\boldsymbol{\Theta}_j^{\top}\boldsymbol{X}_i
\frac{\delta_i(t)\mathbbm{1}\{\widetilde{T}_i \leq t\}}
{H^2\left(t \wedge \widetilde{T}_i\right)} \\
&= O_P\left(N^{-1/2}\right),
\end{aligned}
\end{equation}
}
where the final equality follows from~\eqref{mar1}.

Combining \eqref{thetaphi} and \eqref{thetarhoik}, we have
\begin{equation}\label{offdiag}
\begin{aligned}
\frac{2}{N(N^2 - N)}\boldsymbol{\Theta}_\mathcal{J}\sum_{1 \leq i < k \leq N}\tilde{\xi}_{i,k} & = \frac{1}{N(N^2 - N)}\boldsymbol{\Theta}_j^{\top}
\sum_{i}\sum_{k}\left(\rho_{i,k} + \phi_i\right) - O_P(N^{-2})\\
&= O_P\left(N^{-3/2}\right).
\end{aligned}
\end{equation}
We can now conclude that by \eqref{diag} and \eqref{offdiag}:
\begin{equation}\label{udecom}
\boldsymbol{\Theta}_\mathcal{J} \widehat{\rho}_N(\boldsymbol{\beta}_0)
= A_1 + A_2 = \frac{2}{N^2 - N}\boldsymbol{\Theta}_\mathcal{J}
\sum_{1 \leq i < k \leq N}\tilde{\xi}_{i,k} + o_P\left(N^{-1/2}\right),
\end{equation}
which can be expressed as a U-statistic plus a term of order $o_P(N^{-1/2})$.  

Applying the classical Central Limit Theorem for U-statistics 
(e.g., Theorem~1.1 of \citet{bose2018u}), we obtain
\begin{equation}\label{udis}
\sqrt{N}\left(
\frac{2}{N^2 - N}\boldsymbol{\Theta}_\mathcal{J}
\sum_{1 \leq i < k \leq N}\tilde{\xi}_{i,k}
- \boldsymbol{\theta}
\right)
\xrightarrow{d} \mathcal{N}\left(\boldsymbol{0}, \boldsymbol{V}_\mathcal{J}\right),
\end{equation}
where $\boldsymbol{\theta} = \boldsymbol{\Theta}_\mathcal{J}\mathbbm{E}\left(\tilde{\xi}_{i,k}\right)$.

We next examine the variance term $\boldsymbol{V}_\mathcal{J}$ appearing above. To define the Hoeffding projection of the U-statistic kernel, we follow the
standard convention in the U-statistics literature.  
Let $\left\{\boldsymbol{X}_{(1)},\widetilde T_{(1)},\delta_{(1)}, E_{(1)}\right\}$ and 
$\left\{\boldsymbol{X}_{(2)},\widetilde T_{(2)},\delta_{(2)}, E_{(2)}\right\}$
denote two independent copies drawn from the same distribution as the sample $\left\{\boldsymbol{X}_i,\widetilde T_i,\delta_i,v(\cdot), E_i, i \in [N]\right\}$, and let $\tilde{w}_{(1)} = \frac{\delta_{(1)}(t) \mathbbm{1}{\{\widetilde{T}_{(1)} \leq t\}}}{H^2\left(t \wedge \widetilde{T}_{(1)}\right)}$ and $\tilde{w}_{(2)} = \frac{\delta_{(2)}(t) \mathbbm{1}{\{\widetilde{T}_{(2)} \leq t\}}}{H^2\left(t \wedge \widetilde{T}_{(2)}\right)}$. With the notation, the first-order projection of the kernel
$\tilde{\xi}_{i,k}$ is given by
\tcr{
$$
\begin{aligned}
\mathbbm{E}\left( \tilde{\xi}_{(1),(2)} \mid \boldsymbol{X}_{(1)}, \widetilde{T}_{(1)}, \delta_{(1)} \right) & = \frac{1}{2}\boldsymbol{X}_{(1)}\tilde{w}_{(1)}\mathbbm{E}\left( v\left(\widetilde{T}_{(2)}, \delta_{(2)}, t \wedge \widetilde{T}_{(1)}\right) \bigg| \boldsymbol{X}_{(1)}, \widetilde{T}_{(1)}, \delta_{(1)} \right) \\
& \quad + \frac{1}{2}\left(- f_{(1)}(t) 
  +  \frac{\exp(\boldsymbol{X}_{(1)}^{\top}\boldsymbol{\beta}_0+E_{(1)})}
{1+\exp(\boldsymbol{X}_{(1)}^{\top}\boldsymbol{\beta}_0+E_{(1)})}\right)\boldsymbol{X}_{(1)}\\
& \quad + \frac{1}{2}\mathbbm{E}\left(\boldsymbol{X}_{(2)}\tilde{w}_{(2)}v\left(\widetilde{T}_{(1)}, \delta_{(1)}, t \wedge \widetilde{T}_{(2)}\right) \bigg| \boldsymbol{X}_{(1)}, \widetilde{T}_{(1)}, \delta_{(1)}\right) \\
& \quad + \frac{1}{2}\mathbbm{E}\left(\left(- f_{(2)}(t) 
  +  \frac{\exp(\boldsymbol{X}_{(2)}^{\top}\boldsymbol{\beta}_0+E_{(2)})}
{1+\exp(\boldsymbol{X}_{(2)}^{\top}\boldsymbol{\beta}_0+E_{(2)})}\right)\boldsymbol{X}_{(2)}\right)\\
& = \frac{1}{2}\left(- f_{(1)}(t) 
  +  \frac{\exp(\boldsymbol{X}_{(1)}^{\top}\boldsymbol{\beta}_0+E_{(1)})}
{1+\exp(\boldsymbol{X}_{(1)}^{\top}\boldsymbol{\beta}_0+E_{(1)})}\right)\boldsymbol{X}_{(1)} \\
& \quad + \frac{1}{2}\mathbbm{E}\left(\boldsymbol{X}_{(2)}\tilde{w}_{(2)}v\left(\widetilde{T}_{(1)}, \delta_{(1)}, t \wedge \widetilde{T}_{(2)}\right)\bigg| \boldsymbol{X}_{(1)}, \widetilde{T}_{(1)}, \delta_{(1)}\right),
\end{aligned}
$$
}
where we use that 
\begin{equation}\label{expiffunction}
\mathbbm{E}\!\left( v\left(\widetilde{T}_{(2)}, \delta_{(2)}, t \wedge \widetilde{T}_{(1)}\right) \,\big|\, \boldsymbol{X}_{(1)}, \widetilde{T}_{(1)}, \delta_{(1)} \right) = 0
\end{equation}
which comes from the discussion under Theorem 2 in \citet{appendix_lo1986product} or the asymptotic property of the Kaplan–Meier estimator,
and 
\[
\mathbbm{E}\!\left[\!\left(- f_{(2)}(t) 
  +  \frac{\exp(\boldsymbol{X}_{(2)}^{\top}\boldsymbol{\beta}_0+E_{(2)})}
{1+\exp(\boldsymbol{X}_{(2)}^{\top}\boldsymbol{\beta}_0+E_{(2)})}\right)\boldsymbol{X}_{(2)}\!\right] 
  = \boldsymbol{0},
\]
as established in~\eqref{expf}. 
Furthermore, we have
\tcr{
\[
\begin{aligned}
\mathbbm{E}\left( \tilde{\xi}_{(1),(2)}\right) 
&= \mathbbm{E}\!\left[\frac{1}{2}\left(- f_{(1)}(t) 
  +  \frac{\exp(\boldsymbol{X}_{(1)}^{\top}\boldsymbol{\beta}_0+E_{(1)})}
  {1+\exp(\boldsymbol{X}_{(1)}^{\top}\boldsymbol{\beta}_0+E_{(1)})}\right)\boldsymbol{X}_{(1)}\right] \\
&\quad + \frac{1}{2}\mathbbm{E}\!\left[\mathbbm{E}\!\left(\boldsymbol{X}_{(2)}\tilde{w}_{(2)} v\left(\widetilde{T}_{(1)}, \delta_{(1)}, t \wedge \widetilde{T}_{(2)}\right)\,\big|\, \boldsymbol{X}_{(1)}, \widetilde{T}_{(1)}, \delta_{(1)}\right)\right] \\
&= \boldsymbol{0}
\end{aligned}
\]
}
since
\[
\begin{aligned}
& \mathbbm{E}\!\left[\mathbbm{E}\!\left(\boldsymbol{X}_{(2)}\tilde{w}_{(2)} v\left(\widetilde{T}_{(1)}, \delta_{(1)}, t \wedge \widetilde{T}_{(2)}\right)\,\big|\, \boldsymbol{X}_{(1)}, \widetilde{T}_{(1)}, \delta_{(1)}\right)\right]\\
&= \mathbbm{E}\!\left(\boldsymbol{X}_{(2)}\tilde{w}_{(2)} v\left(\widetilde{T}_{(1)}, \delta_{(1)}, t \wedge \widetilde{T}_{(2)}\right)\right) \\
&= \mathbbm{E}\!\left[\mathbbm{E}\!\left(\boldsymbol{X}_{(2)}\tilde{w}_{(2)} v\left(\widetilde{T}_{(1)}, \delta_{(1)}, t \wedge \widetilde{T}_{(2)}\right)\,\big|\, \boldsymbol{X}_{(2)}, \widetilde{T}_{(2)}, \delta_{(2)}\right)\right] \\
&= \mathbbm{E}\!\left(\boldsymbol{X}_{(2)}\tilde{w}_{(2)}\,\mathbbm{E}\!\left[v\left(\widetilde{T}_{(1)}, \delta_{(1)}, t \wedge \widetilde{T}_{(2)}\right)\,\big|\, \boldsymbol{X}_{(2)}, \widetilde{T}_{(2)}, \delta_{(2)}\right]\right) \\
&= \boldsymbol{0},
\end{aligned}
\]
where the last inequality comes from \eqref{expiffunction} by flipping the indices.
It follows that 
\begin{equation}\label{expu}
\boldsymbol{\theta} = \boldsymbol{0}.
\end{equation}
Hence, the centered version of the first projection of the U-statistic is given by
\tcr{
\[
\begin{aligned}
& \mathbbm{E}\left( \tilde{\xi}_{(1),(2)} \mid \boldsymbol{X}_{(1)}, \widetilde{T}_{(1)}, \delta_{(1)} \right) - \mathbbm{E}\left( \tilde{\xi}_{(1),(2)}\right) \\
& = \frac{1}{2}\left(- f_{(1)}(t) 
  +  \frac{\exp(\boldsymbol{X}_{(1)}^{\top}\boldsymbol{\beta}_0+E_{(1)})}
{1+\exp(\boldsymbol{X}_{(1)}^{\top}\boldsymbol{\beta}_0+E_{(1)})}\right)\boldsymbol{X}_{(1)} + \frac{1}{2}\mathbbm{E}\left(\boldsymbol{X}_{(2)}\tilde{w}_{(2)}v\left(\widetilde{T}_{(1)}, \delta_{(1)}, t \wedge \widetilde{T}_{(2)}\right)\bigg| \boldsymbol{X}_{(1)}, \widetilde{T}_{(1)}, \delta_{(1)}\right).
\end{aligned}
\]
}
Then we apply Theorem~1.1 of \citet{bose2018u}, the variance of the U-statistic is given by
\tcr{
\[
\begin{aligned}
\boldsymbol{V}_\mathcal{J}
&= 4 \, \mathrm{Var} \Bigg(
\frac{1}{2}\boldsymbol{\Theta}_\mathcal{J} 
\left(- f_{(1)}(t) 
  +  \frac{\exp(\boldsymbol{X}_{(1)}^{\top}\boldsymbol{\beta}_0 + E_{(1)})}
         {1 + \exp(\boldsymbol{X}_{(1)}^{\top}\boldsymbol{\beta}_0 + E_{(1)})}\right)\boldsymbol{X}_{(1)} \\
&\qquad + \frac{1}{2} \boldsymbol{\Theta}_\mathcal{J} \, 
\mathbbm{E}\!\left(\boldsymbol{X}_{(2)} \tilde{w}_{(2)} v\left(\widetilde{T}_{(1)}, \delta_{(1)}, t \wedge \widetilde{T}_{(2)}\right) \,\big|\, \boldsymbol{X}_{(1)}, \widetilde{T}_{(1)}, \delta_{(1)}\right)
\Bigg) \\
&= \boldsymbol{\Theta}_\mathcal{J} \, \mathrm{Var} \left(
\left(- f_{i}(t) 
  +  \frac{\exp(\boldsymbol{X}_{i}^{\top}\boldsymbol{\beta}_0 + E_{i})}
         {1 + \exp(\boldsymbol{X}_{i}^{\top}\boldsymbol{\beta}_0 + E_{i})}\right)\boldsymbol{X}_{i} + \xi_{KM, i} 
\right) \boldsymbol{\Theta}_\mathcal{J}^{\top},
\end{aligned}
\]
}
where the last equality comes from the fact that $\left\{\boldsymbol{X}_{(1)},\widetilde T_{(1)},\delta_{(1)}, E_{(1)}\right\}$ has the same distribution as $\left\{\boldsymbol{X}_{i},\widetilde T_{i},\delta_{i}, E_{i}, i \in [N]\right\}$ and
\[
\xi_{KM, i} = 
\mathbbm{E}\!\left(
\boldsymbol{X}_{(2)} \frac{\delta_{(2)}(t)\,\mathbbm{1}\{\widetilde{T}_{(2)} \leq t\}}{H^2(t \wedge \widetilde{T}_{(2)})}v\left(\widetilde{T}_{i}, \delta_{i}, t \wedge \widetilde{T}_{(2)}\right) \,\big|\, \boldsymbol{X}_{i}, \widetilde{T}_{i}, \delta_{i}
\right).
\]
Overall, by Slutsky’s lemma, \eqref{udecom}, \eqref{udis}, and \eqref{expu}, we conclude that
\[
\sqrt{N}\, \boldsymbol{\Theta}_\mathcal{J} \widehat{\rho}_N(\boldsymbol{\beta}_0) 
\xrightarrow{d} \mathcal{N}(\boldsymbol{0}, \boldsymbol{V}_\mathcal{J}).
\]

Furthermore, we give a plug-in variance estimator here
\[
\widehat{\boldsymbol{V}}_\mathcal{J} = \widehat{\boldsymbol{\Theta}}_\mathcal{J}\left(\frac{1}{N}\sum_{i=1}^{N}\widehat{\boldsymbol{\sigma}}_i\widehat{\boldsymbol{\sigma}}_i^{\top}\right)\widehat{\boldsymbol{\Theta}}_\mathcal{J}^{\top},
\]
where
\tcr{
\[
\widehat{\boldsymbol{\sigma}}_i = \left(-\frac{\delta_i(t) \mathbbm{1}\{\widetilde{T}_i \leq t\}}{\widehat{H}(t \wedge \widetilde{T}_i)}+\frac{\exp(\boldsymbol{X}_{i}^{\top}\widehat{\boldsymbol{\beta}})}{1+\exp(\boldsymbol{X}_{i}^{\top}\widehat{\boldsymbol{\beta}}) }\right)\boldsymbol{X}_{i} + \frac{1}{N}\sum_{k=1}^{N}\boldsymbol{X}_{k}\frac{\delta_k(t) \mathbbm{1}{\{\widetilde{T}_k \leq t\}}}{\widehat{H}^2\left(t \wedge \widetilde{T}_k\right)}\hat{v}\left(\widetilde{T}_i, \delta_i,t \wedge \widetilde{T}_k\right),
\]
}
and $\hat{v}(\widetilde{T}_i, \delta_i,\cdot)$ is the estimator of $v(\widetilde{T}_i, \delta_i,\cdot)$, which can be obtained by plugging in the data. Specifically, we have for any positive value $z$
\begin{equation}\label{hat_v_k_discrete_final}
\widehat v\left(\widetilde{T}_i, \delta_i, z\right)
=\widehat{H}(z)\left[\frac{ \mathbbm{1}\left\{\widetilde{T}_i \leq z, \delta_i = 0\right\}}{\bar{P}_N\left(\widetilde{T}_i\right)}- \frac{1}{N}\sum_{j=1}^N\frac{\mathbbm{1}\{\widetilde{T}_j \leq z \wedge \widetilde{T}_i, \delta_j = 0\}}{\bar{P}_N^2(\widetilde{T}_j)}\right],
\end{equation}
where $\bar{P}_N(u) = \frac{1}{N}\sum_{l=1}^N\mathbbm{1}\{\widetilde{T}_l > u\}$.
\end{proof}

\subsection{Technical lemmas}
\begin{lemma}\label{xtheta}
   Let $\mathcal{J} \subseteq [p]$ be a fixed group. Under the conditions of Theorem \ref{the} and Assumptions \ref{asapprerror}, \ref{aseign_2}, and \ref{asrate}, for any $j \in \mathcal{J}$, we have $$
\left|\boldsymbol{X}\widehat{\boldsymbol{\Theta}}_{j}\right|_{\infty} = O_P\left(N^{1.5/q}\right).$$
\end{lemma}
\begin{proof}
We start by writing
 $$
 \begin{aligned}
 \left| \boldsymbol{X}\boldsymbol{\Theta}_{ j}\right|_{\infty} &= \max_{i \in [N]}\left| \boldsymbol{X}_i^{\top}\boldsymbol{\Theta}_{ j}\right| \\
 &  = O_P\left(N^{1.5/q}\right),
 \end{aligned}
$$  
where we use the Markov inequality to bound $\max_{i \in [N]}\left| \boldsymbol{X}_i^{\top}\boldsymbol{\Theta}_{ j}\right|$
$$
\begin{aligned}
P\left(\max_{i \in [N]}\left| \boldsymbol{X}_i^{\top}\boldsymbol{\Theta}_{j}\right| \geq t\right) & \leq \frac{\mathbbm{E}\left(\max_{i \in [N]}\left|  \boldsymbol{X}_i^{\top}\boldsymbol{\Theta}_{ j}\right|^q\right)}{t^q}\\
& \leq \frac{NK_{\Theta}}{t^q},
\end{aligned}
$$
where $K_{\Theta}$ is constant from Assumption \ref{aseign_2} (iv).
Let $t = N^{1.5/q}$, therefore, we have $\max_{i \in [N]}\left| \boldsymbol{X}_i^{\top}\boldsymbol{\Theta}_{ j}\right| = O_P(N^{1.5/q})$. Hence,
 $$
 \begin{aligned}
\left|\boldsymbol{X}\widehat{\boldsymbol{\Theta}}_{j}\right|_{\infty} & \leq \left|\boldsymbol{X}\widehat{\boldsymbol{\Theta}}_{j} - \boldsymbol{X}\boldsymbol{\Theta}_{ j}\right|_{\infty} + \left| \boldsymbol{X}\boldsymbol{\Theta}_{ j}\right|_{\infty}\\
&\leq \left|\boldsymbol{X}\right|_{\infty}\left| \widehat{\boldsymbol{\Theta}}_{j} - \boldsymbol{\Theta}_{ j}\right|_1 + \left| \boldsymbol{X}\boldsymbol{\Theta}_{ j}\right|_{\infty} \\
& = O_P\left(N^{1.5/q}\right),
 \end{aligned}
 $$
where by Lemma \ref{nodewise} and \eqref{pa3}, one can see the order of $ \left|\boldsymbol{X}\right|_{\infty}\left| \widehat{\boldsymbol{\Theta}}_{j} - \boldsymbol{\Theta}_{ j}\right|_1 $ is $o_P\left(1\right)$.
\end{proof}

\begin{lemma}\label{thetasigma}
 Let $\mathcal{J} \subseteq [p]$ be a fixed group. Suppose Assumptions \ref{asapprerror}, \ref{aseign_2}, \ref{asrate}, and all conditions of Theorem \ref{the} are satisfied. If
$p^{\frac{1}{q}}\sqrt{\log p}/N^{\frac{1}{2} - \frac{1}{q}}=O( \lambda )$ and $(p-1)^{\frac{2}{q}}\sqrt{\log p}/N^{\frac{1}{2} - \frac{2}{q}}=O( \lambda_{j} ), j \in \mathcal{J}$, we have for every $j \in \mathcal{J}$
\[
\left| \widehat{\boldsymbol{\Sigma}}_{\widehat{\boldsymbol{\beta}}}\widehat{\boldsymbol{\Theta}}_j - e_j \right|_{\infty}
\leq  \frac{\lambda_{j}}{\hat{\tau}_{j}^2} 
= O_P\left(\frac{p^{2/q} (\log p)^{1/2}}{N^{1/2 - 2/q}}\right),
\]
where $e_j \in \mathbb{R}^{p}$ has $1$ in the $j$-th element and $0$ for others.
\end{lemma}
\begin{proof}
By \eqref{KKTused}, we have 
$$  \left|\widehat{\boldsymbol{\Sigma}}_{\widehat{\boldsymbol{\beta}}}\widehat{\boldsymbol{\Theta}}_j - e_j\right|_{\infty} \leq \lambda_{j}/\hat{\tau}_j^2. 
$$
Since $1/\hat{\tau}_j^2 = O_P(1)$ from \eqref{hattaubound} and $\lambda_j = O\left(\frac{p^{2/q} (\log p)^{1/2}}{N^{1/2 - 2/q}}\right)$, the proof is finished.
\end{proof}
 
\section{Nodewise regression}\label{appnodewise}

\paragraph{Notation.}
For sets $ \mathcal{S}_1$ and $\mathcal{S}_2$, define $\mathcal{S}_1 - \mathcal{S}_2 = \{s: s \in \mathcal{S}_1, s \not\in \mathcal{S}_2\}$.

We start by introducing definitions and some inequalities that are useful for the proof.
Define the empirical version of $\boldsymbol{\eta}_{\boldsymbol{\beta}_0,j}$ from \eqref{populationnodewise} as 
\begin{equation}\label{empnodewise}
\boldsymbol{\eta}_{\widehat{\boldsymbol{\beta}},j} = \boldsymbol{X}_{\widehat{\boldsymbol{\beta}}, j}-\boldsymbol{X}_{\widehat{\boldsymbol{\beta}}, -j} \widehat{\boldsymbol{\gamma}}_{\widehat{\boldsymbol{\beta}},j},
\end{equation} 
and the empirical version of $\boldsymbol{\Sigma}_{\boldsymbol{\beta}_0}$ from \eqref{populationgram} as
\begin{equation}\label{empgram}
\widehat{\boldsymbol{\Sigma}}_{\widehat{\boldsymbol{\beta}}}  = \frac{1}{N} \sum_{i=1}^N \frac{\exp(\boldsymbol{X}_i^{\top}\widehat{\boldsymbol{\beta}})}{ \big(1+\exp(\boldsymbol{X}_i^{\top}\widehat{\boldsymbol{\beta}})\big)^2 } \boldsymbol{X}_{i} \boldsymbol{X}_{i}^{\top}.
\end{equation}
By the definition of $\widehat{\boldsymbol{C}}_j$ in \eqref{hatBC}, we have $\boldsymbol{X}_{\widehat{\boldsymbol{\beta}}, j}-\boldsymbol{X}_{\widehat{\boldsymbol{\beta}}, -j} \widehat{\boldsymbol{\gamma}}_{\widehat{\boldsymbol{\beta}}, j} = \boldsymbol{X}_{\widehat{\boldsymbol{\beta}}}\widehat{\boldsymbol{C}}_j$. Hence,
$$
\hat{\tau}_j^2 :=\frac{\boldsymbol{X}_{\widehat{\boldsymbol{\beta}}, j}^{\top}\left(\boldsymbol{X}_{\widehat{\boldsymbol{\beta}}, j}-\boldsymbol{X}_{\widehat{\boldsymbol{\beta}}, -j} \widehat{\boldsymbol{\gamma}}_{\widehat{\boldsymbol{\beta}}, j}\right)}{N} = \frac{\boldsymbol{X}_{\widehat{\boldsymbol{\beta}}, j}^{\top}\boldsymbol{X}_{\widehat{\boldsymbol{\beta}}}\widehat{\boldsymbol{C}}_j}{N},
$$
which shows 
\begin{equation}\label{kktbefore}
    \frac{\boldsymbol{X}_{\widehat{\boldsymbol{\beta}}, j}^{\top}\boldsymbol{X}_{\widehat{\boldsymbol{\beta}}}\widehat{\boldsymbol{\Theta}}_j}{N} = 1,
\end{equation}
since $\widehat{\boldsymbol{\Theta}}_j:=\widehat{\boldsymbol{C}}_j/\hat{\tau}_j^2$.

By the Karush-Kuhn-Tucker conditions of \eqref{nodewiseregression}, we have
$$
\frac{\left|\boldsymbol{X}_{\widehat{\boldsymbol{\beta}}, -j}^{\top}\left(\boldsymbol{X}_{\widehat{\boldsymbol{\beta}}, j}-\boldsymbol{X}_{\widehat{\boldsymbol{\beta}}, -j} \widehat{\boldsymbol{\gamma}}_{\widehat{\boldsymbol{\beta}}, j}\right)\right|_{\infty}}{N} \leq \lambda_{j},
$$
and dividing each side by $ \hat{\tau}_j^2$, using $\boldsymbol{X}_{\widehat{\boldsymbol{\beta}}, j}-\boldsymbol{X}_{\widehat{\boldsymbol{\beta}}, -j} \widehat{\boldsymbol{\gamma}}_{\widehat{\boldsymbol{\beta}}, j} = \boldsymbol{X}_{\widehat{\boldsymbol{\beta}}}\widehat{\boldsymbol{C}}_j$ and $\widehat{\boldsymbol{\Theta}}_j=\widehat{\boldsymbol{C}}_j / \hat{\tau}_j^2$, we have
\begin{equation}\label{kktafter}
\frac{\left|\boldsymbol{X}_{\widehat{\boldsymbol{\beta}}, -j}^{\top}\boldsymbol{X}_{\widehat{\boldsymbol{\beta}}}\widehat{\boldsymbol{\Theta}}_j\right|_{\infty}}{N} \leq \lambda_{j}/\hat{\tau}_j^2.
\end{equation}
Combining \eqref{kktbefore} and \eqref{kktafter}, one can see that
\begin{equation}\label{KKTused}
\begin{aligned}
 \lambda_{j}/\hat{\tau}_j^2 & \geq   \frac{\left|\boldsymbol{X}_{\widehat{\boldsymbol{\beta}}, -j}^{\top}\boldsymbol{X}_{\widehat{\boldsymbol{\beta}}}\widehat{\boldsymbol{\Theta}}_j\right|_{\infty}}{N}\\
 & = \frac{\left|\boldsymbol{X}_{\widehat{\boldsymbol{\beta}}}^{\top}\boldsymbol{X}_{\widehat{\boldsymbol{\beta}}}\widehat{\boldsymbol{\Theta}}_j - e_j\boldsymbol{X}_{\widehat{\boldsymbol{\beta}},j}^{\top}\boldsymbol{X}_{\widehat{\boldsymbol{\beta}}}\widehat{\boldsymbol{\Theta}}_j\right|_{\infty}}{N}\\
 &  = \left|\widehat{\boldsymbol{\Sigma}}_{\widehat{\boldsymbol{\beta}}}\widehat{\boldsymbol{\Theta}}_j - e_j\right|_{\infty},
\end{aligned}
\end{equation}
where $e_j \in \mathbb{R}^{p}$ has $1$ in the $j$-th element and $0$ for others and by definition $\widehat{\boldsymbol{\Sigma}}_{\widehat{\boldsymbol{\beta}}} = \boldsymbol{X}_{\widehat{\boldsymbol{\beta}}}^{\top}\boldsymbol{X}_{\widehat{\boldsymbol{\beta}}}/N$.

\subsection{Proof of Lemma \ref{nodewise}}
\begin{proof}
We first introduce some specific notations. Let $\boldsymbol{\Theta}_{j,-j}$ denote be the $j$-th row of $\boldsymbol{\Theta}$ after removing its $j$-th column. For the positive definite matrix $\boldsymbol{\Sigma}_{\boldsymbol{\beta}_0}$, let $\boldsymbol{\Sigma}_{j ,\boldsymbol{\beta}_0, j}$ be the $(j,j)$ element of $\boldsymbol{\Sigma}_{\boldsymbol{\beta}_0}$, $\boldsymbol{\Sigma}_{j ,\boldsymbol{\beta}_0, -j}$ be the $j$-th row of $\boldsymbol{\Sigma}_{\boldsymbol{\beta}_0}$ with $j$-th column removed, $\boldsymbol{\Sigma}_{-j ,\boldsymbol{\beta}_0, j}$ be the $j$-th column of $\boldsymbol{\Sigma}_{\boldsymbol{\beta}_0}$ with $j$-th row removed, and $\boldsymbol{\Sigma}_{-j ,\boldsymbol{\beta}_0, -j}$ be the submatrix of $\boldsymbol{\Sigma}_{\boldsymbol{\beta}_0}$ with $j$-th row and $j$-th column removed. 

By the partitioned inverse formula of the matrix $\boldsymbol{\Sigma}_{\boldsymbol{\beta}_0}$, we know that
$$
\begin{aligned}
& \boldsymbol{\Theta}_{j, j}=\left(\boldsymbol{\Sigma}_{j, \boldsymbol{\beta}_0, j}-\boldsymbol{\Sigma}_{j, \boldsymbol{\beta}_0, -j} \boldsymbol{\Sigma}_{-j, \boldsymbol{\beta}_0, -j}^{-1} \boldsymbol{\Sigma}_{-j, \boldsymbol{\beta}_0, j}\right)^{-1}, \\
& \boldsymbol{\Theta}_{j,-j}=-\boldsymbol{\Theta}_{j, j} \boldsymbol{\Sigma}_{j, \boldsymbol{\beta}_0, -j} \boldsymbol{\Sigma}_{-j, \boldsymbol{\beta}_0,-j}^{-1}.
\end{aligned}
$$
It is obvious that 
$$
\boldsymbol{\gamma}_{\boldsymbol{\beta}_0, j}^{\top}=\boldsymbol{\Sigma}_{j,\boldsymbol{\beta}_0,-j} \boldsymbol{\Sigma}_{-j\boldsymbol{\beta}_0,-j}^{-1},
$$
where it comes from the solution of the problem \eqref{populationnodewise1}. Hence, we have
\begin{equation}\label{Theta-j}
\boldsymbol{\Theta}_{j,-j}=-\boldsymbol{\Theta}_{j, j} \boldsymbol{\gamma}_{\boldsymbol{\beta_0}, j}^{\top}.
\end{equation}

Define 
\begin{equation}\label{taupo}
\tau_{ j}^2 := \frac{1}{N} \mathbbm{E} \left[ \boldsymbol{X}_{\boldsymbol{\beta}_0, j}^{\top} \left( \boldsymbol{X}_{\boldsymbol{\beta}_0, j} - \boldsymbol{X}_{\boldsymbol{\beta}_0, -j} \boldsymbol{\gamma}_{\boldsymbol{\beta}_0, j} \right) \right].
\end{equation}
By definition of $ \boldsymbol{\gamma}_{\boldsymbol{\beta}_0, j}$, it holds that
\begin{equation}\label{tauupbound}
\begin{aligned}
\tau_{ j}^2 & = \mathbbm{E}\left(\boldsymbol{X}_{\boldsymbol{\beta}_0, j}^{\top}\left( \boldsymbol{X}_{\boldsymbol{\beta}_0, j}-\boldsymbol{X}_{\boldsymbol{\beta}_0,-j} \boldsymbol{\gamma}_{\boldsymbol{\beta}_0, j}\right) / N\right)\\
& = \mathbbm{E}\left( \left(\boldsymbol{X}_{\boldsymbol{\beta}_0, j} - \boldsymbol{X}_{\boldsymbol{\beta}_0,-j} \boldsymbol{\gamma}_{\boldsymbol{\beta}_0, j}\right)^{\top}\left( \boldsymbol{X}_{\boldsymbol{\beta}_0, j}-\boldsymbol{X}_{\boldsymbol{\beta}_0,-j} \boldsymbol{\gamma}_{\boldsymbol{\beta}_0, j}\right) / N\right) \\
&= \mathbbm{E}\left(\boldsymbol{\eta}_{\boldsymbol{\beta}_0,j}^{\top}\boldsymbol{\eta}_{\boldsymbol{\beta}_0,j}/N\right)\\
& \lesssim C_{\eta},
\end{aligned}
\end{equation}
where the last inequality comes from Assumption \ref{aseign_2} (iii).
By the partitioned inverse formula, we have
\begin{equation}\label{taulowbound}
\begin{aligned}
\tau_{ j}^2 
&= \boldsymbol{\Sigma}_{j ,\boldsymbol{\beta}_0, j} - \boldsymbol{\Sigma}_{j ,\boldsymbol{\beta}_0, -j}\boldsymbol{\Sigma}^{-1}_{-j ,\boldsymbol{\beta}_0, -j}\boldsymbol{\Sigma}_{-j ,\boldsymbol{\beta}_0, j} \\
& = \boldsymbol{\Theta}_{j, j}^{-1}= \frac{1}{(\boldsymbol{\Sigma}_{\boldsymbol{\beta}_0}^{-1})_{j,j}} \\
&\geq \frac{1}{\lambda_{\max}(\boldsymbol{\Sigma}_{\boldsymbol{\beta}_0}^{-1})} = \lambda_{\min}(\boldsymbol{\Sigma}_{\boldsymbol{\beta}_0}) \geq \gamma_{H},
\end{aligned}
\end{equation}
where the first inequality follows from the fact that $\boldsymbol{\Sigma}_{\boldsymbol{\beta}_0}$ is symmetric positive definite, and the last inequality comes from Assumption~\ref{aseign}.

Therefore, by the definitions of $\widehat{\boldsymbol{\Theta}}_j$, \eqref{Theta-j}, and the second equality of \eqref{taulowbound}, we have
$$
\begin{aligned}
\max _{j \in \mathcal{J}}\left|\widehat{\boldsymbol{\Theta}}_j-\boldsymbol{\Theta}_j\right|_1 & \leq \max _{j \in \mathcal{J}}\left|\frac{1}{\hat{\tau}_{j}^2}-\frac{1}{\tau_{ j}^2}\right|+\max _{j \in \mathcal{J}}\left|\frac{\widehat{\boldsymbol{\gamma}}_{\widehat{\boldsymbol{\beta}}, j}}{\hat{\tau}_{j}^2}-\frac{\boldsymbol{\gamma}_{\boldsymbol{\beta}_0, j}}{\tau_{ j}^2}\right|_1 \\
& \leq \max _{j \in \mathcal{J}}\left|\frac{1}{\hat{\tau}_{j}^2}-\frac{1}{\tau_{ j}^2}\right|+\max _{j \in \mathcal{J}} \left| \frac{\widehat{\boldsymbol{\gamma}}_{\widehat{\boldsymbol{\beta}}, j}-\boldsymbol{\gamma}_{\boldsymbol{\beta}_0, j}}{\hat{\tau}_{j}^2} \right|_1  +\max _{j \in \mathcal{J}} \left| \frac{\boldsymbol{\gamma}_{\boldsymbol{\beta}_0, j}}{\hat{\tau}_{j}^2}-\frac{\boldsymbol{\gamma}_{\boldsymbol{\beta}_0, j}}{\tau_{ j}^2} \right|_1 \\
& \leq \max _{j \in \mathcal{J}}\left|\frac{1}{\hat{\tau}_{j}^2}-\frac{1}{\tau_{ j}^2}\right|+\max _{j \in \mathcal{J}}\left|\frac{1}{\hat{\tau}_{j}^2}\right| \max _{j \in \mathcal{J}}\left|\widehat{\boldsymbol{\gamma}}_{\widehat{\boldsymbol{\beta}}, j}-\boldsymbol{\gamma}_{\boldsymbol{\beta}_0, j}\right|_1 \\
& \quad +\max _{j \in \mathcal{J}}\left|\frac{1}{\hat{\tau}_{j}^2}-\frac{1}{\tau_{ j}^2}\right| \max _{j \in \mathcal{J}}\left|\boldsymbol{\gamma}_{\boldsymbol{\beta}_0, j}\right|_1 .
\end{aligned}
$$

By \eqref{taulowbound}, we have $\tau_{ j}^2 \geq \gamma_{H} > 0$ and, since $1/\hat{\tau}_{j}^2$ is a consistent estimator of $1/\tau_{ j}^2$ by Lemma \ref{alemma1}, continuous mapping theorem and Assumption~\ref{asrate}, this implies
\begin{equation}\label{hattaubound}
1/\hat{\tau}_{j}^2 = O_P(1).
\end{equation}
Then, we have
\begin{equation}\label{rate1}
\begin{aligned}
\max _{j \in \mathcal{J}} \left| \frac{1}{\hat{\tau}_{j}^2} - \frac{1}{\tau_{ j}^2} \right|
&= \max _{j \in \mathcal{J}} \left| \frac{\hat{\tau}_{j}^2 - \tau_{ j}^2}{\hat{\tau}_{j}^2 \, \tau_{ j}^2} \right| \\
&= O_P \Bigg( \frac{p^{2/q} (\log p)^{1/2 + 2/q}}{N^{1/2 - 3/q}} s^* + \frac{p^{4/q} (\log p)^{1 + 4/q}}{N^{1 - 5/q}} (s^*)^{3} \Bigg),
\end{aligned}
\end{equation}
where we use Lemma \ref{alemma1}. By Lemma \ref{nwe} and \eqref{hattaubound}, we have
\begin{equation}\label{rate2}
\begin{aligned}
& \max _{j \in \mathcal{J}}\left|\frac{1}{\hat{\tau}_{j}^2}\right| \max _{j \in \mathcal{J}}\left|\widehat{\boldsymbol{\gamma}}_{\widehat{\boldsymbol{\beta}}, j}-\boldsymbol{\gamma}_{\boldsymbol{\beta}_0, j}\right|_1\\
& =O_p(1)O_P\left( \frac{p^{2/q}(\log p)^{1/2+2/q}}{N^{1/2-2/q}}(s^*)^2 \right)  \\
& = O_P\left( \frac{p^{2/q}(\log p)^{1/2+2/q}}{N^{1/2-2/q}}(s^*)^2 \right),
\end{aligned}
\end{equation}
and, by Lemmas \ref{alemma1} and \ref{gammanull},
\begin{equation}\label{rate3}
\begin{aligned}
&\max _{j \in \mathcal{J}}\left|\frac{1}{\hat{\tau}_{j}^2}-\frac{1}{\tau_{ j}^2}\right| \max _{j \in \mathcal{J}}\left|\boldsymbol{\gamma}_{\boldsymbol{\beta}_0, j}\right|_1\\
&=O_P \Bigg( \frac{p^{2/q} (\log p)^{1/2 + 2/q}}{N^{1/2 - 3/q}} s^* + \frac{p^{4/q} (\log p)^{1 + 4/q}}{N^{1 - 5/q}} (s^*)^{3} \Bigg) O\left(\bar{s}^{1 / 2}\right)\\
& = O_P \Bigg( \frac{p^{2/q} (\log p)^{1/2 + 2/q}}{N^{1/2 - 3/q}} (s^*)^{1.5} + \frac{p^{4/q} (\log p)^{1 + 4/q}}{N^{1 - 5/q}} (s^*)^{3.5} \Bigg).
\end{aligned}
\end{equation}
Observe that, by Assumption \ref{asrate}, we have 
$$
\begin{aligned}
\frac{\frac{p^{4/q}(\log p)^{1+4/q}}{N^{1-5/q}}(s^*)^{3.5}}{\frac{p^{2/q}(\log p)^{1/2+2/q}}{N^{1/2-3/q}}(s^*)^{1.5}} & = \frac{p^{2/q}(\log p)^{1/2+2/q}}{N^{1/2-2/q}}(s^*)^2\\
& = o\left(\frac{1}{N^{4/q}p^{2/q}(\log p)^{5/6 +2/q}}\right)=o(1).
\end{aligned}
$$
Hence, taking the slowest rate of \eqref{rate1}, \eqref{rate2}, and \eqref{rate3}, we have 
\begin{equation}\label{tt2rate}
\begin{aligned}
\max _{j \in \mathcal{J}}\left|\widehat{\boldsymbol{\Theta}}_j-\boldsymbol{\Theta}_j\right|_1 & =  O_P \Bigg(\frac{p^{2/q} (\log p)^{1/2 + 2/q}}{N^{1/2 - 3/q}} (s^*)^{1.5}  \Bigg) \\
& =o_P\left(\frac{1}{N^{1/8+3/(2q)}p^{1/q}(\log p)^{1/2+1/q}}\right) = o_P(1),
\end{aligned}
\end{equation}
where we use Assumption \ref{asrate}. The proof is complete.
\end{proof}

\subsection{Auxiliary lemmas}
\begin{lemma}\label{alemma1}
   Let $\mathcal{J} \subseteq [p]$ be a fixed group. Suppose Assumptions \ref{asapprerror}, \ref{aseign_2}, \ref{asrate}, and all conditions of Theorem \ref{the} are satisfied. If there exist sufficiently large constants $\mathcal{K}_1$, $\tilde{\mathcal{K}}_1$, $\mathcal{K}_2$, and $\tilde{\mathcal{K}}_2$ such that
\[
\mathcal{K}_1\frac{p^{1/q} \sqrt{\log p}}{N^{1/2 - 1/q}} \leq \lambda \leq \tilde{\mathcal{K}}_1\frac{p^{1/q} \sqrt{\log p}}{N^{1/2 - 1/q}} ,
\]
and
\[
\mathcal{K}_2\frac{(p-1)^{2/q} \sqrt{\log p}}{N^{1/2 - 2/q}} \leq \lambda_j \leq \tilde{\mathcal{K}}_2\frac{(p-1)^{2/q} \sqrt{\log p}}{N^{1/2 - 2/q}} \quad \text{for all } j \in \mathcal{J},
\] 
then, we have
    \begin{equation}
    \begin{aligned}
\max_{j \in \mathcal{J}}\left|\hat{\tau}_{j}^2 - \tau_{ j}^2\right| = O_P \Bigg( \frac{p^{2/q} (\log p)^{1/2 + 2/q}}{N^{1/2 - 3/q}} s^* + \frac{p^{4/q} (\log p)^{1 + 4/q}}{N^{1 - 5/q}} (s^*)^{3} \Bigg).
        \end{aligned}
    \end{equation}
\end{lemma}
\begin{proof}
 We fix $j \in \mathcal{J}$. Recall the definitions
\[
\hat{\tau}_{j}^2 := \frac{1}{N} \boldsymbol{X}_{\widehat{\boldsymbol{\beta}}, j}^{\top} \left( \boldsymbol{X}_{\widehat{\boldsymbol{\beta}}, j} - \boldsymbol{X}_{\widehat{\boldsymbol{\beta}}, -j} \widehat{\boldsymbol{\gamma}}_{\widehat{\boldsymbol{\beta}}, j} \right),
\]
and \eqref{taupo}
\[
\tau_{ j}^2 := \frac{1}{N} \mathbbm{E} \left[ \boldsymbol{X}_{\boldsymbol{\beta}_0, j}^{\top} \left( \boldsymbol{X}_{\boldsymbol{\beta}_0, j} - \boldsymbol{X}_{\boldsymbol{\beta}_0, -j} \boldsymbol{\gamma}_{\boldsymbol{\beta}_0, j} \right) \right].
\]
By inserting
\[
\boldsymbol{X}_{\widehat{\boldsymbol{\beta}}, j} = W_{\widehat{\boldsymbol{\beta}}} W_{\boldsymbol{\beta}_0}^{-1} \boldsymbol{X}_{\boldsymbol{\beta}_0, j}, \quad 
\boldsymbol{X}_{\widehat{\boldsymbol{\beta}}, -j} = W_{\widehat{\boldsymbol{\beta}}} W_{\boldsymbol{\beta}_0}^{-1} \boldsymbol{X}_{\boldsymbol{\beta}_0, -j},
\]
we can decompose the difference as
\[
\begin{aligned}
\left|\hat{\tau}_{j}^2 - \tau_{ j}^2\right| 
&= \left|\frac{1}{N} \boldsymbol{X}_{\boldsymbol{\beta}_0, j}^{\top} \left( \boldsymbol{X}_{\boldsymbol{\beta}_0, j} - \boldsymbol{X}_{\boldsymbol{\beta}_0, -j} \widehat{\boldsymbol{\gamma}}_{\widehat{\boldsymbol{\beta}}, j} \right) - \tau_{ j}^2\right| \\
&\quad + \left|\frac{1}{N} \boldsymbol{X}_{\boldsymbol{\beta}_0, j}^{\top} \left( W_{\widehat{\boldsymbol{\beta}}}^2 W_{\boldsymbol{\beta}_0}^{-2} - I_N \right) \left( \boldsymbol{X}_{\boldsymbol{\beta}_0, j} - \boldsymbol{X}_{\boldsymbol{\beta}_0, -j} \widehat{\boldsymbol{\gamma}}_{\widehat{\boldsymbol{\beta}}, j} \right)\right| \\
& \leq  (i) + (ii),
\end{aligned}
\]
where
\[
(i) := \left|\frac{1}{N} \boldsymbol{X}_{\boldsymbol{\beta}_0, j}^{\top} \left( \boldsymbol{X}_{\boldsymbol{\beta}_0, j} - \boldsymbol{X}_{\boldsymbol{\beta}_0, -j} \widehat{\boldsymbol{\gamma}}_{\widehat{\boldsymbol{\beta}}, j} \right) - \tau_{ j}^2\right|,
\]
and
\[
(ii) := \left|\frac{1}{N} \boldsymbol{X}_{\boldsymbol{\beta}_0, j}^{\top} \left( W_{\widehat{\boldsymbol{\beta}}}^2 W_{\boldsymbol{\beta}_0}^{-2} - I_N \right) \left( \boldsymbol{X}_{\boldsymbol{\beta}_0, j} - \boldsymbol{X}_{\boldsymbol{\beta}_0, -j} \widehat{\boldsymbol{\gamma}}_{\widehat{\boldsymbol{\beta}}, j} \right)\right|.
\]

First, we analyze $(i)$.
$$
\begin{aligned}
(i)&= \left|\frac{\boldsymbol{X}_{\boldsymbol{\beta}_0, j}^{\top}\left(\boldsymbol{X}_{\boldsymbol{\beta}_0, j}-\boldsymbol{X}_{\boldsymbol{\beta}_0,-j} \widehat{\boldsymbol{\gamma}}_{\widehat{\boldsymbol{\beta}}, j}\right)}{N}-\tau_{ j}^2\right| \\
& =  \left|\frac{\boldsymbol{X}_{\boldsymbol{\beta}_0, j}^{\top}\left(\boldsymbol{X}_{\boldsymbol{\beta}_0, -j}\boldsymbol{\gamma}_{\boldsymbol{\beta}_0, j} + \boldsymbol{\eta}_{\boldsymbol{\beta}_0,j}-\boldsymbol{X}_{\boldsymbol{\beta}_0,-j} \widehat{\boldsymbol{\gamma}}_{\widehat{\boldsymbol{\beta}}, j}\right)}{N}-\tau_{ j}^2\right| \\
& = \left|\frac{\boldsymbol{X}_{\boldsymbol{\beta}_0, j}^{\top}\boldsymbol{X}_{\boldsymbol{\beta}_0, -j}\left(\boldsymbol{\gamma}_{\boldsymbol{\beta}_0, j} - \widehat{\boldsymbol{\gamma}}_{\widehat{\boldsymbol{\beta}}, j}\right)}{N} + \frac{\left(\boldsymbol{X}_{\boldsymbol{\beta}_0, -j}\boldsymbol{\gamma}_{\boldsymbol{\beta}_0, j} + \boldsymbol{\eta}_{\boldsymbol{\beta}_0,j}\right)^{\top}\boldsymbol{\eta}_{\boldsymbol{\beta}_0,j}}{N} -\tau_{ j}^2\right|\\
& \leq \left|\frac{\boldsymbol{X}_{\boldsymbol{\beta}_0, j}^{\top}\boldsymbol{X}_{\boldsymbol{\beta}_0, -j}\left(\boldsymbol{\gamma}_{\boldsymbol{\beta}_0, j} - \widehat{\boldsymbol{\gamma}}_{\widehat{\boldsymbol{\beta}}, j}\right)}{N}\right| + \left|\frac{\boldsymbol{\gamma}_{\boldsymbol{\beta}_0, j}^{\top}\boldsymbol{X}_{\boldsymbol{\beta}_0, -j}^{\top}\boldsymbol{\eta}_{\boldsymbol{\beta}_0,j}}{N}\right| + \left|\frac{\boldsymbol{\eta}_{\boldsymbol{\beta}_0,j}^{\top}\boldsymbol{\eta}_{\boldsymbol{\beta}_0,j}}{N} -\tau_{ j}^2\right|,
\end{aligned}
$$
where we insert \eqref{populationnodewise} to get the first and second equalities.
For the first term on the right-hand side of $(i)$, we have
\[
\begin{aligned}
\left|\frac{1}{N} \boldsymbol{X}_{\boldsymbol{\beta}_0, j}^{\top} \boldsymbol{X}_{\boldsymbol{\beta}_0, -j} \left( \boldsymbol{\gamma}_{\boldsymbol{\beta}_0, j} - \widehat{\boldsymbol{\gamma}}_{\widehat{\boldsymbol{\beta}}, j} \right)\right|
&= \left|\frac{1}{N} \boldsymbol{X}_{\boldsymbol{\beta}_0, j}^{\top} W_{\boldsymbol{\beta}_0} W_{\widehat{\boldsymbol{\beta}}}^{-1} \boldsymbol{X}_{\widehat{\boldsymbol{\beta}}, -j} \left( \boldsymbol{\gamma}_{\boldsymbol{\beta}_0, j} - \widehat{\boldsymbol{\gamma}}_{\widehat{\boldsymbol{\beta}}, j} \right)\right| \\
&\leq \frac{1}{N} \left| \boldsymbol{X}_{\boldsymbol{\beta}_0, j} \right|_{\infty} \left| W_{\boldsymbol{\beta}_0} W_{\widehat{\boldsymbol{\beta}}}^{-1} \boldsymbol{X}_{\widehat{\boldsymbol{\beta}}, -j} \left( \boldsymbol{\gamma}_{\boldsymbol{\beta}_0, j} - \widehat{\boldsymbol{\gamma}}_{\widehat{\boldsymbol{\beta}}, j} \right) \right|_1.
\end{aligned}
\]
Before proceeding, we note the following general inequality: for a diagonal matrix $A = \mathrm{diag}(a_1, \ldots, a_N)$ and a vector $x = (x_1, \ldots, x_N)^{\top}$, it holds that
\[
\frac{|Ax|_1}{N} = \frac{\sum_{i=1}^N |a_i x_i|}{N} \leq \sqrt{\frac{\sum_{i=1}^N a_i^2}{N}} \, \sqrt{\frac{\sum_{i=1}^N x_i^2}{N}} = \| \mathrm{vec}(A) \|_N \, \| x \|_N.
\]
Applying this to $W_{\boldsymbol{\beta}_0} W_{\widehat{\boldsymbol{\beta}}}^{-1} = \mathrm{diag}\left( \frac{w_{1, \boldsymbol{\beta}_0}}{w_{1, \widehat{\boldsymbol{\beta}}}}, \ldots, \frac{w_{N, \boldsymbol{\beta}_0}}{w_{N, \widehat{\boldsymbol{\beta}}}} \right)$, and using \ref{lneed2} which ensures that $\max_{i\in[N]}\left|\frac{w^2_{i, \boldsymbol{\beta}_0}}{w^2_{i, \widehat{\boldsymbol{\beta}}}}\right| \leq 4/3$ with probability $1 - o(1)$, we obtain that, with probability $1 - o(1)$,
\[
\begin{aligned}
&\frac{1}{N} \left| \boldsymbol{X}_{\boldsymbol{\beta}_0, j} \right|_{\infty} \left| W_{\boldsymbol{\beta}_0} W_{\widehat{\boldsymbol{\beta}}}^{-1} \boldsymbol{X}_{\widehat{\boldsymbol{\beta}}, -j} \left( \boldsymbol{\gamma}_{\boldsymbol{\beta}_0, j} - \widehat{\boldsymbol{\gamma}}_{\widehat{\boldsymbol{\beta}}, j} \right) \right|_1 \\
&\quad \lesssim (N \log p)^{1/q} \, \left\| \boldsymbol{X}_{\widehat{\boldsymbol{\beta}}, -j} \left( \boldsymbol{\gamma}_{\boldsymbol{\beta}_0, j} - \widehat{\boldsymbol{\gamma}}_{\widehat{\boldsymbol{\beta}}, j} \right) \right\|_N \\
&\quad \lesssim \frac{p^{2/q} (\log p)^{1/2 + 2/q}}{N^{1/2 - 3/q}} s^*,
\end{aligned}
\]
where the first inequality uses a bound on $\left| \boldsymbol{X}_{\boldsymbol{\beta}_0, j} \right|_{\infty}$ via Markov's inequality (since $W_{\boldsymbol{\beta}_0}$ is bounded), following the argument in \eqref{pa3}, and the second inequality is justified by Lemma \ref{nwe}.

For the second term on the right-hand side of $(i)$, it holds with probability $1 - o(1)$,
\[
\begin{aligned}
\left| \frac{1}{N} \boldsymbol{\gamma}_{\boldsymbol{\beta}_0, j}^{\top} \boldsymbol{X}_{\boldsymbol{\beta}_0, -j}^{\top} \boldsymbol{\eta}_{\boldsymbol{\beta}_0, j} \right|
&\leq \left| \boldsymbol{\gamma}_{\boldsymbol{\beta}_0, j} \right|_1 \, \left| \frac{1}{N} \boldsymbol{X}_{\boldsymbol{\beta}_0, -j}^{\top} \boldsymbol{\eta}_{\boldsymbol{\beta}_0, j} \right|_{\infty} \\
&\lesssim \frac{(p-1)^{2/q} \sqrt{\log p}}{N^{1/2 - 2/q}} \sqrt{s^*},
\end{aligned}
\]
where we use Lemmas \ref{lneed3} and \ref{gammanull} to justify the last inequality.
For the last term on the right-hand side of $(i)$, by Chebyshev's inequality and Assumptions \ref{asX}, \ref{aseign_2} (ii) and \ref{aseign_2} (iii), we have
\[
\left|\frac{1}{N} \boldsymbol{\eta}_{\boldsymbol{\beta}_0, j}^{\top} \boldsymbol{\eta}_{\boldsymbol{\beta}_0, j} - \tau_{ j}^2\right| = O_P\left( \frac{1}{\sqrt{N}} \right).
\]
Combining these results, we conclude that
\[
(i) = O_P\left( \frac{p^{2/q} (\log p)^{1/2 + 2/q}}{N^{1/2 - 3/q}} s^* \right).
\]

Next, we analyze term $(ii)$.
For term $(ii)$, we have with probability $1-o(1)$
\[
\begin{aligned}
& \left|\frac{1}{N} \boldsymbol{X}_{\boldsymbol{\beta}_0, j}^{\top} \left( W_{\widehat{\boldsymbol{\beta}}}^2 W_{\boldsymbol{\beta}_0}^{-2} - I_N \right) \left( \boldsymbol{X}_{\boldsymbol{\beta}_0, j} - \boldsymbol{X}_{\boldsymbol{\beta}_0, -j} \widehat{\boldsymbol{\gamma}}_{\widehat{\boldsymbol{\beta}}, j} \right)\right| \\
&= \left|\frac{1}{N} \boldsymbol{X}_{\boldsymbol{\beta}_0, j}^{\top} \left( W_{\widehat{\boldsymbol{\beta}}}^2 W_{\boldsymbol{\beta}_0}^{-2} - I_N \right) \left( \boldsymbol{X}_{\boldsymbol{\beta}_0, -j} \boldsymbol{\gamma}_{\boldsymbol{\beta}_0, j} + \boldsymbol{\eta}_{\boldsymbol{\beta}_0, j} - \boldsymbol{X}_{\boldsymbol{\beta}_0, -j} \widehat{\boldsymbol{\gamma}}_{\widehat{\boldsymbol{\beta}}, j} \right)\right| \\
&\leq \left| \boldsymbol{X}_{\boldsymbol{\beta}_0, j} \right|_{\infty} \left[ \frac{1}{N} \sum_{i=1}^N \left| \frac{w_{i, \widehat{\boldsymbol{\beta}}}^2 - w_{i, \boldsymbol{\beta}_0}^2}{w_{i, \boldsymbol{\beta}_0}^2} \right| \right] 
\Big( \left| \boldsymbol{\eta}_{\boldsymbol{\beta}_0, j} \right|_{\infty} + \left| \boldsymbol{X}_{\boldsymbol{\beta}_0, -j} \boldsymbol{\gamma}_{\boldsymbol{\beta}_0, j} - \boldsymbol{X}_{\boldsymbol{\beta}_0, -j} \widehat{\boldsymbol{\gamma}}_{\widehat{\boldsymbol{\beta}}, j} \right|_{\infty} \Big) \\
&\leq \left| \boldsymbol{X}_{\boldsymbol{\beta}_0, j} \right|_{\infty} \left( \lambda s_{\boldsymbol{\beta}_0} \right) 
\Big( \left| \boldsymbol{\eta}_{\boldsymbol{\beta}_0, j} \right|_{\infty} + \left| \boldsymbol{X}_{\boldsymbol{\beta}_0, -j} \right|_{\infty} \left| \boldsymbol{\gamma}_{\boldsymbol{\beta}_0, j} - \widehat{\boldsymbol{\gamma}}_{\widehat{\boldsymbol{\beta}}, j} \right|_1 \Big) \\
&\lesssim (N \log p)^{1/q} \lambda s_{\boldsymbol{\beta}_0} \Big[ (N \log p)^{1/q} + (p N \log p)^{1/q} \Big( \lambda_{j} |S_j| + \frac{N^{2/q}(\log p)^{2/q} \lambda^2 s^2_{\boldsymbol{\beta}_0}}{\lambda_{j}} \Big) \Big] \\
& \lesssim \frac{p^{1/q} (\log p)^{1/2 + 1/q}}{N^{1/2 - 2/q}}s^*\Big[ (N \log p)^{1/q} + (p N \log p)^{1/q} \Big( \frac{p^{2/q}(\log p)^{1/2+2/q}}{N^{1/2-2/q}}(s^*)^2 \Big) \Big] \\
&\lesssim \frac{p^{1/q} (\log p)^{1/2 + 2/q}}{N^{1/2 - 3/q}} s^* + \frac{p^{4/q} (\log p)^{1 + 4/q}}{N^{1 - 5/q}} (s^*)^{3},
\end{aligned}
\]
where the first inequality follows from
\[
\left| \boldsymbol{a}^{\top} \boldsymbol{M} \boldsymbol{b} \right| \leq \left| \boldsymbol{a} \right|_{\infty} \left| \boldsymbol{M} \boldsymbol{b} \right|_1 \leq \left| \boldsymbol{a} \right|_{\infty} \left[ \sum_{i=1}^N \left|\boldsymbol{M}_{i,i}\right| \right] \left| \boldsymbol{b} \right|_{\infty},
\]
for $\boldsymbol{a} \in \mathbb{R}^N$, $\boldsymbol{M} \in \mathbb{R}^{N \times N}$ diagonal, and $\boldsymbol{b} \in \mathbb{R}^N$, and the remaining bounds follow from Markov's inequality for $\left| \boldsymbol{X}_{\boldsymbol{\beta}_0, j} \right|_{\infty}$, Lemma \ref{weightbound}, \eqref{nodewisepart1}, and Lemma \ref{nwe}.

Finally, combining the bounds for $(i)$ and $(ii)$, we obtain
\[
\max_{j \in \mathcal{J}}\left|\hat{\tau}_{j}^2 - \tau_{ j}^2\right| = O_P \Bigg( \frac{p^{2/q} (\log p)^{1/2 + 2/q}}{N^{1/2 - 3/q}} s^* + \frac{p^{4/q} (\log p)^{1 + 4/q}}{N^{1 - 5/q}} (s^*)^{3} \Bigg),
\]
since group $\mathcal{J}$ has fixed size. The proof is complete.   
\end{proof}

\begin{lemma} \label{nwe}
Let $\mathcal{J} \subseteq [p]$ be a fixed group. Suppose the conditions of Theorem \ref{the} and Assumptions \ref{asapprerror}, \ref{aseign_2}, and \ref{asrate} hold. If there exist sufficiently large constants $\mathcal{K}_1$, $\tilde{\mathcal{K}}_1$, $\mathcal{K}_2$, and $\tilde{\mathcal{K}}_2$ such that
\[
\mathcal{K}_1\frac{p^{1/q} \sqrt{\log p}}{N^{1/2 - 1/q}} \leq \lambda \leq \tilde{\mathcal{K}}_1\frac{p^{1/q} \sqrt{\log p}}{N^{1/2 - 1/q}} ,
\]
and
\[
\mathcal{K}_2\frac{(p-1)^{2/q} \sqrt{\log p}}{N^{1/2 - 2/q}} \leq \lambda_j \leq \tilde{\mathcal{K}}_2\frac{(p-1)^{2/q} \sqrt{\log p}}{N^{1/2 - 2/q}} \quad \text{for all } j \in \mathcal{J},
\] 
then, with probability going to $1$, we have for every $j \in \mathcal{J}$
$$
\left|\widehat{\boldsymbol{\gamma}}_{\widehat{\boldsymbol{\beta}}, j} - \boldsymbol{\gamma}_{\boldsymbol{\beta}_0, j} \right|_1  \lesssim \lambda_{j}|S_j| +  N^{2/q}(\log p)^{2/q} \lambda^2 s^2_{\boldsymbol{\beta}_0}/\lambda_{j}.$$
Furthermore, it also holds with probability $1-o(1)$
$$
\left|\widehat{\boldsymbol{\gamma}}_{\widehat{\boldsymbol{\beta}}, j} - \boldsymbol{\gamma}_{\boldsymbol{\beta}_0, j} \right|_1  \lesssim \frac{p^{2/q}(\log p)^{1/2+2/q}}{N^{1/2-2/q}}(s^*)^2,
$$
and
$$
\| \boldsymbol{X}_{\widehat{\boldsymbol{\beta}},-j}\left(\widehat{\boldsymbol{\gamma}}_{\widehat{\boldsymbol{\beta}},j}-\boldsymbol{\gamma}_{\boldsymbol{\beta}_0, j}\right) \|_N^2 \lesssim \frac{p^{4/q}(\log p)^{1+2/q}}{N^{1-4/q}}(s^*)^2.
$$
\end{lemma}

\begin{proof}
We divide the proof into two parts. First, we define the events in Lemmas \ref{lneedextra}, \ref{lneed1}, \ref{lneed2}, \ref{lneed3}, and \ref{lneed4} as $\mathcal{E}_1, \mathcal{E}_2, \mathcal{E}_3, \mathcal{E}_4,$ and $\mathcal{E}_5$ correspondingly.
The following proof works on these events, where we know $P\left(\mathcal{E}_1 \cap \mathcal{E}_2 \cap \mathcal{E}_3 \cap \mathcal{E}_4 \cap \mathcal{E}_5 \right) \rightarrow 1$.
\paragraph{Part (i)}
In this part, the goal is to get a key inequality that is useful for the proof. 

By the two-point inequality of \citet[p. 9]{van_de_geer_estimation_2016}, which is derived from the Karush-Kuhn-Tucker conditions of the problem \eqref{nodewiseregression}, we have
\begin{equation}\label{node1}
\frac{\left(\boldsymbol{\gamma}_{\boldsymbol{\beta}_0, j}-\widehat{\boldsymbol{\gamma}}_{\widehat{\boldsymbol{\beta}}, j}\right)^{\top} 
\boldsymbol{X}_{\widehat{\boldsymbol{\beta}},-j}^{\top}\left(\boldsymbol{X}_{\widehat{\boldsymbol{\beta}}, j}-\boldsymbol{X}_{\widehat{\boldsymbol{\beta}},-j} \widehat{\boldsymbol{\gamma}}_{\widehat{\boldsymbol{\beta}}, j}\right)}{N} 
\leq \lambda_{j} \left|\boldsymbol{\gamma}_{\boldsymbol{\beta}_0, j}\right|_1 - \lambda_{j} \left|\widehat{\boldsymbol{\gamma}}_{\widehat{\boldsymbol{\beta}}, j}\right|_1.
\end{equation}
From \eqref{populationnodewise}, we can decompose
\[
\boldsymbol{X}_{\boldsymbol{\beta}_0,j} = \boldsymbol{X}_{\boldsymbol{\beta}_0,-j} \boldsymbol{\gamma}_{\boldsymbol{\beta}_0, j} + \boldsymbol{\eta}_{\boldsymbol{\beta}_0, j}.
\]
Hence, after scaling with $W_{\widehat{\boldsymbol{\beta}}} W_{\boldsymbol{\beta}_0}^{-1}$ and $\boldsymbol{X}_{\widehat{\boldsymbol{\beta}}} = W_{\widehat{\boldsymbol{\beta}}} W_{\boldsymbol{\beta}_0}^{-1} \boldsymbol{X}_{\boldsymbol{\beta}_0}$, we have
\[
\boldsymbol{X}_{\widehat{\boldsymbol{\beta}}, j} = \boldsymbol{X}_{\widehat{\boldsymbol{\beta}}, -j} \boldsymbol{\gamma}_{\boldsymbol{\beta}_0, j} + W_{\widehat{\boldsymbol{\beta}}} W_{\boldsymbol{\beta}_0}^{-1} \boldsymbol{\eta}_{\boldsymbol{\beta}_0, j}.
\]
Substituting the above into \eqref{node1}, we obtain
\begin{align*}
& \frac{\left(\boldsymbol{\gamma}_{\boldsymbol{\beta}_0, j}-\widehat{\boldsymbol{\gamma}}_{\widehat{\boldsymbol{\beta}}, j}\right)^{\top} 
\boldsymbol{X}_{\widehat{\boldsymbol{\beta}},-j}^{\top}\left(\boldsymbol{X}_{\widehat{\boldsymbol{\beta}}, j}-\boldsymbol{X}_{\widehat{\boldsymbol{\beta}},-j} \widehat{\boldsymbol{\gamma}}_{\widehat{\boldsymbol{\beta}}, j}\right)}{N} \\
&= \frac{\left(\boldsymbol{\gamma}_{\boldsymbol{\beta}_0, j}-\widehat{\boldsymbol{\gamma}}_{\widehat{\boldsymbol{\beta}}, j}\right)^{\top} 
\boldsymbol{X}_{\widehat{\boldsymbol{\beta}},-j}^{\top} \boldsymbol{X}_{\widehat{\boldsymbol{\beta}},-j} 
\left(\boldsymbol{\gamma}_{\boldsymbol{\beta}_0, j}-\widehat{\boldsymbol{\gamma}}_{\widehat{\boldsymbol{\beta}}, j}\right)}{N}  + \frac{\left(\boldsymbol{\gamma}_{\boldsymbol{\beta}_0, j}-\widehat{\boldsymbol{\gamma}}_{\widehat{\boldsymbol{\beta}}, j}\right)^{\top} 
\boldsymbol{X}_{\widehat{\boldsymbol{\beta}},-j}^{\top} \left(W_{\widehat{\boldsymbol{\beta}}} W_{\boldsymbol{\beta}_0}^{-1} \boldsymbol{\eta}_{\boldsymbol{\beta}_0, j}\right)}{N} \\
&\leq \lambda_{j} \left|\boldsymbol{\gamma}_{\boldsymbol{\beta}_0, j}\right|_1 - \lambda_{j} \left|\widehat{\boldsymbol{\gamma}}_{\widehat{\boldsymbol{\beta}}, j}\right|_1.
\end{align*}
Rewriting, we have
\begin{equation}\label{nodea}
\begin{aligned}
& \frac{\left(\widehat{\boldsymbol{\gamma}}_{\widehat{\boldsymbol{\beta}}, j}-\boldsymbol{\gamma}_{\boldsymbol{\beta}_0, j}\right)^{\top} 
\boldsymbol{X}_{\widehat{\boldsymbol{\beta}},-j}^{\top} \boldsymbol{X}_{\widehat{\boldsymbol{\beta}},-j} 
\left(\widehat{\boldsymbol{\gamma}}_{\widehat{\boldsymbol{\beta}}, j}-\boldsymbol{\gamma}_{\boldsymbol{\beta}_0, j}\right)}{N} \\
&\leq \frac{\left(\widehat{\boldsymbol{\gamma}}_{\widehat{\boldsymbol{\beta}}, j}-\boldsymbol{\gamma}_{\boldsymbol{\beta}_0, j}\right)^{\top} 
\left(\boldsymbol{X}_{\widehat{\boldsymbol{\beta}},-j}^{\top} W_{\widehat{\boldsymbol{\beta}}} W_{\boldsymbol{\beta}_0}^{-1} \boldsymbol{\eta}_{\boldsymbol{\beta}_0, j} - \boldsymbol{X}_{\boldsymbol{\beta}_0,-j}^{\top} \boldsymbol{\eta}_{\boldsymbol{\beta}_0,j}\right)}{N} \\
&\quad + \frac{\left(\widehat{\boldsymbol{\gamma}}_{\widehat{\boldsymbol{\beta}}, j}-\boldsymbol{\gamma}_{\boldsymbol{\beta}_0, j}\right)^{\top} 
\boldsymbol{X}_{\boldsymbol{\beta}_0,-j}^{\top} \boldsymbol{\eta}_{\boldsymbol{\beta}_0,j}}{N} 
+ \lambda_{j} \left|\boldsymbol{\gamma}_{\boldsymbol{\beta}_0, j}\right|_1 - \lambda_{j} \left|\widehat{\boldsymbol{\gamma}}_{\widehat{\boldsymbol{\beta}}, j}\right|_1.
\end{aligned}
\end{equation}
For the first term on the right-hand side of \eqref{nodea}, since $\boldsymbol{X}_{\widehat{\boldsymbol{\beta}},-j} = W_{\widehat{\boldsymbol{\beta}}} W_{\boldsymbol{\beta}_0}^{-1} \boldsymbol{X}_{\boldsymbol{\beta}_0,-j}$ where both $W_{\widehat{\boldsymbol{\beta}}}$ and $W_{\boldsymbol{\beta}_0}$ are diagonal, we have
\begin{equation}\label{tem1}
\begin{aligned}
& \frac{\left(\widehat{\boldsymbol{\gamma}}_{\widehat{\boldsymbol{\beta}}, j}-\boldsymbol{\gamma}_{\boldsymbol{\beta}_0, j}\right)^{\top} 
\left(\boldsymbol{X}_{\widehat{\boldsymbol{\beta}},-j}^{\top} W_{\widehat{\boldsymbol{\beta}}} W_{\boldsymbol{\beta}_0}^{-1} \boldsymbol{\eta}_{\boldsymbol{\beta}_0, j} - \boldsymbol{X}_{\boldsymbol{\beta}_0,-j}^{\top} \boldsymbol{\eta}_{\boldsymbol{\beta}_0,j}\right)}{N} \\
&= \frac{\left(\widehat{\boldsymbol{\gamma}}_{\widehat{\boldsymbol{\beta}}, j}-\boldsymbol{\gamma}_{\boldsymbol{\beta}_0, j}\right)^{\top} 
\boldsymbol{X}_{\boldsymbol{\beta}_0,-j}^{\top} \left(W_{\widehat{\boldsymbol{\beta}}}^2 W_{\boldsymbol{\beta}_0}^{-2} - I_N\right) \boldsymbol{\eta}_{\boldsymbol{\beta}_0,j}}{N} \\
&\leq \left\| \left(W_{\widehat{\boldsymbol{\beta}}}^2 W_{\boldsymbol{\beta}_0}^{-2} - I_N \right) \boldsymbol{\eta}_{\boldsymbol{\beta}_0, j} \right\|_N \, 
\left\| \boldsymbol{X}_{\boldsymbol{\beta}_0,-j} \left(\widehat{\boldsymbol{\gamma}}_{\widehat{\boldsymbol{\beta}}, j}-\boldsymbol{\gamma}_{\boldsymbol{\beta}_0,j}\right) \right\|_N \\
&\leq \frac{1}{2} \left\| \left(W_{\widehat{\boldsymbol{\beta}}}^2 W_{\boldsymbol{\beta}_0}^{-2} - I_N \right) \boldsymbol{\eta}_{\boldsymbol{\beta}_0, j} \right\|_N^2 
+ \frac{1}{2} \left\| \boldsymbol{X}_{\boldsymbol{\beta}_0,-j} \left(\widehat{\boldsymbol{\gamma}}_{\widehat{\boldsymbol{\beta}}, j}-\boldsymbol{\gamma}_{\boldsymbol{\beta}_0,j}\right) \right\|_N^2,
\end{aligned}
\end{equation}
where we use the Cauchy–Schwarz inequality for the first inequality. Furthermore, notice that
\begin{equation}
    \label{tem2}
\begin{aligned}
\left\| \left(W_{\widehat{\boldsymbol{\beta}}}^2 W_{\boldsymbol{\beta}_0}^{-2}-I_N \right) \boldsymbol{\eta}_{\boldsymbol{\beta}_0, j} \right\|_N^2 
& = \frac{1}{N} \sum_{i=1}^N \left(\frac{w_{i,\widehat{\boldsymbol{\beta}}}^2 - w_{i,\boldsymbol{\beta}_0}^2}{w_{i,\boldsymbol{\beta}_0}^2}\right)^2 \eta_{i,\boldsymbol{\beta}_0,j}^2\\
& \leq \frac{1}{N} \sum_{i=1}^N \left(\frac{w_{i,\widehat{\boldsymbol{\beta}}}^2 - w_{i,\boldsymbol{\beta}_0}^2}{w_{i,\boldsymbol{\beta}_0}^2}\right)^2 \, \max_{j \in \mathcal{J}} \left|\boldsymbol{\eta}_{\boldsymbol{\beta}_0, j}\right|_\infty^2\\
& \leq N^{2/q} (\log p)^{2/q}\lambda^2 s^2_{\boldsymbol{\beta}_0}\\
& = \tau,
\end{aligned}
\end{equation}
where we use Lemma \ref{weightbound}, \eqref{nodewisepart1} and let $\tau \sim N^{2/q} (\log p)^{2/q}\lambda^2 s^2_{\boldsymbol{\beta}_0}$.
Then, combining \eqref{tem1} and \eqref{tem2}, \eqref{nodea} can be rewritten as
\begin{equation}
\begin{aligned}
& \frac{\left(\widehat{\boldsymbol{\gamma}}_{\widehat{\boldsymbol{\beta}}, j} - \boldsymbol{\gamma}_{\boldsymbol{\beta}_0, j}\right)^{\top} 
\boldsymbol{X}_{\widehat{\boldsymbol{\beta}}, -j}^{\top} \boldsymbol{X}_{\widehat{\boldsymbol{\beta}}, -j} 
\left(\widehat{\boldsymbol{\gamma}}_{\widehat{\boldsymbol{\beta}}, j} - \boldsymbol{\gamma}_{\boldsymbol{\beta}_0, j}\right)}{N} \\
& \lesssim \frac{\tau}{2} + \frac{\left\| \boldsymbol{X}_{\boldsymbol{\beta}_0,-j} 
\left(\widehat{\boldsymbol{\gamma}}_{\widehat{\boldsymbol{\beta}}, j} - \boldsymbol{\gamma}_{\boldsymbol{\beta}_0, j}\right) \right\|_N^2}{2} \\
& \quad + \frac{\left(\widehat{\boldsymbol{\gamma}}_{\widehat{\boldsymbol{\beta}}, j} - \boldsymbol{\gamma}_{\boldsymbol{\beta}_0, j}\right)^{\top} 
\boldsymbol{X}_{\boldsymbol{\beta}_0,-j}^{\top} \boldsymbol{\eta}_{\boldsymbol{\beta}_0, j}}{N} 
+ \lambda_{j} \left| \boldsymbol{\gamma}_{\boldsymbol{\beta}_0, j} \right|_1 - \lambda_{j} \left| \widehat{\boldsymbol{\gamma}}_{\widehat{\boldsymbol{\beta}}, j} \right|_1 \\
& \lesssim \frac{\tau}{2} + \frac{2 \left\| \boldsymbol{X}_{\widehat{\boldsymbol{\beta}},-j} 
\left(\widehat{\boldsymbol{\gamma}}_{\widehat{\boldsymbol{\beta}}, j} - \boldsymbol{\gamma}_{\boldsymbol{\beta}_0, j}\right) \right\|_N^2}{3} \\
& \quad + \frac{\left(\widehat{\boldsymbol{\gamma}}_{\widehat{\boldsymbol{\beta}}, j} - \boldsymbol{\gamma}_{\boldsymbol{\beta}_0, j}\right)^{\top} 
\boldsymbol{X}_{\boldsymbol{\beta}_0,-j}^{\top} \boldsymbol{\eta}_{\boldsymbol{\beta}_0, j}}{N} 
+ \lambda_{j} \left| \boldsymbol{\gamma}_{\boldsymbol{\beta}_0, j} \right|_1 - \lambda_{j} \left| \widehat{\boldsymbol{\gamma}}_{\widehat{\boldsymbol{\beta}}, j} \right|_1,
\end{aligned}
\end{equation}
where the second inequality uses  
$\boldsymbol{X}_{\widehat{\boldsymbol{\beta}}, -j} = W_{\widehat{\boldsymbol{\beta}}} W_{\boldsymbol{\beta}_0}^{-1} \boldsymbol{X}_{\boldsymbol{\beta}_0, -j}$ and Lemma \ref{lneed2}.

Finally, we obtain the key inequality
\begin{equation}\label{nodeb}
\begin{aligned}
& \frac{\left\| \boldsymbol{X}_{\widehat{\boldsymbol{\beta}},-j} 
\left(\widehat{\boldsymbol{\gamma}}_{\widehat{\boldsymbol{\beta}}, j} - \boldsymbol{\gamma}_{\boldsymbol{\beta}_0, j}\right) \right\|_N^2}{3} 
+ \lambda_{j} \left| \widehat{\boldsymbol{\gamma}}_{\widehat{\boldsymbol{\beta}}, j} \right|_1 \\
& \lesssim \frac{\tau}{2} + \frac{\left(\widehat{\boldsymbol{\gamma}}_{\widehat{\boldsymbol{\beta}}, j} - \boldsymbol{\gamma}_{\boldsymbol{\beta}_0, j}\right)^{\top} 
\boldsymbol{X}_{\boldsymbol{\beta}_0,-j}^{\top} \boldsymbol{\eta}_{\boldsymbol{\beta}_0, j}}{N} 
+ \lambda_{j} \left| \boldsymbol{\gamma}_{\boldsymbol{\beta}_0, j} \right|_1.
\end{aligned}
\end{equation}

\paragraph{Part (ii)}
We consider two cases in the following.
First, we assume 
\[
\frac{\lambda_{j}\left|\widehat{\boldsymbol{\gamma}}_{\widehat{\boldsymbol{\beta}}, j} - \boldsymbol{\gamma}_{\boldsymbol{\beta}_0, j} \right|_1}{8} + \left\| \boldsymbol{X}_{\widehat{\boldsymbol{\beta}},-j} \left( \widehat{\boldsymbol{\gamma}}_{\widehat{\boldsymbol{\beta}},j} - \boldsymbol{\gamma}_{\boldsymbol{\beta}_0, j} \right) \right\|_N^2 \leq \frac{3\tau}{2}.
\]
Since 
\[
\left\| \boldsymbol{X}_{\widehat{\boldsymbol{\beta}},-j} \left( \widehat{\boldsymbol{\gamma}}_{\widehat{\boldsymbol{\beta}},j} - \boldsymbol{\gamma}_{\boldsymbol{\beta}_0, j} \right) \right\|_N^2 \geq 0,
\] 
it immediately follows that 
\[
\left|\widehat{\boldsymbol{\gamma}}_{\widehat{\boldsymbol{\beta}}, j} - \boldsymbol{\gamma}_{\boldsymbol{\beta}_0, j} \right|_1 \leq \frac{12\tau}{\lambda_{j}},
\] 
and simultaneously,
\[
\left\| \boldsymbol{X}_{\widehat{\boldsymbol{\beta}},-j} \left( \widehat{\boldsymbol{\gamma}}_{\widehat{\boldsymbol{\beta}},j} - \boldsymbol{\gamma}_{\boldsymbol{\beta}_0, j} \right) \right\|_N^2 \leq \frac{3\tau}{2}.
\]

Next, consider the opposite case:
\begin{equation}\label{case2}
\frac{\lambda_{j}\left|\widehat{\boldsymbol{\gamma}}_{\widehat{\boldsymbol{\beta}}, j} - \boldsymbol{\gamma}_{\boldsymbol{\beta}_0, j} \right|_1}{8} + \left\| \boldsymbol{X}_{\widehat{\boldsymbol{\beta}},-j} \left( \widehat{\boldsymbol{\gamma}}_{\widehat{\boldsymbol{\beta}},j} - \boldsymbol{\gamma}_{\boldsymbol{\beta}_0, j} \right) \right\|_N^2 \geq \frac{3\tau}{2}.
\end{equation}
To facilitate further analysis, we introduce notations. Let $$
\boldsymbol{X}_{\boldsymbol{\beta}_0,S_j-\{j\}}^{\top} = \left(\boldsymbol{X}_{\boldsymbol{\beta}_0,k}^{\top}\right)_{k \neq j, k \in S_j} \in \mathbb{R}^{(|S_j-\{j\}|) \times N}
$$ where $\boldsymbol{X}_{\boldsymbol{\beta}_0,k}$ is the $k$-th column of $\boldsymbol{X}_{\boldsymbol{\beta}_0}$ and 
$$\boldsymbol{X}_{\boldsymbol{\beta}_0,S_j^c-\{j\}}^{\top} = \left(\boldsymbol{X}_{\boldsymbol{\beta}_0,k}^{\top}\right)_{k \neq j, k \in S^c_j} \in \mathbb{R}^{(|S^c_j-\{j\}|) \times N}.
$$
Let 
$$
\widehat{\boldsymbol{\gamma}}_{\widehat{\boldsymbol{\beta}}, S_j - \{j\}} = \left(\widehat{\gamma}_{\widehat{\boldsymbol{\beta}},j,k}\right)_{k \neq j, k \in S_j} \in \mathbb{R}^{|S_j-\{j\}|},
$$ and 
$$
\widehat{\boldsymbol{\gamma}}_{\widehat{\boldsymbol{\beta}}, S_j^c - \{j\}} = \left(\widehat{\gamma}_{\widehat{\boldsymbol{\beta}},j,k}\right)_{k \neq j, k \in S_j^c} \in \mathbb{R}^{|S_j^c-\{j\}|}.
$$
Consider the following
\[
\begin{aligned}
& \left|\frac{\left(\widehat{\boldsymbol{\gamma}}_{\widehat{\boldsymbol{\beta}}, j} - \boldsymbol{\gamma}_{\boldsymbol{\beta}_0, j}\right)^{\top} \boldsymbol{X}_{\boldsymbol{\beta}_0,-j}^{\top} \boldsymbol{\eta}_{\boldsymbol{\beta}_0,j}}{N}\right|\\
&= \left|\frac{\left(\widehat{\boldsymbol{\gamma}}_{\widehat{\boldsymbol{\beta}}, S_j - \{j\}} - \boldsymbol{\gamma}_{\boldsymbol{\beta}_0, S_j-\{j\}}\right)^{\top} \boldsymbol{X}_{\boldsymbol{\beta}_0,S_j-\{j\}}^{\top} \boldsymbol{\eta}_{\boldsymbol{\beta}_0,j} + \left(\widehat{\boldsymbol{\gamma}}_{\widehat{\boldsymbol{\beta}}, S_j^c-\{j\}} - \boldsymbol{\gamma}_{\boldsymbol{\beta}_0, S_j^c-\{j\}}\right)^{\top} \boldsymbol{X}_{\boldsymbol{\beta}_0,S_j^c-\{j\}}^{\top} \boldsymbol{\eta}_{\boldsymbol{\beta}_0,j}}{N}\right| \\
&\leq \left|\widehat{\boldsymbol{\gamma}}_{\widehat{\boldsymbol{\beta}}, S_j - \{j\}} - \boldsymbol{\gamma}_{\boldsymbol{\beta}_0, S_j-\{j\}} \right|_1 \left| \frac{\boldsymbol{X}_{\boldsymbol{\beta}_0,S_j-\{j\}}^{\top} \boldsymbol{\eta}_{\boldsymbol{\beta}_0,j}}{N} \right|_{\infty} + \left|\widehat{\boldsymbol{\gamma}}_{\widehat{\boldsymbol{\beta}}, S_j^c-\{j\}} - \boldsymbol{\gamma}_{\boldsymbol{\beta}_0, S_j^c-\{j\}} \right|_1 \left| \frac{\boldsymbol{X}_{\boldsymbol{\beta}_0,S_j^c-\{j\}}^{\top} \boldsymbol{\eta}_{\boldsymbol{\beta}_0,j}}{N} \right|_{\infty} \\
&\leq \lambda_c \left|\widehat{\boldsymbol{\gamma}}_{\widehat{\boldsymbol{\beta}}, S_j - \{j\}} - \boldsymbol{\gamma}_{\boldsymbol{\beta}_0, S_j-\{j\}} \right|_1  + \lambda_c \left|\widehat{\boldsymbol{\gamma}}_{\widehat{\boldsymbol{\beta}}, S_j^c-\{j\}} \right|_1,
\end{aligned}
\]
where we have used $\boldsymbol{\gamma}_{\boldsymbol{\beta}_0, S_j^c-\{j\}} = \boldsymbol{0}$ by definition, and $\lambda_c \sim O\left(\frac{(p-1)^{2/q} \sqrt{\log p}}{N^{1/2 - 2/q}}\right)$ as derived in Lemma \ref{lneed3}.  
Hence, from \eqref{nodeb}, it follows that
\begin{equation}\label{nodec}
\begin{aligned}
\frac{\left\| \boldsymbol{X}_{\widehat{\boldsymbol{\beta}},-j} \left(\widehat{\boldsymbol{\gamma}}_{\widehat{\boldsymbol{\beta}},j} - \boldsymbol{\gamma}_{\boldsymbol{\beta}_0, j}\right) \right\|_N^2}{3} 
&\lesssim \lambda_c \left|\widehat{\boldsymbol{\gamma}}_{\widehat{\boldsymbol{\beta}}, S_j - \{j\}} - \boldsymbol{\gamma}_{\boldsymbol{\beta}_0, S_j-\{j\}} \right|_1 
+ \lambda_c \left|\widehat{\boldsymbol{\gamma}}_{\widehat{\boldsymbol{\beta}}, S_j^c-\{j\}} \right|_1
+ \frac{\tau}{2} \\
&\quad + \lambda_{j} \left| \boldsymbol{\gamma}_{\boldsymbol{\beta}_0, j} \right|_1 - \lambda_{j} \left| \widehat{\boldsymbol{\gamma}}_{\widehat{\boldsymbol{\beta}}, j} \right|_1.
\end{aligned}
\end{equation}
By applying the triangle inequality, we further obtain
\[
\begin{aligned}
\lambda_{j} \left| \boldsymbol{\gamma}_{\boldsymbol{\beta}_0, j} \right|_1 - \lambda_{j} \left| \widehat{\boldsymbol{\gamma}}_{\widehat{\boldsymbol{\beta}}, j} \right|_1
&= \lambda_{j} \left| \boldsymbol{\gamma}_{\boldsymbol{\beta}_0, S_j-\{j\}} \right|_1 - \lambda_{j} \left| \widehat{\boldsymbol{\gamma}}_{\widehat{\boldsymbol{\beta}}, S_j - \{j\}} \right|_1 - \lambda_{j} \left| \widehat{\boldsymbol{\gamma}}_{\widehat{\boldsymbol{\beta}}, S_j^c-\{j\}} \right|_1 \\
&\leq \lambda_{j} \left| \widehat{\boldsymbol{\gamma}}_{\widehat{\boldsymbol{\beta}}, S_j - \{j\}} - \boldsymbol{\gamma}_{\boldsymbol{\beta}_0, S_j-\{j\}} \right|_1 - \lambda_{j} \left| \widehat{\boldsymbol{\gamma}}_{\widehat{\boldsymbol{\beta}}, S_j^c-\{j\}} \right|_1.
\end{aligned}
\]
Substituting the above bound into \eqref{nodec} and because of $\lambda_{j} \gtrsim 2\lambda_c$, we obtain
\begin{equation}\label{noded}
\begin{aligned}
\frac{\left\| \boldsymbol{X}_{\widehat{\boldsymbol{\beta}},-j} \left(\widehat{\boldsymbol{\gamma}}_{\widehat{\boldsymbol{\beta}},j} - \boldsymbol{\gamma}_{\boldsymbol{\beta}_0, j}\right) \right\|_N^2}{3} 
&\lesssim \frac{3}{2} \lambda_{j} \left| \widehat{\boldsymbol{\gamma}}_{\widehat{\boldsymbol{\beta}}, S_j - \{j\}} - \boldsymbol{\gamma}_{\boldsymbol{\beta}_0, S_j-\{j\}} \right|_1 
- \frac{1}{2} \lambda_{j} \left| \widehat{\boldsymbol{\gamma}}_{\widehat{\boldsymbol{\beta}}, S_j^c-\{j\}} \right|_1 + \frac{\tau}{2}.
\end{aligned}
\end{equation}
Recalling the assumed condition \eqref{case2} and \eqref{noded}, we obtain
\[
-\frac{\lambda_{j} \left| \widehat{\boldsymbol{\gamma}}_{\widehat{\boldsymbol{\beta}}, j} - \boldsymbol{\gamma}_{\boldsymbol{\beta}_0, j} \right|_1}{24} 
\leq \frac{3}{2} \lambda_{j}  \left| \widehat{\boldsymbol{\gamma}}_{\widehat{\boldsymbol{\beta}}, S_j - \{j\}} - \boldsymbol{\gamma}_{\boldsymbol{\beta}_0, S_j-\{j\}} \right|_1 - \frac{1}{2} \lambda_{j} \left| \widehat{\boldsymbol{\gamma}}_{\widehat{\boldsymbol{\beta}}, S_j^c-\{j\}} \right|_1.
\]
By noting that 
\[
\left| \widehat{\boldsymbol{\gamma}}_{\widehat{\boldsymbol{\beta}}, j} - \boldsymbol{\gamma}_{\boldsymbol{\beta}_0, j} \right|_1 
= \left| \widehat{\boldsymbol{\gamma}}_{\widehat{\boldsymbol{\beta}}, S_j - \{j\}} - \boldsymbol{\gamma}_{\boldsymbol{\beta}_0, S_j-\{j\}} \right|_1 + \left| \widehat{\boldsymbol{\gamma}}_{\widehat{\boldsymbol{\beta}}, S_j^c-\{j\}} \right|_1,
\] 
it follows that
\[
\frac{37 \lambda_{j} \left| \widehat{\boldsymbol{\gamma}}_{\widehat{\boldsymbol{\beta}}, S_j - \{j\}} - \boldsymbol{\gamma}_{\boldsymbol{\beta}_0, S_j-\{j\}} \right|_1}{24} \geq \frac{11}{24} \lambda_{j} \left| \widehat{\boldsymbol{\gamma}}_{\widehat{\boldsymbol{\beta}}, S_j^c-\{j\}} \right|_1.
\]
 Then we set $L = \frac{37}{11}$ in Lemma \ref{lneed4} and applying Lemma \ref{lneed5}, we have
\[
\hat{\phi}^2(L, S_j-\{j\}) \geq \frac{\phi^2(L, S_j-\{j\})}{2} > \frac{\gamma_{\mathrm{H}}}{2}.
\]
Hence, revisiting \eqref{noded}, we deduce
\begin{equation}\label{nodee1}
\begin{aligned}
&\frac{\left\| \boldsymbol{X}_{\widehat{\boldsymbol{\beta}},-j} \left( \widehat{\boldsymbol{\gamma}}_{\widehat{\boldsymbol{\beta}},j} - \boldsymbol{\gamma}_{\boldsymbol{\beta}_0, j} \right) \right\|_N^2}{3} 
+ \frac{1}{2} \lambda_{j} \left| \widehat{\boldsymbol{\gamma}}_{\widehat{\boldsymbol{\beta}}, S_j - \{j\}} - \boldsymbol{\gamma}_{\boldsymbol{\beta}_0, S_j-\{j\}} \right|_1
+ \frac{1}{2} \lambda_{j} \left| \widehat{\boldsymbol{\gamma}}_{\widehat{\boldsymbol{\beta}}, S_j^c-\{j\}} \right|_1 \\
&\lesssim 2 \lambda_{j} \left| \widehat{\boldsymbol{\gamma}}_{\widehat{\boldsymbol{\beta}}, S_j - \{j\}} - \boldsymbol{\gamma}_{\boldsymbol{\beta}_0, S_j-\{j\}} \right|_1  + \frac{\tau}{2} \\
&\leq 2 \lambda_{j} \sqrt{|S_j|} \frac{\left\| \boldsymbol{X}_{\widehat{\boldsymbol{\beta}},-j} \left( \widehat{\boldsymbol{\gamma}}_{\widehat{\boldsymbol{\beta}},j} - \boldsymbol{\gamma}_{\boldsymbol{\beta}_0, j} \right) \right\|_N}{\hat{\phi}(L, S_j-\{j\})} + \frac{\tau}{2} \\
&\leq \frac{3 \lambda_{j}^2 |S_j|}{\hat{\phi}^2(L, S_j-\{j\})} + \frac{\left\| \boldsymbol{X}_{\widehat{\boldsymbol{\beta}},-j} \left( \widehat{\boldsymbol{\gamma}}_{\widehat{\boldsymbol{\beta}},j} - \boldsymbol{\gamma}_{\boldsymbol{\beta}_0, j} \right) \right\|_N^2}{3} + \frac{\tau}{2} \\
&\leq \frac{6 \lambda_{j}^2 |S_j|}{\gamma_{\mathrm{H}}} + \frac{\left\| \boldsymbol{X}_{\widehat{\boldsymbol{\beta}},-j} \left( \widehat{\boldsymbol{\gamma}}_{\widehat{\boldsymbol{\beta}},j} - \boldsymbol{\gamma}_{\boldsymbol{\beta}_0, j} \right) \right\|_N^2}{3} + \frac{\tau}{2},
\end{aligned}
\end{equation}
where the second inequality follows from the definition of $\hat{\phi}(L, S_j-\{j\})$, the definition of $S_j$ and the fact that $|S_j - \{j\}| \leq |S_j|$. Consequently, from \eqref{nodee1}, it follows that
\begin{equation}\label{nodee2}
\left| \widehat{\boldsymbol{\gamma}}_{\widehat{\boldsymbol{\beta}}, j} - \boldsymbol{\gamma}_{\boldsymbol{\beta}_0, j} \right|_1 
\lesssim \frac{ \lambda_{j} |S_j|}{\gamma_{\mathrm{H}}} + \frac{\tau}{\lambda_{j}}.
\end{equation}
Similarly, we also have
\begin{equation}\label{nodeeX}
\begin{aligned}
&\frac{\left\| \boldsymbol{X}_{\widehat{\boldsymbol{\beta}},-j} \left( \widehat{\boldsymbol{\gamma}}_{\widehat{\boldsymbol{\beta}},j} - \boldsymbol{\gamma}_{\boldsymbol{\beta}_0, j} \right) \right\|_N^2}{3} 
+ \frac{1}{2} \left| \widehat{\boldsymbol{\gamma}}_{\widehat{\boldsymbol{\beta}}, S_j - \{j\}} - \boldsymbol{\gamma}_{\boldsymbol{\beta}_0, S_j-\{j\}} \right|_1 \lambda_{j} 
+ \frac{1}{2} \left| \widehat{\boldsymbol{\gamma}}_{\widehat{\boldsymbol{\beta}}, S_j^c-\{j\}} \right|_1 \lambda_{j} \\
&\leq 2 \lambda_{j} \sqrt{|S_j|} \frac{\left\| \boldsymbol{X}_{\widehat{\boldsymbol{\beta}},-j} \left( \widehat{\boldsymbol{\gamma}}_{\widehat{\boldsymbol{\beta}},j} - \boldsymbol{\gamma}_{\boldsymbol{\beta}_0, j} \right) \right\|_N}{\hat{\phi}(L,S_j-\{j\})} + \frac{\tau}{2} \\
&\leq \frac{8 \lambda_{j}^2 |S_j|}{\hat{\phi}^2(L,S_j-\{j\})} 
+ \frac{\left\| \boldsymbol{X}_{\widehat{\boldsymbol{\beta}},-j} \left( \widehat{\boldsymbol{\gamma}}_{\widehat{\boldsymbol{\beta}},j} - \boldsymbol{\gamma}_{\boldsymbol{\beta}_0, j} \right) \right\|_N^2}{8} + \frac{\tau}{2} \\
&\leq \frac{16 \lambda_{j}^2 |S_j|}{\gamma_{\mathrm{H}}} 
+ \frac{\left\| \boldsymbol{X}_{\widehat{\boldsymbol{\beta}},-j} \left( \widehat{\boldsymbol{\gamma}}_{\widehat{\boldsymbol{\beta}},j} - \boldsymbol{\gamma}_{\boldsymbol{\beta}_0, j} \right) \right\|_N^2}{8} + \frac{\tau}{2}.
\end{aligned}
\end{equation}
Then, it follows that
\[
\left\| \boldsymbol{X}_{\widehat{\boldsymbol{\beta}},-j} \left( \widehat{\boldsymbol{\gamma}}_{\widehat{\boldsymbol{\beta}},j} - \boldsymbol{\gamma}_{\boldsymbol{\beta}_0, j} \right) \right\|_N^2 
\lesssim \frac{\lambda_{j}^2 |S_j|}{\gamma_{\mathrm{H}}} + \tau,
\]
where $ \tau \sim N^{2/q} (\log p)^{2/q}\lambda^2 s^2_{\boldsymbol{\beta}_0}$. 
It immediately follows that
\begin{equation}\label{noderes1}
\left|\widehat{\boldsymbol{\gamma}}_{\widehat{\boldsymbol{\beta}}, j} - \boldsymbol{\gamma}_{\boldsymbol{\beta}_0, j} \right|_1  \lesssim \lambda_{j}|S_j| +  N^{2/q} (\log p)^{2/q}\lambda^2 s^2_{\boldsymbol{\beta}_0}/\lambda_{j},
\end{equation}
and
\begin{equation}\label{noderes2}
\| \boldsymbol{X}_{\widehat{\boldsymbol{\beta}},-j}\left(\widehat{\boldsymbol{\gamma}}_{\widehat{\boldsymbol{\beta}},j}-\boldsymbol{\gamma}_{\boldsymbol{\beta}_0, j}\right) \|_N^2 \lesssim \lambda^2_{j}|S_j| +  N^{2/q}(\log p)^{2/q} \lambda^2 s^2_{\boldsymbol{\beta}_0}.
\end{equation}
Recalling that 
\[
\lambda_{j} \sim O\left( \frac{(p-1)^{2/q} \sqrt{\log p}}{N^{1/2 - 2/q}} \right) \quad \text{and} \quad 
\lambda \sim O\left( \frac{p^{1/q} \sqrt{\log p}}{N^{1/2 - 1/q}} \right),
\] 
and the definition
\[
s^* = \max_j |S_j| \vee s_{\boldsymbol{\beta}_0},
\] 
we immediately deduce from the previous bounds that, for any $j \in \mathcal{J}$,
\begin{equation}\label{noderes3}
\left| \widehat{\boldsymbol{\gamma}}_{\widehat{\boldsymbol{\beta}}, j} - \boldsymbol{\gamma}_{\boldsymbol{\beta}_0, j} \right|_1 
\lesssim \frac{p^{2/q} (\log p)^{1/2 + 2/q}}{N^{1/2 - 2/q}} (s^*)^2,
\end{equation}
and
\begin{equation}\label{noderes4}
\left\| \boldsymbol{X}_{\widehat{\boldsymbol{\beta}},-j} \left( \widehat{\boldsymbol{\gamma}}_{\widehat{\boldsymbol{\beta}}, j} - \boldsymbol{\gamma}_{\boldsymbol{\beta}_0, j} \right) \right\|_N^2 
\lesssim \frac{p^{4/q} (\log p)^{1 + 2/q}}{N^{1 - 4/q}} (s^*)^2.
\end{equation}
Recall that $P\left(\mathcal{E}_1 \cap \mathcal{E}_2 \cap \mathcal{E}_3 \cap \mathcal{E}_4 \cap \mathcal{E}_5 \right) \rightarrow 1$, we then have \eqref{noderes1}, \eqref{noderes2}, \eqref{noderes3}, and \eqref{noderes4} hold with probability $1-o(1)$ .

\end{proof}

\begin{lemma}\label{gammanull}
Let $\mathcal{J} \subseteq [p]$ be a fixed group. Under Assumptions \ref{asX} and \ref{aseign}, for every $j \in \mathcal{J}$, we have $\left|\boldsymbol{\gamma}_{\boldsymbol{\beta}_0, j}\right|_1 = O(\sqrt{\bar{s}})$.
\end{lemma}
\begin{proof}
First, observe that
\[
\frac{\boldsymbol{\gamma}_{\boldsymbol{\beta}_0, j}^{\top} \boldsymbol{\Sigma}_{-j, \boldsymbol{\beta}_0, -j} \boldsymbol{\gamma}_{\boldsymbol{\beta}_0, j}}{\boldsymbol{\gamma}_{\boldsymbol{\beta}_0, j}^{\top} \boldsymbol{\gamma}_{\boldsymbol{\beta}_0, j}} 
\geq \lambda_{\min}(\boldsymbol{\Sigma}_{-j, \boldsymbol{\beta}_0, -j}) 
\geq \lambda_{\min}(\boldsymbol{\Sigma}_{\boldsymbol{\beta}_0}) 
\geq \gamma_{H},
\]
where the last inequality follows from Assumption \ref{aseign}. 
It then immediately follows that
\[
\boldsymbol{\gamma}_{\boldsymbol{\beta}_0, j}^{\top} \boldsymbol{\gamma}_{\boldsymbol{\beta}_0, j} \leq \frac{\boldsymbol{\gamma}_{\boldsymbol{\beta}_0, j}^{\top} \boldsymbol{\Sigma}_{-j, \boldsymbol{\beta}_0, -j} \boldsymbol{\gamma}_{\boldsymbol{\beta}_0, j}}{\gamma_{H}}.
\]

Next, note that for each observation $i\in[N]$,
\[
X_{i, \boldsymbol{\beta}_0, j} = \boldsymbol{X}_{i, \boldsymbol{\beta}_0, -j}^{\top} \boldsymbol{\gamma}_{\boldsymbol{\beta}_0, j} + \eta_{i, \boldsymbol{\beta}_0, j},
\]
where $X_{i, \boldsymbol{\beta}_0, j}$ denotes the $i$-th entry of $\boldsymbol{X}_{\boldsymbol{\beta}_0, j}$, $\boldsymbol{X}_{i, \boldsymbol{\beta}_0, -j}$ is the $i$-th row of $\boldsymbol{X}_{\boldsymbol{\beta}_0, -j}$, and $ \eta_{i, \boldsymbol{\beta}_0, j}$ is the $i$-th entry of $\boldsymbol{\eta}_{ \boldsymbol{\beta}_0, j}$

Since $\boldsymbol{X}_{i, \boldsymbol{\beta}_0, -j}$ is orthogonal to $\eta_{i, \boldsymbol{\beta}_0, j}$ in expectation for all $i \in [N]$, we have
\[
\mathbbm{E}\left( X_{i, \boldsymbol{\beta}_0, j}^2 \right) 
= \boldsymbol{\gamma}_{\boldsymbol{\beta}_0, j}^{\top} \boldsymbol{\Sigma}_{-j, \boldsymbol{\beta}_0, -j} \boldsymbol{\gamma}_{\boldsymbol{\beta}_0, j} + \mathbbm{E}\left( \eta_{i, \boldsymbol{\beta}_0, j}^2 \right),
\]
which implies
\[
\boldsymbol{\gamma}_{\boldsymbol{\beta}_0, j}^{\top} \boldsymbol{\Sigma}_{-j, \boldsymbol{\beta}_0, -j} \boldsymbol{\gamma}_{\boldsymbol{\beta}_0, j} \leq \mathbbm{E}\left( X_{i, \boldsymbol{\beta}_0, j}^2 \right).
\]
Furthermore, since $X_{i, \boldsymbol{\beta}_0, j}^2 = X_{i,j}^2 w_{\boldsymbol{\beta}_0, i}^2 \leq X_{i,j}^2$, and by Assumption \ref{asX}, we have
\[
\mathbbm{E}\left( X_{i, \boldsymbol{\beta}_0, j}^2 \right) \leq \mathbbm{E}\left( X_{i,j}^2 \right) \leq K_0,
\]
it follows that
\[
\boldsymbol{\gamma}_{\boldsymbol{\beta}_0, j}^{\top} \boldsymbol{\gamma}_{\boldsymbol{\beta}_0, j} = O(1).
\]

Finally, by definition, $\boldsymbol{\Theta}_j$ has at most $\bar{s}$ non-zero elements, and hence $\boldsymbol{\gamma}_{\boldsymbol{\beta}_0, j}$ has at most $\bar{s}$ non-zero entries. Therefore, we conclude that
\[
\left| \boldsymbol{\gamma}_{\boldsymbol{\beta}_0, j} \right|_1^2 \leq \bar{s} \, \boldsymbol{\gamma}_{\boldsymbol{\beta}_0, j}^{\top} \boldsymbol{\gamma}_{\boldsymbol{\beta}_0, j} \leq \bar{s} K_0.
\]
The proof is then complete.
\end{proof}

\subsection{Technical lemmas}
\begin{lemma}\label{lneedextra}
    Suppose the conditions of Theorem \ref{the} and Assumption \ref{asapprerror} are satisfied. We have 
    $$
    \left\|\boldsymbol{X}\left(\widehat{\boldsymbol{\beta}}-\boldsymbol{\beta}_0\right) \right\|_N^2 \lesssim \lambda^2 s_{\boldsymbol{\beta}_0}^2
    $$ 
    hold with probability $1-o(1)$.\footnote{\tcr{If we assume either bounded covariates or that $w_{\boldsymbol{\beta}_0,i}^2$, $i \in [N]$, has a uniform positive lower bound, then one can show that $\left\| \boldsymbol{X}\bigl(\widehat{\boldsymbol{\beta}} - \boldsymbol{\beta}_0\bigr) \right\|_N^2
= O_P\!\left( \lambda^2 s_{\boldsymbol{\beta}_0} \right),$
by following a proof strategy similar to that in \citet{van_de_geer_debias}.}}
\end{lemma}
\begin{proof}
Note that
$$
\begin{aligned}
& \left(\widehat{\boldsymbol{\beta}} - \boldsymbol{\beta}_0\right)^{\top}\left(\frac{1}{N}\sum_{i = 1}^{N}\boldsymbol{X}_i\boldsymbol{X}_i^{\top}\right)\left(\widehat{\boldsymbol{\beta}} - \boldsymbol{\beta}_0\right) \\
& \leq \Omega\left(\widehat{\boldsymbol{\beta}} - \boldsymbol{\beta}_0\right)\Omega_{*}\left(\left(\frac{1}{N}\sum_{i = 1}^{N}\boldsymbol{X}_i\boldsymbol{X}_i^{\top}\right)\left(\widehat{\boldsymbol{\beta}} - \boldsymbol{\beta}_0\right)\right)\\
& \leq G^* \left|\frac{1}{N}\sum_{i = 1}^{N}\boldsymbol{X}_i\boldsymbol{X}_i^{\top}\right|_{\infty}\Omega\left(\widehat{\boldsymbol{\beta}} - \boldsymbol{\beta}_0\right)^2\\
& = O_P\left( \lambda^2s_{\boldsymbol{\beta}_0}^2\right),
\end{aligned}
$$
where we use Lemma \ref{cor1} for the first two inequalities,  Lemma \ref{use sample2} and Assumption \ref{asX} to get $\left|\frac{1}{N}\sum_{i = 1}^{N}\boldsymbol{X}_i\boldsymbol{X}_i^{\top}\right|_{\infty} = O_P(1)$ and the rate of $\Omega\left(\widehat{\boldsymbol{\beta}} - \boldsymbol{\beta}_0\right)^2 =O_P\left( \lambda^2s_{\boldsymbol{\beta}_0}^2\right) $ comes from Theorem \ref{the} and Assumption \ref{asapprerror}.
\end{proof}

\begin{lemma}\label{lneed1}
 Let $\mathcal{J} \subseteq [p]$ be a fixed group. Let the conditions of Theorem \ref{the} and Assumptions \ref{asapprerror}, and \ref{aseign_2} hold. With probability approaching $1$, we have
\begin{equation}
\begin{aligned}
& \max_{j \in \mathcal{J}}\left\|\boldsymbol{X}\left(\widehat{\boldsymbol{\beta}}-\boldsymbol{\beta}_0\right) - \boldsymbol{E}\right\|_N^2\left|\boldsymbol{\eta}_{\boldsymbol{\beta}_0, j}\right|_{\infty}^2 \\
& \lesssim N^{2/q}(\log p)^{2/q}\left(\lambda^2 s^2_{\boldsymbol{\beta}_0} \right) + \frac{(\log p)^{2/q}\sum_{i}E_{i}^2}{N^{1-2/q}}\\
& \lesssim N^{2/q}(\log p)^{2/q} \left(\lambda^2 s^2_{\boldsymbol{\beta}_0} \right) + \frac{(\log p)^{2/q}}{N^{1+1/q}},
\end{aligned}
\end{equation}
and
$$
\left\|\boldsymbol{X}\left(\widehat{\boldsymbol{\beta}}-\boldsymbol{\beta}_0\right) - \boldsymbol{E}\right\|_N^2 \lesssim \lambda^2s^2_{\boldsymbol{\beta}_0}.
$$
\end{lemma}
\begin{proof}
Let $\eta_{i, \boldsymbol{\beta}_0,j}$ denote the $i$-th element of $\boldsymbol{\eta}_{\boldsymbol{\beta}_0,j}$. 
Then, we have  
\[
\begin{aligned}
P\left(\max_{i \in [N]}\max_{j \in \mathcal{J}}\left| \eta_{i, \boldsymbol{\beta}_0,j} \right| > t \right) 
& \leq \frac{\mathbbm{E}\left[\max_{i \in [N]}\max_{j \in \mathcal{J}}\left| \eta_{i, \boldsymbol{\beta}_0,j}\right|^q\right]}{t^q} \\
& \leq \frac{N |\mathcal{J}| \max_{i \in [N]}\max_{j \in \mathcal{J}}\mathbbm{E}\left[\left| \eta_{i, \boldsymbol{\beta}_0,j}\right|^q\right]}{t^q} \\
& \leq \frac{N |\mathcal{J}| C_{\eta}}{t^q},
\end{aligned}
\]
where the first inequality follows from Markov’s inequality, and the last inequality follows from Assumption~\ref{aseign_2}~(iii).

Setting $t = N^{1/q}(\log p)^{1/q}$, we obtain that, with probability at least $1 - \frac{C_{\eta} |\mathcal{J}|}{\log p}$,
\begin{equation}\label{nodewisepart1}
\max_{j \in \mathcal{J}}\left| \boldsymbol{\eta}_{\boldsymbol{\beta}_0,j} \right|_{\infty}^{2} 
\leq N^{2/q}(\log p)^{2/q}.
\end{equation}
Obviously, we have
\begin{equation}\label{lemmanodewise1}
\begin{aligned}
\left\|\boldsymbol{X}\left(\widehat{\boldsymbol{\beta}}-\boldsymbol{\beta}_0\right) - \boldsymbol{E}\right\|_N^2 
& \leq 2\left(\widehat{\boldsymbol{\beta}} - \boldsymbol{\beta}_0\right)^{\top} 
\left(\frac{1}{N} \sum_{i=1}^N \boldsymbol{X}_{i}\boldsymbol{X}_{i}^{\top}\right)
\left(\widehat{\boldsymbol{\beta}} - \boldsymbol{\beta}_0\right) 
+ 2\frac{1}{N}\sum_{i=1}^{N}E_{i}^2,
\end{aligned}
\end{equation}
which follows from the fact that, for any two vectors $\boldsymbol{a}$ and $\boldsymbol{b}$ of the same dimension, 
\[
\|\boldsymbol{a}+\boldsymbol{b}\|_N^2 \leq 2\|\boldsymbol{a}\|_N^2 + 2\|\boldsymbol{b}\|_N^2.
\]
For the first term on the right-hand side of~\eqref{lemmanodewise1}, Lemma~\ref{lneedextra} implies that, with probability $1 - o(1)$,
\[
\left\|\boldsymbol{X}\left(\widehat{\boldsymbol{\beta}}-\boldsymbol{\beta}_0\right)\right\|_N^2 
\lesssim \lambda^2 s^2_{\boldsymbol{\beta}_0}.
\]
By Assumption~\ref{aseign_2}~(i), we have  
\begin{equation}\label{lemmanodewise4}
\frac{\sum_{i}E_{i}^2}{N} = o_P\left(\frac{1}{N^{1+3/q}}\right).
\end{equation}
Hence, with probability $1-o(1)$, we have
\[
\begin{aligned}
& \max_{j \in \mathcal{J}}\left\|\boldsymbol{X}\left(\widehat{\boldsymbol{\beta}}-\boldsymbol{\beta}_0\right) - \boldsymbol{E}\right\|_N^2
  \left|\boldsymbol{\eta}_{\boldsymbol{\beta}_0, j}\right|_{\infty}^2 \\
& \quad \lesssim N^{2/q}(\log p)^{2/q} 
\left(\lambda^2 s^2_{\boldsymbol{\beta}_0}\right) 
      +  \frac{(\log p)^{2/q}\sum_{i}E_{i}^2}{N^{1-2/q}}\\
& \quad \lesssim N^{2/q}(\log p)^{2/q}\left(\lambda^2 s^2_{\boldsymbol{\beta}_0}\right) 
      + o_P\left(\frac{(\log p)^{2/q}}{N^{1+1/q}}\right),
\end{aligned}
\]
where the last inequality follows from Assumptions \ref{aseign_2}~(i).
Furthermore, since $\max_i |E_i| \leq \left|\boldsymbol{E}\right|_2 = o_P(1)$ by Assumption \ref{aseign_2} (i), we then define the event $\mathcal{E} = \{\max_i |E_i| \leq 1\}$ where $P\left(\mathcal{E}\right) \rightarrow 1$. Under the event $\mathcal{E}$, we have
\begin{equation}\label{e2e1}
\frac{\sum_{i}E_{i}^2}{N} 
\leq \frac{|\boldsymbol{E}|_1}{N}.
\end{equation}
It follows by  
\begin{equation}
\begin{aligned}
\left\|\boldsymbol{X}\left(\widehat{\boldsymbol{\beta}}-\boldsymbol{\beta}_0\right) - \boldsymbol{E}\right\|_N^2 
& \leq 2\left(\widehat{\boldsymbol{\beta}} - \boldsymbol{\beta}_0\right)^{\top} 
\left(\frac{1}{N} \sum_{i=1}^N \boldsymbol{X}_{i}\boldsymbol{X}_{i}^{\top}\right)
\left(\widehat{\boldsymbol{\beta}} - \boldsymbol{\beta}_0\right) 
+ 2\frac{1}{N}\sum_{i = 1}^{N}E_{i}^2\\
& \leq 2\left\|\boldsymbol{X}\left(\widehat{\boldsymbol{\beta}}-\boldsymbol{\beta}_0\right)\right\|_N^2 
+ 2\frac{|\boldsymbol{E}|_1}{N}\\
& = O_P\left( \lambda^2 s^2_{\boldsymbol{\beta}_0}\right),
\end{aligned}
\end{equation}
where the last equality is justified by Lemma~\ref{lneedextra}, the fact that $\lambda^2 s_{\boldsymbol{\beta}_0} \leq \lambda^2 s^2_{\boldsymbol{\beta}_0}$ since $s_{\boldsymbol{\beta}_0} \geq 1$, and $P\left(\mathcal{E}\right) \rightarrow 1$.
\end{proof}

\begin{lemma}\label{lneed2}
Let the conditions of Theorem \ref{the}, Assumptions \ref{asapprerror}, \ref{aseign_2}, and \ref{asrate} hold. With probability $1-o(1)$, we have
$$
\max_{i \in [N]} \left|\frac{w_{i,\boldsymbol{\beta}_0}^2}{w_{i,\widehat{\boldsymbol{\beta}}}^2}\right| \leq 4/3.
$$
\end{lemma}
\begin{proof}
Note that
\begin{align*}
\max_{i \in [N]} \frac{w_{i,\boldsymbol{\beta}_0}^2}{w_{i,\widehat{\boldsymbol{\beta}}}^2} 
&= \max_{i \in [N]} 
\frac{\exp(\boldsymbol{X}_i^\top \boldsymbol{\beta}_0 + E_i)}{(1+\exp(\boldsymbol{X}_i^\top \boldsymbol{\beta}_0 + E_i))^2} 
\times \frac{(1+\exp(\boldsymbol{X}_i^\top \widehat{\boldsymbol{\beta}}))^2}{\exp(\boldsymbol{X}_i^\top \widehat{\boldsymbol{\beta}})} \\
&= \max_{i \in [N]} 
\exp(\boldsymbol{X}_i^\top \boldsymbol{\beta}_0 + E_i - \boldsymbol{X}_i^\top \widehat{\boldsymbol{\beta}}) 
\frac{(1+\exp(\boldsymbol{X}_i^\top \widehat{\boldsymbol{\beta}}))^2}{(1+\exp(\boldsymbol{X}_i^\top \boldsymbol{\beta}_0 + E_i))^2}.
\end{align*}
Let 
\[
x_0 = \boldsymbol{X}_i^\top \boldsymbol{\beta}_0 + E_i, \quad x_1 = \boldsymbol{X}_i^\top \widehat{\boldsymbol{\beta}}.
\] 
Define
\[
g(x_0,x_1) = \exp(x_0 - x_1) \frac{(1+\exp(x_1))^2}{(1+\exp(x_0))^2}.
\] 
Taking the derivative of the logarithm, we have
\[
\left|\frac{d}{dx_1} \log g(x_0,x_1)\right| = \left|-1 + \frac{2 \exp(x_1)}{1+\exp(x_1)}\right| \leq 1.
\]
 By $\log g(x_0, x_0) = 0$ and the mean-value theorem, it follows that
\[
|\log g(x_0,x_1)| \le |x_1 - x_0|.
\]
Then, we have 
\[
g(x_0,x_1) = \exp(\log g(x_0,x_1)) \le \exp(|x_1 - x_0|),
\] 
and thus
\[
\max_{i \in [N]} \frac{w_{i,\boldsymbol{\beta}_0}^2}{w_{i,\widehat{\boldsymbol{\beta}}}^2} 
\le \exp\left(\max_{i \in [N]} | \boldsymbol{X}_i^\top \boldsymbol{\beta}_0 + E_i - \boldsymbol{X}_i^\top \widehat{\boldsymbol{\beta}} | \right).
\]

Next, we bound the difference:
\begin{align*}
\max_{i\in[N]} |\boldsymbol{X}_i^\top \widehat{\boldsymbol{\beta}} - \boldsymbol{X}_i^\top \boldsymbol{\beta}_0 - E_i|
&\le \max_i |\boldsymbol{X}_i^\top \widehat{\boldsymbol{\beta}} - \boldsymbol{X}_i^\top \boldsymbol{\beta}_0| + \max_i |E_i| \\
&\le G^* \max_i |\boldsymbol{X}_i|_\infty \, \Omega(\widehat{\boldsymbol{\beta}} - \boldsymbol{\beta}_0) + \max_i |E_i| \\
&= O_P\Big( 1/N^{1.5/q} + (Np \log p)^{1/q} \lambda s_{\boldsymbol{\beta}_0} \Big),
\end{align*}
where we use \eqref{maxmax} to get the rate of $\max_i |\boldsymbol{X}_i|_\infty$, 
and apply Theorem \ref{the}, Assumption \ref{asapprerror} and Assumption \ref{aseign_2} (i) which justifies $\max_i |E_i|  = o_P(N^{-1.5/q})$. Then by Assumption~\ref{asrate}, as $N,p \to \infty$, this maximum is $o_P(1)$, and hence smaller than the fixed constant $\log(4/3)$ with probability tending to $1$. Therefore,
\[
\max_{i \in [N]} \frac{w_{i,\boldsymbol{\beta}_0}^2}{w_{i,\widehat{\boldsymbol{\beta}}}^2} \le 4/3.
\] 
Since $w_{i,\boldsymbol{\beta}_0}^2/w_{i,\widehat{\boldsymbol{\beta}}}^2$ is always positive, the proof is complete.

\end{proof}

\begin{lemma}\label{lneed3}
Let $\mathcal{J} \subseteq [p]$ be a fixed group, $\boldsymbol{X}_{\boldsymbol{\beta}_0,S_j-\{j\}}^{\top} = \left(\boldsymbol{X}_{\boldsymbol{\beta}_0,k}^{\top}\right)_{k \neq j, k \in S_j} \in \mathbb{R}^{(|S_j - \{j\}|) \times N}$ where $\boldsymbol{X}_{\boldsymbol{\beta}_0,k}$ is the $k$-th column of $\boldsymbol{X}_{\boldsymbol{\beta}_0}$ and $\boldsymbol{X}_{\boldsymbol{\beta}_0,S_j^c-\{j\}}^{\top} = \left(\boldsymbol{X}_{\boldsymbol{\beta}_0,k}^{\top}\right)_{k \neq j, k \in S^c_j} \in \mathbb{R}^{ (|S^c_j - \{j\}|) \times N}$. Under the conditions of Theorem \ref{the} and Assumption \ref{aseign_2}, it holds with probability $1-o(1)$, that
\begin{equation}
    \begin{aligned}
        & \max_{j \in \mathcal{J}}\left|\frac{\boldsymbol{X}_{\boldsymbol{\beta}_0,S_j-\{j\}}^{\top}\boldsymbol{\eta}_{\boldsymbol{\beta}_0,j}}{N}\right|_{\infty} \leq B_1\frac{\sqrt{\log p}}{\sqrt{N}} + B_2\frac{p^{\frac{2}{q}} \log p}{N^{1-\frac{2}{q}}} + B_3\frac{p^{\frac{2}{q}} \sqrt{\log p}}{\sqrt{N}} \\
        & \max_{j \in \mathcal{J}} \left|\frac{\boldsymbol{X}_{\boldsymbol{\beta}_0,S_j^c-\{j\}}^{\top}\boldsymbol{\eta}_{\boldsymbol{\beta}_0,j}}{N}\right|_{\infty} \leq B_1^{\prime}\frac{\sqrt{\log p}}{\sqrt{N}} + B_2^{\prime}\frac{p^{\frac{2}{q}} \log p}{N^{1-\frac{2}{q}}} + B_3^{\prime}\frac{p^{\frac{2}{q}} \sqrt{\log p}}{\sqrt{N}},
    \end{aligned}
\end{equation}
where
$B_1, B_2, B_3$, $B_1^{\prime}, B_2^{\prime}, B_3^{\prime}$ are strictly positive constants.
\end{lemma}
\begin{proof}
Recall that 
\[
\boldsymbol{\eta}_{\boldsymbol{\beta}_0,j} = \boldsymbol{X}_{\boldsymbol{\beta}_0,j} - \boldsymbol{X}_{\boldsymbol{\beta}_0,-j} \boldsymbol{\gamma}_{\boldsymbol{\beta}_0, j},
\]
and that by the first order moment of \eqref{populationnodewise}, we have
\[
\mathbbm{E}\left[\boldsymbol{X}_{\boldsymbol{\beta}_0,-j}^{\top} \left( \boldsymbol{X}_{\boldsymbol{\beta}_0,j} - \boldsymbol{X}_{\boldsymbol{\beta}_0,-j} \boldsymbol{\gamma}_{\boldsymbol{\beta}_0, j} \right) \right] = \boldsymbol{0}.
\]
This immediately implies that
\[
\mathbbm{E}\left[\boldsymbol{X}_{\boldsymbol{\beta}_0,S_j-\{j\}}^{\top} \boldsymbol{\eta}_{\boldsymbol{\beta}_0,j}\right] = \boldsymbol{0} \quad \text{and} \quad
\mathbbm{E}\left[\boldsymbol{X}_{\boldsymbol{\beta}_0,S_j^c-\{j\}}^{\top} \boldsymbol{\eta}_{\boldsymbol{\beta}_0,j}\right] = \boldsymbol{0}.
\]

Moreover, for any $j \in \mathcal{J}$, we can write
\[
\frac{1}{N} \boldsymbol{X}_{\boldsymbol{\beta}_0,S_j-\{j\}}^{\top}\boldsymbol{\eta}_{\boldsymbol{\beta}_0,j} 
= \left( \frac{1}{N} \sum_{i=1}^{N} X_{\boldsymbol{\beta}_0,S_j-\{j\},[i,1]} \eta_{i,\boldsymbol{\beta}_0,j}, \ldots, \frac{1}{N} \sum_{i=1}^{N} X_{\boldsymbol{\beta}_0,S_j-\{j\},[i,|S_j-\{j\}|]} \eta_{i,\boldsymbol{\beta}_0,j} \right)^{\top},
\]
where $X_{\boldsymbol{\beta}_0,S_j-\{j\},[i,l]}$ denotes the $(i,l)$-th entry of $\boldsymbol{X}_{\boldsymbol{\beta}_0,S_j-\{j\}}$, and $\eta_{i,\boldsymbol{\beta}_0,j}$ is the $i$-th element of $\boldsymbol{\eta}_{\boldsymbol{\beta}_0,j}$.

To apply Theorem \ref{key}, we have to bound the following terms. First, we have 
$$
\begin{aligned}
& \max_{l \in [|S_j-\{j\}|]} \mathbbm{E}\left(\left|X_{\boldsymbol{\beta}_0,S_j-\{j\},[i,k_l]}\eta_{i,\boldsymbol{\beta}_0,j}\right|^2\right) \\
& \leq \max_{l \in [|S_j-\{j\}|]} \mathbbm{E}\left(\left|X_{\boldsymbol{\beta}_0,S_j-\{j\},[i,k_l]}\right|^4\right)^{\frac{1}{2}} \max_{l \in [|S_j-\{j\}|]} \mathbbm{E}\left(\left|\eta_{i,\boldsymbol{\beta}_0,j}\right|^4\right)^{\frac{1}{2}} \\
& \lesssim \sqrt{K_0C_{\eta}},
\end{aligned}
$$
where we use the inequality of Cauchy–Schwarz for the first inequality and Assumptions  \ref{aseign_2} (ii) and (iii) for the second. Furthermore, we have 
$$
\begin{aligned}
& \max_{l \in [|S_j-\{j\}|]} \mathbbm{E}\left(\left|X_{\boldsymbol{\beta}_0,S_j-\{j\},[i,k_l]}\eta_{i,\boldsymbol{\beta}_0,j}\right|^{\frac{q}{2}}\right)^{\frac{2}{q}} \\
& \leq \max_{l \in [|S_j-\{j\}|]} \mathbbm{E}\left(\left|\boldsymbol{X}_{\boldsymbol{\beta}_0,S_j-\{j\},[i,k_l]}\right|^q\right)^{\frac{1}{q}} \mathbbm{E}\left(\left|\eta_{i,\boldsymbol{\beta}_0,j}\right|^q\right)^{\frac{1}{q}} \\
& \lesssim (K_0C_{\eta})^{\frac{1}{q}}.
\end{aligned}
$$
Based on the argument in \eqref{holder}, we have
$$
\begin{aligned}
\mathbbm{E}\left(\max_{l \in [|S_j-\{j\}|]}\left|\boldsymbol{X}_{\boldsymbol{\beta}_0,S_j-\{j\},[i,k_l]}\eta_{i,\boldsymbol{\beta}_0,j}\right|^2\right)^{\frac{1}{2}} & \leq 
\mathbbm{E}\left(\max_{l \in [|S_j-\{j\}|]}\left|\boldsymbol{X}_{\boldsymbol{\beta}_0,S_j-\{j\},[i,k_l]}\eta_{i,\boldsymbol{\beta}_0,j}\right|^{\frac{q}{2}}\right)^{\frac{2}{q}}\\
& \leq \left((|S_j-\{j\}|)\max_{l \in [|S_j-\{j\}|]} \mathbbm{E}\left(\left|\boldsymbol{X}_{\boldsymbol{\beta}_0,S_j-\{j\},[i,k_l]}\eta_{i,\boldsymbol{\beta}_0,j}\right|^{\frac{q}{2}}\right)\right)^{\frac{2}{q}} \\
& \leq (|S_j-\{j\}|)^{\frac{2}{q}}(K_0C_{\eta})^{\frac{1}{q}}\\
& \leq p^{\frac{2}{q}}(K_0C_{\eta})^{\frac{1}{q}}.
\end{aligned}
$$
Hence, by Theorem \ref{key}, for every $j \in \mathcal{J}$, we have probability $1-o(1)$ such that
$$
\left|\frac{\boldsymbol{X}_{\boldsymbol{\beta}_0,S_j-\{j\}}^{\top}\boldsymbol{\eta}_{\boldsymbol{\beta}_0,j}}{N}\right|_{\infty} \leq B_1\frac{\sqrt{\log p}}{\sqrt{N}} + B_2\frac{p^{\frac{2}{q}} \log p}{N^{1-\frac{2}{q}}} + B_3\frac{p^{\frac{2}{q}} \sqrt{\log p}}{\sqrt{N}},
$$
where $B_1, B_2, B_3$ are constants. By a similar argument, we also have 
$$
\left|\frac{\boldsymbol{X}_{\boldsymbol{\beta}_0,S_j^c-\{j\}}^{\top}\boldsymbol{\eta}_{\boldsymbol{\beta}_0,j}}{N}\right|_{\infty} \leq B_1^{\prime}\frac{\sqrt{\log p}}{\sqrt{N}} + B_2^{\prime}\frac{p^{\frac{2}{q}} \log p}{N^{1-\frac{2}{q}}} + B_3^{\prime}\frac{p^{\frac{2}{q}} \sqrt{\log p}}{\sqrt{N}},
$$
where $B_1^{\prime}, B_2^{\prime}, B_3^{\prime}$ are constants. Since the size of $\mathcal{J}$ is fixed, the first statement is obtained. The second statement can be obtained by similar arguments.
\end{proof}
For a vector $\boldsymbol{\delta}_{j}=\left(\delta_{j1},\ldots,\delta_{jj}, \ldots, \delta_{jp}\right)^{\top} \in \mathbb{R}^{p}$, let $\boldsymbol{\delta}_{j,-j}\in \mathbb{R}^{p-1}$ be $\boldsymbol{\delta}_{j}$ with $j$-th element removed, let 
$$
\boldsymbol{\delta}_{j,S_j-\{j\}} = \left(\delta_{jk}\right)_{k \neq j, k \in S_j} \in \mathbb{R}^{|S_j-\{j\}|},
$$ and 
$$\boldsymbol{\delta}_{j,S_j^c-\{j\}} = \left(\delta_{jk}\right)_{k \neq j, k \in S^c_j} \in \mathbb{R}^{|S^c_j-\{j\}|}.
$$ 
For a matrix $\boldsymbol{\Sigma} \in \mathbb{R}^{p \times p}$, let $\boldsymbol{\Sigma}_{-j,-j} \in \mathbb{R}^{(p-1) \times (p-1)}$ be $\boldsymbol{\Sigma}$ with $j$-th row and $j$-th column removed.
The population level of compatibility condition is defined as, for each $j$
$$
\phi^2\left(L, S_j-\{j\}\right)=\min _{\boldsymbol{\delta}_{j,-j}}\left\{ \frac{\left|S_j-\{j\}\right|\left|\boldsymbol{\delta}_{j,-j}^{\top} \boldsymbol{\Sigma}_{\boldsymbol{\beta}_0,-j, -j} \boldsymbol{\delta}_{j,-j}\right|}{\left|\boldsymbol{\delta}_{j,S_j-\{j\}}\right|_1^2}: \quad \left|\boldsymbol{\delta}_{j,S_j^c-\{j\}}\right|_1 \leq L \left|\boldsymbol{\delta}_{j,S_j-\{j\}}\right|_1\right\}.
$$
The empirical version of the compatibility condition is, for each $j$,
$$
\hat{\phi}^2\left(L, S_j-\{j\}\right)=\min _{\boldsymbol{\delta}_{j,-j}}\left\{ \frac{\left|S_j-\{j\}\right|\left|\boldsymbol{\delta}_{j,-j}^{\top} \widehat{\boldsymbol{\Sigma}}_{\widehat{\boldsymbol{\beta}},-j, -j} \boldsymbol{\delta}_{j,-j}\right|}{\left|\boldsymbol{\delta}_{j,S_j-\{j\}}\right|_1^2}: \quad \left|\boldsymbol{\delta}_{j,S_j^c-\{j\}}\right|_1 \leq L \left|\boldsymbol{\delta}_{j,S_j-\{j\}}\right|_1\right\}.
$$
\begin{lemma}\label{lneed4}
Under the conditions of Theorem \ref{the} and Assumptions \ref{asapprerror}, \ref{aseign_2}, and \ref{asrate}, it holds with probability $ 1-o(1)$ that, for each $j \in [p]$, 
$$
\hat{\phi}^2\left(L, S_j-\{j\}\right) \geq \phi^2\left(L, S_j-\{j\}\right) / 2.
$$
\end{lemma}
\begin{proof}
First, we define 
$$
\widehat{\boldsymbol{\Sigma}}_{\boldsymbol{\beta}_0}:= \frac{1}{N} \sum_{i=1}^N \frac{\exp(\boldsymbol{X}_i^{\top}\boldsymbol{\beta}_0+E_i)}{ \big(1+\exp(\boldsymbol{X}_i^{\top}\boldsymbol{\beta}_0 +E_i)\big)^2 } \boldsymbol{X}_{i} \boldsymbol{X}_{i}^{\top}.
$$
Under the constraints $\left|\boldsymbol{\delta}_{j,S_j-\{j\}}\right|_1 = 1$ and $\left|\boldsymbol{\delta}_{j,S_j^c-\{j\}}\right|_1 \leq L$, it follows immediately that
\[
\left|\boldsymbol{\delta}_{j,-j}\right|_1 
= \left|\boldsymbol{\delta}_{j,S_j-\{j\}}\right|_1 + \left|\boldsymbol{\delta}_{j,S_j^c-\{j\}}\right|_1 
\leq (L+1)\left|\boldsymbol{\delta}_{j,S_j-\{j\}}\right|_1.
\]
Then, we have
\begin{equation}
\begin{aligned}
\left|\boldsymbol{\delta}_{j,-j}^{\top} \widehat{\boldsymbol{\Sigma}}_{\widehat{\boldsymbol{\beta}},-j,-j} \boldsymbol{\delta}_{j,-j}\right| 
& \geq 
\left| \boldsymbol{\delta}_{j,-j}^{\top} \boldsymbol{\Sigma}_{\boldsymbol{\beta}_0,-j,-j} \boldsymbol{\delta}_{j,-j} \right|
    - \left| \boldsymbol{\delta}_{j,-j}^{\top} 
        \left(\widehat{\boldsymbol{\Sigma}}_{\widehat{\boldsymbol{\beta}},-j,-j} - 
              \boldsymbol{\Sigma}_{\boldsymbol{\beta}_0,-j,-j}\right) 
        \boldsymbol{\delta}_{j,-j} \right| \\
& \geq 
\left| \boldsymbol{\delta}_{j,-j}^{\top} \boldsymbol{\Sigma}_{\boldsymbol{\beta}_0,-j,-j} \boldsymbol{\delta}_{j,-j} \right|
    - \left|\boldsymbol{\delta}_{j,-j}\right|_1^2 
      \left|\widehat{\boldsymbol{\Sigma}}_{\widehat{\boldsymbol{\beta}},-j,-j} - 
      \boldsymbol{\Sigma}_{\boldsymbol{\beta}_0,-j,-j}\right|_{\infty} \\
& \geq 
\left| \boldsymbol{\delta}_{j,-j}^{\top} \boldsymbol{\Sigma}_{\boldsymbol{\beta}_0,-j,-j} \boldsymbol{\delta}_{j,-j} \right|
    - (L+1)^2 \left|\boldsymbol{\delta}_{j,S_j-\{j\}}\right|_1^2
      \left|\widehat{\boldsymbol{\Sigma}}_{\widehat{\boldsymbol{\beta}},-j,-j} - 
      \boldsymbol{\Sigma}_{\boldsymbol{\beta}_0,-j,-j}\right|_{\infty},
\end{aligned}
\end{equation}
Dividing both sides by $\left|\boldsymbol{\delta}_{j,S_j-\{j\}}\right|_1^2$, we obtain
\begin{equation}\label{com1}
\frac{\left|\boldsymbol{\delta}_{j,-j}^{\top} 
\widehat{\boldsymbol{\Sigma}}_{\widehat{\boldsymbol{\beta}},-j,-j}
\boldsymbol{\delta}_{j,-j}\right|}
{\left|\boldsymbol{\delta}_{j,S_j-\{j\}}\right|_1^2} 
\geq 
\frac{\left|\boldsymbol{\delta}_{j,-j}^{\top} 
\boldsymbol{\Sigma}_{\boldsymbol{\beta}_0,-j,-j} 
\boldsymbol{\delta}_{j,-j}\right|}
{\left|\boldsymbol{\delta}_{j,S_j-\{j\}}\right|_1^2}
- (L+1)^2 
\left|\widehat{\boldsymbol{\Sigma}}_{\widehat{\boldsymbol{\beta}},-j,-j} - 
\boldsymbol{\Sigma}_{\boldsymbol{\beta}_0,-j,-j}\right|_{\infty}.
\end{equation}
Next, we bound the second term on the right-hand side of \eqref{com1}.  
Note that
\begin{equation}\label{com2}
\left|\widehat{\boldsymbol{\Sigma}}_{\widehat{\boldsymbol{\beta}},-j,-j} - 
\boldsymbol{\Sigma}_{\boldsymbol{\beta}_0,-j,-j}\right|_{\infty}
\leq 
\left|\widehat{\boldsymbol{\Sigma}}_{\widehat{\boldsymbol{\beta}}} - 
      \widehat{\boldsymbol{\Sigma}}_{\boldsymbol{\beta}_0}\right|_{\infty} 
+ \left|\widehat{\boldsymbol{\Sigma}}_{\boldsymbol{\beta}_0} - 
      \boldsymbol{\Sigma}_{\boldsymbol{\beta}_0}\right|_{\infty}.
\end{equation}
For $\left|\widehat{\boldsymbol{\Sigma}}_{\boldsymbol{\beta}_0} - \boldsymbol{\Sigma}_{\boldsymbol{\beta}_0}\right|_{\infty}$, 
Lemma~\ref{use sample1} implies that, with probability at least $1 - \frac{C_4}{\log p}$,
\begin{equation}\label{com3}
\left|\widehat{\boldsymbol{\Sigma}}_{\boldsymbol{\beta}_0} - \boldsymbol{\Sigma}_{\boldsymbol{\beta}_0}\right|_{\infty}
\leq 
B_1 \frac{\sqrt{\log p}}{\sqrt{N}}
+ B_2 \frac{p^{2/q} \log p}{N^{1 - 2/q}}
+ B_3 \frac{p^{2/q} \sqrt{\log p}}{\sqrt{N}},
\end{equation}
where $B_1$, $B_2$, $B_3$, and $C_4$ are constants. For the remaining term, $\left|\widehat{\boldsymbol{\Sigma}}_{\widehat{\boldsymbol{\beta}}} - \widehat{\boldsymbol{\Sigma}}_{\boldsymbol{\beta}_0}\right|_{\infty}$, 
by definition, we have
\begin{equation}\label{com4}
\begin{aligned}
\left|\widehat{\boldsymbol{\Sigma}}_{\widehat{\boldsymbol{\beta}}} - 
\widehat{\boldsymbol{\Sigma}}_{\boldsymbol{\beta}_0}\right|_{\infty}
& = 
\left|
\frac{1}{N}\sum_{i=1}^N
\frac{\exp(\boldsymbol{X}_i^{\top}\widehat{\boldsymbol{\beta}})}
{\big(1+\exp(\boldsymbol{X}_i^{\top}\widehat{\boldsymbol{\beta}})\big)^2}  
\boldsymbol{X}_{i}\boldsymbol{X}_{i}^{\top}
-\frac{\exp(\boldsymbol{X}_i^{\top}\boldsymbol{\beta}_0 + E_i)}
     {\big(1+\exp(\boldsymbol{X}_i^{\top}\boldsymbol{\beta}_0 + E_i)\big)^2}
\boldsymbol{X}_{i}\boldsymbol{X}_{i}^{\top}
\right|_{\infty} \\[4pt]
& = 
\max_{m,k \in [p]}\left|
\frac{1}{N}\sum_{i=1}^N
\frac{\exp(\theta_i)(1-\exp(\theta_i))}
     {\big(1+\exp(\theta_i)\big)^3}
\left(\boldsymbol{X}_i^{\top}\widehat{\boldsymbol{\beta}} 
- \boldsymbol{X}_i^{\top}\boldsymbol{\beta}_0 - E_i\right)
X_{i,m} X_{i,k}
\right|,
\end{aligned}
\end{equation}
where the first equality follows from the mean-value theorem and $\theta_i$ is some point lying between $\boldsymbol{X}_i^{\top}\widehat{\boldsymbol{\beta}}$ and $\boldsymbol{X}_i^{\top}\boldsymbol{\beta}_0 + E_i$. 
Next, notice that
\[
\begin{aligned}
& \max_{m,k}\left|
\frac{1}{N}\sum_{i=1}^N
\frac{\exp(\theta_i)(1-\exp(\theta_i))}
     {\big(1+\exp(\theta_i)\big)^3}
\left(\boldsymbol{X}_i^{\top}\widehat{\boldsymbol{\beta}} 
- \boldsymbol{X}_i^{\top}\boldsymbol{\beta}_0 - E_i\right)
X_{i,m} X_{i,k}
\right| \\
&\leq \sqrt{ \frac{ \sum_{i=1}^{N} \left( \boldsymbol{X}_i^{\top} \widehat{\boldsymbol{\beta}} - \boldsymbol{X}_i^{\top} \boldsymbol{\beta}_0 -E_i \right)^2 }{N} } 
      \max_{m,k} \sqrt{ \frac{ \sum_{i=1}^{N} X_{i,m}^2 X_{i,k}^2 }{N} } \\
& \leq \left\|\boldsymbol{X}_i^{\top}\widehat{\boldsymbol{\beta}} 
- \boldsymbol{X}_i^{\top}\boldsymbol{\beta}_0 - E_i\right\|_N \, O_P(1)\\
&= O_P\left( \sqrt{\lambda^2 s^2_{\boldsymbol{\beta}_0}}\right).
\end{aligned}
\]
The first inequality follows from the Cauchy--Schwarz inequality and the bound 
$$\left|\frac{\exp(\theta_i)(1-\exp(\theta_i))}{(1+\exp(\theta_i))^3}\right| \leq 1.$$ 
The second inequality uses an argument similar to Lemma \ref{use sample2} to control 
$$\max_{m,k} \sqrt{\frac{\sum_{i=1}^{N} X_{i,m}^2 X_{i,k}^2}{N}}.$$ The last equality follows from Lemma~\ref{lneed1}.
It follows that
\begin{equation}\label{com7}
    \left|\widehat{\boldsymbol{\Sigma}}_{\widehat{\boldsymbol{\beta}}} - \widehat{\boldsymbol{\Sigma}}_{\boldsymbol{\beta}_0}\right|_{\infty} = o_P(1).
\end{equation}
Combining \eqref{com3} and \eqref{com7}, we see that $\left|\widehat{\boldsymbol{\Sigma}}_{\widehat{\boldsymbol{\beta}},-j,-j} - \boldsymbol{\Sigma}_{\boldsymbol{\beta}_0,-j,-j}\right|_{\infty} = o_{P}(1)$. Hence, with probability $1-o(1)$, it holds that $(L+1)^2\left|\widehat{\boldsymbol{\Sigma}}_{\widehat{\boldsymbol{\beta}},-j,-j} - \boldsymbol{\Sigma}_{\boldsymbol{\beta}_0,-j,-j}\right|_{\infty} \leq \phi^2\left(L, S_j-\{j\}\right) / 2$ since $\phi^2(L,S_j-\{j\})$ is bounded away from $0$ by Lemma \ref{lneed5}, and it follows that
$$
\hat{\phi}^2\left(L, S_j-\{j\}\right) \geq \phi^2\left(L, S_j-\{j\}\right)-\phi^2\left(L, S_j-\{j\}\right)/2
$$
holds with probability $1-o(1)$.

\end{proof}

\begin{lemma}\label{lneed5}
    Suppose that Assumption \ref{aseign} holds. For every $j \in [p]$, it holds that $\phi^2(L,S_j-\{j\}) \geq \gamma_{\mathrm{H}}$.
\end{lemma}
\begin{proof}
By definition of $\boldsymbol{\Sigma}_{\boldsymbol{\beta}_0,-j,-j}$, we have that
\[
\gamma_{\mathrm{H}} \leq \min_{|u|_2 = 1, u \in \mathbb{R}^{p}} 
u^{\top} \mathbbm{E}\left[ \frac{\exp(\boldsymbol{X}_i^{\top}\boldsymbol{\beta}_0 + E_i)}{ \big(1+\exp(\boldsymbol{X}_i^{\top}\boldsymbol{\beta}_0 + E_i)\big)^2 } \boldsymbol{X}_i \boldsymbol{X}_i^{\top} \right] u 
\leq \min_{|v|_2 = 1, v \in \mathbb{R}^{p-1}} v^{\top} \boldsymbol{\Sigma}_{\boldsymbol{\beta}_0,-j,-j} v.
\]
To see the second inequality, note that we can set the $j$-th element of $u$ to $0$ since $\boldsymbol{\Sigma}_{\boldsymbol{\beta}_0}$ is symmetric. 
Next, by the inequality of Cauchy-Schwarz, we observe that
\[
\left|S_j-\{j\}\right| \, \boldsymbol{\delta}_{j,-j}^{\top} \boldsymbol{\delta}_{j,-j} \geq \left|S_j-\{j\}\right| \, \boldsymbol{\delta}_{j,S_j-\{j\}}^{\top} \boldsymbol{\delta}_{j,S_j-\{j\}} \geq \left|\boldsymbol{\delta}_{j,S_j-\{j\}}\right|_1^2.
\]
Hence, for every $j \in [p]$, we have
\[
\frac{\left|S_j-\{j\}\right| \, \left|\boldsymbol{\delta}_{j,-j}^{\top} \boldsymbol{\Sigma}_{\boldsymbol{\beta}_0,-j,-j} \boldsymbol{\delta}_{j,-j}\right|}{\left|\boldsymbol{\delta}_{j,S_j-\{j\}}\right|_1^2} 
= \left| \frac{\sqrt{|S_j-\{j\}|}}{|\boldsymbol{\delta}_{j,S_j-\{j\}}|_1} \, \boldsymbol{\delta}_{j,-j}^{\top} \boldsymbol{\Sigma}_{\boldsymbol{\beta}_0,-j,-j} \boldsymbol{\delta}_{j,-j} \, \frac{\sqrt{|S_j-\{j\}|}}{|\boldsymbol{\delta}_{j,S_j-\{j\}}|_1} \right|
\geq \gamma_{\mathrm{H}},
\]
where the last inequality follows because 
\[
\frac{|S_j-\{j\}| \, \boldsymbol{\delta}_{j,-j}^{\top} \boldsymbol{\delta}_{j,-j}}{|\boldsymbol{\delta}_{j,S_j-\{j\}}|_1^2} \geq 1.
\]
Finally, by the definition of $\phi^2(L,S_j-\{j\})$, this completes the proof.
\end{proof}

\begin{lemma}\label{weightbound}
     Under the conditions of Theorem \ref{the} and Assumptions \ref{asapprerror}, \ref{aseign_2}, and \ref{asrate}, we have 
$$
\frac{1}{N} \sum_{i=1}^N\left|\frac{w_{i, \widehat{\boldsymbol{\beta}}}^2-w_{i, \boldsymbol{\beta}_0}^2}{w_{i, \boldsymbol{\beta}_0}^2}\right|^2 \lesssim \lambda^2 s^2_{\boldsymbol{\beta}_0},
$$
and 
   $$
\frac{1}{N} \sum_{i=1}^N\left|\frac{w_{i, \widehat{\boldsymbol{\beta}}}^2-w_{i, \boldsymbol{\beta}_0}^2}{w_{i, \widehat{\boldsymbol{\beta}}}^2}\right|^2 \lesssim \lambda^2 s^2_{\boldsymbol{\beta}_0},
$$
hold with probability $1-o(1)$. Furthermore, we have 
\[
\frac{1}{N}\sum_{i=1}^N \left|
\frac{w_{i,\widehat{\boldsymbol{\beta}}}^2 - w_{i,\boldsymbol{\beta}_0}^2}{w_{i,\boldsymbol{\beta}_0}^2}
\right| \leq \sqrt{\frac{1}{N}\sum_{i=1}^N \left|
\frac{w_{i,\widehat{\boldsymbol{\beta}}}^2 - w_{i,\boldsymbol{\beta}_0}^2}{w_{i,\boldsymbol{\beta}_0}^2}
\right|^2} \lesssim \lambda s_{\boldsymbol{\beta}_0},
\]
and
\[
\frac{1}{N} \sum_{i=1}^N\left|\frac{w_{i, \widehat{\boldsymbol{\beta}}}^2-w_{i, \boldsymbol{\beta}_0}^2}{w_{i, \widehat{\boldsymbol{\beta}}}^2}\right| \leq \sqrt{\frac{1}{N} \sum_{i=1}^N\left|\frac{w_{i, \widehat{\boldsymbol{\beta}}}^2-w_{i, \boldsymbol{\beta}_0}^2}{w_{i, \widehat{\boldsymbol{\beta}}}^2}\right|^2} \lesssim \lambda s_{\boldsymbol{\beta}_0}.
\]
also hold with probability $1-o(1)$.
\end{lemma}
\begin{proof}
Note that
\[
\begin{aligned}
\left|
\frac{w_{i,\widehat{\boldsymbol{\beta}}}^2 - w_{i,\boldsymbol{\beta}_0}^2}{w_{i,\boldsymbol{\beta}_0}^2}
\right|
&=
\left|
\exp\!\big(\boldsymbol{X}_i^{\top}\widehat{\boldsymbol{\beta}} - \boldsymbol{X}_i^{\top}\boldsymbol{\beta}_0 - E_i\big)
\frac{\big(1+\exp(\boldsymbol{X}_i^{\top}\boldsymbol{\beta}_0 + E_i)\big)^2}{\big(1+\exp(\boldsymbol{X}_i^{\top}\widehat{\boldsymbol{\beta}})\big)^2} - 1
\right|.
\end{aligned}
\]
To simplify the analysis, define, for $x_0, x_1 \in \mathbb{R}$,
\[
g(x_0,x_1) := \exp(x_1 - x_0)\frac{(1+\exp(x_0))^2}{(1+\exp(x_1))^2}.
\]
Note that $g(x,x) = 1$ for all $x$. Fix $x_0$ and write $x_1 = x_0 + t$. Then
\[
\log g(x_0, x_0+t) = t + 2\log(1+\exp(x_0)) - 2\log(1+\exp(x_0+t)).
\]
Taking the derivative with respect to $t$ gives
\[
\frac{d}{dt}\log g(x_0, x_0+t) = 1 - 2\frac{\exp(x_0+t)}{1+\exp(x_0+t)} = \frac{1 - \exp(x_0+t)}{1 + \exp(x_0+t)},
\]
from which it follows that
\[
\left|\frac{d}{dt}\log g(x_0,x_0+t)\right| \le 1.
\]
By $\log g(x_0, x_0) = 0$ and the mean-value theorem, for any $x_1$,
\[
|\log g(x_0,x_1)| \le |x_1 - x_0|.
\]
Using the elementary inequality $|\exp(u) - 1| \le \exp(|u|) - 1$ for any $u \in \mathbb{R}$ and setting $u = \log g(x_0,x_1)$, we obtain
\[
|g(x_0,x_1) - 1| \le \exp(|x_1 - x_0|) - 1.
\]
Substituting $x_1 = \boldsymbol{X}_i^{\top}\widehat{\boldsymbol{\beta}}$ and $x_0 = \boldsymbol{X}_i^{\top}\boldsymbol{\beta}_0 + E_i$, we have
\[
\left|
\frac{w_{i,\widehat{\boldsymbol{\beta}}}^2 - w_{i,\boldsymbol{\beta}_0}^2}{w_{i,\boldsymbol{\beta}_0}^2}
\right|
\le \exp\Big(\big|\boldsymbol{X}_i^{\top}\widehat{\boldsymbol{\beta}} - \boldsymbol{X}_i^{\top}\boldsymbol{\beta}_0 - E_i\big|\Big) - 1.
\]
Averaging over $i$ and squaring both sides gives
\[
\frac{1}{N}\sum_{i=1}^N \left|
\frac{w_{i,\widehat{\boldsymbol{\beta}}}^2 - w_{i,\boldsymbol{\beta}_0}^2}{w_{i,\boldsymbol{\beta}_0}^2}
\right|^2
\le \frac{1}{N}\sum_{i=1}^N \left(\exp\Big(\big|\boldsymbol{X}_i^{\top}\widehat{\boldsymbol{\beta}} - \boldsymbol{X}_i^{\top}\boldsymbol{\beta}_0 - E_i\big|\Big) - 1\right)^2.
\]

Next, observe that
\[
\begin{aligned}
\max_{i\in[N]} \big|\boldsymbol{X}_i^{\top}\widehat{\boldsymbol{\beta}} - \boldsymbol{X}_i^{\top}\boldsymbol{\beta}_0 - E_i\big|
& \le \max_i \big|\boldsymbol{X}_i^{\top}\widehat{\boldsymbol{\beta}} - \boldsymbol{X}_i^{\top}\boldsymbol{\beta}_0\big| + \max_i |E_i| \\
& \le G^* \max_i |\boldsymbol{X}_i|_\infty \, \Omega(\widehat{\boldsymbol{\beta}} - \boldsymbol{\beta}_0) + \max_i |E_i| \\
& = O_P\Big( 1/N^{1.5/q} + (Np \log p)^{1/q} \lambda s_{\boldsymbol{\beta}_0} \Big),
\end{aligned}
\]
where we use \eqref{maxmax} to get the rate of $\max_i |\boldsymbol{X}_i|_\infty$, 
Theorem \ref{the}, Assumption \ref{asapprerror} and Assumption \ref{aseign_2} (i) which justifies $\max_i E_i  = o_P(N^{-1.5/q})$.
Then by Assumption~\ref{asrate}, as $N, p \to \infty$, this maximum is $o_P(1)$, and hence smaller than the fixed constant $1$ with probability tending to $1$. Let
\[
\mathcal{E} = \Big\{ \max_{i \in [N]} |\boldsymbol{X}_i^{\top}\widehat{\boldsymbol{\beta}} - \boldsymbol{X}_i^{\top}\boldsymbol{\beta}_0 - E_i| \le 1 \Big\}.
\]
On the event $\mathcal{E}$, the function $\exp(x)-1$ can be linearized, giving
\[
\exp\big(|\boldsymbol{X}_i^{\top}\widehat{\boldsymbol{\beta}} - \boldsymbol{X}_i^{\top}\boldsymbol{\beta}_0 - E_i|\big) - 1 \le e \, |\boldsymbol{X}_i^{\top}\widehat{\boldsymbol{\beta}} - \boldsymbol{X}_i^{\top}\boldsymbol{\beta}_0 - E_i|.
\]
Consequently,
\[
\begin{aligned}
\frac{1}{N}\sum_{i=1}^N \left|
\frac{w_{i,\widehat{\boldsymbol{\beta}}}^2 - w_{i,\boldsymbol{\beta}_0}^2}{w_{i,\boldsymbol{\beta}_0}^2}
\right|^2
& \lesssim \frac{1}{N}\sum_{i=1}^N \big|\boldsymbol{X}_i^{\top}\widehat{\boldsymbol{\beta}} - \boldsymbol{X}_i^{\top}\boldsymbol{\beta}_0 - E_i\big|^2 \\
& \lesssim \frac{1}{N}\sum_{i=1}^N \big|\boldsymbol{X}_i^{\top}(\widehat{\boldsymbol{\beta}} - \boldsymbol{\beta}_0)\big|^2 + \frac{1}{N}\sum_{i=1}^N |E_i|^2 \\
& = O_P\left( \lambda^2 s_{\boldsymbol{\beta}_0}^2\right),
\end{aligned}
\]
where the last step follows from Lemma~\ref{lneedextra}, ~\eqref{e2e1}, and given that $s_{\boldsymbol{\beta}_0} \geq 1$.

Since $P(\mathcal{E}) \to 1$, we conclude that
\[
\frac{1}{N}\sum_{i=1}^N \left|
\frac{w_{i,\widehat{\boldsymbol{\beta}}}^2 - w_{i,\boldsymbol{\beta}_0}^2}{w_{i,\boldsymbol{\beta}_0}^2}
\right|^2 \lesssim \lambda^2 s_{\boldsymbol{\beta}_0}^2
\]
holds with probability tending to 1. By a similar argument, we have 
\[
\frac{1}{N} \sum_{i=1}^N\left|\frac{w_{i, \widehat{\boldsymbol{\beta}}}^2-w_{i, \boldsymbol{\beta}_0}^2}{w_{i, \widehat{\boldsymbol{\beta}}}^2}\right|^2  \lesssim \lambda^2 s^2_{\boldsymbol{\beta}_0}
\]
holds with probability going to $1$. By the Cauchy–Schwarz inequality, we also have the following hold with probability going to $1$
\[
\frac{1}{N}\sum_{i=1}^N \left|
\frac{w_{i,\widehat{\boldsymbol{\beta}}}^2 - w_{i,\boldsymbol{\beta}_0}^2}{w_{i,\boldsymbol{\beta}_0}^2}
\right| \leq \sqrt{\frac{1}{N}\sum_{i=1}^N \left|
\frac{w_{i,\widehat{\boldsymbol{\beta}}}^2 - w_{i,\boldsymbol{\beta}_0}^2}{w_{i,\boldsymbol{\beta}_0}^2}
\right|^2} \lesssim \lambda s_{\boldsymbol{\beta}_0},
\]
and
\[
\frac{1}{N} \sum_{i=1}^N\left|\frac{w_{i, \widehat{\boldsymbol{\beta}}}^2-w_{i, \boldsymbol{\beta}_0}^2}{w_{i, \widehat{\boldsymbol{\beta}}}^2}\right| \leq \sqrt{\frac{1}{N} \sum_{i=1}^N\left|\frac{w_{i, \widehat{\boldsymbol{\beta}}}^2-w_{i, \boldsymbol{\beta}_0}^2}{w_{i, \widehat{\boldsymbol{\beta}}}^2}\right|^2} \lesssim \lambda s_{\boldsymbol{\beta}_0}.
\]
The proof is then finished.
\end{proof}
}

\section{Data pre-processing}\label{data process appendix}
Since the raw empirical dataset contains many missing values across different financial variables, we propose Algorithm \ref{algorithm 1} to extract a complete sub-dataset where all firms have survived for at least $s$ years.

\begin{algorithm}
    \caption{Data Processing Algorithm}
    \label{algorithm 1}
  \begin{algorithmic}[1]
    \REQUIRE Raw dataset, $s$, two initial values $c_1 = 25$ and $c_2 = 25$, step $l = 50$
    \STATE In the raw dataset, we select firms that satisfy $\mathbbm{1}\{\widetilde{T} \geq s\} = 1$ to form a new dataset and define the dimension of the new dataset as $(N,p)$ 
    \STATE For each firm $i$, calculate the number of missing values in variables, say $M_{1,i}$
    \STATE For each variable $k$ and its lags, calculate the minimum number of missing values across firms for each lagged covariate, say $M_{2,k}$
    \FOR{$a \in (c_1, c_1 + 1 \times l, c_1 + 2 \times l, \ldots,  N)$}
     \FOR{$b \in (c_2, c_2 + 1 \times l, c_2 + 2 \times l, \ldots,  p)$}
      \STATE Delete all firms $i$ with $M_{1,i} \geq b$
      \STATE Find all variables $k$ with $M_{2,k} \geq a$, then delete these variables and their lags
      \STATE Delete firms that still have missing values in variables and then calculate the dimensions of this sub-dataset, say $N_{a,b}$ and $p_{a,b}$
      \STATE if $N_{a,b} / N \geq 0.5$ and $p_{a,b} / p \geq 0.5$, calculate the number of uncensored firms $C_{a,b}$ which satisfy $\mathbbm{1}\{T_i \leq C_i\} = 1$.
    \ENDFOR
\ENDFOR
\STATE We finally picked up the sub-dataset, which has the most uncensored firms.
  \end{algorithmic}
\end{algorithm}
Algorithm \ref{algorithm 1} is designed to balance the number of firms, variables, and uncensored firms in the sub-dataset. In the raw dataset, removing firms with more than one missing variable results in no firms remaining. A similar issue arises if we remove all variables with missing values, leaving only a limited number of variables. To maximize the number of uncensored firms without considering the dimensionality of the processed dataset, we still find that only a few variables remain. Steps $6$ and $7$ of Algorithm \ref{algorithm 1} address this by removing firms and variables with excessive missing values. Step $9$ further optimizes the balance between the dimensionality of the dataset and the number of uncensored firms. Table \ref{data process} presents the dimensions of a sub-dataset with $s = 6$ years, selected through the following methods:
\begin{itemize}
    \item Method $1$: Delete all firms that have missing values.
    \item Method $2$: Delete all variables that have missing values.
    \item Method $3$: Using Algorithm \ref{algorithm 1}, without considering the dimension of processed dataset in Step $9$.
    \item Method $4$: Using Algorithm \ref{algorithm 1}.
\end{itemize}
It is evident that Methods $1$, $2$, and $3$ remove an excessive number of observations or covariates from the raw dataset.
\begin{table}[H]
  \centering
  \caption{Dimensions of sub-dataset with $s = 6$ years extracted by different methods.}
  \label{data process}
  \medskip
    \begin{tabular}{ccccc}
    \toprule
    Method & \multicolumn{1}{l}{Method 1} & \multicolumn{1}{l}{Method 2} & \multicolumn{1}{l}{Method 3} & \multicolumn{1}{l}{Method 4} \\
\cmidrule{2-5}    Number of firms $N$     & 0     & 1403 & 951 & 901 \\
   Number of variables $K$ (without lags)     & 57   & 0     & 3     & 32 \\
    \bottomrule
    \end{tabular}%
\end{table}%

\section{Additional simulations and empirical results}\label{additional} 

Tables \ref{recovery1}, \ref{recovery2}, \ref{recovery3}, \ref{recovery4}, and \ref{recovery5} present the average mean integrated squared error of the true weight coefficients $\boldsymbol{\theta}_{0,1+j}, j \in [d]$ and $\boldsymbol{\theta}_{0,1+d+j}, j \in [d]$ across different simulation scenarios.

\begin{table}[htbp]
\color{black}
  \centering
  \caption{Shape of weights estimation accuracy of the three different methods: LASSO-U, LASSO-M, sg-LASSO-M. Entries are the average mean squared error. $s = 6$ and $t = \{t_1 = 10\%, t_2 = 30\%, t_3 = 50\%\}$ percentile of the set $\{T_i: T_i \text{ is uncensored }, i \in [N]\}$.}
  \medskip
  \scalebox{0.75}{
    \begin{tabular}{cccccccc}
    \toprule
          & \multicolumn{7}{c}{Scenario \ref{s1}} \\
\cmidrule{2-8}
          & \multicolumn{3}{c}{$N = 800$} &       & \multicolumn{3}{c}{$N = 1200$} \\
\cmidrule{2-8}
          & \multicolumn{7}{c}{$t = t_1$} \\
\cmidrule{2-8}
          & LASSO-U & LASSO-M & sg-LASSO-M &       & LASSO-U & LASSO-M & sg-LASSO-M \\
\cmidrule{2-8}
$(1+log(t-s))\boldsymbol{\operatorname{Beta}}(1,3)$
& 0.964 & 0.762 & 0.738 &       & 0.920 & 0.658 & 0.621 \\
$(-1+log(t-s))\boldsymbol{\operatorname{Beta}}(2,3)$
& 1.972 & 1.503 & 1.533 &       & 1.928 & 1.379 & 1.384 \\
\cmidrule{2-8}
          & \multicolumn{7}{c}{$t = t_2$} \\
\cmidrule{2-8}
          & LASSO-U & LASSO-M & sg-LASSO-M &       & LASSO-U & LASSO-M & sg-LASSO-M \\
\cmidrule{2-8}
$(1+log(t-s))\boldsymbol{\operatorname{Beta}}(1,3)$
& 1.228 & 0.954 & 0.904 &       & 1.155 & 0.848 & 0.823 \\
$(-1+log(t-s))\boldsymbol{\operatorname{Beta}}(2,3)$
& 1.618 & 1.241 & 1.249 &       & 1.561 & 1.170 & 1.176 \\
\cmidrule{2-8}
          & \multicolumn{7}{c}{$t = t_3$} \\
\cmidrule{2-8}
          & LASSO-U & LASSO-M & sg-LASSO-M &       & LASSO-U & LASSO-M & sg-LASSO-M \\
\cmidrule{2-8}
$(1+log(t-s))\boldsymbol{\operatorname{Beta}}(1,3)$
& 1.455 & 1.171 & 1.120 &       & 1.393 & 1.049 & 0.992 \\
$(-1+log(t-s))\boldsymbol{\operatorname{Beta}}(2,3)$
& 1.420 & 1.143 & 1.152 &       & 1.392 & 1.078 & 1.074 \\
    \bottomrule
    \end{tabular}
    }
  \label{recovery1}
\end{table}

\begin{table}[htbp]
\color{black}
  \centering
  \caption{Shape of weights estimation accuracy of the three different methods: LASSO-U, LASSO-M, sg-LASSO-M. Entries are the average mean squared error. $s = 6$ and $t = \{t_1 = 10\%, t_2 = 30\%, t_3 = 50\%\}$ percentile of the set $\{T_i: T_i \text{ is uncensored }, i \in [N]\}$.}
  \medskip
  \scalebox{0.75}{
    \begin{tabular}{cccccccc}
    \toprule
          & \multicolumn{7}{c}{Scenario \ref{s2}} \\
\cmidrule{2-8}          & \multicolumn{3}{c}{$N = 800$} &       & \multicolumn{3}{c}{$N = 1200$} \\
\cmidrule{2-8}          & \multicolumn{7}{c}{$t = t_1$} \\
\cmidrule{2-8}          & LASSO-U & LASSO-M & sg-LASSO-M &       & LASSO-U & LASSO-M & sg-LASSO-M \\
\cmidrule{2-8}
$(1+log(t-s))\boldsymbol{\operatorname{Beta}}(1,3)$ & 0.369 & 0.314 & 0.309 &       & 0.349  & 0.304  & 0.289  \\
$(-1+log(t-s))\boldsymbol{\operatorname{Beta}}(2,3)$ & 2.918& 2.749 & 2.768 &       & 2.819 & 2.674 & 2.654 \\
\cmidrule{2-8}          & \multicolumn{7}{c}{$t = t_2$} \\
\cmidrule{2-8}          & LASSO-U & LASSO-M & sg-LASSO-M &       & LASSO-U & LASSO-M & sg-LASSO-M \\
\cmidrule{2-8}
$(1+log(t-s))\boldsymbol{\operatorname{Beta}}(1,3)$ &0.676 & 0.64 & 0.619 &       & 0.638 & 0.609 & 0.582 \\
$(-1+log(t-s))\boldsymbol{\operatorname{Beta}}(2,3)$ & 2.119 & 2.035 & 2.029 &       & 2.037 & 1.982 & 1.954 \\
\cmidrule{2-8}          & \multicolumn{7}{c}{$t = t_3$} \\
\cmidrule{2-8}          & LASSO-U & LASSO-M & sg-LASSO-M &       & LASSO-U & LASSO-M & sg-LASSO-M \\
\cmidrule{2-8}
$(1+log(t-s))\boldsymbol{\operatorname{Beta}}(1,3)$ & 1.027 & 0.998 & 0.977 &       & 0.989 & 0.957 & 0.941 \\
$(-1+log(t-s))\boldsymbol{\operatorname{Beta}}(2,3)$  & 1.665 & 1.603 & 1.612 &       & 1.607 & 1.555 & 1.565\\
\bottomrule
\end{tabular}%
  }
  \label{recovery2}%
\end{table}%

\begin{table}[htbp]
\color{black}
  \centering
  \caption{Shape of weights estimation accuracy of the three different methods: LASSO-U, LASSO-M, sg-LASSO-M. Entries are the average mean squared error. $s = 6$ and $t = \{t_1 = 10\%, t_2 = 30\%, t_3 = 50\%\}$ percentile of the set $\{T_i: T_i \text{ is uncensored }, i \in [N]\}$.}
  \medskip
  \scalebox{0.75}{
    \begin{tabular}{cccccccc}
    \toprule
          & \multicolumn{7}{c}{Scenario \ref{s3}} \\
\cmidrule{2-8}          & \multicolumn{3}{c}{$N = 800$} &       & \multicolumn{3}{c}{$N = 1200$} \\
\cmidrule{2-8}          & \multicolumn{7}{c}{$t = t_1$} \\
\cmidrule{2-8}          & LASSO-U & LASSO-M & sg-LASSO-M &       & LASSO-U & LASSO-M & sg-LASSO-M \\
\cmidrule{2-8}
$(1+log(t-s))\boldsymbol{\operatorname{Beta}}(1,3)$
& 0.829 & 0.775 & 0.771 &       & 0.815 & 0.740 & 0.730 \\
$(-1+log(t-s))\boldsymbol{\operatorname{Beta}}(2,3)$
& 2.240 & 2.154 & 2.166 &       & 2.242 & 2.094 & 2.104 \\
\cmidrule{2-8}          & \multicolumn{7}{c}{$t = t_2$} \\
\cmidrule{2-8}          & LASSO-U & LASSO-M & sg-LASSO-M &       & LASSO-U & LASSO-M & sg-LASSO-M \\
\cmidrule{2-8}
$(1+log(t-s))\boldsymbol{\operatorname{Beta}}(1,3)$
& 1.192 & 1.080 & 1.073 &       & 1.183 & 1.051 & 1.043 \\
$(-1+log(t-s))\boldsymbol{\operatorname{Beta}}(2,3)$
& 1.752 & 1.605 & 1.618 &       & 1.748 & 1.564 & 1.582 \\
\cmidrule{2-8}          & \multicolumn{7}{c}{$t = t_3$} \\
\cmidrule{2-8}          & LASSO-U & LASSO-M & sg-LASSO-M &       & LASSO-U & LASSO-M & sg-LASSO-M \\
\cmidrule{2-8}
$(1+log(t-s))\boldsymbol{\operatorname{Beta}}(1,3)$
& 1.492 & 1.378 & 1.378 &       & 1.483 & 1.356 & 1.348 \\
$(-1+log(t-s))\boldsymbol{\operatorname{Beta}}(2,3)$
& 1.463 & 1.357 & 1.368 &       & 1.463 & 1.338 & 1.351 \\
    \bottomrule
    \end{tabular}%
    }
  \label{recovery3}%
\end{table}%

\begin{table}[htbp]
\color{black}
  \centering
  \caption{Shape of weights estimation accuracy of the three different methods: LASSO-U, LASSO-M, sg-LASSO-M. Entries are the average mean squared error. $s = 6$ and $t = \{t_1 = 10\%, t_2 = 30\%, t_3 = 50\%\}$ percentile of the set $\{T_i: T_i \text{ is uncensored }, i \in [N]\}$.}
  \medskip
  \scalebox{0.75}{
    \begin{tabular}{cccccccc}
    \toprule
          & \multicolumn{7}{c}{Scenario \ref{s4}} \\
\cmidrule{2-8}          & \multicolumn{3}{c}{$N = 800$} &       & \multicolumn{3}{c}{$N = 1200$} \\
\cmidrule{2-8}          & \multicolumn{7}{c}{$t = t_1$} \\
\cmidrule{2-8}          & LASSO-U & LASSO-M & sg-LASSO-M &       & LASSO-U & LASSO-M & sg-LASSO-M \\
\cmidrule{2-8}
$(1+log(t-s))\boldsymbol{\operatorname{Beta}}(1,3)$
& 0.370 & 0.363 & 0.360 &       & 0.367 & 0.360 & 0.356 \\
$(-1+log(t-s))\boldsymbol{\operatorname{Beta}}(2,3)$
& 3.085 & 3.057 & 3.059 &       & 3.070 & 3.037 & 3.032 \\
\cmidrule{2-8}          & \multicolumn{7}{c}{$t = t_2$} \\
\cmidrule{2-8}          & LASSO-U & LASSO-M & sg-LASSO-M &       & LASSO-U & LASSO-M & sg-LASSO-M \\
\cmidrule{2-8}
$(1+log(t-s))\boldsymbol{\operatorname{Beta}}(1,3)$
& 0.758 & 0.749 & 0.747 &       & 0.755 & 0.746 & 0.740 \\
$(-1+log(t-s))\boldsymbol{\operatorname{Beta}}(2,3)$
& 2.233 & 2.214 & 2.218 &       & 2.220 & 2.202 & 2.201 \\
\cmidrule{2-8}          & \multicolumn{7}{c}{$t = t_3$} \\
\cmidrule{2-8}          & LASSO-U & LASSO-M & sg-LASSO-M &       & LASSO-U & LASSO-M & sg-LASSO-M \\
\cmidrule{2-8}
$(1+log(t-s))\boldsymbol{\operatorname{Beta}}(1,3)$
& 1.179 & 1.170 & 1.166 &       & 1.173 & 1.161 & 1.161 \\
$(-1+log(t-s))\boldsymbol{\operatorname{Beta}}(2,3)$
& 1.706 & 1.691 & 1.697 &       & 1.694 & 1.677 & 1.687 \\
    \bottomrule
    \end{tabular}%
    }
  \label{recovery4}%
\end{table}%

\begin{table}[htbp]
\color{black}
  \centering
  \caption{Shape of weights estimation accuracy of the three different methods: LASSO-U, LASSO-M, sg-LASSO-M. Entries are the average mean squared error. $s = 6$ and $t = \{t_1 = 10\%, t_2 = 30\%, t_3 = 50\%\}$ percentile of the set $\{T_i: T_i \text{ is uncensored }, i \in [N]\}$.}
  \medskip
  \scalebox{0.75}{
    \begin{tabular}{cccccccc}
    \toprule
          & \multicolumn{7}{c}{Scenario \ref{s5}} \\
\cmidrule{2-8}          & \multicolumn{3}{c}{$N = 800$} &       & \multicolumn{3}{c}{$N = 1200$} \\
\cmidrule{2-8}          & \multicolumn{7}{c}{$t = t_1$} \\
\cmidrule{2-8}          & LASSO-U & LASSO-M & sg-LASSO-M &       & LASSO-U & LASSO-M & sg-LASSO-M \\
\cmidrule{2-8}
$(1+log(t-s))\boldsymbol{\operatorname{Beta}}(1,3)$
& 0.766 & 0.679 & 0.646 &       & 0.743 & 0.626 & 0.602 \\
$(-1+log(t-s))\boldsymbol{\operatorname{Beta}}(2,3)$
& 2.111 & 1.905 & 1.872 &       & 2.055 & 1.789 & 1.772 \\
\cmidrule{2-8}          & \multicolumn{7}{c}{$t = t_2$} \\
\cmidrule{2-8}          & LASSO-U & LASSO-M & sg-LASSO-M &       & LASSO-U & LASSO-M & sg-LASSO-M \\
\cmidrule{2-8}
$(1+log(t-s))\boldsymbol{\operatorname{Beta}}(1,3)$
& 1.068 & 0.916 & 0.888 &       & 1.026 & 0.845 & 0.861 \\
$(-1+log(t-s))\boldsymbol{\operatorname{Beta}}(2,3)$
& 1.665 & 1.442 & 1.421 &       & 1.613 & 1.370 & 1.403 \\
\cmidrule{2-8}          & \multicolumn{7}{c}{$t = t_3$} \\
\cmidrule{2-8}          & LASSO-U & LASSO-M & sg-LASSO-M &       & LASSO-U & LASSO-M & sg-LASSO-M \\
\cmidrule{2-8}
$(1+log(t-s))\boldsymbol{\operatorname{Beta}}(1,3)$
& 1.347 & 1.199 & 1.206 &       & 1.298 & 1.076 & 1.100 \\
$(-1+log(t-s))\boldsymbol{\operatorname{Beta}}(2,3)$
& 1.444 & 1.305 & 1.328 &       & 1.406 & 1.191 & 1.226 \\
    \bottomrule
    \end{tabular}%
    }
  \label{recovery5}%
\end{table}%

Figures \ref{sel 6} and \ref{sel 10} present the financial covariates selected by sg-LASSO-MIDAS when $s = 6$ years in the empirical application. Only covariates with at least one nonzero estimated coefficient in the group corresponding to this covariate are selected and marked in red.
Furthermore, Figure \ref{selected6} illustrates the proportion of selected covariates for each financial type identified by the sg-LASSO-MIDAS method when $s = 6$ and $10$ years.

 \begin{figure}[htbp]
    \centering
    \includegraphics[width=1\linewidth]{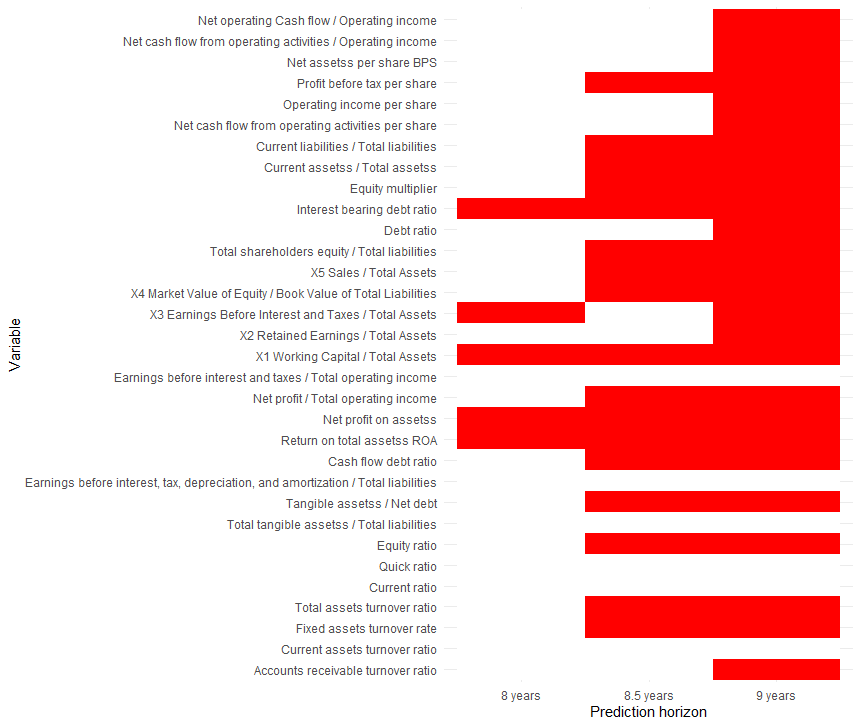}
    \caption{Selected financial variables by sg-LASSO-MIDAS when $s = 6$ years.}
    \label{sel 6}
\end{figure}

\begin{figure}[htbp]
    \centering
    \includegraphics[width=1\linewidth]{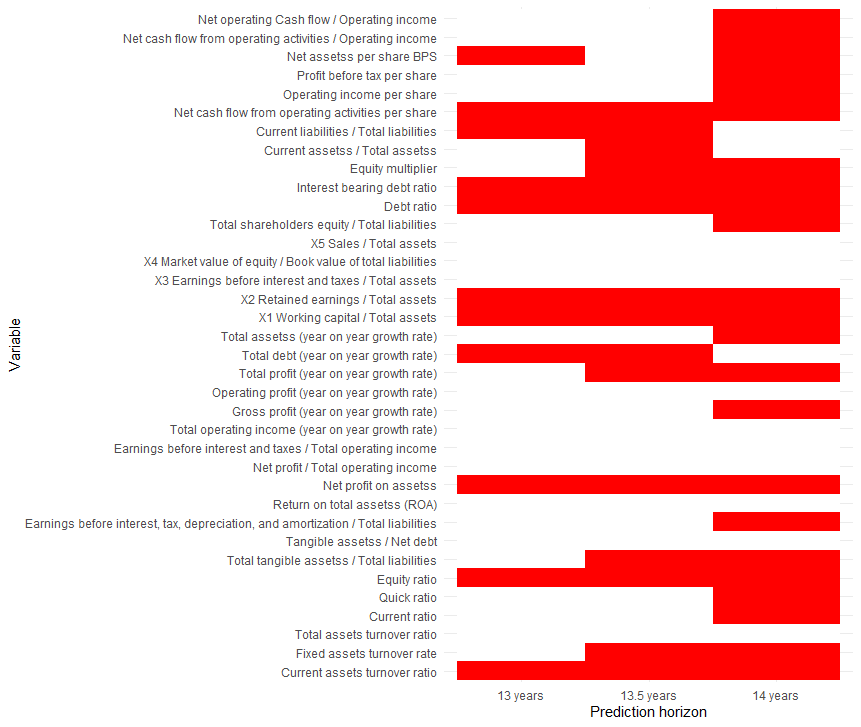}
    \caption{Selected financial variables by sg-LASSO-MIDAS when $s = 10$ years.}
    \label{sel 10}
\end{figure}

\begin{figure}[htbp]
    \centering
    \begin{minipage}{1\textwidth}
        \centering
        \includegraphics[width=\textwidth]{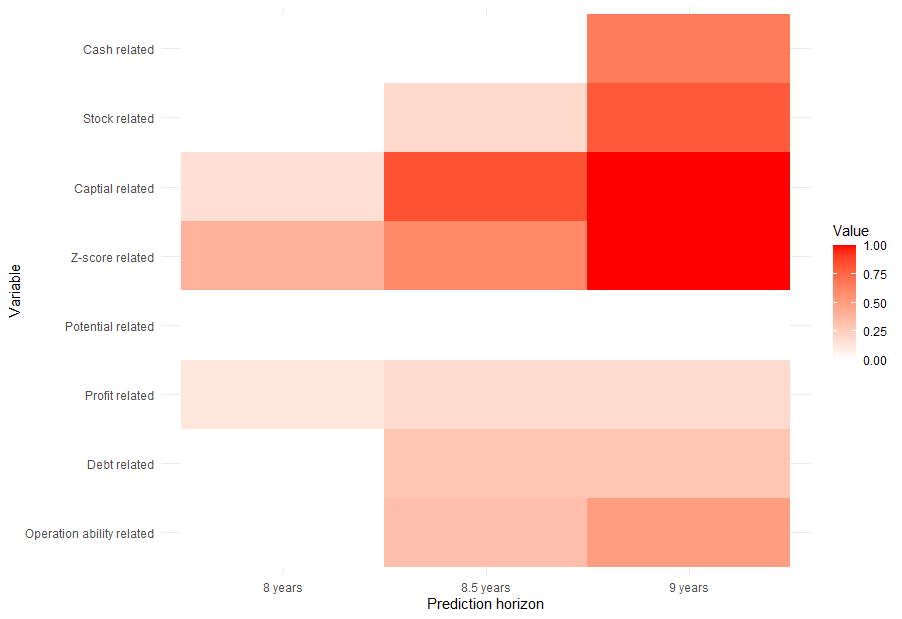}
    \end{minipage}
    \begin{minipage}{1\textwidth}
        \centering
        \includegraphics[width=\textwidth]{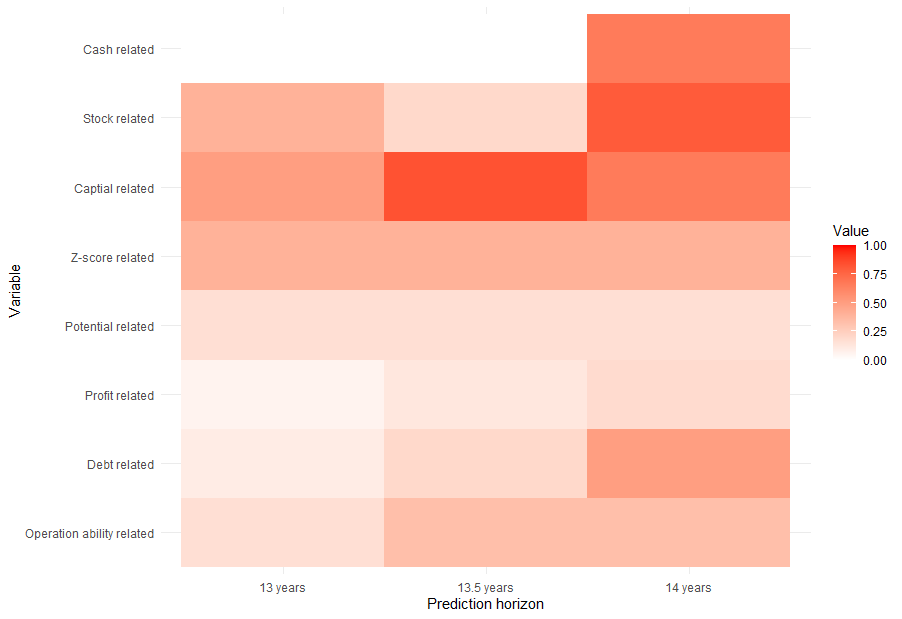}
    \end{minipage}
    \caption{The proportion of selected financial variables in each type when $s = 6$ and $10$ years.}
    \label{selected6}
\end{figure}

 \begin{figure}[htbp]
 \color{black}
    \centering
    \includegraphics[width=1\linewidth]{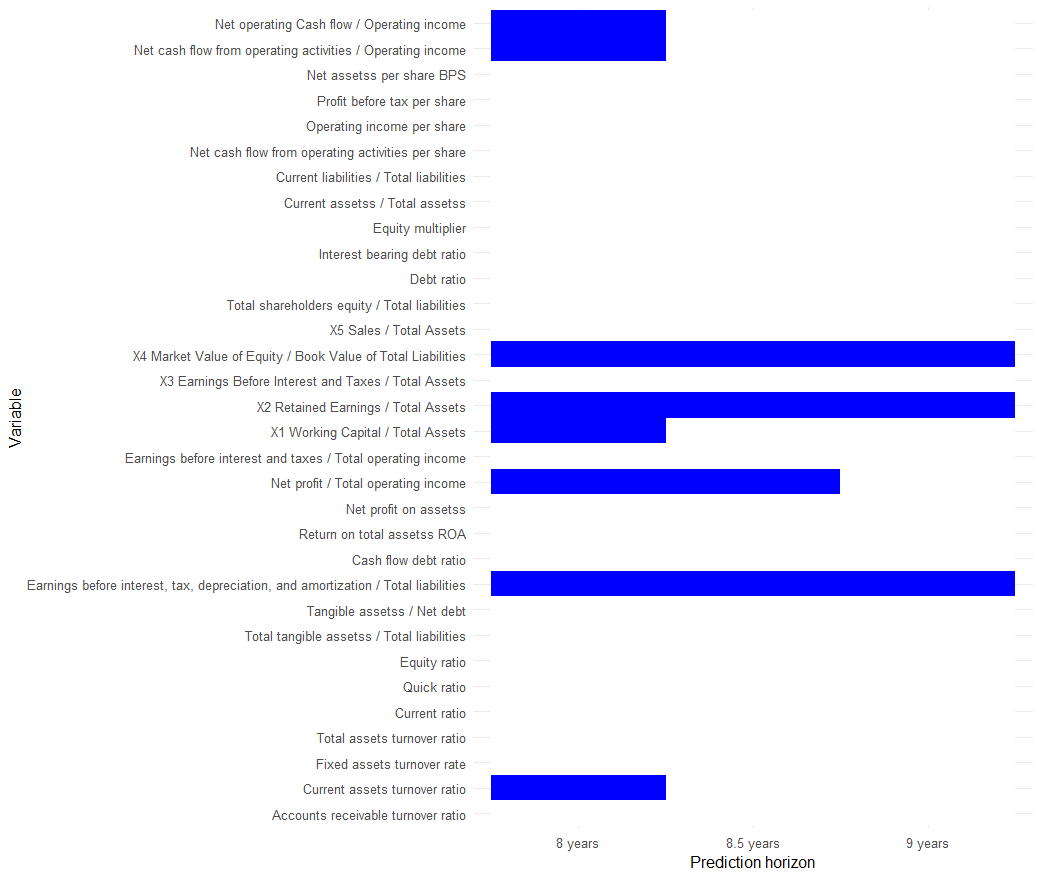}
    \caption{Significant financial variables identified by de-sparsified sg-LASSO-MIDAS when $s = 6$ years.}
    \label{sig 6}
\end{figure}

 \begin{figure}[htbp]
 \color{black}
    \centering
    \includegraphics[width=1\linewidth]{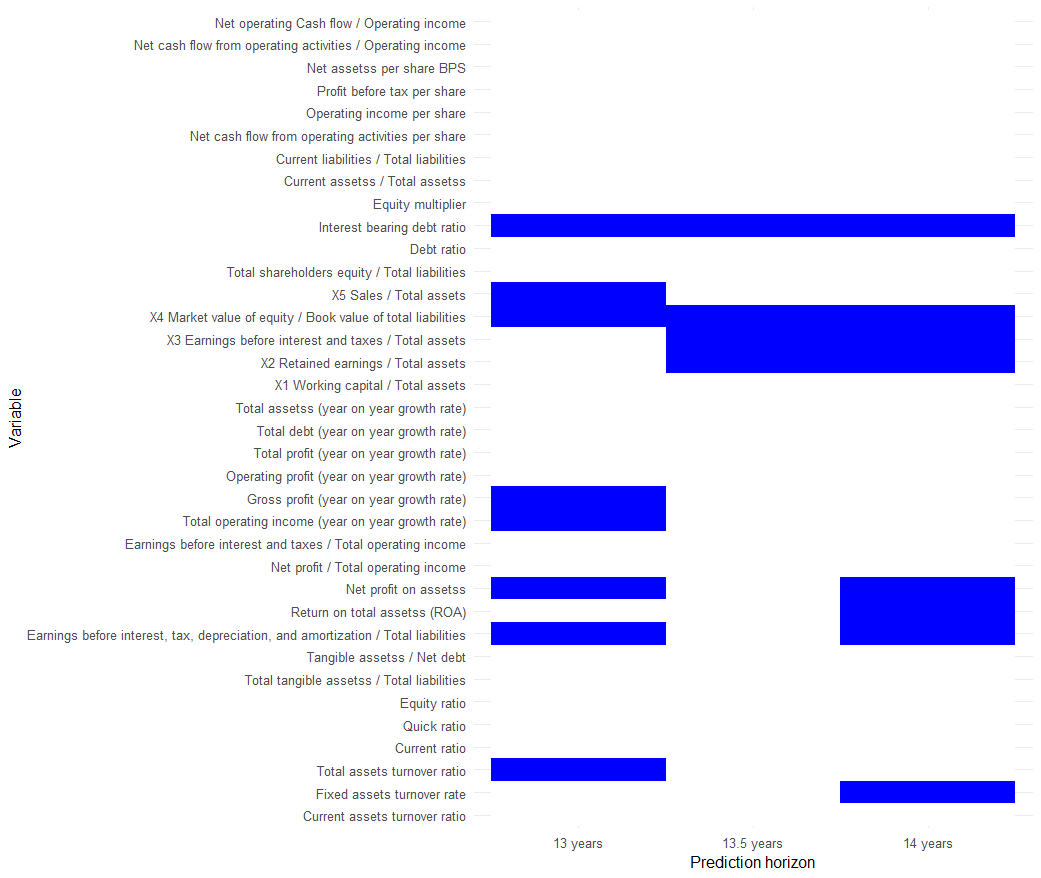}
    \caption{Significant financial variables identified by de-sparsified sg-LASSO-MIDAS when $s = 10$ years.}
    \label{sig 10}
\end{figure}
{\color{black}
Following the results in Tables \ref{empirical AUC Augment6} and \ref{empirical AUC Augment10}, Tables \ref{alt group structures s=6} and \ref{alt group structures s=10} present the performance of sg-LASSO-MIDAS using different group structures:

\medskip

\noindent\textbf{Grouped by covariate types.} We group covariates by financial type, with all MIDAS-weighted covariates of the same type forming a group. The types are described in Section \ref{sec real}.

\medskip

\noindent\textbf{Grouped by correlation matrix.} We compute the correlation matrix of the MIDAS-weighted covariates and apply hierarchical agglomerative clustering using the average linkage method. The number of clusters is treated as a tuning choice. Specifically, we consider $\{1,2,3,4,5,6\}$.

The performance of sg-LASSO-MIDAS appears to be robust across alternative group structures, with sg-LASSO-MIDAS consistently outperforming LASSO-UMIDAS. As discussed in Section \ref{sec real}, when $s = 6$ years, it is difficult to distinguish between the performance of LASSO-MIDAS and sg-LASSO-MIDAS. However, under alternative group specifications, such as grouping the correlation matrix into three groups, sg-LASSO-MIDAS exhibits slightly better numerical performance than LASSO-MIDAS at two prediction horizons. For $s = 10$ years, sg-LASSO-MIDAS still achieves the best predictive performance in most prediction horizons, also demonstrating stable performance across all group structures. Although the true group structure of covariates is unknown in practice, one advantage of sg-LASSO-MIDAS over LASSO-MIDAS is that, with a modest increase in computation, exploring different group structures may lead to improved predictive performance.
}
\begin{table}[htbp]
  \centering
  \color{black}
  \caption{(Alternative group structures for sg-LASSO-MIDAS) Estimated average AUCs ($95\%$ confidence interval) in the out-of-sample set with $s = 6$ years and prediction horizons $t = 8, 8.5, 9$ years.}
  \medskip
    \begin{tabular}{cccc}
    \toprule
          & \multicolumn{3}{c}{$s = 6$ years} \\
\cmidrule{2-4}          & $t = 8$ years & $t = 8.5$ years & $t = 9$ years \\
\cmidrule{2-4}    LASSO-U &  0.797 [0.755, 0.844]   &  0.756 [0.717, 0.801]   & 0.765 [0.734, 0.802]\\
    LASSO-M &  0.838 [0.793, 0.872]   &   0.817 [0.774, 0.864]    & 0.811 [0.778, 0.843] \\
\cmidrule{2-4}          & \multicolumn{3}{c}{Grouped by covariate types} \\
\cmidrule{2-4}    sg-LASSO-M &  0.827
 [0.789, 0.863]   &  0.830 [0.790, 0.860]   &  0.792 [0.753, 0.833]\\
\cmidrule{2-4}          & \multicolumn{3}{c}{Grouped by correlation matrix where number of groups $=1$} \\
\cmidrule{2-4}   sg-LASSO-M &  0.822 [0.782, 0.863]   &  0.843 [0.811, 0.874]   &  0.802 [0.774, 0.838]\\
    \cmidrule{2-4}          & \multicolumn{3}{c}{Grouped by correlation matrix where number of groups $=2$} \\
\cmidrule{2-4}  
    sg-LASSO-M &  0.829 [0.785, 0.864]    &   0.816 [0.780, 0.860]  &  0.814 [0.785, 0.847] \\
    \cmidrule{2-4}          & \multicolumn{3}{c}{Grouped by correlation matrix where number of groups $=3$} \\
\cmidrule{2-4}  
    sg-LASSO-M &  0.826 [0.788, 0.861]    &   0.830 [0.795, 0.873]   &  0.818 [0.788, 0.848] \\
\cmidrule{2-4}          & \multicolumn{3}{c}{Grouped by correlation matrix where number of groups $=4$} \\
\cmidrule{2-4}  
    sg-LASSO-M    & 0.808 [0.763, 0.849]     & 0.834 [0.799, 0.874]  &  0.802 [0.774, 0.838]\\
    \cmidrule{2-4}          & \multicolumn{3}{c}{Grouped by correlation matrix where number of groups $=5$} \\
\cmidrule{2-4}   
    sg-LASSO-M    & 0.807 [0.770, 0.846]     & 0.835 [0.801, 0.867]  &  0.805 [0.770, 0.837]\\
            \cmidrule{2-4}          & \multicolumn{3}{c}{Grouped by correlation matrix where number of groups $=6$} \\
\cmidrule{2-4}    
    sg-LASSO-M &   0.817 [0.776, 0.854]  &  0.824 [0.791, 0.868]   &  0.814 [0.781, 0.847] \\
    \bottomrule
    \end{tabular}%
  \label{alt group structures s=6}%
\end{table}%

\begin{table}[htbp]
  \centering
  \color{black}
  \caption{(Alternative group structures for sg-LASSO-MIDAS) Estimated average AUCs ($95\%$ confidence interval) in the out-of-sample set with $s = 10$ years and prediction horizons $t = 13, 13.5, 14$ years.}
  \medskip
    \begin{tabular}{cccc}
    \toprule
          & \multicolumn{3}{c}{$s = 10$ years} \\
\cmidrule{2-4}          & $t = 13$ years & $t = 13.5$ years & $t = 14$ years \\
\cmidrule{2-4}    LASSO-U &  0.566 [0.514, 0.633]   &  0.628 [0.591, 0.674]   & 0.669 [0.637, 0.709] \\
    LASSO-M &  0.773 [0.738, 0.818]    &   0.653 [0.615, 0.700]    & 0.688 [0.654, 0.726]\\
\cmidrule{2-4}          & \multicolumn{3}{c}{Grouped by covariate types} \\
\cmidrule{2-4}    sg-LASSO-M &  0.795
 [0.758, 0.832]   &  0.644 [0.609, 0.690]   &  0.730 [0.705, 0.768]\\
\cmidrule{2-4}          & \multicolumn{3}{c}{Grouped by correlation matrix where number of groups $=1$} \\
\cmidrule{2-4}   sg-LASSO-M &  0.800 [0.763, 0.839]   &  0.686 [0.649, 0.727]   &  0.753 [0.720, 0.782]\\
    \cmidrule{2-4}          & \multicolumn{3}{c}{Grouped by correlation matrix where number of groups $=2$} \\
\cmidrule{2-4}  
    sg-LASSO-M &  0.790 [0.752, 0.832]    &   0.657 [0.617, 0.704]  &  0.712 [0.676, 0.747] \\
    \cmidrule{2-4}          & \multicolumn{3}{c}{Grouped by correlation matrix where number of groups $=3$} \\
\cmidrule{2-4}  
    sg-LASSO-M &  0.766 [0.733, 0.814]    &   0.649 [0.600, 0.693]   &  0.725 [0.689, 0.752] \\
\cmidrule{2-4}          & \multicolumn{3}{c}{Grouped by correlation matrix where number of groups $=4$} \\
\cmidrule{2-4}  
    sg-LASSO-M    & 0.755 [0.717, 0.803]     & 0.676 [0.642, 0.721]  &  0.709 [0.677, 0.737]\\
    \cmidrule{2-4}          & \multicolumn{3}{c}{Grouped by correlation matrix where number of groups $=5$} \\
\cmidrule{2-4}   
    sg-LASSO-M    & 0.769 [0.730, 0.812]     & 0.677 [0.634, 0.719]  &  0.737 [0.706, 0.766]\\
            \cmidrule{2-4}          & \multicolumn{3}{c}{Grouped by correlation matrix where number of groups $=6$} \\
\cmidrule{2-4}    
    sg-LASSO-M &   0.753 [0.714, 0.803]  &  0.635 [0.598, 0.689]   &  0.720 [0.692, 0.753] \\
    \bottomrule
    \end{tabular}%
  \label{alt group structures s=10}%
\end{table}%
To further show the prediction performance of the proposed method, we present another application of distress prediction using the same dataset, but with a different method for dividing the in-sample and out-of-sample sets compared to the previous subsection, which is closer to real-time prediction. 

Recall that the financial dataset spans from January $1^{\text{st}}$, $1985$, to December $31^{\text{st}}$, $2020$. To better align with practical applications, when $s = 6$ years, we first select the time point $2016/12/31$. Firms that had already survived 6 years prior to this date are used as the in-sample set ($544$ firms), while the remaining firms that had not yet survived 6 years by $2016/12/31$ are placed in the out-of-sample set ($357$ firms). Thus, the actual observation period ranges from $1985/01/01$ to $2016/12/31$. The prediction horizons are set as $t = 8, 8.5, 9$ years, as in previous analyses. The regularization parameters $\lambda$ and $\alpha$ using $5$-fold stratified cross-validation for AUC. Specifically, we use a grid of $\{0.9, 0.91, 0.92, \dots, 1\}$ to search for the optimal regularization parameter $\alpha$ in the sparse-group LASSO penalty. As before, $\lambda$ is chosen in a grid which follows \citet{liang_sparsegl}. All other settings are consistent with those in the previous section, except that we use a dictionary $W$ composed of Gegenbauer polynomials shifted to $[0,1]$ with parameter $\alpha_{\text{poly}} = \beta_{\text{poly}} =  \frac{1}{2}$ and size $L = 3$.
\begin{table}[H]
\scalebox{1.15}{
\begin{threeparttable}
  \centering
  \caption{(Additional application) Estimated AUCs ($95\%$ confidence interval) in the out-of-sample set with $s = 6$ years and prediction horizons $t = 8, 8.5, 9$ years. }
  \medskip
    \begin{tabular}{cccc}
    \toprule
          & \multicolumn{3}{c}{$s = 6$ years} \\
\cmidrule{2-4}          & $t = 8$ years & $t = 8.5$ years & $t = 9$ years \\
\cmidrule{2-4}          & \multicolumn{3}{c}{\tcr{Benchmark} } \\
\cmidrule{2-4}    Logistic reg. &  \tcr{0.529 [0.443, 0.684]}   &  \tcr{0.595 [0.422, 0.716]}   &  \tcr{0.602 [0.517, 0.693]}\\
    \cmidrule{2-4}          & \multicolumn{3}{c}{Cross-Validation for the AUC} \\
\cmidrule{2-4}    LASSO-U &  0.734 [0.565, 0.880]  &  0.688 [0.558, 0.824]    & 0.666 [0.539, 0.798]\\
    LASSO-M &  0.898 [0.829, 0.952]     &   0.866 [0.774, 0.940]    & 0.812 [0.728, 0.910] \\
    sg-LASSO-M &  0.898 [0.829, 0.952]     &   0.845 [0.746 ,0.932]   &  0.812 [0.728, 0.910] \\
        \cmidrule{2-4}          & \multicolumn{3}{c}{Cross-Validation for the likelihood score} \\
\cmidrule{2-4}    LASSO-U &  0.736 [0.566, 0.881]  &  0.714 [0.588, 0.838]    & 0.552 [0.445, 0.669]  \\
    LASSO-M &  0.898 [0.828, 0.956]     &   0.841 [0.753, 0.938]     & 0.789 [0.680, 0.908] \\
    sg-LASSO-M &  0.898 [0.828, 0.956]     &   0.835 [0.748, 0.932]    &  0.760 [0.648, 0.898] \\
    \cmidrule{2-4}          & \multicolumn{3}{c}{Macro Data Augmented} \\
\cmidrule{2-4}    LASSO-U    &  0.734 [0.645, 0.812]      & 0.688 [0.616, 0.764] &  0.666 [0.583, 0.751]\\
    LASSO-M     & 0.898 [0.868, 0.933]     & 0.809 [0.745, 0.877]  &  0.769 [0.739, 0.813]\\
    sg-LASSO-M    &    0.899 [0.869, 0.933]   & 0.874 [0.838, 0.915]  &  0.759 [0.731, 0.811]\\
   \cmidrule{2-4}          & \multicolumn{3}{c}{Data without censored firms satisfying $C_i < t$} \\
\cmidrule{2-4}    LASSO-U &  0.680 [0.543, 0.822]   &  0.745 [0.648, 0.829]    & 0.628 [0.512, 0.787] \\
    LASSO-M &  0.807 [0.679, 0.926]  &   0.767 [0.671, 0.884]    & 0.778 [0.682, 0.885] \\
    sg-LASSO-M &  0.807 [0.679, 0.926]   &   0.767 [0.671, 0.884] &  0.773 [0.658, 0.888]\\
          \cmidrule{3-3}  &  \multicolumn{3}{c}{sg-LASSO-M $vs.$ LASSO-U} \\
        p-value   &          $0.001^{**}$ & $0.000^{**}$ & $0.000^{**}$ \\
                  \cmidrule{3-3}  &  \multicolumn{3}{c}{ sg-LASSO-M with $vs.$ without censored firms satisfying $C_i < t$} \\
        p-value   &          $0.009^{**}$ & $0.042^{**}$ & $0.107$  \\
    \bottomrule
    \end{tabular}%
    \label{empirical AUC a2 6}%
    \begin{tablenotes}
      \small
      \item Notes: The second-to-last row reports the p-value from the pairwise difference test across the three prediction horizons, with the null hypothesis that the estimated AUC of sg-LASSO-MIDAS is no larger than LASSO-UMIDAS's. The last row presents the p-value from the pairwise test across the three prediction horizons, comparing sg-LASSO-MIDAS applied to data with and without censored firms satisfying $C_i < t$.  We use $*$ and $**$ to indicate $10\%$ and $5\%$ significance,  respectively.
    \end{tablenotes}
    \end{threeparttable}
    }
\end{table}%

\begin{table}[H]
\scalebox{1.15}{
\begin{threeparttable}
  \centering
  \caption{(Additional application) Estimated AUCs ($95\%$ confidence interval) in the out-of-sample set with $s = 10$ years and prediction horizons $t = 13, 13.5, 14$ years.}
  \medskip
     \begin{tabular}{cccc}
    \toprule
            & \multicolumn{3}{c}{$s = 10$ years} \\
\cmidrule{2-4}          & $t = 13$ years & $t = 13.5$ years & $t = 14$ years \\
\cmidrule{2-4}          & \multicolumn{3}{c}{\tcr{Benchmark} } \\
\cmidrule{2-4}    Logistic reg. &  \tcr{0.467 [0.260, 0.759]}   &  \tcr{0.504 [0.281, 0.759]}   &  \tcr{0.520 [0.320, 0.773]}\\
    \cmidrule{2-4}          & \multicolumn{3}{c}{Cross-Validation for the AUC} \\
\cmidrule{2-4}    LASSO-U &  0.685 [0.483, 0.879]   &  0.753 [0.529, 0.895]    & 0.681 [0.541, 0.847]  \\
    LASSO-M &  0.775 [0.528, 0.943]    &   0.784 [0.685, 0.878]    & 0.801 [0.704, 0.880]\\
    sg-LASSO-M &  0.758 [0.499, 0.946]    &   0.803 [0.706, 0.879]   &  0.801 [0.704, 0.880]\\
        \cmidrule{2-4}          & \multicolumn{3}{c}{Cross-Validation for the likelihood score} \\
\cmidrule{2-4}    LASSO-U &  0.500 [0.500, 0.500]  &  0.500 [0.500, 0.500]   & 0.690 [0.547, 0.846]\\
    LASSO-M &  0.789 [0.687, 0.902]     &   0.806 [0.694, 0.907]    & 0.696 [0.577, 0.865]\\
    sg-LASSO-M &  0.789 [0.687, 0.902]    &   0.806 [0.694, 0.907]   &  0.696 [0.577, 0.865]\\
        \cmidrule{2-4}          & \multicolumn{3}{c}{Macro Data Augmented} \\
\cmidrule{2-4}    LASSO-U    &  0.685 [0.574, 0.806]      & 0.643 [0.559, 0.756] &  0.738 [0.650, 0.827]\\
    LASSO-M     & 0.775 [0.718, 0.842]      & 0.787 [0.741, 0.845]  &  0.788 [0.715, 0.848]\\
    sg-LASSO-M    &   0.788 [0.724, 0.837]   & 0.797 [0.750, 0.855]  &  0.796 [0.716, 0.853]\\
        \cmidrule{2-4}          & \multicolumn{3}{c}{Data without censored firms satisfying $C_i < t$} \\
\cmidrule{2-4}    LASSO-U &  0.662 [0.491, 0.893]   &  0.507 [0.288, 0.701]    & 0.662 [0.474, 0.806]\\
    LASSO-M &  0.737 [0.609, 0.900]   &   0.653 [0.453, 0.859]    & 0.801 [0.704, 0.880]\\
    sg-LASSO-M &  0.745 [0.605, 0.902]  &   0.731 [0.544, 0.894]  &  0.678 [0.521, 0.802] \\
              \cmidrule{3-3}  &  \multicolumn{3}{c}{sg-LASSO-M $vs.$ LASSO-U} \\
        p-value   &          0.318 & 0.255 & $0.097^{*}$ \\
                  \cmidrule{3-3}  &  \multicolumn{3}{c}{ sg-LASSO-M with $vs.$ without censored firms satisfying $C_i < t$} \\
        p-value   &          $0.459$ & $0.290$ & $0.072^{*}$  \\
    \bottomrule
    \end{tabular}%
  \label{empirical AUC a2 10}%
      \begin{tablenotes}
      \small
      \item Notes: The second-to-last row reports the p-value from the pairwise difference test across the three prediction horizons, with the null hypothesis that the estimated AUC of sg-LASSO-MIDAS is no larger than LASSO-UMIDAS's. The last row presents the p-value from the pairwise test across the three prediction horizons, comparing sg-LASSO-MIDAS applied to data with and without censored firms satisfying $C_i < t$. We use $*$ and $**$ to indicate $10\%$ and $5\%$ significance,  respectively.
    \end{tablenotes}
    \end{threeparttable}
    }
\end{table}%

For $s = 10$ years, which is relatively large, we select a new time point of $2013/12/31$ to allow for more years of prediction after this date. The prediction horizons are set to $t = 13, 13.5, 14$ years. Firms that had survived for $s = 10$ years before $2013/12/31$ are used as the in-sample set ($311$ firms), while those that had not survived $s = 10$ years by this time are treated as the out-of-sample set ($473$ firms).

Tables \ref{empirical AUC a2 6} and \ref{empirical AUC a2 10} report the estimated AUCs in the out-of-sample set. The second-to-last row presents the pairwise test between sg-LASSO-MIDAS and LASSO-UMIDAS, while the last row presents the comparison between sg-LASSO-MIDAS applied to data with and without censored firms satisfying $C_i < t$. 

For $s = 6$ years, sg-LASSO-MIDAS significantly outperforms LASSO-UMIDAS, whereas LASSO-MIDAS performs similarly to sg-LASSO-MIDAS. Furthermore, the macroeconomic data-augmented prediction appears comparable to the purely financial model in most scenarios. However, the prediction performance is observed to be more stable when macroeconomic data is included compared to using only financial data. Additionally, sg-LASSO-MIDAS performs statistically better at the $5\%$ level when the dataset includes censored firms with $C_i < t$, except for the $t = 9$ years prediction horizon, highlighting the advantage of accounting for censoring in the prediction model. For $s = 10$ years, sg-LASSO-MIDAS applied to the full dataset is numerically superior to both sg-LASSO-MIDAS applied to the dataset without censored firms satisfying $C_i < t$ and LASSO-UMIDAS. However, both of the differences are statistically significant at $10\%$ level only for $t = 14$ years.

\section{Description of the empirical dataset}\label{data details}

The raw dataset contains $1614$ Chinese publicly traded firms in the manufacturing industry with observed financial status. The first listing date, the first date at which the firm is under ST, and industry code of each firm are contained in the dataset. All the financial variables are quarterly measured from January 1, 1985, to December 31, 2020. Table \ref{variable} shows the names and units of these financial variables.
\begin{longtable}{ll}
\cmidrule{1-2}   \textbf{Financial Variable}  & \textbf{Unit} \\
    \midrule
    \textbf{Operation ability Related} &  \\
    \midrule
    Inventory Turnover & times \\
    Accounts Receivable Turnover Ratio & $\%$ \\
    Accounts Payable Turnover Ratio & $\%$ \\
    \multicolumn{1}{p{36.285em}}{Current Assets Turnover Ratio } & $\%$ \\
    Fixed Assets Turnover Ratio  & $\%$ \\
    Total Assets Turnover Ratio  & $\%$ \\
    \midrule
    \textbf{Debt Related} &  \\
    \midrule
    Current Ratio  & $\%$ \\
    Quick Ratio  & $\%$ \\
    Equity Ratio  & $\%$ \\
    Total Tangible Assets / Total Liabilities & $\%$ \\
    Total Tangible Assets / Interest-Bearing Debt & $\%$ \\
    Total Tangible Assets / Net Debt & $\%$ \\
    Earnings Before Interest, Tax, Depreciation, and Amortization / Total Liabilities & $\%$ \\
    Cash Flow Debt Ratio  & $\%$ \\
    Time Interest Earned Ratio  & $\%$ \\
    Long-Term Debt to Capitalization Ratio  & $\%$ \\
    \midrule
    \textbf{Profit Related} &  \\
    \midrule
    Weighted Return on Equity (ROE)  & $\%$ \\
    Deducted Return on Equity (ROE)  & $\%$ \\
    Return on Assets (ROA)  & $\%$ \\
    Net Profit on Assets  & $\%$ \\
    Return on Invested Capital (ROIC) & $\%$ \\
    Sales Margin & $\%$ \\
    Gross Profit Margin & $\%$ \\
    Net Profit / Total Operating Income & $\%$ \\
    Earnings Before Interest and Taxes / Total Operating Income & $\%$ \\
    Basic Earnings Per Share (year-on-year growth rate) & $\%$ \\
    Diluted Earnings Per Share (year-on-year growth rate)  & $\%$ \\
    Total Operating Income (year-on-year growth rate) & $\%$ \\
    Gross Profit (year-on-year growth rate) & $\%$ \\
    Operating Profit (year-on-year growth rate) & $\%$ \\
    Total Profit (year-on-year growth rate) & $\%$ \\
    Net Profit (year-on-year growth rate) & $\%$ \\
    \midrule
    \textbf{Potential Related} &  \\
    \midrule
    Net Cash Flow From Operating Activities  & $\%$ \\
    Cash in Net Profit (year-on-year growth rate) & $\%$ \\
    Net Assets (year-on-year growth rate) & $\%$ \\
    Total Debt (year-on-year growth rate) & $\%$ \\
    Total Assets (year-on-year growth rate) & $\%$ \\
    Net Cash Flow (year-on-year growth rate) & $\%$ \\
    \midrule
    \textbf{Z-score Related}&  \\
    \midrule
    X1 - Working Capital / Total Assets  & $\%$ \\
    X2 - Retained Earnings / Total Assets  & $\%$ \\
    X3 - Earnings Before Interest and Taxes / Total Assets  & $\%$ \\
    X4 - Market Value of Equity / Book Value of Total Liabilities  & $\%$ \\
    X5 - Sales / Total Assets  & $\%$ \\
    \midrule
    \textbf{Capital Related} &  \\
    \midrule
    Total Shareholders’ Equity / Total Liabilities & $\%$ \\
    Debt Ratio  & $\%$ \\
    Interest-Bearing Debt Ratio & $\%$ \\
    Equity Multiplier  & $\%$ \\
    Current Assets / Total Assets & $\%$ \\
    Current Liabilities / Total Liabilities & $\%$ \\
    \midrule
    \textbf{Stock Related} &  \\
    \midrule
    Earnings Per Share EPS - Basic & Yuan \\
    Net Cash Flow from Operating Activities Per Share  & Yuan \\
    Operating Income Per Share & Yuan \\
    Profit before Tax Per Share & Yuan \\
    Net Assets Per Share BPS  & Yuan \\
    \midrule
    \textbf{Cash Related} &  \\
    \midrule
    Net Cash Flow from Operating Activities / Operating Income & $\%$ \\
    Net Cash Flow from Operating Activities / Net Income from Operating Activities & $\%$ \\
    Net Operating Cash Flow / Operating Income & $\%$ \\
    \bottomrule
    \\[-2ex]
    \caption{Quarterly financial data of Chinese publicly listed firms.}
  \label{variable}%
\end{longtable}%

Table \ref{macro variabe} presents the names and descriptions of the macro-economic variables used in the paper. See more data details in the 'readme' file \url{https://www.atlantafed.org/cqer/research/china-macroeconomy#Tab2}.
\begin{longtable}{lp{10.6cm}}
\cmidrule{1-2} \textbf{Macro Variable} & \textbf{Description} \\ 
\midrule
CPI & Consumer price index \\ 
RetailPriceIndex & Retail Price Index \\ 
FAIPriceIndex & Fixed asset investment price index \\ 
GFCFPriceIndex & Price index for gross fixed capital formation \\ 
NominalRetailGoodsC & Retail sales of consumer goods \\ 
NominalFAI & Fixed asset investment by eliminating the 1994Q4 outlier \\ 
NominalFAIGovt & Fixed asset investment: government \\ 
NominalFAIPriv & Fixed asset investment: private sector excluding SOEs and other non-SOEs \\ 
NominalFAISOEexGovt & Fixed asset investment: SOEs excluding government \\ 
NominalFAINonSOE & Fixed asset investment: other non-SOE enterprises \\ 
NominalGDP & GDP by expenditure \\ 
NominalNetExports & Net exports by expenditure \\ 
NominalExportsGoods & Exports of goods reported by the Chinese customs  \\ 
NominalImportsGoods & Imports of goods reported by the Chinese customs  \\ 
NominalHHC & Household consumption by expenditure  \\ 
NominalGovtC & Government consumption by expenditure  \\ 
NominalGCF & Nominal gross capital formation  \\ 
NominalInvty & Changes in inventories  \\ 
NominalGFCF & Gross fixed capital formation with no inventories  \\ 
NominalGovtGFCF & Gross fixed capital formation: government  \\ 
NominalPrivGFCF & Gross fixed capital formation: private sector—excluding government, households, SOEs, and other non-SOEs \\ 
NominalHHGFCF & Gross fixed capital formation: households \\ 
NominalSOEGFCF & Gross fixed capital formation: SOE \\ 
NominalSOEexGovtGFCF & Gross fixed capital formation: SOE excluding government \\ 
NominalNonSOEGFCF & Gross fixed capital formation: other non-SOE enterprises \\ 
NominalBusGFCF & Total business investment \\ 
NominalNarOutput & Narrow definition of output (NominalBusGFCF + NominalHHC) \\ 
RatioGFCFPrice2CPI & Relative prices of investment goods (to CPI) \\ 
LaborIncome & Interpolated labor income with extrapolation in early years \\ 
LaborIncomeShare & Labor income as a share of total value added \\ 
LaborCompSumProvinces & Sum of interpolated labor compensations across provinces according to the provincial GDP by income \\ 
DPI & Interpolated disposable personal income with extrapolation in early years \\ 
AvgNominalWage & Aggregate average nominal wages \\ 
ReserveMoney & Reserve money \\ 
M0 & M0 \\ 
M2 & M2 \\ 
RRR & Required reserve ratio \\ 
ARR & Actual reserve ratio \\ 
ERR & Excess reserve ratio \\ 
R3mDeposit & Time deposits rate: 3 months \\ 
R1dRepo & The spliced series of 1-day Repo rate and 1-day Chibor rate \\ 
BankLoansMLT & Medium and long-term bank loans outstanding \\ 
NewBankLoansNFEST & New bank loans to non-financial enterprises: short term \\ 
NewBankLoansNFESTBF & New bank loans to non-financial enterprises: short term and bill financing \\ 
NewBankLoansNFEMLT & New bank loans to non-financial enterprises: medium and long terms \\ 
logrealHHC & Log of real household consumption expenditure (Nominal HHC divided by CPI) \\ 
logrealBusI & Log of real business gross fixed capital formation (Nominal Bus GFCF divided by GFCF Price Index) \\ 
logrealHHC\_nipa & Log of real household consumption expenditure (Nominal HHC divided by GDP Deflator) \\ 
logrealBusI\_nipa & Log of real business gross fixed capital formation (Nominal Bus GFCF divided by GDP Deflator) \\ 
logrealNarrowY\_nipa & Log of real narrow output (Nominal Narrow Output divided by GDP Deflator) \\ 
logrealGDP\_nipa & Log of real GDP (Nominal GDP divided by GDP Deflator) \\ 
logrealGDP\_va & Log of real GDP value-added (Nominal GDPva divided by GDP Deflator) \\ 
logrealLaborIncome & Log of real labor income (Labor Income divided by GDP Deflator) \\ 
logrealDPI & Log of real disposable personal income (DPI divided by GDP Deflator) \\ 
logM2 & Log of M2 money supply \\ 
ratioNewLoansNFEST2GDP & Ratio of new bank loans to non-financial enterprises (short term) to nominal GDP \\ 
ratioNewLoansNFESTBF2GDP & Ratio of new bank loans to non-financial enterprises (short term and bill financing) to nominal GDP \\ 
ratioNewLoansNFEMLT2GDP & Ratio of new bank loans to non-financial enterprises (medium and long term) to nominal GDP \\ 
LendingRatePBC1year & One-year PBOC benchmark lending rate \\ 
DepositRatePBC1year & One-year PBOC benchmark deposit rate \\ 
Employment & Average of contemporaneous and 1-period lag of employment level \\ 
NVA\_InpOut\_Heavy & Value-added input-output for the heavy sector (seasonally adjusted) \\ 
NVA\_InpOut\_Light & Value-added input-output for the light sector (seasonally adjusted) \\ 
NHeavyFAI & Heavy sector net fixed asset investment \\ 
NLightFAI & Light sector net fixed asset investment \\ 
PPI & Producer price index \\
LandPrice & Nominal land price \\ 
FAInvPrice & Fixed asset investment price \\ 
NominalGDPva & Nominal GDP value-added \\ 
RealGDPva & Real GDP value-added \\ 
GDPDeflator & GDP deflator \\ 
R7dRepo & 7-day repo rate \\ 
BankLoansTotal & Total bank loans outstanding \\ 
BankLoansST & Short-term bank loans outstanding \\ 
NGDPva\_Heavy & Nominal GDP value-added for the heavy industry \\ 
NGDPva\_Light & Nominal GDP value-added for the light industry \\ 
EntrustedLoans & Total entrusted lending  \\ 
TrustedLoans & Total trusted lending  \\ 
BankAccts & Bank acceptance bills  \\ 
ShadowBanking & Total lending in the shadow banking industry  \\ 
AggFinancing & Total aggregate social financing  \\ 
RealEstateDomesticLoanFAI & New loans to real estate sector  \\ 
HeavyIndustryDomesticLoanFAI & New loans to the heavy industry sector  \\ 
LightIndustryDomesticLoanFAI & New loans to the light industry sector\\ 
pop & Total population \\ 
CPriceExHousing & Consumer goods price, excluding housing investment \\ 
NonFinBusinessLoans & Bank loans outstanding to non-financial firms \\ 
ResidentialInvestment & Investment in residential sector \\ 
GFCFPrice & Price index for fixed gross capital formation \\ 
NonConstrEmp & Employment in the non-housing (non-construction) sector \\ 
ConstrEmp & Employment in the housing (construction) sector \\ 
NonConstrWage & Average urban wage in the non-housing (non-construction) sector \\ 
ConstrWage & Average urban wage in the housing (construction) sector \\ 
InvRETotal & Total real estate investment (seasonally adjusted) \\ 
FAIRETotal & Fixed asset investment in real estate (seasonally adjusted) \\ 
FAIRETotalBack & Backcasted fixed asset investment in real estate (using ratio extrapolation with InvRETotal) \\ 
NSTRGFCF & Gross fixed capital formation in structures (seasonally adjusted) \\ 
NRESSTRGFCF & Gross fixed capital formation in residential structures (seasonally adjusted) \\ 
NNONRESSTRGFCF & Gross fixed capital formation in non-residential structures (seasonally adjusted) \\ 
\bottomrule
\\[-2ex]
\caption{Quarterly Macro data of China.}
\label{macro variabe}
\end{longtable}
\end{appendices}

\end{document}